\documentclass[12pt]{article}

\usepackage[colorlinks, linkcolor=blue,citecolor=blue]{hyperref}           
\usepackage{color}

\usepackage{graphicx,subfigure,amsmath,amssymb,amsfonts,bm,epsfig,epsf,url,dsfont}
\usepackage{amsthm,mathrsfs}
\usepackage{tikz}

\usetikzlibrary{arrows.meta,positioning}

\usetikzlibrary{positioning, calc}
\usepackage{multirow}
\usepackage{bbm}      
\usepackage{booktabs}
\usepackage{cases}
\usepackage{fullpage,enumitem}
\usepackage[small,bf]{caption}
\usepackage{natbib}
\usepackage[top=1in,bottom=1in,left=1in,right=1in]{geometry}
\usepackage{fancybox}
\usepackage{hyperref}
\usepackage{algorithm}
\usepackage{verbatim}
\usepackage{algorithmic}
\usepackage{amsthm}
\usepackage{mathtools,pgfmath}

\usepackage{scalerel,stackengine}
\stackMath
\newcommand\reallywidehat[1]{%
\savestack{\tmpbox}{\stretchto{%
0  \scaleto{%
    \scalerel*[\widthof{\ensuremath{#1}}]{\kern-.6pt\bigwedge\kern-.6pt}%
    {\rule[-\textheight/2]{1ex}{\textheight}}%WIDTH-LIMITED BIG WEDGE
  }{\textheight}% 
}{0.5ex}}%
\stackon[1pt]{#1}{\tmpbox}%
}
\newenvironment{nscenter}
 {\parskip=0pt\par\nopagebreak\centering}
 {\par\noindent\ignorespacesafterend}
\newcommand{\bE}{\mathbb{E}}

\newtheorem{definition}{Definition}
\newtheorem{example}{Example}

\newtheorem{theorem}{Theorem}
\newtheorem{corollary}{Corollary}

\newtheorem{lemma}{Lemma}

\newenvironment{fminipage}%
  {\begin{Sbox}\begin{minipage}}%
  {\end{minipage}\end{Sbox}\fbox{\TheSbox}}

\makeatletter
\newcommand*{\rom}[1]{\expandafter\@slowromancap\romannumeral #1@}
\makeatother
\newtheorem{remark}{Remark}

\newcommand{\adv}{\mathsf{{Adv}}}
\newcommand{\Ind}{\mathbbm{1}}

\newcommand{\sd}{\mu}

\newcommand{\Jb}{\mathbf{J}}

\newcommand{\abs}[1]{\left|#1\right|}
\newcommand{\R}{\mathbb{R}} 
\newcommand{\N}{\mathbb{N}}

\newcommand{\E}{\mathbb{E}}

\def\P{{\mathbb P}}

\newcommand {\pr} {\mathbb{P}}

\newcommand{\calA}{{\cal A}}
\newcommand{\calB}{{\cal B}}
\newcommand{\calC}{{\cal C}}
\newcommand{\calD}{{\cal D}}
\newcommand{\calE}{{\cal E}}
\newcommand{\calF}{{\cal F}}
\newcommand{\calG}{{\cal G}}
\newcommand{\calH}{{\cal H}}

\newcommand{\calK}{{\cal K}}

\newcommand{\calM}{{\cal M}}
\newcommand{\calN}{{\cal N}}
\newcommand{\calO}{{\cal O}}
\newcommand{\calP}{{\cal P}}

\newcommand{\calR}{{\cal R}}
\newcommand{\calS}{{\cal S}}
\newcommand{\calT}{{\cal T}}

\newcommand{\calW}{{\cal W}}
\newcommand{\calX}{{\cal X}}

\DeclarePairedDelimiterX{\set}[1]{\{}{\}}{\setargs{#1}}
\DeclarePairedDelimiterX{\cond}[1]{[}{]}{\setargs{#1}}
\NewDocumentCommand{\setargs}{>{\SplitArgument{1}{;}}m}
{\setargsaux#1}
\NewDocumentCommand{\setargsaux}{mm}
{\IfNoValueTF{#2}{#1} {#1\,\delimsize|\,\mathopen{}#2}}%{#1\:;\:#2}

\newcommand{\be}{\begin{equation}}
\newcommand{\ee}{\end{equation}}
\newcommand{\beqna}{\begin{eqnarray}}
\newcommand{\eeqna}{\end{eqnarray}}

\newcommand{\p}[1]{\left(#1\right)}
\newcommand{\pp}[1]{\left[#1\right]}
\newcommand{\ppp}[1]{\left\{#1\right\}}
\newcommand{\norm}[1]{\left\|#1\right\|}
\newcommand{\innerP}[1]{\left\langle#1\right\rangle}

\usepackage{xspace}

\newcommand{\s}[1]{\mathsf{#1}}

\makeatletter
\def\thanks#1{\protected@xdef\@thanks{\@thanks
        \protect\footnotetext{#1}}}
\makeatother

\makeatletter
\renewcommand{\paragraph}{%
  \@startsection{paragraph}{4}%
  {\z@}{1.25ex \@plus 1ex \@minus .2ex}{-1em}%
  {\normalfont\normalsize\bfseries}%
}
\makeatother

\allowdisplaybreaks

\begin{document}
%\title{Recovering Arbitrary Planted Subgraphs in Random Graphs}
\title{Recovery of Planted Subgraphs}
\author{Wasim Huleihel\thanks{W. Huleihel is with the School of Electrical Engineering and Computer Engineering, at Tel Aviv University, {T}el {A}viv 6997801, Israel (e-mail:  \texttt{wasimh@tauex.tau.ac.il}).}}

\maketitle

\begin{abstract}

Understanding the fundamental limits of recovering planted subgraphs in random graphs is a central challenge in high-dimensional statistics and theoretical computer science. While existing work has largely focused on special subgraph families such as cliques, bicliques, or dense blocks, the exact recovery of a general planted subgraph in Erd\H{o}s--R\'enyi random graphs remains poorly understood. In this paper, we study the exact recovery of an arbitrary planted subgraph $\Gamma = \Gamma_n$ embedded in a dense Erd\H{o}s--R\'enyi random graph $\mathcal{G}(n,q_n)$, where edges within $\Gamma$ are present independently with probability $p_n > q_n$.

Our main results identify sharp conditions under which exact recovery is possible with high probability, and we establish matching lower bounds showing the necessity of these conditions. The resulting statistical threshold is characterized by a new graph-theoretic quantity, which we term the \emph{minimal maximum subgraph density}. This quantity is defined as the maximum subgraph density of the smallest induced balanced subgraph of $\Gamma$.

We then turn to the problem of recovery under polynomial-time constraints. We propose a computationally efficient recovery algorithm that applies to arbitrary planted subgraphs and analyze its performance in terms of certain spectral properties of the adjacency matrix. In addition, we derive computational lower bounds for recovery using the low-degree polynomial framework, establishing regimes where recovery is statistically possible but computationally hard. Finally, we consider several extensions of our setting, including recovery in semi-random models and weaker notions of recovery.

\end{abstract}

\newpage
\tableofcontents

%\newpage

\section{Introduction}\label{sec:intro}

The study of structured signals in networks lies at the intersection of graph theory, computer science, and statistics, with applications ranging from social networks to computational biology. A central question in this area is whether one can reliably reconstruct hidden or anomalous structures embedded in otherwise random graphs. Much of the existing literature has focused on identifying communities or clusters of vertices with unusually high internal connectivity. While early work emphasized \emph{detection}, namely determining whether such a structure is present, an equally fundamental and arguably more challenging task is \emph{recovery}: given an observed network, can one pinpoint the precise location of the hidden structure? This recovery perspective is crucial in applications such as anomaly localization, motif discovery in biological networks, and targeted monitoring in social networks, and it raises new theoretical questions regarding the statistical and computational limits of inference.

As in many inference problems, recovery exhibits both statistical and computational facets. The statistical aspect concerns the feasibility of recovery given unlimited computational power, whereas the computational aspect asks whether recovery can be achieved efficiently. Traditionally, these aspects were studied separately, with information-theoretic tools providing fundamental limits. However, recent work has highlighted the critical role of computational constraints in high-dimensional inference. A growing body of research (see, e.g., \cite{berthet2013complexity,ma2015computational,cai2015computational,chen2016statistical,hopkins2017bayesian,Hopkins18,gamarnik2020lowdegree,barak2016nearly,Lenka16,Lesieur_2015,hajek2015computational,bandeira2018notes,brennan18a,brennan19,brennan20a}, among many others) has identified striking statistical--computational gaps in planted combinatorial problems, that is, regimes where recovery is information-theoretically possible but no known polynomial-time algorithm succeeds.

In this paper, we investigate the problem of recovering an \emph{arbitrary} subgraph planted in an Erd\H{o}s--R\'enyi random graph. Formally, let $n \in \mathbb{N}$, let $p_n \in (0,1)$ and $q_n \in (0,1)$ satisfy $q_n < p_n \leq 1$, and let $\Gamma = \Gamma_n$ denote an arbitrary sequence of undirected graphs, referred to as the planted subgraph. We observe a graph $\s{G}$ generated as follows: first, a copy $\Gamma_n^\star$ of $\Gamma_n$ is chosen uniformly at random among all embeddings in the complete graph on $n$ vertices; edges within $\Gamma_n^\star$ are then included independently with probability $p_n$, while all remaining edges are included independently with probability $q_n$. The inferential task is to recover the precise location of $\Gamma_n^\star$, namely to achieve \emph{exact recovery}, with high probability.

Over the past several decades, numerous special cases of planted subgraph recovery have been studied. The canonical example is the planted clique problem \cite{jerrum1992large}, where the goal is to recover a hidden $k$-clique embedded in $\calG(n,1/2)$. Other notable examples include the planted dense subgraph problem \cite{arias2014community,arias2015detecting,hajek2015computational}, the recovery of planted trees \cite{massoulie19a}, Hamiltonian cycles \cite{Bagaria20}, matchings \cite{10.1214/20-AAP1660}, and bipartite structures \cite{HuleihelBip}. Despite shared methodological features, the statistical and computational behaviors vary significantly across different planted structures. Some subgraphs exhibit sharp recovery thresholds and conjectured computationally hard regimes, as in the planted clique problem, whereas others, such as paths or stars, appear not to exhibit any computational barrier \cite{massoulie19a}.

There has been recent progress toward unified frameworks for planted subgraph inference. For the \emph{detection problem}, where the goal is to distinguish between a pure Erd\H{o}s--R\'enyi random graph and the planted model described above, \cite{addario2010combinatorial,Huleihel2022,pmlr-v247-yu24a,elimelech2025detecting,ElimelechHuleihel2025SemiRandom} studied general detection models for arbitrary planted subgraphs in various settings and established both statistical and computational thresholds, primarily in the dense regime. Collectively, these results provide a nearly complete understanding of detection: when it is statistically possible and when it can be achieved efficiently.

By contrast, the landscape for the \emph{recovery problem} remains far less understood. Important progress was made in \cite{pmlr-v195-mossel23a,LeePerniceRajaramanZadik2025}, which studied weak recovery for arbitrary planted subgraphs and derived a general variational formula for the limiting minimum mean squared error (MMSE) under mild structural density assumptions. These results yield a powerful statistical characterization of recovery in the weak sense. Nevertheless, several key questions remain open. First, the focus there is on approximate or fractional recovery, as captured by small MMSE, rather than on exact recovery or other natural notions of recovery. Different recovery criteria can lead to qualitatively different thresholds, governed by fundamentally different graph-theoretic measures. Second, while these works characterize statistical limits, they do not address computational tractability: general computational lower and upper bounds for recovery are still missing. As a result, it remains unclear when a statistical--computational gap arises and how efficient recovery can be achieved.

In this paper, we take a step toward addressing these challenges by developing a general framework for recovery in the arbitrary planted subgraph setting. Our approach simultaneously addresses statistical and computational aspects and accommodates multiple notions of recovery. Motivated by the gaps described above, we pose the following guiding questions:

\vspace{0.2cm}
\centerline{\noindent\fbox{\parbox{0.9\textwidth}{
\begin{nscenter}
\emph{What graph-theoretic properties of $\Gamma$ govern the statistical and computational limits of recovery?}

  \emph{For which planted structures does a statistical–computational gap emerge?}  

\end{nscenter}
}}}
\vspace{0.2cm}

\subsection{Main contributions}

In this paper, we aim to answer the questions posed above. We begin with what is arguably the most canonical setting of the problem: the dense regime, where the edge probabilities $(p_n,q_n)$ are fixed constants. Our first objective is to characterize the statistical (information-theoretic) limits of exact recovery of the planted subgraph $\Gamma = \Gamma_n$, as defined above (see Section~\ref{sec:problem} for formal details).

A key ingredient in our analysis is the \emph{onion decomposition} of $\Gamma = \Gamma_n$, a concept recently introduced in \cite{LeePerniceRajaramanZadik2025}. Informally, this decomposition iteratively peels off subgraphs achieving the maximum subgraph density of $\Gamma$ (see Definition~\ref{def:maxsubden}), layer by layer, until the entire graph is decomposed into the union of such layers. A precise definition is given in Definition~\ref{def:onionDec}. Our results are expressed in terms of a graph-theoretic quantity that we call the \emph{minimal maximum subgraph density}, defined as
\begin{align*}
\mu_{\s{min}}(\Gamma_n) \triangleq \min_{\s{S}\subseteq \Gamma_n}\max_{\s{S}\subsetneq\s{F}\subseteq\Gamma_n}\eta(\s{F}\vert\s{S}).
\end{align*}
Here, $\eta(\cdot\vert\cdot)$ denotes the relative density, as defined in Definition~\ref{def:relative}. This quantity admits a natural operational interpretation: it coincides with the maximum subgraph density of the final, and hence minimal, layer in the onion decomposition of $\Gamma$.

Our first main result shows that the statistical threshold for exact recovery is governed by
\centerline{\noindent\fbox{\parbox[t][1cm][c]{0.4\textwidth}{
\begin{align*}
\mu_{\s{min}}(\Gamma_n)\mathop{\gtrless}\limits_{\s{Impossible}}^{\s{Possible}}\s{C}\frac{\log n}{d_{\s{KL}}(p_n||q_n)}\,,
\end{align*}}}}
\vspace{0.2cm}

\noindent where $d_{\s{KL}}(p_n||q_n)$ denotes the Kullback--Leibler divergence between $\s{Bern}(p_n)$ and $\s{Bern}(q_n)$, and $\s{C}\geq 1$ is a universal constant. In particular, if the minimal maximum subgraph density of $\Gamma$ exceeds the right-hand side, then exact recovery is achievable with high probability; otherwise, exact recovery is information-theoretically impossible. The achievability result is obtained by applying a maximum-likelihood estimator (MLE) layer by layer along the onion decomposition, while the converse follows from a genie-aided argument combined with Fano's inequality. Together, these results fully characterize the statistical limits of exact planted subgraph recovery in the dense regime.

We next turn to the computational aspect of the problem. We propose a general, computationally efficient recovery algorithm and analyze its performance guarantees. This method can be viewed as a semidefinite relaxation of the optimal MLE. We show that it succeeds whenever certain spectral properties of an appropriate transformation of the adjacency matrix of $\Gamma_n$, specifically its rank and coherence number, satisfy suitable conditions. Notably, the resulting algorithm and its guarantees recover, as special cases, the best-known polynomial-time procedures for classical planted subgraph problems, including cliques, bipartite graphs, and related models. For certain choices of $\Gamma$, however, there remains a gap between the information-theoretic limits established above and the performance of this efficient algorithm. We conjecture that this gap is intrinsic, in the sense that below the corresponding computational threshold, no polynomial-time algorithm can achieve recovery.

To provide evidence for this conjecture, we adopt the low-degree polynomial framework of \cite{SchrammWein2022} and establish computational lower bounds for recovery. Roughly speaking, we show that all polynomials of degree $O(\log n)$ fail to recover $\Gamma$ (even in a weak sense) whenever its density satisfies

\centerline{\noindent\fbox{\parbox[t][1cm][c]{0.2\textwidth}{
\begin{align*}
\eta(\Gamma_n)\triangleq\frac{|e(\Gamma_n)|}{|v(\Gamma_n)|}\ll \sqrt{n}\,.
\end{align*}}}}
\vspace{0.2cm}

\noindent For instance, when $\Gamma_n$ is a clique, this condition predicts computational hardness whenever $|v(\Gamma_n)| \ll \sqrt{n}$. Similarly, for a complete bipartite graph with $k_{\s{L}}$ left and $k_{\s{R}}$ right vertices, recovery is conjectured to be computationally hard when $\min\{k_{\s{L}},k_{\s{R}}\} \ll \sqrt{n}$. These predictions align with well-known folklore conjectures. We also establish complementary upper bounds, showing that low-degree polynomials succeed whenever $\eta(\Gamma_n)\gg \sqrt{n}$ for a broad class of nearly regular subgraphs, including cliques, balanced bipartite graphs, and sparse expander graphs. In these cases, the low-degree predictions match the performance of the best known polynomial-time algorithms.

Finally, we consider several extensions of the baseline planted subgraph model. The standard formulation assumes a purely random generative process, which may be fragile under adversarial perturbations. To address this issue, we study semi-random models in which an adversary may delete edges outside the planted subgraph and add edges inside it prior to observation. Importantly, the statistician does not know which edges have been modified. For this setting, known as the monotone adversary model \cite{feige2000finding,feige2001heuristics}, we show that both the optimal MLE and the efficient algorithm proposed above remain robust and achieve the same recovery guarantees as in the non-adversarial model.

Motivated by the observation that exact recovery is fundamentally impossible for planted subgraphs consisting of a dense core with a very sparse appendage, such as kite graphs, since their minimal maximum subgraph density $\mu_{\s{min}}$ is sub-logarithmic in the number of vertices, we also study weaker notions of recovery. Specifically, we consider exact layer recovery, where the goal is to recover a fixed number of initial layers in the onion decomposition, as well as almost-exact recovery \cite{hajek2016information,wu2018statistical}, where the estimator is required to be close to the planted subgraph in Hamming distance. For both recovery criteria, we establish asymptotically tight thresholds.

\subsection{Related work}

This paper lies within a broad literature on planted combinatorial structure in random graphs and matrices, studied from both statistical and computational perspectives. We highlight work most pertinent to recovery; for a fuller survey, see \cite{elimelech2025detecting}.

\paragraph{Planted subgraphs and matrices.} One of the earliest and most influential recovery problems is the planted clique. Alon, Krivelevich, and Sudakov \cite{alon1998finding} showed that spectral methods can identify a clique of size $k = \Omega(\sqrt{n})$. Since then, a variety of algorithmic approaches have been studied—including combinatorial heuristics \cite{feige2000finding, mcsherry2001spectral, feige2010finding}, convex relaxations such as semidefinite programming and nuclear-norm minimization \cite{ames2011nuclear, deshpande2015finding}, and approximate message passing \cite{chen2016statistical}. Despite these advances, all known polynomial-time algorithms require $k = \Omega(\sqrt{n})$, giving rise to the widely accepted planted clique conjecture, which asserts that recovery is computationally intractable when $k = o(\sqrt{n})$. Beyond cliques, researchers have explored other planted subgraph models:
\begin{itemize}[leftmargin=*]
    \item \emph{Independent sets.} Feige and Krauthgamer \cite{feige2005finding} proposed a spectral method for recovering planted independent sets, while Coja-Oghlan \cite{coja2003finding} established polynomial-time recovery guarantees in sparse regimes where $q = \Theta(n^{-\alpha})$, provided the independence number scales appropriately. Additional work has analyzed the limits of greedy and local algorithms in these settings \cite{coja2015independent, gamarnik2014limits, rahman2017local}.
    \item \emph{Dense subgraphs and communities.} The problem of recovering a dense community embedded in an Erd\H{o}s--R\'enyi background has attracted extensive attention \cite{arias2014community, butucea2013detection, verzelen2015community, Hajek2015, montanari2015finding, candogan2018finding, hajek2016information, chen2016statistical}. A key milestone was the reduction from planted clique to planted dense subgraph, due to Hajek, Wu, and Xu \cite{hajek2015computational}, which established hardness in regimes where $p = cq$ and $q = \Theta(n^{-\alpha})$. Brennan et al. \cite{brennan18a} later extended this reduction to nearly the full range of $p > q$, with transitions to denser regimes analyzed in \cite{bhaskara2010detecting}.
    \item \emph{Other planted structures.} A variety of recovery problems have been studied for specific combinatorial templates, including: planted trees \cite{massoulie19a}, Hamiltonian cycles \cite{Bagaria20}, perfect matchings \cite{10.1214/20-AAP1660}, bipartite subgraphs \cite{HuleihelBip}, and cycles \cite{ChengWein, pmlr-v195-mao23a}. These studies demonstrate that the algorithmic and statistical behavior can differ drastically depending on the underlying subgraph: some exhibit sharp thresholds and conjectured computational barriers, while others allow efficient algorithms down to information-theoretic limits.
    \item \emph{Matrices and Gaussian models.} Beyond graphs, analogous recovery problems appear in high-dimensional statistics. A prime example is Gaussian biclustering, where one seeks to identify a planted submatrix with elevated mean. Detection aspects were analyzed in \cite{butucea2013detection, ma2015computational, montanari2015limitation}, while recovery guarantees were developed in \cite{shabalin2009finding, kolar2011minimax, balakrishnan2011statistical, cai2015computational, chen2016statistical, hajek2016information, brennan19, 10447320, 10387722}. Closely related are spectral analyses of the spiked Wigner model, beginning with \cite{peche2006largest, feral2007largest, capitaine2009largest}, and later work on spectral thresholds and Bayesian algorithms \cite{montanari2015limitation, perry2016statistical, perry2016optimality, banks2018information, hopkins2017power}.
\end{itemize}
Together, this body of research illustrates the diversity of recovery phenomena: in some cases (e.g., planted cliques and dense subgraphs), computational-statistical gaps appear central, while in others (e.g., certain trees or paths), efficient recovery is achievable down to statistical thresholds.

\paragraph{General planting models.} Research has long studied planting arbitrary substructures in random graphs, exploring both detection and recovery under varied generative models. Early works like \cite{addario2010combinatorial} developed general probabilistic bounds for testing in the planted-vector setting—though their results were broad and not always tight across all cases. More recent lines of inquiry have introduced two foundational paradigms:
\begin{enumerate}[leftmargin=*]
    \item \emph{Detection.} Several works have developed frameworks for the detection of arbitrary planted subgraphs. The work \cite{Huleihel2022} introduced two models for planting arbitrary subgraphs in random graphs: the union model (where the planted edges are superimposed on an Erd\H{o}s--R\'enyi background) and the induced model (where the planted subgraph appears as an induced copy). They analyzed detection thresholds in the dense regime, providing both information-theoretic bounds and computational insights. Recently, \cite{pmlr-v247-yu24a} studied the computational limits of detection in the dense regime and showed that optimal constant-degree polynomial tests are always given by counting stars. This result highlights the limitations of low-degree polynomials for detection of arbitrary planted subgraphs. Recently, \cite{elimelech2025detecting} extended this line of work by analyzing detection in the union model more generally. Their results provide sharp characterizations across both sparse and dense regimes, demonstrating how the feasibility of detection depends delicately on graph density and subgraph structure. Together, these works clarify the statistical and algorithmic landscape for detection in general planted subgraph models, identifying both tractable and hard regimes. Finally, \cite{ElimelechHuleihel2025SemiRandom} established fundamental statistical limits for detecting arbitrary planted subgraph under a semi-random model where an adversary is allowed to remove edges outside the planted subgraph before the graph is provided to the statistician; the goal is to derive robust detection algorithms. 
    \item \emph{Recovery.} The problem of recovery, where the goal is to reconstruct the precise location of the planted subgraph, has also been addressed in several recent studies. In \cite{Huleihel2022} considered recovery in the induced model, deriving bounds in the dense regime. While their focus was primarily on detection, they established statistical thresholds for recovery as well. In \cite{pmlr-v195-mossel23a} investigated recovery for arbitrary planted subgraphs in the dense regime. They provided tight upper and lower bounds for recovery in this setting, focusing on specific families of subgraphs, and highlighted where statistical–computational gaps arise. Finally, \cite{LeePerniceRajaramanZadik2025} advanced this direction by deriving an exact formula for the asymptotic MMSE curve for recovering arbitrary planted subgraphs. They also proposed an efficient algorithm for approximating this threshold in dense graphs, based on a novel decomposition technique, extending “all-or-nothing” phenomena to general subgraphs. Taken together, these works establish the first general frameworks for recovery in arbitrary planted subgraph models, with \cite{pmlr-v195-mossel23a} and \cite{LeePerniceRajaramanZadik2025} providing sharp results in the dense setting, and \cite{Huleihel2022} bridging induced subgraphs and recovery.
\end{enumerate}

\paragraph{Computational hardness.} Over the past decade, major progress has been made toward a rigorous understanding of the fundamental limits of efficient algorithms for high-dimensional inference problems with planted structure. A recurring theme in this line of work is the emergence of a statistical–computational gap: the number of samples required by any known polynomial-time algorithm is strictly larger than the information-theoretic minimum \cite{berthet2013complexity,ma2015computational,cai2015computational,krauthgamer2015semidefinite,hajek2015computational,chen2016statistical,wang2016average,wang2016statistical,gao2017sparse,brennan18a,brennan19,wu2018statistical,brennan20a,hopkins2017bayesian,Hopkins18,Kunisky19,Cherapanamjeri20,gamarnik2020lowdegree,barak2016nearly,deshpande2015improved,meka2015sum,TengyuWig15,kothari2017sum,hopkins2016integrality,raghavendra2019highdimensional,hopkins2017power,mohanty2019lifting,Feldman17,feldman2018complexity,Diakonikolas17,DiakonikolasKong19,Lenka16,Lesieur_2015,Lesieur_2016,Krzakala10318,Ricci_Tersenghi_2019,bandeira2018notes,SchrammWein2022,pmlr-v134-brennan21a,Wein2021Independent,Abhishek23,elimelech2025detecting}. The evidence for these gaps typically falls into two broad categories:
\begin{enumerate}[leftmargin=*]
    \item \emph{Failure of classes of algorithms.} One line of evidence comes from showing that broad algorithmic paradigms fail in the conjectured hard regime. For instance, low-degree polynomials provide a unifying lens for analyzing high-dimensional inference, and their failure below certain thresholds suggests sharp computational barriers \cite{hopkins2017bayesian,Hopkins18,Kunisky19,Cherapanamjeri20,gamarnik2020lowdegree}. Similarly, the sum-of-squares hierarchy, despite its power as a family of semidefinite relaxations, has been shown to fall short in planted clique, planted dense subgraph, and related problems \cite{barak2016nearly,deshpande2015improved,meka2015sum,TengyuWig15,kothari2017sum,hopkins2016integrality,raghavendra2019highdimensional,hopkins2017power,mohanty2019lifting}. The statistical query model, which captures a wide class of algorithms accessing only expectations of data-dependent functions, has also yielded strong lower bounds in this context \cite{Feldman17,feldman2018complexity,Diakonikolas17,DiakonikolasKong19,pmlr-v134-brennan21a}. Finally, message-passing algorithms such as belief propagation and approximate message passing often achieve optimal performance in ``easy" regimes but exhibit provable failures in the hard phase across multiple planted models \cite{Lenka16,Lesieur_2015,Lesieur_2016,Krzakala10318,Ricci_Tersenghi_2019,bandeira2018notes}.
    \item \emph{Average-case reductions.} Another powerful approach is to establish hardness by reducing one planted problem to another conjectured-to-be-hard task, most prominently planted clique. This technique has been used to show that recovery (or even detection) in models such as planted dense subgraph or biclustering is at least as hard as planted clique in certain parameter regimes \cite{berthet2013complexity,ma2015computational,cai2015computational,chen2016statistical,hajek2015computational,wang2016average,wang2016statistical,gao2017sparse,brennan18a,brennan19,wu2018statistical,brennan20a}.
\end{enumerate}

\paragraph{Semi-random models.} The notion of robustness via semi-random perturbations dates back to Blum and Spencer \cite{blum1995coloring}, with stronger adversarial variants introduced in the semi-random planted clique model \cite{feige2000finding,feige2001heuristics}. These ideas have since been generalized to community detection in the SBM \cite{charikar2017learning,moitra2016robustCD,banks2021local,bhaskara2024robustness}, and to recovery of cliques \cite{steinhardt2017does,mckenzie2020new,buhai2023algorithms,blasiok2024semirandom,guruswami2025semirandom}, bipartite graphs \cite{Kumar22}, and $r$-colorable graphs \cite{louis2025robust}. While convex relaxations such as SDP often remain robust, spectral and other classical methods can fail under semi-random noise. In the planted dense subgraph model, Brennan et al. \cite{brennan20a} conjectured that below certain density thresholds even weak recovery is computationally hard with adversarial edge deletions, supporting their claim via reductions from the planted clique with side information.

\subsection{Notation}

For an integer $n\in\N$, we write $[n]\triangleq\{1,\dots,n\}$ and define $n^{(2)}\triangleq\binom{n}{2}$. For $i\in\N$, the collection of all subsets of $[n]$ of size $i$ is denoted by $\binom{[n]}{i}$. For real numbers $a,b\in\R$, we write $a\vee b$ for their maximum. We use the standard asymptotic notations $O(\cdot)$, $o(\cdot)$, $\Omega(\cdot)$, and $\omega(\cdot)$ to describe growth rates of sequences, and write $a_n\ll b_n$ to indicate that $a_n$ is polynomially smaller than $b_n$, i.e., $\limsup_{n\to\infty}\log_n a_n < \liminf_{n\to\infty}\log_n b_n$.

For matrices and vectors, we denote by $\mathbf{J}_{m\times n}$ the all-ones $m\times n$ matrix, and by $\mathbf{1}_m$ and $\mathbf{0}_m$ the all-ones and all-zeros vectors in $\R^m$, respectively; subscripts are omitted when dimensions are clear from the context. For square matrices $\s{A}$ and $\s{B}$ of the same size, $\innerP{\s{A},\s{B}}$ denotes the Hilbert--Schmidt inner product, defined as $\s{Tr}(\s{A}^\top\s{B})$. We write $\|\s{A}\|_\star$ for the nuclear (trace, Schatten-$1$) norm of $\s{A}$, given by the sum of its singular values, and $\|\s{A}\|_F$ for its Frobenius norm. Additional matrix norms used throughout include the spectral norm $\|\s{X}\|_{\s{op}}$, the entrywise $\ell_1$ norm $\|\s{X}\|_{\ell_1}=\sum_{i,j}|\s{X}_{ij}|$, the entrywise $\ell_\infty$ norm $\|\s{X}\|_{\ell_\infty}=\max_{i,j}|\s{X}_{ij}|$, the maximum absolute row sum $\|\s{X}\|_{\ell_\infty\to\ell_\infty}\triangleq\max_i\sum_j|\s{X}_{ij}|$, and the maximum row $\ell_2$ norm $\|\s{X}\|_{2,\infty}\triangleq\max_{i\in[n]}\|\s{X}_{i,:}\|_2$.

For probability measures $\P_0$ and $\P_1$ on the same measurable space, we use the total variation distance
$d_{\s{TV}}(\P_0,\P_1)=\frac{1}{2}\int|\mathrm{d}\P_0-\mathrm{d}\P_1|$, the $\chi^2$-divergence $\chi^2(\P_0||\P_1)=\int \frac{(\mathrm{d}\P_0-\mathrm{d}\P_1)^2}{\mathrm{d}\P_1}$, and the Kullback--Leibler (KL) divergence $d_{\s{KL}}(\P_0||\P_1)=\bE_{\P_0}\log\frac{\mathrm{d}\P_0}{\mathrm{d}\P_1}$. 
When $\P_0=\s{Bern}(p)$ and $\P_1=\s{Bern}(q)$, we write $\chi^2(p||q)=\frac{(p-q)^2}{q(1-q)}$ and $d_{\s{KL}}(p||q)=p\log\frac{p}{q}+(1-p)\log\frac{1-p}{1-q}$. For a finite set $S$, $\s{Unif}(S)$ denotes the uniform distribution over $S$.

We use standard graph-theoretic notation. An undirected graph is denoted by $\s{G}=(v(\s{G}),e(\s{G}))$, with vertex set $v(\s{G})$ and edge set $e(\s{G})$. Since we primarily consider graphs without isolated vertices, for a subgraph $\s{H}\subseteq\s{G}$ we write $|\s{H}|$ to denote the number of edges in $e(\s{H})$. If $|v(\s{G})|\le n$, we let $\calS_{\s{G}}$ denote the set of all isomorphic copies of $\s{G}$ in the complete graph $\calK_n$. That is, $\s{H}\subseteq\calK_n$ belongs to $\calS_{\s{G}}$ if there exists a bijection $f:v(\s{G})\to v(\s{H})$ such that $(u,v)\in e(\s{G})$ if and only if $(f(u),f(v))\in e(\s{H})$. Given $\Gamma\subseteq\calK_n$ and a subgraph $\s{H}\subseteq\Gamma$, we define $\calN(\s{H},\Gamma)$ as the set of copies of $\s{H}$ in $\Gamma$, and $\calM(\s{H},\Gamma)$ as the set of copies of $\Gamma$ in $\calK_n$ that contain $\s{H}$. For any graph $\s{G}$, $\s{A}_{\s{G}}$ denotes its adjacency matrix, indexed by $v(\s{G})$, where $\s{A}_{\s{G}}(i,j)=\Ind\{\{i,j\}\in e(\s{G})\}$.

The rest of this paper is organized as follows. In Section~\ref{sec:problem}, we introduce the problem setup and provide some necessary preliminaries. Section~\ref{sec:mainresults} presents our main results, discussions, and examples. Sections~\ref{sec:lowerbounds}--\ref{sec:complowerbounds} are devoted to the the derivation of our lower and upper bounds. Finally, in Section~\ref{sec:Out} we conclude our paper, and discuss a few directions for future research.

\section{Problem Setup and Preliminaries}\label{sec:problem}

In this section, we describe the setting we study, along with several important preliminaries. Let $\Gamma = (\Gamma_n)_{n \in \mathbb{N}}$ be a sequence of graphs such that, for each $n \in \mathbb{N}$, $\Gamma_n = (v(\Gamma_n), e(\Gamma_n))$ is an undirected graph without isolated vertices and with $|v(\Gamma_n)| \leq n$. Let $\mathcal{S}_{\Gamma_n}$ denote the set of all isomorphic copies of $\Gamma_n$ in the complete graph on $n$ vertices. We refer to $\Gamma_n$ as the \emph{planted} (or \emph{hidden}) structure. Fix parameters $p_n,q_n$ satisfying $0 < q_n < p_n \leq 1$. The \emph{planted subgraph model} $\mathcal{G}_{\Gamma_n}(n,p_n,q_n)$ is defined as the distribution of a random graph $\mathsf{G}$ on $n$ vertices generated as follows: first draw an arbitrary but fixed copy $\Gamma_n^\star \in \mathcal{S}_{\Gamma_n}$; then include each edge $e \in e(\Gamma_n^\star)$ independently with probability $p_n$, and include each edge $e \notin e(\Gamma_n^\star)$ independently with probability $q_n$. Equivalently, $\mathsf{G}$ can be viewed as the union of a noisy copy of $\Gamma_n^\star$ and an Erd\H{o}s--R\'enyi random graph $\mathcal{G}(n,q_n)$.

A learner observes a single sample $\mathsf{G} \sim \mathcal{G}_{\Gamma_n}(n,p_n,q_n)$, and the goal is to recover the hidden copy $\Gamma_n^\star$. We study this framework in the asymptotic regime where $n \to \infty$. Figure~\ref{fig:recovery-figs} illustrates typical graph observations. Given $\mathsf{G}$, an estimator $\widehat\Gamma:\{0,1\}^{\binom{n}{2}}\to\mathcal S_{\Gamma_n}$ aims to output $\Gamma_n^\star$.
Define the worst-case error probability associated with an estimator $\widehat{\Gamma}$ as
\begin{align}
    \s{E}_n(\hat{\Gamma})\triangleq\sup_{\Gamma^\star\in\mathcal{S}_{\Gamma_n}}\pr_{\mathcal{G}_{\Gamma_n^\star}(n, p_n, q_n)}[\hat{\Gamma}(\s{G})\neq\Gamma^\star],
\end{align}
and the optimal error probability as
\begin{align}
    \s{E}_n^\star \triangleq \inf_{\hat{\Gamma}:\{0,1\}^{n^{(2)}} \to \calS_{\Gamma_n}} \sup_{\Gamma^\star \in \calS_\Gamma} \pr_{\mathcal{G}_{\Gamma_n^\star}(n, p_n, q_n)}[\hat{\Gamma}(\s{G})\neq\Gamma^\star].
\end{align}
We say that a sequence of estimators $(\widehat{\Gamma}_n)_{n \in \mathbb{N}}$, where $\widehat{\Gamma}_n : \{0,1\}^{\binom{n}{2}} \to \mathcal{S}_{\Gamma_n}$, achieves \emph{exact recovery} if $\limsup_{n\to\infty}\s{E}_n(\widehat{\Gamma}_n)= 0$. Conversely, we say that \emph{exact recovery is impossible} if $\liminf_{n \to \infty} \s{E}_n^\star>0$.

\begin{remark}
In the above, we focused on a worst-case definition of error probability. However, as we show in Appendix~\ref{app:EquivBayesWorst}, the Bayesian (average-case) definition, where an expectation over $\Gamma^\star \sim \s{Unif}(\mathcal{S}_{\Gamma_n})$ is taken instead of a supremum, is equivalent. Indeed, the uniform measure over $\mathcal{S}_{\Gamma_n}$ induces a permutation-invariant statistical model for which the least favorable prior is also uniform.
\end{remark}

\begin{figure}[t]
\centering

%%% LEFT: BIPARTITE %%%
\begin{minipage}[t]{0.48\textwidth}
\centering
\scalebox{0.3}{\begin{tikzpicture}
  [blue node/.style={circle, draw=blue!50, fill=blue!20, thick, minimum size=5mm},
   red node/.style={circle, draw=red!50, fill=red!20, thick, minimum size=3mm},
   red edge/.style={draw=red!70, thick},
   blue edge/.style={draw=blue!70, thick}]

  % Define two blue partitions A and B (for bipartite)
  \foreach \i in {1,...,7} {
    \node[blue node] (a\i) at ({-2 + rand}, {2 - \i}) {};
  }

  \foreach \i in {1,...,8} {
    \node[blue node] (b\i) at ({2 - rand}, {2 - \i}) {};
  }

  % Define red nodes (background ER nodes)
  \foreach \i in {1,...,35} {
    \node[red node] (r\i) at ({rand * 10 - 5}, {rand * 10 - 5}) {};
  }

  % Connect blue nodes as a complete bipartite graph between sets A and B
  \foreach \i in {1,...,7} {
    \foreach \j in {1,...,8} {
      \draw[blue edge] (a\i) -- (b\j);
    }
  }

  % Randomly connect red nodes (~100 edges)
  \foreach \i in {1,...,35} {
  \foreach \j in {1,...,35} {
    \ifnum\i<\j
      \pgfmathsetmacro\randval{rnd}
      \ifdim\randval pt < 0.08pt
        \draw[red edge] (r\i) -- (r\j);
      \fi
    \fi
  }
}

  % Randomly connect red nodes to blue nodes (~50 edges)
  \foreach \i in {1,...,35} {
    \foreach \j in {1,...,7} {
      \pgfmathsetmacro\randval{rnd}
      \ifdim\randval pt < 0.08pt
        \draw[red edge] (r\i) -- (a\j);
      \fi
    }
    \foreach \j in {1,...,8} {
      \pgfmathsetmacro\randval{rnd}
      \ifdim\randval pt < 0.08pt
        \draw[red edge] (r\i) -- (b\j);
      \fi
    }
  }

\end{tikzpicture}}
\caption*{(a) Planted $\Gamma^\star_n=\calK_{k_{\s{L}},k_{\s{R}}}$.}
\end{minipage}~
\begin{minipage}[t]{0.48\textwidth}
\centering
\scalebox{0.3}{\begin{tikzpicture}
  [blue node/.style={circle, draw=blue!50, fill=blue!20, thick, minimum size=5mm},
   red node/.style={circle, draw=red!50, fill=red!20, thick, minimum size=3mm},
   red edge/.style={draw=red!70, thick},
   blue edge/.style={draw=blue!70, thick}]

  % Define blue nodes (forming a path inside the graph)
  \foreach \i in {1,...,15} {
    \node[blue node] (b\i) at ({rand * 6 - 3}, {rand * 6 - 3}) {}; % More central placement
  }

  % Define red nodes (spread randomly around the blue nodes)
  \foreach \i in {1,...,35} {
    \node[red node] (r\i) at ({rand * 10 - 5}, {rand * 10 - 5}) {};
  }

  % Connect blue nodes in a path
  \foreach \i in {1,...,14} {
    \pgfmathtruncatemacro{\next}{\i + 1}  % Compute next node correctly
    \draw[blue edge] (b\i) -- (b\next);
  }

  % Randomly connect red nodes (~100 edges)
  \foreach \i in {1,...,35} {
  \foreach \j in {1,...,35} {
    \ifnum\i<\j
      \pgfmathsetmacro\randval{rnd}
      \ifdim\randval pt < 0.08pt
        \draw[red edge] (r\i) -- (r\j);
      \fi
    \fi
  }
}

  % Randomly connect red nodes to blue nodes (~50 edges)
  \foreach \i in {1,...,35} {
    \foreach \j in {1,...,15} {
      \pgfmathsetmacro\randval{rnd}  % Generate random probability
      \ifdim\randval pt < 0.12pt  % Adjust probability for ~50 edges
        \draw[red edge] (r\i) -- (b\j);
      \fi
    }
  }

\end{tikzpicture}}
\caption*{(b) Planted $\Gamma^\star_n=\calP_{k}$}
\end{minipage}

%%% Main Caption
\caption{The observed graph $\s{G}$ is the union of an Erd\H{o}s--R\'enyi graph and a planted subgraph $\Gamma_n^\star$. In (a) $\Gamma_n^\star$ is a bipartite subgraph, and in (b) $\Gamma_n^\star$ is a path.}
\label{fig:recovery-figs}
\end{figure}

Our results will be expressed in terms of the following graph-theoretic measures. We let $\eta\left(\Gamma_n\right)\triangleq|e(\Gamma_n)|/|v(\Gamma_n)|$ denote the density of $\Gamma_n$, and we recall the following definition of the \emph{maximum subgraph density}.
\begin{definition}[Maximum subgraph density \cite{bollobas_2001}]\label{def:maxsubden}
Let $\s{G}$ be an undirected graph. The maximum subgraph density of $\s{G}$ is
\begin{align}
\sd(\s{G})\triangleq\max\ppp{\eta(\s{H}):\s{H}\subseteq\s{G},\s{H}\neq\emptyset}.\label{eqn:maxDensity}
\end{align}
\end{definition}
Next, we introduce the notion of a graph-cut.
\begin{definition}[Graph-cut]\label{def:gcut}
Let $n$ be a positive integer. A graph-cut on $n$ vertices is a triplet $\s{H} = (V,S,E)$, where $S \subseteq V \subseteq [n]$, and
\begin{align}
E \subseteq \calK_V \setminus \calK_S \triangleq \left\{(u,v) : u,v \in V \;\s{and}\;\s{at}\;\s{most}\;\s{one}\;\s{of}\;u,v\;\s{belongs}\;\s{to}\; S \right\}.
\end{align}
We define the number of edges of the graph-cut as $|\s{H}|\triangleq|E|$, and the number of non-selected vertices as $|v(\s{H})|\triangleq|V \setminus S|$.
%\begin{equation}
%\eta(\s{H})\triangleq\frac{|\s{H}|}{|v(\s{H})|}.
%\end{equation}
%In the special case where $V \setminus S = \emptyset$, we define $\eta(\s{H})\triangleq \infty$.
\end{definition}
We note that throughout the paper, we use $\s{H}$ to denote either a graph (identified with its edge set) or a graph-cut, depending on the context. When $\s{H}$ is a graph, $v(\s{H})$ denotes its vertex set; when $\s{H}=(V,S,E)$ is a graph-cut, $v(\s{H})$ denotes the set of non-selected vertices $V\setminus S$. This distinction will be clear from the context. An important graph-theoretic quantity that will play a central role is the \emph{relative density} and the \emph{maximum subgraph relative density}, defined as follows.
\begin{definition}[Relative densities]\label{def:relative}
Given graphs $\s{H}' \subseteq \s{H}$ (viewed as edge sets), we define the induced graph-cut by $\s{H} \vert \s{H}' \triangleq (v(\s{H}), v(\s{H}'), \s{H} \setminus \s{H}')$. The relative density is defined as
\begin{align}
\eta(\s{H} \vert \s{H}') \triangleq \frac{|\s{H}| - |\s{H}'|}{|v(\s{H}) \setminus v(\s{H}')|}.
\end{align}
%In general, for graphs $\s{H}' \subseteq \s{H}$, we define
%\begin{align}
%\eta(\s{H} \vert \s{H}') \triangleq \frac{|\s{H} \setminus \s{H}'|}{|v(\s{H}) \setminus v(\s{H}')|}.
%\end{align}
If $\s{H}\setminus \s{H}' = \emptyset$, we define $\eta(\s{H} \vert \s{H}')\triangleq \infty$. The maximum subgraph relative density is
\begin{align}
\sd(\s{H}\vert\s{H}')\triangleq\max\ppp{\eta(\s{J} \vert \s{H}'):\s{H}'\subsetneq\s{J}\subseteq\s{H}}.\label{eqn:maxDensityRelative}
\end{align}
\end{definition}
Note that, given $\s{H}'$, we choose $\s{J} \supsetneq \s{H}'$ in \eqref{eqn:maxDensityRelative} to maximize the ``cut density'', namely the number of new edges per new vertex, counting also edges from the new vertices back into $\s{H}'$. This corresponds to the densest subgraph in the cut, rather than in the induced remainder. For example, let $\s{H}' = \calK_4$ be a clique on $4$ vertices (with $6$ edges). Add many leaf vertices, each joined to $\calK_4$ by exactly one edge and with no edges among the leaves, and denote the resulting graph by $\s{H}$. Then every added vertex contributes exactly one cross-edge to $\s{H}'$, and hence $\sd(\s{H} \vert \s{H}') = 1$.

Next, our statistical lower and upper bounds will rely, both in the statements and in the proofs, on a canonical decomposition of the planted subgraph $\Gamma$, introduced in \cite[Definition~3.3]{LeePerniceRajaramanZadik2025}.
\begin{definition}[Onion decomposition]\label{def:onionDec}
Let $\Gamma = \Gamma_n$ be an arbitrary graph. The \emph{onion decomposition} of $\Gamma$ is an increasing sequence of subgraphs $\Gamma^{(0)},\Gamma^{(1)},\dots$ constructed as follows:
\begin{itemize}
    \item[(i)] Initialize with $\Gamma^{(0)}\triangleq\emptyset$.
    \item[(ii)] For each $\ell> 0$, let $\Gamma^{(\ell)}$ be a maximal subgraph that maximizes $\eta(\s{H} \vert \Gamma^{(\ell-1)})$ among all $\s{H}$ such that $ \Gamma^{(\ell-1)}\subsetneq \s{H}$.\footnote{Maximality implies that there does not exist a subset $\Gamma^{(\ell)}\subsetneq\Gamma'$ for which $\eta(\Gamma' \vert \Gamma^{(\ell-1)})$ is also maximized.}
    \item[(iii)] Stop if $\Gamma^{(\ell)} = \Gamma$.
\end{itemize}
Let $M = M(\Gamma) \leq |\Gamma|$ denote the total number of steps until termination. The sequence $\{ \Gamma^{(\ell)} \}_{\ell=0}^M$ is referred to as the onion decomposition of $\Gamma$. Finally, define the remainder subgraphs $\calD^{(\ell)}\triangleq\Gamma^{(\ell)}\setminus\Gamma^{(\ell-1)}$ for $\ell=1,2,\ldots,M$.
\end{definition}
Intuitively, this process iteratively selects the densest remaining subgraph, removes it, and continues on the remainder. In the case of balanced graphs \cite{bollobas_1981}, where $\max_{\s{H} \subseteq \Gamma} \frac{|\s{H}|}{|v(\s{H})|} = \frac{|\Gamma|}{|v(\Gamma)|}$, the process terminates in a single step. Note that the decomposition elements $\{\Gamma^{(\ell)}\}_{\ell\geq0}$ are, by definition, edge-disjoint; however, they may share vertices. The following result, proved in Appendix~\ref{app:onionUn}, shows that the decomposition is unique. An alternative proof appears in \cite[Thm.~3.6]{LeePerniceRajaramanZadik2025} as well.
\begin{lemma}[Uniqueness]\label{lem:onionUnique}
Let $\Gamma = \Gamma_n$ be an arbitrary graph, and let $\{ \Gamma^{(\ell)} \}_{\ell=0}^{M}$ (and $\{ \calD^{(\ell)} \}_{\ell=0}^{M}$) denote its onion decomposition in Definition~\ref{def:onionDec}. Then, for every $\ell = 0, 1, \dots, M-1$, the choice of $\Gamma^{(\ell+1)}$ in step (ii) is unique. Thus, the onion decomposition of any graph is uniquely determined.
\end{lemma}
We note that computing $\mu(\Gamma\vert\Gamma^{(\ell-1)})=\eta(\Gamma^{(\ell)}\vert \Gamma^{(\ell-1)})$ for $\ell\in[M(\Gamma)]$ requires, at least naively, constructing the entire sequence $\{\Gamma^{(\ell)}\}_{\ell=0}^{M}$. Nonetheless, the following result, proved in Appendix~\ref{app:equiv}, provides a non-sequential equivalent characterization. In particular, a key quantity in our results and analysis is the \emph{minimal maximum subgraph density}, given by $\eta(\Gamma^{(M)}\vert \Gamma^{(M-1)})$, corresponding to the relative density of the final layer in the onion decomposition of $\Gamma$.
\begin{lemma}[Equivalent characterization]\label{lem:equivminimalMSD}
Let $\Gamma = \Gamma_n$ be an arbitrary graph, and denote by $\{ \Gamma^{(\ell)} \}_{\ell=0}^{M}$ its onion decomposition in Definition~\ref{def:onionDec}. Let
\begin{align}
\Lambda(\Gamma)\triangleq\{\mu(\Gamma\vert\s{J}): \s{J}\subseteq \Gamma\},
\end{align}
and list the distinct values of $\Lambda(\Gamma)$ in strictly decreasing order
\begin{align}
\lambda_1>\lambda_2>\cdots>\lambda_T.
\end{align}
Then $T=M(\Gamma)$ and, for every $\ell=1,\dots,M$,
\begin{align}
\eta\left(\Gamma^{(\ell)}\vert\Gamma^{(\ell-1)}\right)
=\lambda_\ell,
\end{align}
i.e., the $\ell^{\s{th}}$ layer value is the $\ell^{\s{th}}$ largest distinct value taken by $\s{J}\mapsto\mu(\Gamma\vert\s{J})$. %Equivalently,
%\begin{align}
%\eta \left(\Gamma^{(\ell)}\vert\Gamma^{(\ell-1)}\right)
%=\operatorname{ord}_\ell\bigl\{\mu(\Gamma\vert\s{J}):\s{J}\subseteq \Gamma\bigr\}.
%\end{align}
In particular, define
\begin{align}
\mu_{\s{min}}(\Gamma) \triangleq \min_{\s{S}\subseteq \Gamma}\max_{\s{S}\subsetneq\s{F}\subseteq\Gamma}\eta(\s{F}\vert\s{S}).\label{eqn:minimaxSubDen}
\end{align}
Then, $\eta(\Gamma^{(M)}\vert \Gamma^{(M-1)})=\mu_{\s{min}}(\Gamma)$.
\end{lemma}
Thus, for any fixed layer index $\ell\in[M(\Gamma)]$, we have
\begin{align}
  \eta \left(\Gamma^{(\ell)}\vert\Gamma^{(\ell-1)}\right)
  =\s{the}\;\ell^{\s{th}}\;\s{largest}\;\s{distinct}\;\s{value}\;\s{of}\; \{\mu(\Gamma\vert\s{J}):\s{J}\subseteq\Gamma\},
\end{align}
and this value can be defined without explicit reference to the onion construction. At the extremes:
\begin{itemize}
    \item For $\ell=1$ we have $\lambda_1=\mu(\Gamma\vert\emptyset)=\max_{\s{F}\subseteq\Gamma}\eta(\s{F}\vert\emptyset)$, which is the classical maximum subgraph density $\mu(\Gamma)$.
    \item For $\ell=M$ we have $\lambda_M=\min_{\s{J}\subseteq\Gamma}\mu(\Gamma\vert\s{J})$, which is the minimal maximum subgraph density $\mu_{\s{min}}(\Gamma)$.
\end{itemize}
Throughout this paper, we sometimes suppress the explicit dependence of various parameters on the index $n$. For example, we write the sequence of planted graphs as $\Gamma = (\Gamma_n)_n$, the sequences of edge probabilities as $p = (p_n)_n$ and $q = (q_n)_n$, and so on.

\section{Main Results}\label{sec:mainresults}

In this section, we present our main results. We begin by analyzing the statistical limits of the recovery problem, setting computational considerations aside, and identify a sharp threshold at which the optimal recovery error probability undergoes a phase transition from $0$ to $1$ as $n\to\infty$. We then address recovery under polynomial-time constraints by proposing a general, computationally efficient procedure and providing statistical guarantees on its performance. Next, we establish several computational lower bounds, showing that polynomial-time algorithms can incur an inherent gap relative to the optimal (information-theoretic) solution. Finally, we consider extensions of the vanilla planted model, including recovery under semi-random (monotone-adversary) perturbations as well as other weaker notions of recovery.

\subsection{Statistical limits}

In this subsection, we establish the statistical limits of the problem, beginning with information-theoretic lower bounds and continuing with matching upper bounds.

\paragraph{Lower bound.} We start with the following general lower bound, which holds for an arbitrary planted subgraph $\Gamma$.
\begin{theorem}[Statistical lower bound]\label{thm:lowerp1}
Fix a sequence of subgraphs $\Gamma=(\Gamma_n)_n$, and assume $p_n,q_n=\Theta(1)$. %Let $\{\calD^{\ell}\}_{\ell=0}^M$ be the reminder decomposition of $\Gamma$ in Definition~\ref{def:onionDec}. Then, exact recovery is impossible if,
Exact recovery is statistically impossible if
\begin{align} 
        \mu_{\s{min}}(\Gamma_n)\leq \frac{(1-\varepsilon)\cdot \log n}{d_{\s{KL}}(p||q)},\label{eq:condDenseStat}
\end{align}
for any $\varepsilon>0$.
\end{theorem}

To better understand Theorem~\ref{thm:lowerp1}, we now consider a few examples.
\begin{example}[Planted clique]
    Consider the case where $\Gamma = \calK_k$ is a clique with $k = |v(\calK_k)|$ vertices. A clique is balanced (i.e., $\mu(\calK_k) = \eta(\calK_k) = \frac{k-1}{2}$), and thus $\mu_{\s{min}}(\calK_k) = \frac{k-1}{2}$. Therefore, Theorem~\ref{thm:lowerp1} tells us that exact recovery is statistically impossible if $k \leq (1 - \varepsilon)\frac{2\log n}{d_{\s{KL}}(p||q)}$, which is consistent with folklore results (e.g., \cite{jerrum1992large}).
\end{example}

\begin{example}[Union of disjoint cliques]
Consider the case where $\Gamma$ is a union of $L$ disjoint cliques with sizes $k_1,\dots,k_L$. Define $k_{\min}\triangleq \min_{i\in[L]} k_i$. A straightforward calculation reveals that $\mu_{\s{min}}(\Gamma) = \frac{k_{\min}-1}{2}$, and so Theorem~\ref{thm:lowerp1} implies that exact recovery is statistically impossible if $k_{\min} \leq (1 - \varepsilon)\frac{2\log n}{d_{\s{KL}}(p||q)}$.
\end{example}

\begin{example}[A kite]\label{examp:kite}
    Consider the case where $\Gamma$ is a kite on $k+1$ vertices, namely, a clique with $k = |v(\calK_k)|$ vertices, with one of its vertices connected by an edge to an additional vertex. It is not difficult to see that, in this case, the onion decomposition layers are $\calD^{(1)} = \calK_k$ and $\calD^{(2)} = \{\cdot, k+1\}$, which is a single edge. Thus, $\mu_{\s{min}}(\calK_k) = \frac{1}{2}$. Therefore, Theorem~\ref{thm:lowerp1} tells us that exact recovery of a kite is statistically impossible. This aligns with intuition: the noisy random background can alter (or ``tweak'') the single edge with positive probability. 
\end{example}

Example~\ref{examp:kite} illustrates that if $\Gamma$ contains a dense core together with a very sparse appendage, then the latter can preclude exact recovery. In such cases, it is reasonable to consider alternative recovery criteria (such as weak recovery), which allow for a nonzero fraction of errors. We discuss such criteria in Subsection~\ref{subsec:exten}. Next, it is instructive to compare and contrast the above results with the detection (hypothesis-testing) variant, given by
\begin{align}
\calH_0: \s{G} \sim \calG(n,q) \quad \s{vs.} \quad \calH_1: \s{G} \sim \calG_{\Gamma_n}(n,p,q).\label{eqn:super_hypo}   
\end{align}
Here, the task is to decide whether or not a copy of $\Gamma$ was planted in $\s{G}$. Specifically, \cite[Thm.~1]{elimelech2025detecting} shows that 
\begin{itemize}
    \item If $\Gamma_n$ has sub-logarithmic density, i.e., $\mu(\Gamma_n)=o(\log|v(\Gamma_n)|)$, then detection is statistically impossible provided that
    \begin{align}
        |e(\Gamma_n)|\vee d_{\max}^2(\Gamma_n)\ll n.\label{eqn:Dec1IMP}
    \end{align}
    \item If $\Gamma_n$ has super-logarithmic density, i.e., $\mu(\Gamma_n)=\Omega(\log|v(\Gamma_n)|)$, then detection is statistically impossible provided that
    \begin{align}
        \mu(\Gamma_n) \leq \underline{C} \cdot \log n,\label{eqn:Dec2IMP}
    \end{align}
    for some $\underline{C}>0$.
\end{itemize}
It is further shown in the same paper that these bounds are asymptotically tight, namely, there exist algorithms that achieve them. Comparing the above with Theorem~\ref{thm:lowerp1}, we see that the statistical barriers for detection and recovery hinge on different graph-theoretic quantities, revealing a distinct \emph{detection--recovery gap}. Specifically, in the sub-logarithmic density regime, which includes graphs such as paths, stars, trees, and related structures, detection can be statistically possible, in stark contrast to exact recovery, which is always statistically impossible in this regime. In fact, this impossibility extends beyond subgraphs with sub-logarithmic density to (arguably) all subgraphs with sub-logarithmic minimal maximum subgraph density $\mu_{\s{min}}$. Thus, a clear \emph{detection--recovery gap} emerges in this setting. By contrast, the super-logarithmic density regime includes sufficiently dense subgraphs, such as cliques and bipartite graphs, where detection and recovery thresholds may coincide when the planted structure is ``sufficiently nice'' (e.g., when $\Gamma$ is a clique). However, recovery may still be statistically impossible even when detection is feasible, mirroring the behavior in the sub-logarithmic regime, as illustrated by the kite example. 

Finally, for the special case where $\Gamma$ is balanced \cite{bollobas_1981} (e.g., a clique), we can establish tight information-theoretic lower bounds using a different technique, namely Bayes risk analysis, which relates exact recovery to a variant of the detection problem. The details are provided in Appendix~\ref{app:secProofofITLOWEr}.\footnote{As noted in \cite[pg.~7]{wein2025computational}, even in the special case of a planted clique, the ``$2\log_2 n$'' information-theoretic threshold had been established for detection but not for exact recovery, although the latter was expected to hold. In this work, we confirm this expectation by providing a lecture-style proof that establishes $2\log_2 n$ as the sharp threshold for exact recovery as well, and in fact we generalize this result to arbitrary planted subgraphs.} We now move to our statistical upper bounds.

\paragraph{Upper bound.} It is relatively straightforward to show that the optimal maximum-likelihood estimator (MLE) of $\Gamma^\star$ is given by (see Appendix~\ref{app:0} for details)
\begin{align}
\hat{\Gamma}_{\s{MLE}}(\s{G}) = \arg\max_{\Gamma'\in\calS_{\Gamma}}\abs{e(\Gamma')\cap\s{G}}.\label{eqn:MLEGen}
\end{align}
Interestingly, it turns out to be both simpler and more intuitive to analyze the following recursive peeling MLE procedure. Specifically, let $\{\Gamma^{(\ell)}\}_{\ell\geq1}$ denote the onion decomposition of $\Gamma$ in Definition~\ref{def:onionDec}, and recall that $\calM(\s{H},\Gamma)$ denotes the set of copies of $\Gamma$ in $\calK_n$ which contain $\s{H}$. At step $\ell \in [M]$, the recovery algorithm produces the MLE estimate for $\Gamma^{(\ell)}$ as follows:
\begin{align}
\hat{\Gamma}^{(\ell)}_{\s{MLE}} = \arg\max_{\calD \in \calM(\hat{\Gamma}^{(\ell-1)}_{\s{MLE}},\Gamma^{(\ell)})} \abs{e(\calD) \cap \left(\s{G}\setminus\hat{\Gamma}^{(\ell-1)}_{\s{MLE}}\right)}.\label{eqn:PeelMLE0}
\end{align}
That is, the estimator searches for a copy of $\Gamma^{(\ell)}$, restricted to those copies that extend the previously estimated layer $\hat{\Gamma}^{(\ell-1)}_{\s{MLE}}$, within the observed graph $\s{G}$, after removing the previously estimated layer $\hat{\Gamma}^{(\ell-1)}_{\s{MLE}}$. We define $\hat{\Gamma}^{(0)}_{\s{MLE}}=\emptyset$. At the final step, the algorithm outputs the complete estimate
\begin{align}
    \hat{\Gamma}_{\s{pMLE}} = \hat{\Gamma}^{(M)}_{\s{MLE}}.\label{eqn:PeelMLE}
\end{align}
The following result complements the information-theoretic lower bound in Theorem~\ref{thm:lowerp1}.
\begin{theorem}[Optimal algorithm]\label{thm:MLE}
Fix a sequence of subgraphs $\Gamma=(\Gamma_n)_n$, and assume $p_n,q_n=\Theta(1)$. 
%Let $\{\calD^{\ell}\}_{\ell=0}^M$ be the reminder decomposition of $\Gamma$ in Definition~\ref{def:onionDec}. 
Exact recovery of $\Gamma^\star$, via the likelihood peeling algorithm in \eqref{eqn:PeelMLE}, is possible if
\begin{align}
    \mu_{\s{min}}(\Gamma_n)\geq \s{C}\frac{(1+\varepsilon)\cdot \log n}{d_{\s{KL}}(p||q)},\label{eqn:MLEcond}
\end{align}
for any $\varepsilon>0$ and some constant $\s{C}>0$. %then, the maximum likelihood estimator satisfies $\pr_{\calG_{\Gamma_n}}[\hat{\Gamma}_{\s{MLE}}(\s{G})=\Gamma^\star]\to1$, as $n\to\infty$.
\end{theorem}

Thus, Theorem~\ref{thm:MLE} matches the lower bound in Theorem~\ref{thm:lowerp1} up to a constant factor. For $p=1$, the proof of Theorem~\ref{thm:MLE} yields $\s{C}=4$, and for any $p<1$ we obtain $\s{C}=16$. We have not made a serious effort to optimize this constant, but we conjecture that $\s{C}=1$ for any $p\in[0,1]$.\footnote{As remarked in the proof, improving the constant would require splitting the various sums involved into different asymptotic regimes and analyzing each case separately.}

It is worth emphasizing that the proof of Theorem~\ref{thm:MLE} relies on an important property of the layers obtained in the onion decomposition of $\Gamma$, namely a uniqueness property of the planted copy on its full vertex set. Specifically, one of the steps in the proof analyzes $|v(\calD' \cap \calD^{(\ell),\star})|$, the intersection between an isomorphic copy $\calD'\neq\calD^{(\ell),\star}$ of the actual planted $\ell$th layer $\calD^{(\ell),\star}$. In principle, it could happen that $|v(\calD' \cap \calD^{(\ell),\star})| = |v(\calD^{(\ell),\star})|$ yet $\calD' \neq \calD^{(\ell),\star}$; for example, this can occur in the case of a kite. When this happens, exact recovery becomes impossible. Fortunately, we show that this cannot occur for the layers selected by the onion procedure when the parameters lie in the achievability regime.

\subsection{Computationally efficient algorithm}\label{sec:compConvex}

We now introduce a polynomial-time recovery algorithm. To describe the procedure, we first establish notation and record several definitions. The Hilbert--Schmidt inner product (also referred to as the Frobenius inner product) between two matrices $\s{A}$ and $\s{B}$ of identical dimensions is
\begin{align}
    \innerP{\s{A},\s{B}} \triangleq \s{Tr}(\s{A}\s{B}^\top)=\sum_{i,j}\s{A}_{i,j}\s{B}_{i,j}.
\end{align}
This inner product will be used throughout our analysis.

The formulation in \eqref{eqn:MLEGen} often serves as a starting point for constructing computationally tractable recovery algorithms, typically via relaxation. In our work, we employ semidefinite programming (SDP), a matrix-based generalization of linear programming, as one such relaxation. A standard form of an SDP is
\begin{equation}
\begin{aligned}
&\underset{\s{Z}\in\mathbb{R}^{n\times n}}{\max}
& & \innerP{\s{C},\s{Z}} \\
& \ \text{s.t.}
& &  \s{Z}\succeq 0\\
&&& \innerP{\s{B}_i,\s{Z}}\leq \beta_i\quad\forall i\in[m],
\end{aligned}\label{eqn:GenSDP}
\end{equation}
for a collection of matrices $\s{C},\{\s{B}_i\}_{i=1}^m$, constants $\{\beta_i\}_{i=1}^m$, and $m\in\mathbb{N}$. Here, $\s{Z} \succeq 0$ means that $\s{Z}$ is symmetric and positive semidefinite. SDPs are convex optimization problems and therefore admit efficient solution techniques, including interior-point methods and first-order algorithms; see, for instance, \cite{nesterov1987interior,boyd2004convex}. A typical use case is to approximate solutions to optimization problems with nonconvex constraints, such as integer programs, by relaxing them into a semidefinite form.

Recall that the nuclear norm $\norm{\s{B}}_{\star}$ of a matrix $\s{B}$ is defined as the $\ell_1$-norm of its singular value vector. When $\s{B}$ is symmetric, this equals the sum of the absolute values of its eigenvalues. Slightly abusing notation, we denote by $\norm{\s{G}}_\star$ the nuclear norm of the adjacency matrix associated with a graph $\s{G}$, interpreted as the adjacency matrix of the subgraph induced on $|v(\s{G})|$ vertices. Importantly, if $|v(\s{G})| \leq n$, then padding the adjacency matrix with zeros to embed $\s{G}$ into an $n$-vertex graph does not change its nuclear norm.

To facilitate the discussion, we associate the planted subgraph $\Gamma^\star$ with a clustered adjacency matrix $\s{X}^\star \in \{0,1\}^{n \times n}$, defined such that $\s{X}^\star_{ii} = 0$ for all $i \in [n]$, $\s{X}^\star_{ij} = 1$ if $(i,j) \in e(\Gamma)$, and $\s{X}^\star_{ij} = 0$ otherwise. Thus, $\s{X}^\star$ can be interpreted as the adjacency matrix of $\Gamma^\star$ embedded in a graph on $n$ vertices, where the additional $n-|v(\Gamma^\star)|$ vertices are isolated. Note that there is a one-to-one correspondence between each $\Gamma^\star$ and its matrix representation $\s{X}^\star$.

Our main result depends on spectral and structural properties of real-valued symmetric matrices. Let $\s{Q}$ be a generic real-valued symmetric matrix. Denote its singular value decomposition (SVD) by $\s{Q} = \s{U}_{\s{Q}} \Sigma \s{U}_{\s{Q}}^T$, where $\Sigma \triangleq \s{Diag}(\sigma_1, \ldots, \sigma_r) \in \mathbb{R}^{r \times r}$ contains the (non-negative) singular values $\{\sigma_i\}_{i=1}^r$. The columns of $\s{U}_{\s{Q}} = [u_1, \ldots, u_r] \in \mathbb{R}^{n \times r}$ are orthonormal and are referred to as the singular vectors. The quantity $r \triangleq \s{rank}(\s{Q})$ denotes the rank of $\s{Q}$.

In our setting, $\s{Q}$ will exhibit a block structure; for example,
\begin{align}
\s{Q} =
\begin{bmatrix}
\s{Q}_\calS & \mathbf{0} \\
\mathbf{0} & \mathbf{0}
\end{bmatrix},
\end{align}
where $\s{Q}_\calS$ denotes the restriction of $\s{Q}$ to $\calS\subseteq[n]$. Accordingly, $\mathsf{U}_{\s{Q}}$ is supported on $\calS$, in the sense that $[\mathsf{U}_{\s{Q}}]_{i,:} = 0$ for all $i \notin\calS$. Finally, we define the \emph{coherence parameter} of $\s{U}_{\s{Q}}$, which measures how ``spread out'' (or ``localized'') its rows are:
\begin{align}
    \s{coh}(\s{U}_{\s{Q}})\triangleq\frac{n}{\s{rank}(\s{Q})}\cdot\max_{i\in[n]}\norm{[\mathsf{U}_{\s{Q}}]_{i,:}}_2^2 = \frac{n}{\s{rank}(\s{Q})}\cdot\norm{[\mathsf{U}_{\s{Q}}]_{i,:}}_{2,\infty}^2.
\end{align}
We prove in Appendix~\ref{app:cohereBounds} that $\frac{n}{|\calS|}\leq \s{coh}(\s{U}_{\s{Q}})\leq \frac{n}{\s{rank}(\s{Q})}$.\footnote{Note that the lower bound is achieved when $\Gamma^\star$ is a clique; in this case $\s{X}^\star$ has rank one with $\s{U} = \frac{1}{\sqrt{|v(\Gamma^\star)|}}\mathbf{1}_{v(\Gamma^\star)}$, and hence $\s{coh}(\s{U}) = \frac{n}{|v(\Gamma^\star)|}$.} With a slight abuse of notation, we also write $\s{coh}(\s{Q}) = \s{coh}(\s{U}_{\s{Q}})$. In our presentation and analysis, it is convenient to work with the following transformation of $\s{A}$:
\begin{align}
\s{W}_{ij}\triangleq\begin{cases}
q^{-1}\s{A}_{ij}-1\ &i\neq j\\
0\ &i=j.
\end{cases}\label{eqn:centeredform}
\end{align}

To present the proposed algorithm and its performance guarantees, we introduce a few additional pieces of notation. Fix a hyper-parameter $\alpha\in\mathbb{R}_+$. For any symmetric matrix $\s{X}\in\mathbb{R}^{n\times n}$ with $\s{Diag}(\s{X})=\mathbf{0}$, and any vector $\mathbf{s}\in\mathbb{R}^n$, define
\begin{align}
    \s{Ext}(\s{X},\mathbf{s};\alpha)\triangleq \s{X}+\alpha\cdot\s{Diag}(\mathbf{s}).
\end{align}
Let $\mathbf{s}^\star = \s{supp}(v(\Gamma^\star))$ denote the support of the vertex set of the planted subgraph. At $(\s{X},\mathbf{s}) = (\s{X}^\star,\mathbf{s}^\star)$, the matrix $\s{Ext}(\s{X}^\star,\mathbf{s}^\star;\alpha)$ is therefore a diagonally shifted version of the adjacency matrix of $\Gamma^\star$. We denote the nuclear norm of the diagonally shifted version of an arbitrary planted subgraph $\Gamma\in\calS_\Gamma$ by
\begin{align}
    \norm{\s{Ext}(\s{X}^\star,\mathbf{s}^\star;\alpha)}_\star
    =\norm{\Gamma+\alpha\cdot\s{Diag}(v(\Gamma))}_\star.
\end{align}
This quantity is an a priori known constant, taking the same value for all $\Gamma \in \mathcal{S}_\Gamma$, and is independent of the particular latent planted subgraph.

We now consider the following convex optimization problem:
\begin{equation}
\begin{aligned}
\hat{\s{X}}_{\mathrm{con}}^{(\alpha)}=& \underset{\s{X}\in\mathbb{R}^{n\times n},\mathbf{s}\in[0,1]^n}{\arg\max}
& & \innerP{\s{W},\s{X}} \\
& \ \text{s.t.}
& & \;\mathbf{s}^\top\mathbf{1} = |v(\Gamma)|,\;\s{X}_{ij}\leq \min(\mathbf{s}_i,\mathbf{s}_j),\;\forall i\neq j\\
&&&  \norm{\s{Ext}(\s{X},\mathbf{s};\alpha)}_\star\leq\norm{\Gamma+\alpha\cdot\s{Diag}(v(\Gamma))}_\star\\
&&& \;\mathbf{0}\leq\s{X}\leq\mathbf{J},\;\s{X} = \s{X}^\top,\;\s{X}_{ii}=0,\;\forall i\in[n]\\
&&& \innerP{\Jb,\s{X}}=2|e(\Gamma)|,
\end{aligned}\label{eqn:convex}
\end{equation}
where the inequality $\mathbf{0}\leq\s{X}\leq\mathbf{J}$ is to be interpreted entrywise. In several special cases, such as when $\Gamma$ is a clique or a bipartite graph, the program in~\eqref{eqn:convex} can be viewed as a direct relaxation of the maximum-likelihood estimator. 

Before proceeding, we briefly interpret the variables and constraints in~\eqref{eqn:convex}. The vector $\mathbf{s}\in[0,1]^n$ serves as a relaxed selector for the vertex set of the planted subgraph: each coordinate $\mathbf{s}_i$ indicates whether vertex $i$ belongs to $v(\Gamma^\star)$. At the planted solution, $\mathbf{s}_i^\star=\Ind\ppp{i\in v(\Gamma^\star)}$. Similarly, $\s{X}\in\mathbb{R}^{n\times n}$ is a convex relaxation of the planted adjacency matrix $\s{X}^\star$. One might ask why the auxiliary variable $\mathbf{s}$ is needed, rather than encoding the vertex set directly via self-loops in $\s{X}$. The issue is that such an encoding can significantly weaken the discriminative power of the nuclear norm. For instance, certain subgraphs, such as complete bipartite graphs, become full-rank under this representation, rendering the nuclear-norm prior ineffective. Introducing $\mathbf{s}$ allows us to control the number of selected vertices through the linear constraint $\sum_{i\in[n]}\mathbf{s}_i = |v(\Gamma)|$, while applying the low-rank prior to the shifted matrix $\s{Ext}(\s{X},\mathbf{s};\alpha)$. 

In addition, we impose the edge--vertex coupling constraints $\s{X}_{ij}\leq\min(\mathbf{s}_i,\mathbf{s}_j)$ for all $i\neq j$, ensuring that an edge can exist only if both endpoints are selected. These are convex (McCormick envelope-type) relaxations of the nonconvex constraint $\s{X}_{ij}\leq \mathbf{s}_i\mathbf{s}_j$, enforcing that $\s{X}$ is supported on the selected vertex set. The remaining constraints ensure that $\s{X}$ encodes a relaxed simple-graph adjacency matrix: it is symmetric, has zero diagonal, and has fixed off-diagonal mass equal to $2|e(\Gamma)|$, corresponding to the number of edges in the planted subgraph. Finally, we constrain the nuclear norm of the shifted matrix $\s{Ext}(\s{X},\mathbf{s};\alpha)$ to be no larger than that of the ground-truth planted subgraph. This promotes low-rank structure in the solution. By choosing $\alpha$ appropriately, we can avoid rank inflation caused by forced self-loops. For example, setting $\alpha=1$ for cliques ensures that $\Gamma+\alpha\cdot\s{Diag}(v(\Gamma))$ has rank one, while setting $\alpha=0$ for complete bipartite graphs ensures it has rank two.

The program~\eqref{eqn:convex} is convex. The objective $\innerP{\s{W},\s{X}}$ is linear in $\s{X}$, and all constraints on $\s{X}$ and $\mathbf{s}$ are linear, except for the nuclear-norm constraint. The latter defines a convex set because $\|\cdot\|_\star$ is a convex norm and $\s{Ext}(\cdot,\cdot;\alpha)$ is affine. Moreover, the nuclear-norm constraint admits a standard semidefinite representation, allowing~\eqref{eqn:convex} to be formulated as a semidefinite program (SDP). In particular, the following classical result applies.
\begin{lemma}[{\cite[Lemma 2]{fazel2002matrix}}]
 Fix $t\geq0$. Let $\s{X} \in \mathbb{R}^{m\times n}$. Then, $\norm{\s{X}}_\star \leq t$ if and only if there exist $\s{P}_1 \in \mathbb{R}^{m \times m}$ and $\s{P}_2 \in \mathbb{R}^{n \times n}$ such that,
    \begin{align}
\begin{bmatrix}
\s{P}_1 & \s{X} \\
\s{X}^\top & \s{P}_2
\end{bmatrix} \succeq 0 \quad \s{and} \quad \s{Tr}(\s{P}_1) + \s{Tr}(\s{P}_2) \leq 2t.\label{eqn:ccooc}
\end{align}
\end{lemma}
As a result, since all constraints in \eqref{eqn:convex} consist of linear equalities, linear inequalities, and linear matrix inequalities, the program is a semidefinite program (SDP), and therefore admits a polynomial-time solution via standard convex optimization methods. We now state the following result.
%Substituting \eqref{eqn:nuc-epigraph} into \eqref{eqn:convex-shifted} yields an SDP:
%\begin{equation}
%\begin{aligned}
%\hat{\s{X}}_{\mathrm{SDP}}=& \underset{\substack{\s{X}\in\mathbb{R}^{n\times n},\ s\in\mathbb{R}^n\\ \s{P},\s{Q}\in\mathbb{S}^n}}{\arg\max}
%& & \innerP{\s{W},\s{X}} \\begin{align}2pt]
%& \ \text{s.t.}
%& & \begin{bmatrix}
%\s{P} & \s{X}+\alpha\mathrm{Diag}(s)\\begin{align}2pt]
%\s{X}+\alpha\mathrm{Diag}(s) & \s{Q}
%\end{bmatrix}\succeq 0,\quad \s{P}\succeq 0,\ \s{Q}\succeq 0,\quad \mathrm{tr}(\s{P})+\mathrm{tr}(\s{Q})\leq  2\tau\\begin{align}4pt]
%&&& \mathbf{0}\leq\s{X}\leq\mathbf{J},\quad \s{X}=\s{X}^\top,\quad \s{X}_{ii}=0\ \ (\forall i)\\begin{align}2pt]
%&&& 0\leq  \mathbf{s}_i\leq  1\ \ (\forall i),\qquad \sum_{i=1}^n \mathbf{s}_i\ =\ |v(\Gamma)|\\begin{align}2pt]
%&&& X_{ij}\leq  \mathbf{s}_i,\ \ X_{ij}\leq  \mathbf{s}_j\ \ \ (\forall i\neq j)\\begin{align}2pt]
%&&& \innerP{\Jb,\s{X}}\ =\ 2|e(\Gamma)|\;.
%\end{aligned}
%\label{eqn:convex-shifted-sdp}
%\end{equation}
\begin{theorem}[Efficient algorithm]\label{thm:effAlg}
Fix a sequence of subgraphs $\Gamma=(\Gamma_n)_n$, $\alpha\in\mathbb{R}_+$, and assume $p_n,q_n=\Theta(1)$. Define the diagonally shifted adjacency matrix $\s{S}^{(\alpha)}_\Gamma\triangleq\s{X}^\star+\alpha\cdot\s{Diag}(v(\Gamma^\star))$. Exact recovery of $\s{X}^\star$, using the convex program $\hat{\s{X}}_{\mathrm{con}}$ in \eqref{eqn:convex}, is possible if,
\begin{align}
        &\s{coh}(\s{S}^{(\alpha)}_\Gamma)\cdot\s{rank}(\s{S}^{(\alpha)}_\Gamma)\leq \min\ppp{c_1\sqrt{n},c_2\frac{n}{\sqrt{|v(\Gamma)|\log n}}},\label{eqn:TwoCondMain}
\end{align}
for some constants $c_1,c_2>0$.
\end{theorem}
To appreciate Theorem~\ref{thm:effAlg} let us consider a few examples.
\begin{example}[Planted clique]\label{exmp:1}
Consider the case where $\Gamma=\calK_k$ is a clique with $k\triangleq |v(\calK_k)|$ vertices. The adjacency $\s{X}^\star$ in this case has eigenvalues $k-1$ and $-1$ with multiplicity $k-1$. Hence, if we take $\alpha=1$, we get that the shifted matrix
\begin{align}
\s{S}^{(1)}_{\calK_k}=\s{X}^\star+\s{Diag}(v(\calK_k)),
\end{align}
has rank one, namely, $\s{rank}(\s{S}^{(1)}_{\calK_k})=1$. Furthermore, the corresponding singular vector $\s{U}$ associated with $\s{S}^{(1)}_{\calK_k}$ has a single column equal to $k^{-1/2}\mathbf{1}_{v(\calK_k)}$. This implies that $\max_i \|\s{U}_{i,:}\|_2^2= \frac{1}{k}$, and thus,
\begin{align}
\s{coh}(\s{S}^{(1)}_{\calK_k}) = \s{coh}(\s{U})=\frac{n}{\s{rank}(\s{S}^{(1)}_{\calK_k})}\cdot\max_{i\in[n]}\norm{\mathsf{U}_{i,:}}_2^2 = \frac{n}{k}.
\end{align}
Accordingly, it can be seen that Theorem~\ref{thm:effAlg} implies that strong recovery using \eqref{eqn:convex} is possible provided that $k\geq C\sqrt{n}$, for some $C>0$. This is consistent with how state-of-the-art recovery algorithms perform on the planted clique problem, e.g., \cite{alon1998finding,dekel2014finding,montanari2015finding,hajek2016achieving,chen2016statistical}.
\end{example}
\begin{example}[Complete bipartite]\label{exmp:bip}
Consider the case where $\Gamma=\calK_{k_{\s{L}},k_{\s{R}}}$ is a complete bipartite graph with partitions of size $k_{\s{L}}$ and $k_{\s{R}}$, and define $k\triangleq k_{\s{L}}+k_{\s{R}}$. Let us choose $\alpha=0$, and then $\s{S}_{\calK_{k_{\s{L}},k_{\s{R}}}}^{(0)} = \s{X}^\star$ is of the form
\begin{align}
\s{S}_{\calK_{k_{\s{L}},k_{\s{R}}}}^{(0)}=\begin{bmatrix} \mathbf{0} & \mathbf{J}_{k_{\s{L}},k_{\s{R}}}\\ \mathbf{J}_{k_{\s{R}},k_{\s{L}}} & \mathbf{0}\end{bmatrix}.
\end{align}
Thus, we see that $\s{S}_{\calK_{k_{\s{L}},k_{\s{R}}}}^{(0)}$ has rank 2, i.e., $\s{rank}(\s{S}^{(0)}_{\calK_{k_{\s{L}},k_{\s{R}}}})=2$, with non-zero eigenvalues $\{\sqrt{k_{\s{L}}k_{\s{R}}},-\sqrt{k_{\s{L}}k_{\s{R}}}\}$. The projector $\s{P}\triangleq\s{U}\s{U}^\top$ onto this $2$-dimensional subspace satisfies
\begin{align}
[\s{P}]_{ii}=\begin{cases}
\frac{1}{k_{\s{L}}}, & i\in \s{left}\;\s{side},\\
\frac{1}{k_{\s{R}}}, & i\in \s{right}\;\s{side},
\end{cases}
\end{align}
which implies that $\max_i[\s{P}]_{ii}=\frac{1}{\min\{k_{\s{L}},k_{\s{R}}\}}$. Therefore
\begin{align}
\s{coh}(\s{S}^{(0)}_{\calK_{k_{\s{L}},k_{\s{R}}}})\cdot\s{rank}(\s{S}^{(0)}_{\calK_{k_{\s{L}},k_{\s{R}}}})=\frac{n}{\min\{k_{\s{L}},k_{\s{R}}\}}.
\end{align}
Accordingly, it can be seen that Theorem~\ref{thm:effAlg} implies that strong recovery using \eqref{eqn:convex} is possible provided that
\begin{align}
    \min\{k_{\s{L}},k_{\s{R}}\}\geq C_1\sqrt{n}\quad\s{and}\quad\frac{\min\{k_{\s{L}},k_{\s{R}}\}}{\sqrt{\max\{k_{\s{L}},k_{\s{R}}\}}}\geq C_2\sqrt{\log n},\label{eqn:CondBip}
\end{align}
for some $C_1,C_2>0$. In the balanced case, where $k_{\s{L}}\approx k_{\s{R}}$, the condition at the left-hand side of \eqref{eqn:CondBip} dominates, and exact recovery is possible once $\min\{k_{\s{L}},k_{\s{R}}\}\geq C_1\sqrt{n}$. This is consistent with how state-of-the-art recovery algorithms perform on the planted balanced bipartite (or bi-clique) problem, e.g., \cite{Levanzov2018,Kumar2022}.
\end{example}

\begin{example}[Balanced Tur\'{a}n graph]
Consider the case where $\Gamma=\calT(k,r)$ is the balanced Tur\'{a}n graph, namely, a complete $r$-partite graph on $k=|v(\Gamma)|$ vertices that avoids a $(r+1)$-clique. Denote the size of each partition by $m\triangleq k/r$. Let us choose $\alpha=0$. As is well-known, the adjacency matrix $\s{S}_{\calT(k,r)}^{(0)} = \s{X}^\star$ has the following eigenvalues (see, e.g., \cite{nikiforov2017norms}): $(r-1)m$ with multiplicity 1, $-m$ with multiplicity $r-1$, and $0$ with multiplicity $k - r$. Thus, $\s{rank}(\s{S}_{\calT(k,r)}^{(0)})=r$. The range is spanned by the $r$ part-indicator vectors. As in the bipartite case, for any vertex $i\in\Gamma$, we have $[\s{P}]_{ii}=\frac{1}{m}$, and thus, $\max_i[\s{P}]_{ii}=\frac{1}{m}$. Hence
\begin{align}
\s{coh}(\s{S}_{\calT(k,r)}^{(0)})\cdot\s{rank}(\s{S}_{\calT(k,r)}^{(0)})=\frac{n}{m}=\frac{nr}{k}.
\end{align}
Accordingly, it can be seen that Theorem~\ref{thm:effAlg} implies that strong recovery using \eqref{eqn:convex} is possible provided that $k\geq \max(C_1r\sqrt{n},C_2r^2\log n)$. 
\end{example}

\begin{example}[Union of disjoint cliques]
Consider the case where $\Gamma$ is a union of $L$ disjoint cliques with sizes sizes $k_1,\dots,k_L$. Define $k\triangleq \sum_{i=1}^L k_i$ and $k_{\min}\triangleq \min_{i\in[L]} k_i$. Take $\alpha=1$. Similarly to Example~\ref{exmp:1}, it can be shown that $\s{rank}(\s{S}_{\Gamma}^{(1)})=L$, and that 
\begin{align}
\s{coh}(\s{S}_{\Gamma}^{(1)})\cdot\s{rank}(\s{S}_{\Gamma}^{(1)})=\frac{n}{k_{\min}}.
\end{align}
Accordingly, Theorem~\ref{thm:effAlg} implies that strong recovery using \eqref{eqn:convex} is possible provided that $k_{\min}\geq \max(C_1\sqrt{n},C_1\sqrt{k\log n})$.
\end{example}

\begin{example}[Triangular graph]
Consider the case where $\Gamma=T(m)$ is the triangular graph, namely, the line graph of $\calK_m$. The triangular graph $T(m)$ has vertex set $\binom{[m]}{2}$, and thus $k=|v(\Gamma)|=\binom{m}{2}$, and edges between pairs of $2$-subsets that intersect in exactly one element. It is the Johnson graph $J(m,2)$ (a strongly regular, vertex-transitive graph). It is well-known that the eigenvalues of its adjacency matrix $\s{X}^\star$ are given by (see, e.g., \cite{nikiforov2017norms}): 
\begin{align}
\s{Spectrum}(\s{X}^\star)=\ppp{\underbrace{2(m-2)}_{\s{mult.}\; 1},\underbrace{m-4}_{\s{mult.}\; m-1},\underbrace{-2}_{\s{mult.}\;(m-1)(m-2)/2}}.
\end{align}
Taking $\alpha=2$ we obtain
\begin{align}
\s{Spectrum}(\s{S}_{T(m)}^{(2)})=\ppp{\underbrace{2m-2}_{\s{mult.}\; 1},\underbrace{m-2}_{\s{mult.}\; m-1},\underbrace{0}_{\s{mult.}\;(m-1)(m-2)/2}}.
\end{align}
Hence $\s{rank}(\s{S}_{T(m)}^{(2)})=m$. Because $T(m)$ is vertex-transitive and the range of $\s{S}_{T(m)}^{(2)}$ is a direct sum of irreducible representations invariant under $\mathrm{Aut}(T(m))\cong \mathbb{S}_m$, the projector $\s{P} = \s{U}\s{U}^\top$ has constant diagonal:
\begin{align}
[\s{P}]_{ii}\equiv \frac{\s{rank}(\s{S}_{T(m)}^{(2)})}{k}=\frac{m}{\binom{m}{2}}=\frac{2}{m-1}\;\forall i\in [k].\label{eq:proj-constant-diag}
\end{align}
Indeed, conjugating $\s{P}$ by any permutation matrix induced by an automorphism leaves $\s{P}$ invariant; by vertex-transitivity, all diagonal entries must match, and $\s{trace}(\s{P})=\s{rank}(\s{S}_{T(m)}^{(2)})$ gives \eqref{eq:proj-constant-diag}. Therefore
\begin{align}
\s{coh}(\s{S}_{T(m)}^{(2)})\cdot\s{rank}(\s{S}_{T(m)}^{(2)})=\frac{2n}{m-1}.
\end{align}
Thus, it is not difficult to check that Theorem~\ref{thm:effAlg} implies that strong recovery using \eqref{eqn:convex} is possible provided that $k\geq C\sqrt{n}$, for some $C>0$.
\end{example}

\subsection{Computational lower bounds}\label{subsec:ldp-intro}

In this subsection, we derive computational lower bounds for recovery. We begin by introducing the low-degree polynomial (LDP) framework for \emph{recovery}, setting up the notation used throughout, and recording several basic identities that will be invoked later. Our presentation adapts the general LDP methodology to the planted-graph model from Section~\ref{sec:problem}, and follows \cite{SchrammWein2022} (see also \cite{BKYLowDegreeSurvey}).

\subsubsection{Background and preliminaries}

\paragraph{Problem setup.} Fix $n$, $0<q_n<p_n\leq 1$, and a planted structure $\Gamma_n=(v(\Gamma_n),e(\Gamma_n))$ with $|v(\Gamma_n)|\le n$ and no isolated vertices. A planted copy $\Gamma_n^\star\in\mathcal{S}_{\Gamma_n}$ is selected uniformly at random and then observed through the binary edge--channel
\begin{align}
\left.\s{Y}_e \right| \Gamma_n^\star \sim \s{Bern}(\s{X}_e), 
\qquad
\s{X}_e=\begin{cases}
p_n, & e\in e(\Gamma_n^\star)\\
q_n, & e\notin e(\Gamma_n^\star),
\end{cases}
\end{align}
independently over $e\in \binom{[n]}{2}$. We write $N\triangleq\binom{n}{2}$ and view $\s{Y}\in\{0,1\}^{N}$ as the input to any estimator $\hat{\Gamma}:\{0,1\}^{N}\to \mathcal{S}_{\Gamma_n}$.

\paragraph{LDP and performance metrics.} We observe a random vector $\s{Y} \in \mathcal{Y}^{N}$ and seek to estimate an \emph{anchor} scalar parameter $x \in \mathbb{R}$, defined as a function of the latent signal generating $\s{Y}$ (e.g., one coordinate of the signal). For a degree budget $D \in \mathbb{N}$, let
\begin{align}
\mathbb{R}[\s{Y}]_{\leq D}\triangleq\left\{ f:\mathcal{Y}^{N}\to\mathbb{R}\; \s{polynomial}\;\s{of}\;\s{total}\;\s{degree}\;\s{at}\;\s{most}\;D \right\},
\end{align}
namely, the space of polynomials of degree at most $D$. The \emph{degree-$D$ minimum mean-squared error} is defined as
\begin{align}\label{eq:ldp-mmse-def}
\s{MMSE}_{\leq D}
\triangleq
\inf_{f \in \mathbb{R}[\s{Y}]_{\le D}}\mathbb{E}\left[(f(\s{Y})-x)^2\right],
\end{align}
where the expectation is taken with respect to the joint law of $(x,\s{Y})$. Furthermore, define the associated \emph{degree-$D$ maximum correlation} as
\begin{align}\label{eq:corr-def}
\s{Corr}_{\le D}\triangleq
\sup_{f\in\mathbb{R}[\s{Y}]_{\le D}}
\frac{\mathbb{E}\pp{f(\s{Y})\cdot x}}{\sqrt{\mathbb{E}\pp{f(\s{Y})^2}}}.
\end{align}
A basic identity links these two quantities:
\begin{align}\label{eq:mmse-def}
\s{MMSE}_{\le D} = \mathbb{E}[x^2]-\s{Corr}_{\le D}^{2}.
\end{align}
This identity allows us to lower bound $\s{MMSE}_{\le D}$ (and hence derive computational lower bounds) by upper bounding $\s{Corr}_{\le D}$. Polynomials of degree $D=\s{polylog}(n)$ capture a broad class of efficient estimators, including many spectral and AMP-type procedures via polynomial approximation, and the LDP framework has proved predictive for computational thresholds across a range of high-dimensional problems \cite{SchrammWein2022}. In particular, if the right-hand side of \eqref{eq:corr-def} is $o(1)$ for some $D=\s{polylog}(n)$, then $\s{Corr}_{\le D}=o(1)$ and hence $\s{MMSE}_{\le D}\geq \mathbb{E}[x^2]-o(1)$, ruling out any degree-$D$ polynomial from achieving nontrivial recovery of $x$. This yields concrete, model-specific computational lower bounds that often match (up to $\s{polylog}$ factors) the best known algorithmic thresholds in classical inference problems.

To study recovery via low-degree polynomials, we focus on a one-bit anchor that is symmetric across the vertices of the planted copy. Fix an \emph{ambient anchor vertex} $v^\star=\{1\}$ and define
\begin{align}\label{eq:x-edge-anchor}
x\triangleq\Ind\{v^\star\in v(\Gamma_n^\star)\}\in\{0,1\}.
\end{align}
Because $\Gamma_n^\star$ is drawn uniformly among all of its isomorphic copies, the joint law of $(x,\s{Y})$ does not depend on which ambient vertex is chosen, and the anchor is without loss of generality. %This edge--anchored choice matches the symmetry of general planted subgraphs (where vertex degrees may vary) and is the key distinction from the vertex--anchored targets used in \cite{SchrammWein2022}.

\paragraph{Binary observation model.} The planted graph problem in this paper fits the general binary model considered in \cite[Sec. 2.3]{SchrammWein2022}, where coordinates of $\s{Y}$ are conditionally independent Bernoulli variables with means $\s{X}=(\s{X}_i)_{i\in[N]}$ taking values in $\{\tau_0,\tau_1\}$ with $0<\tau_0<\tau_1<1$. In this setting, \cite[Thm. 2.7]{SchrammWein2022} derived an upper bound on $\s{Corr}_{\le D}$ expressed in terms of \emph{joint cumulants}:
\begin{align}\label{eq:sw-master}
\s{Corr}_{\le D}^{2}
\le
\sum_{\substack{\alpha\subseteq\binom{[n]}{2}\\ |\alpha|\le D}}
\frac{\kappa_\alpha^{  2}}{\big(q_n(1-p_n)\big)^{|\alpha|}},
\end{align}
where $\kappa_\alpha=\kappa(x,\{\s{X}_e\}_{e\in\alpha})$ is the joint cumulant of $x$ and $\{\s{X}_e\}_{e\in\alpha}$. It is shown in \cite[Thm. 2.7]{SchrammWein2022} that the cumulants $\{\kappa_\alpha\}$ are defined recursively as
\begin{align}\label{eq:cumulant-recursion}
\kappa_\alpha
&=\mathbb{E}[x\cdot\s{X}^\alpha]
-\sum_{0\leq \beta\prec\alpha}
\kappa_\beta\mathbb{E}[\s{X}^{\alpha-\beta}],\\
\s{X}^\alpha&\triangleq\prod_{e:\;\alpha_e=1}\s{X}_e,
\end{align}
with $\beta\prec\alpha$ meaning coordinate--wise inequality and $\beta\neq\alpha$.

\subsubsection{Main result}

The following result shows that $O(\log n)$-degree polynomials fail at recovery whenever the graph-density satisfies $\eta(\Gamma_n)=\frac{|e(\Gamma_n)|}{|v(\Gamma_n)|}\ll\sqrt{n}$. For reference, the trivial estimator $f(\s{Y})=\bE[x]$ attains mean-squared error $\bE[f(\s{Y})-x]^2=\s{Var}(x)$. The proof is given in Section~\ref{sec:complowerbounds}.
\begin{theorem}[Computational lower bound]\label{thm:compLower}
Fix a sequence of subgraphs $\Gamma=(\Gamma_n)_n$, and assume $p_n,q_n=\Theta(1)$. If 
\begin{align}
    \eta(\Gamma_n) \leq n^{\frac{1}{2}-\varepsilon},
\end{align}
for any fixed $\varepsilon>0$, and $D=D_n$ scales as $D\leq(\log n)^{\alpha}$ for some fixed $\alpha<1$, then
\begin{align}
    \s{MMSE}_{\leq D}\geq (1-o(1))\cdot\s{Var}(x).
\end{align}
\end{theorem}

Next, we prove upper bounds on $\s{MMSE}_{\leq D}$. To that end, following \cite[Sec. 4.2]{SchrammWein2022}, we analyze one and multiple rounds of the power iteration method starting from the all-ones vector, followed by thresholding. We now describe this estimator in more detail. Fix the number of iterations $L\in\mathbb{N}$. For simplicity of notation, let $\s{Z}_{ij}\triangleq\s{Y}_{ij}-q$, for any $i,j\in[n]$. Let $\calP_L$ denote the set of all simple undirected paths of length $L$ in the ambient complete graph starting at vertex $1$ and visiting pairwise distinct vertices, namely, $P=(u_0,u_1,\dots,u_\ell)$ with $u_0=1$. For each $P\in\calP_L$, define
\begin{align}
    \s{Z}(P)\triangleq\prod_{\ell=0}^{L-1} \s{Z}_{u_t u_{t+1}},
\end{align}
and the degree-$L$ \emph{walk polynomial}
\begin{equation}
\s{W}_L\triangleq\sum_{P\in\mathcal{P}_L}\s{Z}(P).
\end{equation}
For $u\in v(\Gamma)$ let
\begin{align}
W_L(\Gamma;u)&\triangleq\abs{\ppp{\s{simple}\;\s{paths}\;\s{of}\;\s{length}\;L\;\s{in}\;\Gamma\;\s{ starting}\;\s{at}\;u}},\\
    W_L^{\min}(\Gamma)&\triangleq\min_{u\in v(\Gamma)} W_L(\Gamma;u).
\end{align}
Define the rescaled statistic 
\begin{align}
\mathscr{Z}_L\triangleq \frac{2}{(p-q)^L}\frac{\s{W}_L}{W_L^{\min}(\Gamma)}.    
\end{align}
Our estimator is defined as
\begin{align}
f_{L,m}(\s{Y})=\tau_m(\mathscr{Z}_L),\label{eqn:EstL}
\end{align}
where $\tau_m$ of degree $D=2m+1$ is a \emph{polynomial threshold} that approximates the threshold function using a polynomial of degree $2m+1$ (see Lemma~\ref{lem:poly-approx}). In the special case of a single-iteration power method, the estimator above reduces to
\begin{align}
    f_{1,m}(\s{Y}) = \tau_k\p{\frac{1}{(p-q)\eta(\Gamma)}\sum_{i=2}^n\p{\s{Y}_{1i}-q}}.\label{eqn:EstL1}
\end{align}
We have the following result.
\begin{theorem}[LDP upper bound]\label{thm:compUpper}
Assume $p_n,q_n=\Theta(1)$. 
\begin{enumerate}
    \item Single iteration: consider the estimator in \eqref{eqn:EstL1}. Fix $\epsilon>0$, and let $\Gamma_n$ be any sequence of subgraphs with
\begin{align}
\s{Dis}(\Gamma)\triangleq\max_{v\in v(\Gamma)} \frac{|d_\Gamma(v)-\eta(\Gamma)|}{\eta(\Gamma)}&\leq\frac{r}{12},\label{eqn:condDstar1Main}\\
\eta(\Gamma_n)&\geq n^{\frac12+\epsilon},\label{eqn:condDstar2Main}
\end{align}
for some fixed $0<r<1$ and all sufficiently large $n$. If $D=D(n)\leq (\log n)^{\alpha}$ for any fixed $\alpha>0$, then
\begin{align}
\bE\pp{(f_{1,m}(\s{Y})-x)^2}\leq CD^2r^{D-1},
\end{align}
for some $C>0$.

\item Multiple iteration: consider the estimator in \eqref{eqn:EstL}. Fix $L\in\mathbb{N}$ and let $r\in(0,1/4]$. Let $\Gamma_n$ be any sequence of subgraphs with
\begin{equation}
W_L^{\min}(\Gamma)\geq C^\star(L,p,q)\pp{n^{L/2}+k^{L-1/2}}\sqrt{\log n},\label{eqn:condDstar3Main}
\end{equation}
for some $C^\star(L,p,q)>0$. Assume that $m=\omega(1)$ and $D\le C\log\log n$, for some constant $C>0$. Then
\begin{align}
\bE\pp{(f_{L,m}(\s{Y})-x)^2}\leq (\log n)^{-\Omega(1)}.
\end{align}
\end{enumerate}
\end{theorem}

We see that for ``almost-regular'' structures satisfying \eqref{eqn:condDstar1Main}, the single-iteration bound complements the computational lower bound in Theorem~\ref{thm:compLower}. However, there are natural examples for which \eqref{eqn:condDstar1Main} fails. We next discuss a few such examples to further illustrate the LDP-based lower and upper bounds.

\begin{example}[Planted clique]
Consider the case where $\Gamma_n = \calK_k$. Since $\frac{|e(\Gamma_n)|}{|v(\Gamma_n)|} = \frac{|v(\Gamma_n)| - 1}{2}$, Theorem~\ref{thm:compLower} implies that recovery is computationally hard whenever $|v(\Gamma_n)| \leq n^{\frac{1}{2} - \varepsilon}$, matching the state-of-the-art algorithmic threshold (see Example~\ref{exmp:1}) and consistent with the folklore planted clique conjecture. Moreover, the first item of Theorem~\ref{thm:compUpper} shows that a single iteration of the power method, followed by thresholding, achieves this computational barrier: the regularity condition in \eqref{eqn:condDstar1Main} is clearly satisfied, and \eqref{eqn:condDstar2Main} reduces to $|v(\Gamma_n)| \geq n^{\frac{1}{2} + \varepsilon}$. Finally, we note that the multi-iteration bound is also applicable. In this case, $W_L(\Gamma;u)=(k-1)_{L}\sim k^L$, and hence Theorem~\ref{thm:compUpper} implies that recovery is possible whenever $k\gtrsim \sqrt n(\log n)^{1/(2L)}$.
\end{example}

\begin{example}[Complete bipartite clique]
Consider the case where $\Gamma_n=\calK_{k_{\s{L}},k_{\s{R}}}$ is a complete bipartite
graph with partitions of size $k_{\s{L}}$ and $k_{\s{R}}$, and define $k=k_{\s{L}}+k_{\s{R}}$. Here $\frac{|e(\Gamma_n)|}{|v(\Gamma_n)|} = \frac{k_{\s{L}}\cdot k_{\s{R}}}{k_{\s{L}}+k_{\s{R}}}\in\pp{\min\{k_{\s{L}},k_{\s{R}}\}/2,\min\{k_{\s{L}},k_{\s{R}}\}}$. Accordingly, Theorem~\ref{thm:compLower} implies that recovery is computationally hard whenever $\min\{k_{\s{L}},k_{\s{R}}\}\leq n^{\frac{1}{2}-\varepsilon}$, matching the state-of-the-art algorithmic threshold (see Example~\ref{exmp:bip}). As in the previous example, the first item of Theorem~\ref{thm:compUpper} shows that a single iteration of the power method, followed by thresholding, achieves this computational barrier: the regularity condition in \eqref{eqn:condDstar1Main} is clearly satisfied, and \eqref{eqn:condDstar2Main} reduces to $\min\{k_{\s{L}},k_{\s{R}}\}\leq n^{\frac{1}{2}-\varepsilon}$. Finally, we note that the multi-iteration bound is also applicable. In this case, $W_1(\Gamma;u)=\min\{k_{\s{L}},k_{\s{R}}\}$, and hence Theorem~\ref{thm:compUpper} implies that recovery is possible under the same condition.
\end{example}

\begin{example}[Half clique half independent set]
Split $v(\Gamma)$ into two sets $\calA$ and $\calB$ with $|\calA| = |\calB| = k/2$: $\calA$ forms a clique, $\calB$ is an independent set, and each $b \in \calB$ is adjacent to exactly $\sqrt{k}$ vertices in $\calA$. Accordingly, Theorem~\ref{thm:compLower} implies that recovery is computationally hard whenever $k \leq n^{\frac{1}{2}-\varepsilon}$. On the other hand, the first item of Theorem~\ref{thm:compUpper} does not apply in this example because \eqref{eqn:condDstar1Main} is not satisfied. Nonetheless, for $u \in \calB$ we have
$W_L(\Gamma;u) \asymp \sqrt{k}\,(k/2)^{L-1} \asymp k^{L-\tfrac{1}{2}}$, and therefore the second item of Theorem~\ref{thm:compUpper} implies that recovery is possible whenever $k \gtrsim n^{\tfrac{L}{2L-1}} (\log n)^{\tfrac{1}{2L-1}}$, which complements our recovery lower bound for any $L = \omega(1)$.
\end{example}

\subsection{Extensions}\label{subsec:exten}

In this subsection, we present several direct extensions of the results above. We begin by showing that both the optimal and the efficient algorithms we introduce are robust to a simple monotone adversary, which is allowed to remove arbitrary edges outside the planted subgraph and add edges within it. We then develop guarantees for weaker notions of recovery.

\subsubsection{Semi-random model}

We now consider a framework for the \emph{recovery} problem in a semi-random graph setting known as the \emph{monotone adversary} or ``Sandwich Model'' introduced in \cite{feige2000finding}. As before, let $\Gamma = (\Gamma_n)_n$ be a sequence of graphs with $|v(\Gamma_n)| \leq n$. To exclude trivial settings, we assume that $\Gamma_n$ contains no isolated vertices. The recovery task is to exactly identify a hidden subgraph embedded in a noisy background graph under semi-random perturbations.

We assume the data is generated according to the following semi-random process. First, we pick a subgraph $\Gamma^\star\in\calS_\Gamma$. A random graph $\s{G}$ is then formed by retaining each edge of $\Gamma^\star$ independently with probability $p$, and adding each edge outside $\Gamma^\star$ independently with probability $q$. The resulting distribution over graphs is denoted $\calG_{\Gamma_n}(n,p,q)$. An adversary may then modify $\s{G}$ by removing edges that do not belong to the planted structure $\Gamma^\star$, and by adding edges within $\Gamma^\star$. The final observed graph is denoted by $\s{G}_{\adv}$. We model this perturbation via a (possibly randomized) function $\adv$ acting on both $\s{G}$ and $\Gamma^\star$, such that $\s{G}_{\adv} = \adv(\s{G},\Gamma^\star) \in \{0,1\}^{\binom{n}{2}}$. The adversary satisfies
\begin{align}
    \P\pp{\bigcap_{(i,j)\in\binom{[n]}{2}\setminus e(\Gamma^\star)}\s{A}_{\s{G}_{\adv}}(i,j)\leq \s{A}_{\s{G}}(i,j),\bigcap_{(i,j)\in e(\Gamma^\star)}\s{A}_{\s{G}_{\adv}}(i,j)\geq \s{A}_{\s{G}}(i,j)}=1.\label{eq:recovery_adversary}
\end{align}
We denote by $\calA$ the collection of all such valid (possibly randomized) functions $\adv$ satisfying the condition above. When randomized, these functions define conditional distributions over graphs. Thus, the observed graph is drawn from a distribution in the family $\s{Adv}(\calG_{\Gamma_n}(n,p,q))$, induced by applying some $\adv\in \calA$ to a sample from the planted model. Our goal is to design a recovery algorithm that identifies the underlying subgraph $\Gamma^\star$ exactly.

A \emph{recovery algorithm} is a function $\hat{\Gamma}_n: \{0,1\}^{\binom{n}{2}} \to \calS_\Gamma$ that maps an observed graph to a candidate subgraph. The (worst-case) error of such an algorithm is given by
\begin{align}
    \s{E}_{\s{adv}}(\hat{\Gamma}) \triangleq \sup_{\Gamma^\star\in\calS_\Gamma}\sup_{\adv_1 \in \calA_1} \P[\hat{\Gamma}(\s{G}_{\adv_1}) \neq \Gamma^\star].
\end{align}
We define the worst-case optimal error as
\begin{align}
    \s{E}_{\s{adv}}^{\star} \triangleq \inf_{\hat{\Gamma}:\{0,1\}^{n^{(2)}} \to \calS_\Gamma} \s{E}_{\s{adv}}(\phi).
\end{align}
As before, we say that \textit{exact recovery} is possible if $\limsup_{n \to \infty} \s{E}_{\s{adv}}^{\star} = 0$, and impossible otherwise.

It is straightforward to show that the optimal MLE in the vanilla setting (i.e., without an adversary) is robust to the adversarial model described above. This is summarized in the following result.
\begin{theorem}[MLE is robust]\label{thm:MLERobust}
    Let $\hat{\Gamma}_{\s{MLE}}$ denote the MLE for the recovery problem in Section~\ref{sec:problem}. Then, whenever $\hat{\Gamma}_{\s{MLE}}=\Gamma^\star$ is the unique solution to the recovery problem under $\calG_{\Gamma_n}(n,p,q)$ (i.e., in the absence of an adversary), it remains the unique solution under the adversarial mode $\s{Adv}(\calG_{\Gamma_n}(n,p,q))$.
\end{theorem}
In other words, $\s{E}_{\s{adv}}(\hat{\Gamma}_{\s{MLE}})\to0$ whenever $\s{E}(\hat{\Gamma}_{\s{MLE}})\to0$. In light of the fact that $\s{E}^{\star}\leq\s{E}_{\s{adv}}^{\star}$, we see that, at least statistically, the monotone adversary does not alter the performance. 
\begin{proof}[Proof of Theorem~\ref{thm:MLERobust}]
    Recall that the MLE is given by
    \begin{align}
        \hat{\Gamma}_{\s{MLE}}(\s{G}) = \arg\max_{\Gamma'\in\calS_{\Gamma}}\abs{e(\Gamma')\cap\s{G}}.
    \end{align}
    Denote the corresponding $\s{G}$ after modification as $\s{G}_{\s{adv}}$. Then, for all feasible $\Gamma\neq\Gamma^\star$, we have the following chain of inequalities:
    \begin{align}
        \abs{e(\Gamma)\cap\s{G}_{\s{adv}}} &= \abs{e(\Gamma\cap(\Gamma^\star)^c)\cap\s{G}_{\s{adv}}}+\abs{e(\Gamma\cap\Gamma^\star)\cap\s{G}_{\s{adv}}}\\
        &\leq\abs{e(\Gamma\cap(\Gamma^\star)^c)\cap\s{G}}+\abs{e(\Gamma\cap\Gamma^\star)\cap\s{G}_{\s{adv}}}\\
        & = \abs{e(\Gamma)\cap\s{G}}+\abs{e(\Gamma\cap\Gamma^\star)\cap\s{G}_{\s{adv}}}-\abs{e(\Gamma\cap\Gamma^\star)\cap\s{G}}\\
        &<\abs{e(\Gamma^\star)\cap\s{G}}+\abs{e(\Gamma\cap\Gamma^\star)\cap\s{G}_{\s{adv}}}-\abs{e(\Gamma\cap\Gamma^\star)\cap\s{G}}\\
        & = \abs{e(\Gamma^c\cap\Gamma^\star)\cap\s{G}}+\abs{e(\Gamma\cap\Gamma^\star)\cap\s{G}_{\s{adv}}}\\
        & \leq  \abs{e(\Gamma^c\cap\Gamma^\star)\cap\s{G}_{\s{adv}}}+\abs{e(\Gamma\cap\Gamma^\star)\cap\s{G}_{\s{adv}}}\\
        & = \abs{e(\Gamma^\star)\cap\s{G}_{\s{adv}}},
    \end{align}
    where the second inequality follows from the fact that outside the planted subgraph (i.e., in  $\Gamma\cap(\Gamma^\star)^c$), the monotone adversary can only remove edges. Therefore, replacing $\s{G}_{\s{adv}}$ with $\s{G}$ can only increase the intersection. The second \emph{strict} inequality holds due to the assumption that $\Gamma^\star$ is the unique global maximizer of the original problem in the absence of an adversary. The third inequality follows from the fact that within the planted subgraph (i.e., in $\Gamma^c\cap\Gamma^\star$), the adversary can only add edges, so replacing $\s{G}$ with $\s{G}_{\s{adv}}$ can only increase the intersection. Hence, $\Gamma^\star$ stays optimal even when the graph is modified by the adversary.
\end{proof}

Next, we show that the computationally efficient algorithm proposed in \eqref{eqn:convex} is also robust to a monotone adversary. 
\begin{theorem}[Convex program is robust]\label{thm:ConvRobust}
        Let $\hat{\s{X}}^{(\alpha)}_{\s{con}}$ denote the convex program in \eqref{eqn:convex}. Then, whenever $\Gamma^\star$ is the unique solution of \eqref{eqn:convex} to the recovery problem under $\calG_{\Gamma_n}(n,p,q)$ (i.e., in the absence of an adversary), it remains the unique solution under the adversarial mode $\s{Adv}(\calG_{\Gamma_n}(n,p,q))$.
\end{theorem}
\begin{proof}[Proof of Theorem~\ref{thm:ConvRobust}]
    The proof follows the same ideas as in the proof of Theorem~\ref{thm:MLERobust}. Denote the corresponding $\s{W}$ after modification as $\s{W}_{\s{adv}}$. Also, recall that $\s{X}^\star$ denotes the adjacency matrix of the planted subgraph $\Gamma^\star$. Let $\s{X}^\star_{\s{comp}}$ denote the adjacency matrix of $(\Gamma^\star)^c$, and for any feasible $\s{X}$ let $\s{X}_{\s{comp}}\triangleq\mathbf{J}-\s{X}$. Finally, recall that $\odot$ denotes the element-wise Hadamard product. Then, for all feasible $\s{X}\neq\s{X}^\star$, we have the following chain of inequalities:
    \begin{align}
        \innerP{\s{W}_{\s{adv}},\s{X}} &= \innerP{\s{W}_{\s{adv}},\s{X}\odot\s{X}^\star_{\s{comp}}}+\innerP{\s{W}_{\s{adv}},\s{X}\odot\s{X}^\star}\\
        &\leq\innerP{\s{W},\s{X}\odot\s{X}^\star_{\s{comp}}}+\innerP{\s{W}_{\s{adv}},\s{X}\odot\s{X}^\star}\\
        & = \innerP{\s{W},\s{X}}+\innerP{\s{W}_{\s{adv}}-\s{W},\s{X}\odot\s{X}^\star}\\
        &<\innerP{\s{W},\s{X}^\star}+\innerP{\s{W}_{\s{adv}}-\s{W},\s{X}\odot\s{X}^\star}\\
        & = \innerP{\s{W},\s{X}^\star\odot\s{X}_{\s{comp}}}+\innerP{\s{W}_{\s{adv}},\s{X}\odot\s{X}^\star}\\
        & \leq  \innerP{\s{W}_{\s{adv}},\s{X}^\star\odot\s{X}_{\s{comp}}}+\innerP{\s{W}_{\s{adv}},\s{X}\odot\s{X}^\star}\\
        & = \innerP{\s{W}_{\s{adv}},\s{X}^\star},
    \end{align}
    where the second inequality arises because outside the planted subgraph (that is, over the region $\s{X} \odot \s{X}^\star_{\s{comp}}$), the monotone adversary is limited to deleting edges. As a result, substituting $\s{W}_{\s{adv}}$ with $\s{W}$ can only lead to a larger objective value over this region. The second strict inequality follows from the premise that $\s{X}^\star$ is the unique global maximizer of the original problem when no adversary is present. The third inequality holds because within the planted subgraph (specifically, over $\s{X}^\star \odot \s{X}_{\s{comp}}$), the adversary is only allowed to insert edges. Consequently, replacing $\s{W}$ with $\s{W}_{\s{adv}}$ can only increase the objective contribution over this region, implying that $\s{X}^\star$ remains the optimal solution even under adversarial modifications to the graph.
\end{proof}

\subsubsection{Individual layers recovery} 

The strict notion of exact recovery precludes recovering subgraphs with very sparse layers (e.g., as in Example~\ref{examp:kite}). One way to bypass this inherent limitation is to consider weaker notions of recovery. For instance, one might be interested in exactly recovering only the subgraph $\s{H} \subseteq \Gamma$ that achieves the maximum subgraph density of $\Gamma$. More generally, for any $\kappa \in [M(\Gamma)]$---where $M(\Gamma)$ denotes the number of components in the onion decomposition of $\Gamma$, i.e., $\Gamma = \bigcup_{\ell=1}^{M}\calD^{(\ell)}$ (see Definition~\ref{def:onionDec})---we could aim to exactly recover the first $\kappa$ layers of $\Gamma$, namely $\Gamma^{(\kappa)}= \bigcup_{\ell=1}^{\kappa}\calD^{(\ell)}$. Accordingly, we define the worst-case error probability associated with an estimator $\hat{\Gamma}$ as follows:
\begin{align}
    \s{E}_n^{(\kappa)}(\hat{\Gamma})\triangleq\sup_{\Gamma^\star\in\calS_\Gamma}\pr_{\mathcal{G}_{\Gamma_n^\star}(n,p_n,q_n)}[\hat{\Gamma}(\s{G})\neq\Gamma^{(\kappa),\star}],
\end{align}
and the optimal error probability as $\s{E}_n^{(\kappa),\star} = \inf_{\hat{\Gamma}}\s{E}_n^{(\kappa)}(\hat{\Gamma})$. We say that a sequence of estimators $(\hat{\Gamma}_n)_{n \in \mathbb{N}}$ achieves \emph{$\kappa$-exact recovery} if $\limsup_{n\to\infty}\s{E}_n^{(\kappa)}(\hat{\Gamma}_n)= 0$. Conversely, we say that \emph{$\kappa$-exact recovery is impossible} if $\liminf_{n \to \infty} \s{E}_n^{(\kappa),\star}>0$. Finally, we define the \emph{truncated likelihood peeling algorithm} as
\begin{align}
    \hat{\Gamma}^{(\kappa)}_{\s{pMLE}} = \hat{\Gamma}^{(\kappa)}_{\s{MLE}},\label{eqn:PeelMLEKappa}
\end{align}
where $\{\hat{\Gamma}^{(\ell)}_{\s{MLE}}\}_{\ell\geq1}$ is defined in \eqref{eqn:PeelMLE0}. Recall that $\mu(\Gamma\vert\Gamma^{(\kappa-1)})=\eta(\Gamma^{(\kappa)}\vert\Gamma^{(\kappa-1)})$. We have the following result.
\begin{theorem}[$\kappa$-exact recovery]\label{thm:kappaExactRec}
Fix a sequence of subgraphs $\Gamma=(\Gamma_n)_n$, $\kappa\in[M(\Gamma)]$, and assume $p_n,q_n=\Theta(1)$. Consider the onion decomposition of $\Gamma$ in Definition~\ref{def:onionDec}.
\begin{enumerate}
    \item $\kappa$-exact recovery is statistically impossible if
\begin{align}
        \mu(\Gamma\vert\Gamma^{(\kappa-1)})\leq \frac{(1-\varepsilon)\cdot \log n}{d_{\s{KL}}(p||q)},\label{eq:condDenseStatKappa}
\end{align}
for any $\varepsilon>0$.
\item $\kappa$-exact recovery of $\Gamma^\star$, via the truncated likelihood peeling algorithm in \eqref{eqn:PeelMLEKappa}, is possible if
\begin{align}
    \mu(\Gamma\vert\Gamma^{(\kappa-1)})\geq \s{C}\frac{(1+\varepsilon)\cdot \log n}{d_{\s{KL}}(p||q)},\label{eqn:MLEcondKappa}
\end{align}
for any $\varepsilon>0$ and some constant $\s{C}>0$.
\end{enumerate}
\end{theorem}
Going back to Example~\ref{examp:kite}, if we take $\kappa = 1$, this corresponds to recovering the clique while ignoring the single edge attached to it (which precluded strong recovery). Since $\calD^{(1)} = K_k$, Theorem~\ref{thm:kappaExactRec} tells us that $1$-exact recovery is possible (or impossible) when $k \gtrless C \log n$ for some $C > 0$, just as in the planted clique recovery problem. Finally, the proof of Theorem~\ref{thm:kappaExactRec} follows verbatim from the arguments used in the proofs of Theorems~\ref{thm:lowerp1} and~\ref{thm:MLE}, with the recovery restricted to $\Gamma^{(\kappa)}$. The details are therefore omitted.

\subsubsection{Almost-exact recovery} 

Next, we turn to the notion of \emph{almost-exact recovery}, defined as follows (see, e.g., \cite{hajek2016information,wu2018statistical}).
\begin{definition}[Almost-exact recovery]
An estimator $\hat{\Gamma}$ almost-exactly recovers $\Gamma^\star$ if, as $n \to \infty$, $\s{d}_{\s{H}}(\hat{\Gamma},\Gamma^\star)/|e(\Gamma)| \to 0$ in probability, where $\s{d}_{\s{H}}$ denotes the Hamming distance between the adjacency matrices of $\hat{\Gamma}$ and $\Gamma^\star$. 
\end{definition}

In general, lower bounds for exact recovery do not directly imply lower bounds for almost-exact recovery, but any estimator that achieves exact recovery also achieves almost-exact recovery. To state our main results, we first introduce some notation. Consider the onion decomposition of $\Gamma_n$ in Definition~\ref{def:onionDec}, and define for $\ell\in0\cup[M(\Gamma)]$ and $n\in\mathbb{N}$ the $\ell^{\s{th}}$ \emph{leftover-edge fraction} as
\begin{align}
\s{Res}_\ell^{(n)}\triangleq\frac{|e(\Gamma_n\setminus \Gamma_n^{(\ell)})|}{|e(\Gamma_n)|}.
\end{align}
Fix any null sequence $\varepsilon_n\downarrow 0$ and define
\begin{align}
\ell_{\s{LB}}(n)&\triangleq\max\{\ell\in0\cup[M(\Gamma)]:\; \s{Res}_\ell^{(n)}>\varepsilon_n\},\label{eqn:ellLB}\\
\ell_{\s{UB}}(n)&\triangleq\min\{\ell\in0\cup[M(\Gamma)]:\; \s{Res}_\ell^{(n)}\le \varepsilon_n\}.\label{eqn:ellUB}
\end{align}
Heuristically, $\ell_{\s{LB}}$ is the last index for which the leftover-edge fraction is $\Omega(1)$, whereas $\ell_{\s{UB}}$ is the first index for which the leftover-edge fraction is $o(1)$. Since $\s{Res}_\ell^{(n)}$ is nonincreasing, it follows that $\ell_{\s{UB}}(n) = \ell_{\s{LB}}(n)+1$. We nevertheless introduce both notations for clarity.

\paragraph{Lower bounds.} We begin with our impossibility results. For almost-exact recovery, an interesting dichotomy emerges. Roughly speaking, if the last non-negligible leftover component has sub-logarithmic maximum subgraph relative density, then almost-exact recovery is impossible. On the other hand, when this relative density is super-logarithmic, almost-exact recovery remains impossible whenever it is at most logarithmic in $n$. This dichotomy is summarized in the following result, proved in Appendix~\ref{app:weakRec}.
\begin{theorem}\label{thm:subLogAlmostExactMain}
Fix a sequence of subgraphs $\Gamma=(\Gamma_n)_n$, assume $p_n,q_n=\Theta(1)$, and consider the onion decomposition of $\Gamma$ in Definition~\ref{def:onionDec}. 
\begin{enumerate}
    \item If $|v(\Gamma\setminus\Gamma^{(\ell_{\s{LB}})})|=o(n)$ and
\begin{align}
\mu(\Gamma\vert\Gamma^{(\ell_{\s{LB}})})=o\p{\frac{\log|v(\Gamma\setminus\Gamma^{(\ell_{\s{LB}})})|}{\log \log|v(\Gamma\setminus\Gamma^{(\ell_{\s{LB}})})|}},\label{eq:condCopyLBMain}
\end{align}
then almost-exact recovery is impossible. 
\item If $\mu(\Gamma\vert\Gamma^{(\ell_{\s{LB}})})\geq \alpha_n\cdot \log|v(\Gamma\vert\Gamma^{(\ell_{\s{LB}})})|$, for some $\alpha_n=\Omega(1)$, then there exists a constant $\underline{C}>0$ such that almost-exact recovery is impossible if
    \begin{align}
        \mu(\Gamma\vert\Gamma^{(\ell_{\s{LB}})})\leq \underline{C}\cdot \log n.\label{eq:condDenseStatweakMain}
    \end{align} 
\end{enumerate}
\end{theorem}

A few remarks are in order. First, the proofs of \eqref{eq:condCopyLBMain} and \eqref{eq:condDenseStatweakMain} rely on different techniques, each tailored to its own regime of relative density and typically suboptimal in the other. In both cases, the first step is to reduce almost-exact recovery for $\Gamma$ to recovering the last non-negligible $\Omega(1)$ leftover-edge fraction of $\Gamma$, namely, $\Gamma \setminus \Gamma^{(\ell_{\s{LB}})}$. 

To prove \eqref{eq:condCopyLBMain}, one of the main ingredients is a generalization of the subgraph expectation threshold \cite{kahn2007thresholds}, and more specifically, of the modified subgraph expectation threshold studied in \cite{mossel2022second}, which analyzes the threshold for the appearance of \emph{any} isomorphic copy of $\Gamma$ in $\s{G} \sim \calG(n,q)$. For our purposes, however, not all copies are admissible. Indeed, recall that the recovery problem here is supplied with $\Gamma^{(\ell_{\s{LB}})}$ and is tasked with finding $\Gamma$. Thus, the admissible copies are precisely those extending $\Gamma^{(\ell_{\s{LB}})}$, i.e., the copies contained in $\calM(\Gamma^{(\ell_{\s{LB}})},\Gamma)$. To handle this, we derive a generalization of the modified subgraph expectation threshold that accounts for the appearance of such constrained copies. 

To prove \eqref{eq:condDenseStatweakMain}, we establish a connection between almost-exact recovery and a hypothesis-testing variant of the problem (see \eqref{eqn:super_hypo}), and then leverage the known impossibility result \eqref{eqn:Dec2IMP} for the latter. Most notably, as discussed right after \eqref{eqn:Dec1IMP}--\eqref{eqn:Dec2IMP}, in the sub-logarithmic regime there exists a region where detection is statistically possible while, as we show above, almost-exact recovery is always impossible. This highlights why hypothesis-testing--based bounds are insufficient in this regime.

Let us now explain the scope of Theorem~\ref{thm:subLogAlmostExactMain}. Comparing Theorems~\ref{thm:lowerp1} and~\ref{thm:subLogAlmostExactMain}, we see that if $\Gamma$ is balanced and has super-logarithmic maximum density, i.e., $\mu(\Gamma_n) = \Omega(\log |v(\Gamma_n)|)$, then the two bounds coincide (up to constant factors). In this case, almost-exact and exact recovery are statistically equivalent. In general, however, the bounds need not coincide. As noted after Theorem~\ref{thm:lowerp1} (see, e.g., Example~\ref{examp:kite}), the strict exact-recovery criterion rules out recovering graphs with a dense core and a very sparse appendage, whereas almost-exact recovery may remain possible since one may effectively ``discard'' the problematic appendage, provided its contribution is $o(|e(\Gamma_n)|)$. We illustrate this phenomenon with two examples.

\begin{example}[A kite]
Recall Example~\ref{examp:kite}, where $\Gamma$ is a kite on $k+1$ vertices---namely, a clique with $k = |v(\calK_k)|$ vertices, with one of its vertices connected by an edge to an additional vertex. In this case, $\ell_{\s{LB}} = 0$, so $\Gamma^{(\ell_{\s{LB}})} = \emptyset$ and the core is a clique on $k$ vertices. Thus, $\Gamma\setminus\Gamma^{(\ell_{\s{LB}})} = \Gamma$ and $\mu(\Gamma\vert\Gamma^{(\ell_{\s{LB}})}) = \mu(\Gamma) = \mu(\calK_k)=\frac{k-1}{2}$. Hence, almost-exact recovery is impossible precisely when recovery of a $k$-clique is impossible, namely for $k \leq C \log n$ for some constant $C > 0$, in agreement with folklore results. Therefore, while the dangling edge precludes exact recovery, it has no effect on almost-exact recovery.
\end{example}

\begin{example}[A kite graph with a long tail]
Consider the case where $\Gamma$ is a clique $\calK_k$ on $k = |v(\calK_k)|$ vertices, with one of its vertices attached to a path of length $k^2$ (i.e., $k^2-1$ edges and $k^2$ vertices). In this case, $\Gamma^{(1)} = \calK_k$ and $\Gamma^{(2)} = \calP_{k^2}$ (a path on $k^2$ vertices). Accordingly, we have $\ell_{\s{LB}} = 1$, so $\Gamma^{(\ell_{\s{LB}})} = \calK_k$. Thus, $\Gamma\setminus\Gamma^{(\ell_{\s{LB}})} = \calP_{k^2}$ and $\mu(\Gamma\vert\Gamma^{(\ell_{\s{LB}})}) = \frac{|e(\calP_{k^2})|}{k+k^2-k-1} = \frac{k^2-1}{k^2-1} = 1$. Hence, we are in the sub-logarithmic regime of \eqref{eq:condCopyLBMain}, which implies that almost-exact recovery is impossible, in agreement with the exact-recovery criterion. We see that a sparse appendage may still deteriorate almost-exact recovery when its ``size'' is comparable to that of the dense core. Notably, if the path above had length $o(k^2)$, then almost-exact recovery would be impossible only when recovery of the clique itself is impossible (whereas exact recovery remains impossible).
\end{example}

We now turn to achievability results and show that the lower bounds above are tight.

\paragraph{Upper bounds.} Recall the truncated likelihood peeling algorithm in \eqref{eqn:PeelMLEKappa}. For almost-exact recovery, we propose
\begin{align}
    \hat{\Gamma}^{\s{almost}}_{\s{pMLE}} \triangleq \hat{\Gamma}^{(\ell_{\s{UB}})}_{\s{MLE}},\label{eqn:PeelMLEAlmostExact}
\end{align}
where $\{\hat{\Gamma}^{(\ell)}_{\s{MLE}}\}_{\ell\geq1}$ is defined in \eqref{eqn:PeelMLE0}, and $\ell_{\s{UB}}$ is defined in \eqref{eqn:ellUB}. Indeed, \eqref{eqn:PeelMLEAlmostExact} aims to recover all but the last layers whose leftover-edge fraction is $o(1)$, since these contribute negligibly to the Hamming error. We have the following result.
\begin{theorem}\label{thm:AlmostExactRec}
Almost-exact recovery of $\Gamma^\star$, via the truncated likelihood peeling algorithm in \eqref{eqn:PeelMLEAlmostExact}, is possible if
\begin{align}
    \mu(\Gamma\vert\Gamma^{(\ell_{\s{UB}}-1)})\geq \s{C}\frac{(1+\varepsilon)\cdot \log n}{d_{\s{KL}}(p||q)},\label{eqn:MLEcondAlmost}
\end{align}
for any $\varepsilon>0$ and some constant $\s{C}>0$.
\end{theorem}
Since $\ell_{\s{UB}} = \ell_{\s{LB}}+1$, we see that Theorem~\ref{thm:AlmostExactRec} matches the lower bound in \eqref{eq:condDenseStatweakMain} up to a constant factor. Finally, the proof of Theorem~\ref{thm:AlmostExactRec} follows verbatim from the arguments used in the proof of Theorem~\ref{thm:MLE}, with the recovery restricted to $\Gamma^{(\ell_{\s{UB}})}$. The details are therefore omitted.

\section{Statistical Lower Bound}\label{sec:lowerbounds}

In this subsection, we prove Theorem~\ref{thm:lowerp1}. As it turns out, this result can be established using two different approaches. The first relies on information-theoretic tools, while the second is based on a second-moment Bayes risk analysis. Here, we present the first approach and relegate the second to Appendix~\ref{app:secProofofITLOWEr}.

We adopt an information-theoretic strategy based on Fano's inequality combined with a genie-aided argument. Specifically, we lower bound the optimal recovery error probability by the error probability of an estimator that is granted partial access to the latent planted subgraph $\Gamma^\star$. Let $\{\Gamma^{(\ell),\star}\}_{\ell=1}^M$ denote the (unique) onion decomposition of $\Gamma^\star$ as defined in Definition~\ref{def:onionDec}. Suppose that a genie reveals $\Gamma^{(M-1),\star}$ to the recovery algorithm. Under this additional information, recovering $\Gamma^\star$ reduces to recovering only the final layer
\begin{align}
\calD^{(M),\star}\triangleq \Gamma^{(M),\star}\setminus\Gamma^{(M-1),\star}
= \Gamma^\star\setminus\Gamma^{(M-1),\star}.
\end{align}
Accordingly, we consider the recovery problem of $\Gamma^{(M),\star}$ \emph{given} the observation
\begin{align}
\calO\triangleq(\s{G},\Gamma^{(M-1),\star}),
\end{align}
and let $\hat{\Gamma}^{(M)}(\calO)$ denote any recovery algorithm that is allowed to depend on this augmented information. Then, by definition,
\begin{align}
    \inf_{\hat{\Gamma}}\max_{\Gamma^\star\in\calS_{\Gamma}}\pr[\hat{\Gamma}(\s{G})\neq\Gamma^\star]
    \geq
    \inf_{\hat{\Gamma}^{(M)}}\max_{\Gamma^\star\in\calS_{\Gamma}}\pr[\hat{\Gamma}^{(M)}(\calO)\neq\Gamma^{\star}],
    \label{eqn:genie-aidedIT}
\end{align}
where the infimum on the left-hand side of \eqref{eqn:genie-aidedIT} ranges over all measurable functions of $\s{G}$.

Recall that $\calM(\Gamma^{(M-1),\star},\Gamma)$ denotes the number of ways in which $\Gamma^{(M-1),\star}$ can be extended to a copy $\Gamma$ of $\Gamma^\star$ in the complete graph $\s{K}_n$. Equivalently, it is the number of copies of $\Gamma^\star$ in $\s{K}_n$ that contain $\Gamma^{(M-1),\star}$. In words, $\calM(\Gamma^{(M-1),\star},\Gamma)$ characterizes the set of all \emph{valid} embeddings of $\Gamma^\star$ consistent with the revealed partial structure, and thus captures the sole remaining uncertainty in recovering $\Gamma^\star$ given $\Gamma^{(M-1),\star}$. Accordingly, let $n'\triangleq n-|v(\Gamma^{(M-1),\star})|$ and $k'\triangleq|v(\Gamma^{(M),\star})\setminus v(\Gamma^{(M-1),\star})|$. Crucially, note that $k'$ represents the number of vertices not yet known to the recovery algorithm. Conversely, recall that although recall that although $\Gamma^{(M),\star}\setminus\Gamma^{(M-1),\star}$ and $\Gamma^{(M-1),\star}$ are edge-disjoint, they may share common vertices. These shared vertices are available to the informed recovery algorithm, since it is provided with $\Gamma^{(M-1),\star}$. Next, we observe that
\begin{align}
    \inf_{\hat{\Gamma}^{(M)}}\max_{\Gamma^\star\in\calS_{\Gamma}}\pr[\hat{\Gamma}^{(M)}(\calO)\neq\Gamma^{\star}]
    &= \inf_{\hat{\Gamma}^{(M)}}\max_{\Gamma^{(M-1),\star}}\max_{\Gamma^{\star}\in\calM(\Gamma^{(M-1),\star},\Gamma^{\star})}\pr[\hat{\Gamma}^{(M)}(\calO)\neq\Gamma^{\star}]\\
    &\geq \inf_{\hat{\Gamma}^{(M)}}\max_{\Gamma^{(M-1),\star}}\bE_{\Gamma^{\star}\sim\pi}\pr[\hat{\Gamma}^{(M)}(\calO)\neq\Gamma^{\star}],
\end{align}
where $\pi \triangleq \s{Unif}(\calM(\Gamma^{(M-1),\star},\Gamma))$ denotes the uniform distribution over $\calM(\Gamma^{(M-1),\star},\Gamma)$, and the inequality follows from lower bounding the worst-case risk by the corresponding average risk. Applying conditional Fano's inequality (see, e.g., \cite[Theorem~3]{scarlett2021fano}), we obtain
\begin{align}
    &\inf_{\hat{\Gamma}^{(M)}}\max_{\Gamma^{(M-1),\star}}\bE_{\Gamma^{\star}\sim\pi}\pr[\hat{\Gamma}^{(M)}(\calO)\neq\Gamma^{\star}]\nonumber\\
    &\hspace{3cm}\geq 
    \max_{\gamma^{(M-1)}}\pp{1-\frac{I(\Gamma^\star; \s{G} \vert \Gamma^{(M-1),\star}=\gamma^{(M-1)}) + \log 2}{\log|\calM(\gamma^{(M-1)},\Gamma)|}},
    \label{eqn:fanoCond}
\end{align}
where $I(X;Z\mid V=v)$ denotes the conditional mutual information between the random variables $X$ and $Z$ given $V=v$, namely, $I(X;Z\vert V=v) = H(X\vert V=v)-H(X\vert Z,V=v)$. We begin by lower bounding $|\calM(\gamma^{(M-1)},\Gamma)|$ for an arbitrary realization $\gamma^{(M-1)}$. Although sharper bounds on $|\calM(\gamma^{(M-1)},\Gamma)|$ can be derived, the following simple bound suffices for our purposes. Specifically, we claim that $|\calM(\gamma^{(M-1)},\Gamma)| \geq \binom{n'}{k'}$. Indeed, fixing $\Gamma^{(M-1),\star}=\gamma^{(M-1)}$, we may first select $k'$ vertices from the remaining $n'$ vertices to host the vertices of the final layer $\Gamma^\star=\Gamma^{(M),\star}$. This yields $\binom{n'}{k'}$ distinct choices. In principle, one should also account for the possible labelings of $\Gamma^\star$, whose number is given by $\frac{k'!}{|\s{Aut}(\Gamma^\star\setminus\gamma^{(M-1)})|}$. For our purposes, however, it is sufficient to note that this factor is always at least $1$, and thus the crude lower bound $\binom{n'}{k'}$ is adequate.

Next we upper bound $I(\Gamma^\star; \s{G} \vert\gamma^{(M-1)})$. With a slight abuse of notation, for any $e\in\binom{[n]}{2}$ we let $\s{G}_e$ denote the indicator for the presence of an edge in $\s{G}$, i.e., $\s{G}_e=1$ if $e\in\s{G}$, and zero otherwise. We define $\Gamma^\star_e$ similarly. Then, note that conditioned on the edges $\gamma^{(M-1)}$, the corresponding edges in $\s{G}$ are deterministically fixed, and thus do not contribute to the mutual information. Specifically, recall that 
\begin{align}
    I(\Gamma^\star; \s{G} \vert\Gamma^{(M-1),\star}=\gamma^{(M-1)}) = H(\s{G} \vert\Gamma^{(M-1),\star}=\gamma^{(M-1)})-H(\s{G} \vert\Gamma^{(M-1),\star}=\gamma^{(M-1)},\Gamma^\star).
\end{align}
Then, applying the chain rule for entropy we get
\begin{align}
    H(\s{G} \vert\Gamma^{(M-1),\star}=\gamma^{(M-1)}) &= H\p{\left.\s{G}\setminus\gamma^{(M-1)} \right|\Gamma^{(M-1),\star}=\gamma^{(M-1)}}\\
    &\leq \pp{\binom{n}{2}-|e(\gamma^{(M-1)})|}H(\s{G}_e),
\end{align}
for any arbitrary edge $e\in\binom{[n]}{2}\setminus\gamma^{(M-1)}$, where the inequality follows from the fact that $\s{G}_{e}$'s are identically distributed by symmetry. Similarly
\begin{align}
    H(\s{G} \vert\Gamma^{(M-1),\star}=\gamma^{(M-1)},\Gamma^\star)= \pp{\binom{n}{2}-|e(\gamma^{(M-1)})|}H(\s{G}_e\vert\Gamma^\star_e),
\end{align}
again, for any arbitrary edge $e\in\binom{[n]}{2}\setminus\gamma^{(M-1)}$, where we have used the fact that $\s{G}_{e}$'s are independent conditioned on $\Gamma^\star$. Therefore 
\begin{align}
    I(\Gamma^\star; \s{G} \vert\Gamma^{(M-1),\star})\leq \pp{\binom{n}{2}-|e(\gamma^{(M-1)})|}I(\s{G}_e;\Gamma^\star_e),\label{eqn:MutualTotal}
\end{align}
for an arbitrary edge $e\in\binom{[n]}{2}\setminus\gamma^{(M-1)}$. Note that
\begin{align}
    \s{G}_e\vert\Gamma^\star_e\sim\begin{cases}
        \s{Bern}(p_n),\ &\s{if}\;\Gamma^\star_e=1\\
        \s{Bern}(q_n),\ &\s{if}\;\Gamma^\star_e=0.
    \end{cases}
\end{align}
Next, we find $\eta\triangleq\pr[\Gamma^\star_e=1]$ for $e\in\binom{[n]}{2}\setminus\gamma^{(M-1)}$. Recall that $\Gamma^\star\sim\s{Unif}(\calM(\Gamma^{(M-1),\star},\Gamma))$ given $\Gamma^{(M-1),\star}=\gamma^{(M-1)}$. Thus, we have
\begin{align}
    \eta= \frac{1}{|\calM(\gamma^{(M-1)},\Gamma)|}\sum_{\Gamma'\in\calM(\gamma^{(M-1)},\Gamma)}\Ind\{e\in\Gamma'\}.
\end{align}
Observe that each embedding $\Gamma'\in\calM(\gamma^{(M-1)},\Gamma)$ contains the same number of edges $|e(\Gamma^\star\setminus\Gamma^{(M-1),\star})|$, and so
\begin{align}
    \sum_{\Gamma'\in\calM(\gamma^{(M-1)},\Gamma)}\sum_{\{e\}\subseteq\binom{[n]}{2}\setminus\gamma^{(M-1)}} \Ind\{e\in\Gamma'\} = |\calM(\gamma^{(M-1)},\Gamma)|\cdot|e(\Gamma^\star\setminus\Gamma^{(M-1),\star})|.
\end{align}
Since every embedding $\Gamma'\in\calM(\gamma^{(M-1)},\Gamma)$ is obtained by first choosing an \emph{unordered} $k'$-subset and then relabeling the vertices of $\Gamma'$, every unordered host pair $\{i,j\}$ appears in precisely the \emph{same} number of embeddings as any other pair. Therefore, combining the above, we may conclude that each fixed pair $\{i,j\}$ appears as an edge in $\frac{|\calM(\gamma^{(M-1)},\Gamma)|\cdot|e(\Gamma^\star\setminus\Gamma^{(M-1),\star})|}{\binom{n}{2}-|e(\gamma^{(M-1)})|}$ embeddings, and as so
\begin{align}
    \eta &=\frac{\frac{|\calM(\gamma^{(M-1)},\Gamma)|\cdot|e(\Gamma^\star\setminus\Gamma^{(M-1),\star})|}{\binom{n}{2}-|e(\gamma^{(M-1)})|}}
       {|\calM(\gamma^{(M-1)},\Gamma)|}\\
       & = \frac{|e(\Gamma^\star\setminus\Gamma^{(M-1),\star})|}{\binom{n}{2}-|e(\gamma^{(M-1)})|}\\
       & = \frac{\eta(\Gamma^{(M),\star}\vert\Gamma^{(M-1),\star})k'}{\binom{n}{2}-|e(\gamma^{(M-1)})|} \\
       &=\frac{\mu_{\s{min}}(\Gamma)k'}{\binom{n}{2}-|e(\gamma^{(M-1)})|} = o(1),\label{eqn:etaDef}
\end{align}
where in the third equality we have used the fact that, the onion decomposition in Definition~\ref{def:onionDec} forms the sequence $\{\Gamma^{(\ell),\star}\}_{\ell\geq0}$ in a way such that each layer $\Gamma^{(\ell)}$ is the maximal subgraph that maximizes $\eta(\s{H} \vert \Gamma^{(\ell-1)})$ among all subgraphs $\Gamma^{(\ell-1)}\subsetneq \s{H}$, and therefore, $|e(\Gamma^{(M),\star} \setminus \Gamma^{(M-1),\star})| = \eta(\Gamma^{(M),\star} \vert \Gamma^{(M-1),\star})\cdot|v(\Gamma^{(M),\star}) \setminus v(\Gamma^{(M-1),\star})| = \mu_{\s{min}}(\Gamma)k'$. The last equality is because $k'\leq n$ and \eqref{eq:condDenseStat}. Next, note that by Bayes theorem, $\zeta\triangleq\pr[\s{G}_e=1]=\eta p_n + (1-\eta)q_n$. Combining the above, we finally obtain
\begin{align}
   I(\s{G}_e;\Gamma^\star_e) &= \eta d_{\s{KL}}(p_n||\zeta)+(1-\eta) d_{\s{KL}}(q_n||\zeta)\\
   &\leq \eta d_{\s{KL}}(p_n||\zeta)+(1-\eta) \frac{(q_n-\zeta)^2}{\zeta(1-\zeta)}.
\end{align}
Since $\eta=o(1)$, we note that $\zeta = q_n+O(\eta)$, and so, $\frac{(q_n-\zeta)^2}{\zeta(1-\zeta)} = O(\eta^2)$. Furthermore
\begin{align}
    d_{\s{KL}}(p_n||\zeta) = d_{\s{KL}}(p_n||q_n)+O(\eta).
\end{align}
Thus
\begin{align}
    I(\s{G}_e;\Gamma^\star_e)\leq \eta d_{\s{KL}}(p_n||q_n)+O(\eta^2),
\end{align}
which in light of \eqref{eqn:MutualTotal} implies that
\begin{align}
    I(\Gamma^\star; \s{G} \vert\Gamma^{(M-1),\star})&\leq \pp{\binom{n}{2}-|e(\gamma^{(M-1)})|}I(\s{G}_e;\Gamma^\star_e)\\
    &\leq \pp{\binom{n}{2}-|e(\gamma^{(M-1)})|}\eta d_{\s{KL}}(p_n||q_n)\pp{1+O(\eta)}\\
    & = \mu_{\s{min}}(\Gamma)k'd_{\s{KL}}(p_n||q_n)\pp{1+O(\eta)},
\end{align}
where the last equality follows from \eqref{eqn:etaDef}. Substituting in \eqref{eqn:fanoCond} we obtain
\begin{align}
    \inf_{\hat{\Gamma}^{(M)}}\max_{\Gamma^{(M-1),\star}}\bE_{\Gamma^{\star}\sim\pi}\pr[\hat{\Gamma}^{(M)}(\calO)\neq\Gamma^{\star}]&\geq 1-\frac{\mu_{\s{min}}(\Gamma)k'd_{\s{KL}}(p_n||q_n)\pp{1+O(\eta)} + \log 2}{\log\binom{n'}{k'}}.\label{eqn:fanOfinal}
\end{align}
For any $\delta>0$, the term at the right-hand side of \eqref{eqn:fanOfinal} is bounded below by $\delta$ provided that
\begin{align}
    \mu_{\s{min}}(\Gamma)k'd_{\s{KL}}(p_n||q_n)\pp{1+O(\eta)}\leq (1-\delta)\log\binom{n'}{k'}.\label{eqn:fanOfinal2}
\end{align}
Using the facts that $\log\binom{n'}{k'}\geq k'\log(n'/k')$, $k'\leq k = o(n)$, $n' = (1-o(1))n$, and $\eta=o(1)$, it is evident that \eqref{eqn:fanOfinal} is implied by \eqref{eq:condDenseStat}, for large enough $n$, which concludes the proof.

\section{Upper Bounds}\label{sec:upperbounds}

In this section, we establish the sufficient conditions for the planted graph recovery.

\subsection{Peeling MLE}\label{sec:upperbounds_MLE}
\begin{proof}[Proof of Theorem~\ref{thm:MLE}]
We analyze the MLE peeling algorithm in \eqref{eqn:PeelMLE}. Let $\{\Gamma^{(\ell),\star}\}_{\ell=0}^M$ denote the onion decomposition of the planted subgraph $\Gamma^\star$ as defined in Definition~\ref{def:onionDec}. We have
\begin{align}
    \pr[\hat{\Gamma}_{\s{pMLE}}\neq\Gamma^{\star}] &= \pr[\hat{\Gamma}^{(M)}_{\s{MLE}}\neq\Gamma^{\star}]\\
    & = \pr\pp{\bigcup_{\ell=1}^M\ppp{\hat{\Gamma}^{(\ell)}_{\s{MLE}}\neq{\Gamma}^{(\ell),\star}}}\\
    & \leq \sum_{\ell=1}^M\pr\pp{\left.\hat{\Gamma}^{(\ell)}_{\s{MLE}}\neq{\Gamma}^{(\ell),\star}\right|\hat{\Gamma}^{(\ell-1)}_{\s{MLE}}={\Gamma}^{(\ell-1),\star}},\label{eqn:UnionSmartBound}
\end{align}
where the inequality follows from the fact that for a set of monotonically increasing nested events $(\calA_i)_{i=1}^n$ we have $\pr[\bigcup_{i=1}^n\calA_i]\leq\sum_{i=1}^n\pr[\calA_i\vert \calA_{i-1}^c]$.\footnote{Define $\calB_1 = \Omega$, $\calB_i = \bigcap_{j < i} \calA_j^{c}$ for $i \geq2$. Then the union splits into disjoint pieces $\bigcup_{i=1}^n \calA_i=\bigsqcup_{i=1}^n (\calA_i \cap \calB_i),$ and so $\pr(\bigcup_{i=1}^n \calA_i) = \sum_{i=1}^n \pr(\calA_i \cap \calB_i)= \sum_{i=1}^n \pr(\calA_i \vert \calB_i)\pr(\calB_i)\leq\sum_{i=1}^n \pr(\calA_i\vert \calB_i),$ because $\pr(\calB_i) \leq 1$ for each $i$.} Accordingly, in what follows, we analyze the error probability at the $\ell^{\s{th}}$ iteration of the algorithm, conditioned on the event that at the $(\ell-1)^{\s{th}}$ iteration the estimated subgraph is correct. It should be emphasized that, in the case where $v({\Gamma}^{(\ell),\star}\setminus{\Gamma}^{(\ell-1),\star})\cap v({\Gamma}^{(\ell-1),\star})\neq\emptyset$, the previously estimated layer ${\Gamma}^{(\ell-1),\star}$ correctly identified those common vertices. 

Recall that $\hat{\Gamma}^{(\ell)}_{\s{MLE}}$ in \eqref{eqn:PeelMLE0} scans over all copies $\calM\triangleq\calM(\hat{\Gamma}^{(\ell-1)}_{\s{MLE}},\Gamma^{(\ell)})$, i.e., copies of $\Gamma^{(\ell)}$ in $\calK_n$ which contain $\hat{\Gamma}^{(\ell-1)}_{\s{MLE}}$. Loosely speaking, we shall refer to this set as the collection of all possible copies of $\calD^{(\ell),\star}$, each taken in union with $\hat{\Gamma}^{(\ell-1)}_{\s{MLE}}$. Now, the MLE $\hat{\Gamma}^{(\ell)}_{\s{MLE}}$ of $\Gamma^{(\ell,\star)}$, given $\hat{\Gamma}^{(\ell-1)}_{\s{MLE}}={\Gamma}^{(\ell-1),\star}$, outputs any element (copy) of $\calM(\hat{\Gamma}^{(\ell-1)}_{\s{MLE}},\Gamma^{(\ell)})$ in the observed graph $\s{G}$. Accordingly, this estimator succeeds if, with high probability, the observed graph contains such a \emph{unique} copy. Let $\{\calW_j\}_{j=1}^{|\calM|}$ denote the $|\calM(\hat{\Gamma}^{(\ell-1)}_{\s{MLE}},\Gamma^{(\ell)})|$ possible copies (note that each copy is an extension of $\hat{\Gamma}^{(\ell-1)}_{\s{MLE}}$). Without loss of generality, we may assume that the actual planted copy is $\Gamma^{(\ell,\star)} = \calW_1$. Given $\s{G}$, let $\s{N}_{\calM}(\s{G}) \triangleq \sum_{j=1}^{|\calM|}\Ind\ppp{\calW_j\in\s{G}}$ denote the number of copies in $\calM(\hat{\Gamma}^{(\ell-1)}_{\s{MLE}},\Gamma^{(\ell)})$ that appear in $\s{G}$. Define the event $\calF_\ell\triangleq\{\hat{\Gamma}^{(\ell-1)}_{\s{MLE}}=\Gamma^{(\ell-1),\star}\}$. Using the above argument and Markov's inequality, we have:
\begin{align}
\pr\pp{\left.\hat{\Gamma}^{(\ell)}_{\s{MLE}}\neq{\Gamma}^{(\ell),\star}\right|\calF_\ell}&\leq\pr[\s{At}\;\s{least}\;\s{two}\;\s{extensions}\;\s{in}\;\calG_{\Gamma_n}(n,p_n,q_n)\vert\calF_\ell]\\
& = \pr_{\calG_{\Gamma_n}(n,p_n,q_n)}\pp{\s{N}_{\calM}(\s{G})\geq2\vert\calF_\ell}\\
& = \pr_{\calG_{\Gamma_n}(n,p_n,q_n)}\pp{\s{N}_{\calM\setminus\calW_1}(\s{G})\geq1\vert\calF_\ell}\\
    &\leq \bE\pp{\s{N}_{\calM\setminus\calW_1}(\s{G})\vert\calF_\ell},\label{eqn:FirstMoment}
\end{align}
where the second equality follows because under $\calG_{\Gamma_n}(n,p_n,q_n)$ the copy $\calW_1$ is planted—and thus appear—in $\s{G}$, and $\s{N}_{\calM\setminus\calW_1}(\s{G})$ counts the number of copies in $\s{G}$ except the planted one $\calW_1$, i.e., $\s{N}_{\calM\setminus\calW_1}(\s{G})\triangleq\sum_{j=2}^{|\calM|}\Ind\ppp{\calW_j\in\s{G}}$. 

We first consider the case where $p=1$, and show that if $\mu_{\s{min}}(\Gamma_n)\geq \frac{(4+\varepsilon)\cdot \log n}{d_{\s{KL}}(1||q)}$, then the error probability is small; this differs by a factor of 4 from the lower bound stated in Theorem~\ref{thm:lowerp1}. To that end, we first note that each copy $\calW_j$ can be written as $\calW_j = \hat{\Gamma}^{(\ell-1)}_{\s{MLE}}\cup(\calW_j\setminus\hat{\Gamma}^{(\ell-1)}_{\s{MLE}})$, for all $j\in[|\calM|]$. Conditioned on $\hat{\Gamma}^{(\ell-1)}_{\s{MLE}}=\hat{\Gamma}^{(\ell-1),\star}$, we have $\hat{\Gamma}^{(\ell-1)}_{\s{MLE}}\in\s{G}$ with probability one. Thus, for any $j\in[|\calM|]$
\begin{align}
\pr\pp{\calW_j\in\s{G}\vert\calF_\ell} &= \pr\pp{\left.\hat{\Gamma}^{(\ell-1)}_{\s{MLE}}\cup(\calW_j\setminus\hat{\Gamma}^{(\ell-1)}_{\s{MLE}})\in\s{G}\right|\calF_\ell}\\
& = \pr\pp{\left.\calW_j\setminus\hat{\Gamma}^{(\ell-1)}_{\s{MLE}}\in\s{G}\right|\calF_\ell}\\
& = \pr\pp{\left.\calW_j\setminus\Gamma^{(\ell-1),\star}\in\s{G}\right|\calF_\ell}\\
&=q^{|e(\Gamma^{(\ell),\star} \setminus \Gamma^{(\ell-1),\star})|-|e((\calW_j\setminus\Gamma^{(\ell-1),\star})\cap(\Gamma^{(\ell),\star}\setminus\Gamma^{(\ell-1),\star}))|}\\
& = q^{|e(\Gamma^{(\ell),\star} \setminus \Gamma^{(\ell-1),\star})|-|e((\calW_j\cap\Gamma^{(\ell),\star})\setminus\Gamma^{(\ell-1),\star})|},
\end{align}
where $|e(\Gamma^{(\ell),\star} \setminus \Gamma^{(\ell-1),\star})|-|e((\calW_j\cap\Gamma^{(\ell),\star})\setminus\Gamma^{(\ell-1),\star})|$ is the number of edges in $\calW_j\setminus\Gamma^{(\ell-1),\star}$ that are not part of the planted subgraph $\Gamma^{(\ell),\star} \setminus \Gamma^{(\ell-1),\star}$, whose edges appear in the graph with probability one. Thus
\begin{align}
    \bE\pp{\s{N}_{\calM\setminus\calW_1}(\s{G})\vert\calF_\ell} &= \sum_{j\geq2}\pr\pp{\calW_j\in\s{G}\vert\calF_\ell}\\
    & = \sum_{j\geq2}q^{|e(\Gamma^{(\ell),\star} \setminus \Gamma^{(\ell-1),\star})|-|e((\calW_j\cap\Gamma^{(\ell),\star})\setminus\Gamma^{(\ell-1),\star})|}.\label{eqn:momentGenInt0}
\end{align}
The onion decomposition in Definition~\ref{def:onionDec} forms the sequence $\{\Gamma^{(\ell),\star}\}_{\ell\geq0}$ in a way such that each layer $\Gamma^{(\ell)}$ is the maximal subgraph that maximizes $\eta(\s{H} \vert \Gamma^{(\ell-1)}) = \frac{|e(\s{H})|-|e(\Gamma^{(\ell-1)})|}{|v(\s{H})\setminus v(\Gamma^{(\ell-1)})|}$ among all subgraphs $\Gamma^{(\ell-1)}\subsetneq \s{H}$. Therefore,
\begin{align}
   \eta(\Gamma^{(\ell),\star} \vert \Gamma^{(\ell-1),\star})&= \sup_{\Gamma^{(\ell-1),\star}\subsetneq \s{H}} \eta(\s{H} \vert \Gamma^{(\ell-1)})\\
   &\geq\frac{|e(\Gamma^{(\ell-1),\star}\cup(\calW_j\cap\Gamma^{(\ell),\star}))|-|e(\Gamma^{(\ell-1),\star})|}{|v(\Gamma^{(\ell-1),\star}\cup(\calW_j\cap\Gamma^{(\ell),\star}))\setminus v(\Gamma^{(\ell-1),\star})|}\\
   & = \frac{|e((\calW_j\cap\Gamma^{(\ell),\star})\setminus\Gamma^{(\ell-1),\star})|}{|v(\calW_j\cap\Gamma^{(\ell),\star})\setminus v(\Gamma^{(\ell-1),\star})|}.
\end{align}
Furthermore, $|e(\Gamma^{(\ell),\star} \setminus \Gamma^{(\ell-1),\star})| = \eta(\Gamma^{(\ell),\star} \vert \Gamma^{(\ell-1),\star})\cdot|v(\Gamma^{(\ell),\star}) \setminus v(\Gamma^{(\ell-1),\star})|$. Thus, we get
\begin{align}
     \bE\pp{\s{N}_{\calM\setminus\calW_1}(\s{G})\vert\calF_\ell}&\leq \sum_{j\geq2}q^{\eta(\Gamma^{(\ell),\star} \vert \Gamma^{(\ell-1),\star})\cdot\pp{|v(\Gamma^{(\ell),\star}) \setminus v(\Gamma^{(\ell-1),\star})|-|v(\calW_j\cap\Gamma^{(\ell),\star})\setminus v(\Gamma^{(\ell-1),\star})|}}.\label{eqn:momentGenInt0Mid}
\end{align}
Denote $\calD^{(\ell),\star} \triangleq \Gamma^{(\ell),\star} \setminus \Gamma^{(\ell-1),\star}$. It can be seen that \eqref{eqn:momentGenInt0Mid} depends on each possible embedding $\{\calW_j\}_{j\geq 2}$ only through the number of vertices it shares with $\Gamma^{(\ell),\star}$ but with the vertices of $\Gamma^{(\ell-1),\star}$ removed. Accordingly, let us count the number of embeddings $\{\calW_j\}_{j\geq 2}$ with overlap size $i = |v(\calW_j\cap\Gamma^{(\ell),\star})\setminus v(\Gamma^{(\ell-1),\star})|$, and then sum over all possible values of $i$. This count can be upper bounded quite easily. Indeed, we first consider the number of ways to choose $i$ vertices from $v(\Gamma^{(\ell),\star}) \setminus v(\Gamma^{(\ell-1),\star})$ to overlap with $v(\calW_j) \setminus v(\Gamma^{(\ell-1),\star})$—that is, we are selecting which $i$ vertices of $v(\Gamma^{(\ell),\star})\setminus v(\Gamma^{(\ell-1),\star})$ will align with the vertex set of $v(\calW_j) \setminus v(\Gamma^{(\ell-1),\star})$. This count is given by $\binom{|v(\Gamma^{(\ell),\star}) \setminus v(\Gamma^{(\ell-1),\star})|}{i}$. Next, we count the number of ways to select the remaining $|v(\Gamma^{(\ell),\star}) \setminus v(\Gamma^{(\ell-1),\star})|-i$ fresh vertices from the $n - |v(\Gamma^{(\ell),\star}) \setminus v(\Gamma^{(\ell-1),\star})|$ vertices not in $v(\Gamma^{(\ell),\star})\setminus v(\Gamma^{(\ell-1),\star})$. This is given by $\binom{n-|v(\Gamma^{(\ell),\star}) \setminus v(\Gamma^{(\ell-1),\star})|}{|v(\Gamma^{(\ell),\star}) \setminus v(\Gamma^{(\ell-1),\star})|-i}$. Therefore, the total number of injective vertex mappings of $v(\calW_j\cap\Gamma^{(\ell),\star})$ into the vertex set $[n]$ such that exactly $i$ of the vertices are shared with $v(\Gamma^{(\ell),\star})\setminus v(\Gamma^{(\ell-1),\star})$, and the remaining $|v(\Gamma^{(\ell),\star}) \setminus v(\Gamma^{(\ell-1),\star})|-i$ vertices are disjoint from it, is bounded above by $\binom{k_\ell}{i}\binom{n-k_\ell}{k_\ell-i}$, where $k_\ell\triangleq |v(\Gamma^{(\ell),\star}) \setminus v(\Gamma^{(\ell-1),\star})|$. 

Before we plug in the above and sum over the index $i$, there is a subtle point that must be taken into account—namely, the uniqueness property of the planted copy on its full vertex set. Specifically, in theory, it could be the case that $|v(\calW_j\cap\Gamma^{(\ell),\star})\setminus v(\Gamma^{(\ell-1),\star})| = |v(\Gamma^{(\ell),\star}) \setminus v(\Gamma^{(\ell-1),\star})|$, yet $\calW_j \neq \Gamma^{(\ell),\star}$. For example, consider the case where $\Gamma^{(\ell),\star}$ is a clique plus a single edge attached to one of its vertices. This scenario implies that the range of $i$ should be $0 \leq i \leq |v(\Gamma^{(\ell),\star}) \setminus v(\Gamma^{(\ell-1),\star})|$. However, this possibility is precluded by the property that, at each step of the onion decomposition (see Definition~\ref{def:onionDec}), the selected layer maximizes the relative subgraph density. We now prove that this is indeed the case. The following result is general and holds for any step in the onion decomposition of $\Gamma$ as defined in Definition~\ref{def:onionDec}. Accordingly, consider the $\ell$th step in the onion decomposition, and recall that $\calD^{(\ell),\star}\triangleq \Gamma^{(\ell),\star} \setminus \Gamma^{(\ell-1),\star}$.
\begin{lemma}[Uniqueness at full overlap]
\label{lem:unique-full-overlap}
Consider the onion decomposition of $\Gamma^\star$ in Definition~\ref{def:onionDec}.
Let $\calW'\in\calM({\Gamma}^{(\ell-1),\star},\Gamma^{(\ell)})$ and $\calD' = \calW'\setminus{\Gamma}^{(\ell-1),\star}$. Assume that $\calD'$ satisfies $v(\calD')=v(\calD^{(\ell),\star})$. Then $e(\calD')=e(\calD^{(\ell),\star})$; hence
$\calD'=\calD^{(\ell),\star}$. In particular, there exists no copy distinct from the planted one whose vertex overlap with $\calD^{(\ell),\star}$ equals $|v(\calD^{(\ell),\star})|$.
\end{lemma}

\begin{proof}[Proof of Lemma~\ref{lem:unique-full-overlap}]
By definition of $\calM({\Gamma}^{(\ell-1),\star},\Gamma^{(\ell)})$ there is an isomorphism $\phi_1: v(\calD^{(\ell)})\to v(\calD^{(\ell),\star})$ with $\phi_1(e(\calD^{(\ell)}))=e(\calD^{(\ell),\star})$. Furthermore, because we assume that there exist $\calW'\in\calM({\Gamma}^{(\ell-1),\star},\Gamma^{(\ell)})$ such that $\calD'=\calW'\setminus{\Gamma}^{(\ell-1),\star}$ with $v(\calD')=v(\calD^{(\ell),\star})$, there is an isomorphism
$\phi_2:v(\calD^{(\ell)}) \to v(\calD^{(\ell),\star})$ with $\phi_2(e(\calD^{(\ell)}))=e(\calD')$.
Compose them to obtain an automorphism $\psi = \phi_1^{-1}\circ\phi_2:v(\calD^{(\ell)})\to v(\calD^{(\ell)})$. 
Now, let us show that $\psi$ fixes $v(\Gamma^{(\ell-1)})$. Indeed, since
$v(\Gamma^{(\ell-1)})\subsetneq v(\Gamma^{(\ell)})=v(\calD^{(\ell),\star})$, both $\phi_1$ and $\phi_2$ coincide with the identity on $v(\Gamma^{(\ell-1)})$, hence so does $\psi$. Next, we claim that $\psi$ is the identity on $v(\calD^{(\ell)})$. Otherwise set $\widetilde{\Gamma}^{(\ell)}=\Gamma^{(\ell-1)}\cup\psi(\calD^{(\ell)})$. The vertex set of $\widetilde{\Gamma}^{(\ell)}$ equals that of
$\Gamma^{(\ell)}$ and $e(\psi(\calD^{(\ell)})) = e(\calD')\neq e(\calD^{(\ell),\star}) =e(\calD^{(\ell)})$. Consequently, $|e(\widetilde{\Gamma}^{(\ell)})|>|e(\Gamma^{(\ell)})|$. This contradicts the \emph{maximality} of $\Gamma^{(\ell)}$ in step (ii) of Definition~\ref{def:onionDec}. Hence $\psi$ is the identity, and
$e(\calD')=e(\calD^{(\ell),\star})$.
\end{proof}
Returning to \eqref{eqn:momentGenInt0Mid} and its notation, we observe that, by Lemma~\ref{lem:unique-full-overlap}, the set
\begin{align}
\bigl\{\calW'\in\calM: \calW'\neq\Gamma^{(\ell),\star},\;|v(\calW'\cap\Gamma^{(\ell),\star})|=|v(\Gamma^{(\ell),\star})|\bigr\}
\end{align}
is empty. Recall that $k_\ell\triangleq |v(\Gamma^{(\ell),\star}) \setminus v(\Gamma^{(\ell-1),\star})|$, and further denote $\eta_\ell\triangleq\eta(\Gamma^{(\ell),\star} \vert \Gamma^{(\ell-1),\star})$. Therefore
\begin{align}
    \bE\pp{\s{N}_{\calM\setminus\calW_1}(\s{G})\vert\calF_\ell}&\leq \sum_{j\geq2}q^{\eta_\ell\cdot\pp{k_\ell-|v(\calW_j\cap\Gamma^{(\ell),\star})\setminus v(\Gamma^{(\ell-1),\star})|}}\label{eqn:samePP0}\\
    &\leq\sum_{i=0}^{k_\ell-1}q^{\eta_\ell\cdot(k_\ell-i)}\binom{k_\ell}{i}\binom{n-k_\ell}{k_\ell-i}\\
    &\leq \sum_{i=0}^{k_\ell-1}k_\ell^{k_\ell-i}n^{k_\ell-i} q^{\eta_\ell\cdot(k_\ell-i)}\\
    & = \sum_{i=0}^{k_\ell-1}\pp{k_\ell n q^{\eta_\ell}}^{k_\ell-i}\\
    &\leq \frac{\delta}{M(\Gamma)},\label{eqn:samePP1}
\end{align}
where the second inequality is because $\binom{k_\ell}{i} = \binom{k_\ell}{k_\ell-i}\leq k_\ell^{k_\ell-i}$ and $\binom{n-|v(\calD)|}{k_\ell-i}\leq n^{k_\ell-i}$, the last inequality holds provided that $k_\ell nq^{\eta_\ell}\leq \frac{\delta/2}{M(\Gamma)}$, for any $\delta>0$, and $M(\Gamma)$ is the number of subgraphs in the onion decomposition of $\Gamma^\star$. This concludes the analysis of the $\ell^{\s{th}}$ step of the MLE peeling algorithm. Applying the same argument for each step, by using \eqref{eqn:UnionSmartBound} we get
\begin{align}
    \pr[\hat{\Gamma}_{\s{pMLE}}\neq\Gamma^{\star}] &\leq \sum_{\ell=1}^M\pr\pp{\left.\hat{\Gamma}^{(\ell)}_{\s{MLE}}\neq{\Gamma}^{(\ell),\star}\right|\hat{\Gamma}^{(\ell-1)}_{\s{MLE}}={\Gamma}^{(\ell-1),\star}}\leq \delta,\label{eqn:FinalMLEPeelP1ana}
\end{align}
provided that $k_\ell nq^{\eta_\ell}\leq \frac{\delta/2}{M(\Gamma)}$, for all $\ell\geq1$. Since $M(\Gamma)\leq|e(\Gamma)|\leq n^2$ and $k_\ell\leq n$, this condition is clearly satisfied when $\eta_\ell\geq \frac{(4+\varepsilon)\cdot \log n}{d_{\s{KL}}(1||q)}$, for all $\ell\geq1$, namely, $\mu_{\s{min}}(\Gamma_n)\geq \frac{(4+\varepsilon)\cdot \log n}{d_{\s{KL}}(1||q)}$, for any $\varepsilon>0$, which concludes the proof for $p=1$.
\begin{remark}[Why the argument fails for the whole graph $\Gamma$]
The conclusion of Lemma~\ref{lem:unique-full-overlap} relies on the \emph{edge--maximality} property~\eqref{eqn:maxDensity} of each layer $\Gamma^{(\ell)}$. The full graph $\Gamma$ need not be edge--maximal on its vertex set.  If $\Gamma$ contains a dense core and a sparse appendage (e.g., a clique $K_k$ plus one extra edge), one can keep the entire vertex set fixed and relocate only the sparse part, producing $\Theta(n^2)$ distinct copies $\Gamma'\neq\Gamma^\star$ with full vertex overlap $i=|v(\Gamma^\star)|$. These copies must be included in the global first--moment sum, and their aggregate contribution is already unbounded, so the counting/Markov argument used for a peeling layer cannot be transferred verbatim to the global MLE.
\end{remark}

Finally, we adapt the above proof to the case where planted edges appear with probability $q<p\leq1$. Recall \eqref{eqn:UnionSmartBound}, and as above, let us analyze the $\ell^{\s{th}}$ step of the MLE peeling algorithm. To that end, for any $j\in[|\calM(\hat{\Gamma}^{(\ell-1)}_{\s{MLE}},\Gamma^{(\ell)})|]$, define
\begin{align}
    \calA(\calW_j)\triangleq \sum_{(i,j)\in\calW_j}\s{A}_{ij},
\end{align}
and $\triangle(\calW_j)\triangleq\calA(\calW_1)-\calA(\calW_j)$, where $\calW_1$ denotes the actual planted subgraph, namely, $\calW_1 = \Gamma^{(\ell),\star}$. We next find conditions under which $\triangle(\calW_j)>0$, for any $j\neq 1$. By definition, note that $\hat{\Gamma}^{(\ell-1)}_{\s{MLE}}\subsetneq\calW_j$ and since we condition on $\calF_\ell$ (see \eqref{eqn:UnionSmartBound}) we have $\Gamma^{(\ell-1),\star}\subsetneq\calW_j$, for any $\calW\in\calM$. For simplicity of notations, we define $\calD^{(\ell),\star}\triangleq\calW_1\setminus\Gamma^{(\ell-1),\star}=\Gamma^{(\ell),\star}\setminus\Gamma^{(\ell-1),\star}$ and $\calD_j\triangleq\calW_j\setminus\Gamma^{(\ell-1),\star}$, for all $j\geq2$. Thus,
\begin{align}
    \triangle(\calW_j) &= \sum_{(i,j)\in\calW_1}\s{A}_{ij}-\sum_{(i,j)\in\calW_j}\s{A}_{ij}\\
    & = \sum_{(i,j)\in\calD^{(\ell),\star}}[\s{A}_{ij}-\bE\s{A}_{ij}]-\sum_{(i,j)\in\calD_j}[\s{A}_{ij}-\bE\s{A}_{ij}]+\sum_{(i,j)\in\calD^{(\ell),\star}}\bE\s{A}_{ij}-\sum_{(i,j)\in\calD_j}\bE\s{A}_{ij}\\
    & = (p-q)\cdot |\calD_j\setminus\calD^{(\ell),\star}|+2\sum_{(i,j)\in\calD^{(\ell),\star}\setminus\calD_j,i<j}[\s{A}_{ij}-p]-2\sum_{(i,j)\in\calD_j\setminus\calD^{(\ell),\star},i<j}[\s{A}_{ij}-q]\\
    & \triangleq (p-q)\cdot |\calD_j\setminus\calD^{(\ell),\star}|+2\s{B}_1-2\s{B}_2,
\end{align}
where $\s{B}_1$ and $\s{B}_2$ are independent random variables, each composed of a sum of $\frac{1}{2}|\calD^{(\ell),\star}\setminus\calD_j| = \frac{1}{2}|\calD_j\setminus\calD^{(\ell),\star}| = \frac{1}{2}\p{|e(\calD^{(\ell),\star})|-|\calD_j\cap\calD^{(\ell),\star}|}$ i.i.d. centered Bernoulli random variables with parameters $p$ and $q$, respectively. Let $\s{J}(\calD_j)\triangleq|e(\calD^{(\ell),\star})|-|\calD_j\cap\calD^{(\ell),\star}|$. Chernoff's bound implies that,
\begin{align}
    \pr_{\s{Bern}(p)^{\otimes\s{J}(\calD_j)}}\pp{\s{B}_1\leq-\frac{p-q}{4}\s{J}(\calD_j)}\leq \exp\pp{-\frac{\s{J}(\calD_j)}{2}d_{\s{KL}}\p{\frac{p+q}{2}\Big\| p}},
\end{align}
and
\begin{align}
    \pr_{\s{Bern}(q)^{\otimes\s{J}(\calD_j)}}\pp{\s{B}_2\geq\frac{p-q}{4}\s{J}(\calD_j)}\leq \exp\pp{-\frac{\s{J}(\calD_j)}{2}d_{\s{KL}}\p{\frac{p+q}{2}\Big\| q}}.
\end{align}
For simplicity of notation define $\kappa\triangleq\frac{p+q}{2}$. Then
\begin{align}
    \pr\pp{\left.\hat{\Gamma}^{(\ell)}_{\s{MLE}}\neq{\Gamma}^{(\ell),\star}\right|\calF_\ell}&\leq\pr\pp{\bigcup_{j\geq2}\triangle(\calW_j)<0}\\
    &=\pr\pp{\bigcup_{j\geq2}\ppp{\s{B}_2-\s{B}_1<\frac{(p-q)}{2}\cdot |\calD_j\setminus\calD^{(\ell),\star}|}}\\
    &\leq \pr\pp{\bigcup_{j\geq2}\ppp{\s{B}_1\leq-\frac{p-q}{4}\s{J}(\calD_j)}\cup\ppp{\s{B}_2\geq\frac{p-q}{4}\s{J}(\calD_j)}}\\
    &\leq \sum_{j\geq2}\pp{\pr\p{\s{B}_1\leq-\frac{p-q}{4}\s{J}(\calD_j)}+\pr\p{\s{B}_2\geq\frac{p-q}{4}\s{J}(\calD_j)}}\\
    &\leq\sum_{j\geq2}\pp{e^{-\frac{\s{J}(\calD_j)}{2}d_{\s{KL}}\p{\kappa\parallel q}}+e^{-\frac{\s{J}(\calD_j)}{2}d_{\s{KL}}\p{\kappa\parallel p}}}\\
    & = \sum_{j\geq2}\pp{e^{-\frac{\s{J}(\calD_j)}{2}d_{\s{KL}}\p{\kappa\parallel q}}+e^{-\frac{\s{J}(\calD_j)}{2}d_{\s{KL}}\p{\kappa\parallel p}}},\label{eqn:errorProbabilityML}
\end{align}
where the first inequality follows from the fact that for any two random variables $X,Y$ and any $\zeta\in\mathbb{R}$, we have $\pr[X\geq Y]\leq\pr[X\geq\zeta\cup Y\leq\zeta]$. Since $\s{J}(\calD_j) = |e(\Gamma^{(\ell),\star} \setminus \Gamma^{(\ell-1),\star})|-|e((\calW_j\cap\Gamma^{(\ell),\star})\setminus\Gamma^{(\ell-1),\star})|$, we notice the resemblance between \eqref{eqn:errorProbabilityML} and \eqref{eqn:momentGenInt0}. Accordingly, by following the same arguments as in \eqref{eqn:momentGenInt0}--\eqref{eqn:FinalMLEPeelP1ana}, we obtain that $\pr[\hat{\calD}_{\s{MLE}}\neq\calD^{\star}]\leq\frac{\delta}{|M(\Gamma)|}$, provided that 
\begin{align}
    k_\ell \cdot n\cdot\exp\p{-\frac{d_{\s{KL}}(\kappa||q)\wedge d_{\s{KL}}(\kappa||p)}{2}\eta_\ell}\leq \frac{\delta/2}{M(\Gamma)}.\label{eqn:condGenP1}
\end{align}
Accordingly, $\pr[\hat{\Gamma}_{\s{pMLE}}\neq\Gamma^{\star}]\to0$, 
provided by \eqref{eqn:condGenP1}, for all $\ell\geq1$. This is equivalent to $\mu_{\s{min}}(\Gamma_n)\geq \frac{(8+\varepsilon)\cdot \log n}{d_{\s{KL}}(\kappa||q)\wedge d_{\s{KL}}(\kappa||p)}$, for any $\varepsilon>0$.
\end{proof}

\subsection{Convex relaxation}\label{sec:upperbounds_conv}

\begin{proof}[Proof of Theorem~\ref{thm:effAlg}]
We will show that, with high probability, the objective function of any feasible $\s{X}\neq\s{X}^\star$ is inferior, i.e.,
\begin{align}
    \innerP{\s{X},\s{W}} < \innerP{\s{X}^\star,\s{W}}.
\end{align}
To that end, we start by establishing some notations. Denote the singular value decomposition (SVD) of the underlying \emph{diagonally shifted} adjacency matrix $\s{S}^\star\triangleq\s{Ext}(\s{X}^\star,\mathbf{s}^\star;\alpha) = \s{X}^\star + \alpha\s{Diag}(\mathbf{s}^\star)$ associated with the planted subgraph $\Gamma^\star$ by
\begin{align}
\s{S}^\star = \s{U}\Sigma \s{U}^\top,
\end{align}
and define the projections projections onto the tangent space $\calT$ of the manifold of matrices as $\calP_{\calT}(M)\triangleq \s{U}\s{U}^TM+M\s{U}\s{U}^T-\s{U}\s{U}^TM\s{U}\s{U}^T$ and $\calP_{\calT^\perp}(M)\triangleq M-\calP_{\calT}(M)$. We note that $\bE\s{W} = c\s{X}^\star$ with $c\triangleq \frac{p}{q}-1>0$. With these notations, we can write,
\begin{align}
    \innerP{\s{X}^\star-\s{X},\s{W}} = \underbrace{\langle \s{X}^\star - \s{X}, c\s{X}^\star \rangle}_{\text{(a)}}
+ \underbrace{\langle \s{X}^\star - \s{X}, \calP_{\calT^\perp}(\s{W} - \bE[\s{W}]) \rangle}_{\text{(b)}}
+ \underbrace{\langle \s{X}^\star - \s{X}, \calP_{\calT}(\s{W} - \bE[\s{W}]) \rangle}_{\text{(c)}},
\end{align}
and we next bound each one of the above terms (a)--(c). As for (a), using the feasibility condition $\innerP{\Jb,\s{X}}=2|e(\Gamma)|$ in \eqref{eqn:convex}, we get
\begin{align}
    \langle \s{X}^\star - \s{X}, c\s{X}^\star \rangle = \frac{c}{2}\norm{\s{X}^\star - \s{X}}_{\ell_1}.
\end{align}
Next, we analyze term (b). To that end, we follow a standard approach and analyze the subgradient of $\norm{\cdot}_{\star}$ at $\s{S}^\star$. Using \cite[Corollary 6.1]{chen2016statistical}, we have $\partial \norm{\cdot}_{\star}(\s{S}^\star) = \{\s{U}\s{U}^\top + \calP_{\calT^\perp}(\s{Y}) : \|\s{Y}\|_{\s{op}} \leq 1\}$.
Thus, using the above fact and the nuclear-norm feasibility condition $\|\s{X}+\alpha\s{Diag}(\mathbf{s})\|_\star\leq \|\s{S}^\star\|_\star$, we obtain
\begin{align}
0 &\geq \|\s{X}+\alpha\s{Diag}(\mathbf{s})\|_\star-\|\s{S}^\star\|_\star\\
&\geq \langle \s{X} - \s{X}^\star, \s{U}\s{U}^\top \rangle + \langle \s{X} - \s{X}^\star, \calP_{\calT^\perp}(\s{Y}) \rangle + \alpha\langle \s{Diag}(\mathbf{s}-\mathbf{s}^\star), \s{U}\s{U}^\top + \calP_{\calT^\perp}(\s{Y}) \rangle,
\end{align}
and thus, by taking $\s{Y} = \frac{\s{W}-\bE[\s{W}]}{\|\s{W}-\bE[\s{W}]\|_{\s{op}}}$, and using H\"older's inequality, we have
\begin{align}
    &\abs{\langle \s{X}-\s{X}^\star, \calP_{\calT^\perp}(\s{W} - \bE[\s{W}]) \rangle}\leq \|\s{W}-\bE[\s{W}]\|_{\s{op}}\abs{\langle \s{X} - \s{X}^\star, \s{U}\s{U}^\top \rangle}\nonumber\\
    &\hspace{1.5cm}+\alpha\|\s{W}-\bE[\s{W}]\|_{\s{op}}\Big(\norm{\s{U}\s{U}^\top}_\infty \| \s{Diag}(\mathbf{s}-\mathbf{s}^\star)\|_{\ell_1}
    +\abs{\langle \s{Diag}(\mathbf{s}-\mathbf{s}^\star), \calP_{\calT^\perp}(\s{Y}) \rangle}\Big).
\end{align}
Since $(\s{W}-\bE[\s{W}])$ has zero diagonal, so does $\s{Y}$. Hence
\begin{align}
\langle \s{Diag}(\mathbf{s}-\mathbf{s}^\star), \calP_{\calT^\perp}(\s{Y}) \rangle
= - \langle \s{Diag}(\mathbf{s}-\mathbf{s}^\star), \calP_{\calT}(\s{Y}) \rangle,
\end{align}
and therefore, by duality,
\begin{align}
\abs{\langle \s{Diag}(\mathbf{s}-\mathbf{s}^\star), \calP_{\calT^\perp}(\s{Y}) \rangle}
\leq  \|\calP_{\calT}(\s{Y})\|_{\ell_\infty}\|\s{Diag}(\mathbf{s}-\mathbf{s}^\star)\|_{\ell_1}.
\end{align}
Finally, as for (c), using H\"older's inequality once again
\begin{align}
    \langle \s{X}^\star - \s{X}, \calP_{\calT}(\s{W}-\bE[\s{W}]) \rangle\leq \norm{\calP_{\calT}(\s{W} - \bE[\s{W}])}_{\ell_\infty} \norm{\s{X}-\s{X}^\star}_{\ell_1}.
\end{align}
Combining the above we obtain
\begin{align}
    \innerP{\s{X}^\star-\s{X},\s{W}}\geq\pp{\frac{c}{2}-\norm{\s{U}\s{U}^\top}_{\ell_\infty}\|\s{W}-\bE[\s{W}]\|_{\s{op}}-\norm{\calP_{\calT}(\s{W} - \bE[\s{W}])}_{\ell_\infty}}\norm{\s{X}-\s{X}^\star}_{\ell_1}\nonumber\\
    \hspace{1.8cm}-\alpha\|\s{W}-\bE[\s{W}]\|_{\s{op}}\Big(\norm{\s{U}\s{U}^\top}_\infty+\|\calP_{\calT}(\s{Y})\|_{\ell_\infty}\Big) \| \s{Diag}(\mathbf{s}-\mathbf{s}^\star)\|_{\ell_1},
    \label{eqn:ZerothBoundCon}
\end{align}
where $\s{Y}=\frac{\s{W}-\bE[\s{W}]}{\|\s{W}-\bE[\s{W}]\|_{\s{op}}}$. Let us bound each one of the terms at the right-hand side of \eqref{eqn:ZerothBoundCon}. Since $\s{W}-\bE\s{W}$ is a symmetric and i.i.d. matrix, with zero mean, it is well-known that (see, e.g., \cite[Corollary 4.4.7]{vershynin2010introduction}) with probability at least $1-\delta$,
\begin{align}
    \|\s{W}-\bE\s{W}\|_{\s{op}}\leq \frac{C}{q}\sqrt{n}+\frac{C}{q}\sqrt{\log\frac{4}{\delta}}.\label{eqn:FirstBoundCon00}
\end{align}
Thus,
\begin{align}
    \|\s{W}-\bE[\s{W}]\|_{\s{op}}\leq C'_q\sqrt{n},\label{eqn:FirstBoundCon}
\end{align}
for some large enough constant $C'_q>0$. Furthermore, we have the following lemma, which we will prove later on.

\begin{lemma}\label{lem:maxNormCohe}
Let $\s{U}\in\mathbb{R}^{n\times r}$ have orthonormal columns and set $\s{P}\triangleq \s{U}\s{U}^{\top}$.  
For each $i$, let $\ell_i\triangleq\|\s{U}_{i,:}\|_2^2$ and thus $\s{coh}(\s{U})=\frac{n}{r}\max_{1\leq i\leq n}\ell_i$. Then
\begin{align}
\frac{\sqrt{r}}{n}\leq\norm{\s{P}}_{\ell_\infty}\leq\frac{r}{n}\s{coh}(\s{U}).
\end{align}
\end{lemma}

Lemma~\ref{lem:maxNormCohe} implies that 
\begin{align}
    \norm{\s{U}\s{U}^\top}_{\ell_\infty}\leq \frac{\s{rank}(\s{S}^\star)}{n}\s{coh}(\s{U}).\label{eqn:SecondBoundCon}
\end{align}
We finally upper bound $\norm{\calP_{\calT}(\s{W}-\bE[\s{W}])}_{\ell_\infty}$ and $\|\calP_{\calT}(\s{Y})\|_{\ell_\infty}$ where $\s{Y} = \frac{\s{W}-\bE[\s{W}]}{\|\s{W}-\bE[\s{W}]\|_{\s{op}}}$. To that end, we let, for simplicity of notation, $\s{P}\triangleq \s{U}\s{U}^\top$, namely, orthogonal projector onto $\s{range}(\s{S}^\star)$, with $\s{U} \in \mathbb{R}^{n \times r}$ having orthonormal columns. Denote $r\triangleq\s{rank}(\s{S}^\star)$, and recall that
\begin{align}
\|\s{U}\|_{2,\infty} \triangleq \max_{i \in [n]} \|\s{U}_{i,:}\|_2, \qquad
\s{coh}(\s{U})=\frac{n}{r} \cdot \|\s{U}\|_{2,\infty}^2
\end{align}
Accordingly, we note that $\s{P}_{ii} = \|\s{U}^\top e_i\|_2^2 \leq \|\s{U}\|_{2,\infty}^2 = \s{coh}(\s{U}) r/n$, for any $i\in[n]$, and furthermore,
\begin{align}
|\s{P}_{ij}| &\leq \|\s{U}\|_{2,\infty}^2 = \frac{r}{n}\s{coh}(\s{U}),
\end{align}
for all $i,j\in[n]$. We have the following lemma.
\begin{lemma}\label{lem:infBoundSymY}
    Let $\s{Z}$ be a real-valued symmetric matrix, and $\calP_{\calT}$ be the projection matrix onto the tangent space, i.e., $\calP_{\calT}(\s{Z}) = \s{P}\s{Z}+\s{Z}\s{P}-\s{P}\s{Z}\s{P}$. Then
    \begin{align}
\norm{\calP_{\calT}(\s{Z})}_{\ell_\infty} &\leq(2 + \|\s{P}\|_{\ell_\infty \to \ell_\infty}) \cdot \norm{\s{P}\s{Z}}_{\ell_\infty}\\
&\leq\p{2 + \sqrt{ \frac{|v(\Gamma)| r \s{coh}(\s{U})}{n}}} \cdot \norm{\s{P}\s{Z}}_{\ell_\infty}\label{eqn:Epropo2}
\end{align}
Furthermore, for a diagonal matrix $D$
\begin{align}
\norm{\calP_{\calT}(\s{D})}_{\ell_\infty} \leq3\frac{\s{coh}(\s{U}) r}{n} \cdot\norm{\s{D}}_{\ell_\infty}.
\end{align}
\end{lemma}
Apply Lemma~\ref{lem:infBoundSymY} with $\s{Z} =  \s{W} - \bE[\s{W}]$ we have
\begin{align}
   \norm{\calP_{\calT}(\s{W}-\bE[\s{W}])}_{\ell_\infty}\leq\p{2 + \sqrt{ \frac{|v(\Gamma)|  \s{coh}(\s{U})r}{n}}} \cdot \norm{\s{P}(\s{W}-\bE[\s{W}])}_{\ell_\infty}.\label{eqn:ProjectionBound}
\end{align}
Now, note that for any $(i,j)$
\begin{align}
[\s{P}(\s{W}-\bE[\s{W}])]_{ij} = \sum_{s=1}^n \s{P}_{is}[\s{W}-\bE[\s{W}]]_{sj}.
\end{align}
Furthermore $([\s{W}-\bE[\s{W}]]_{sj})_{s=1}^n$ has independent, mean-zero, bounded entries with $\psi_2$-norm $\leq K_q$). Then a weighted Bernstein's bound gives with probability at least $1-n^{-3}$ that,
\begin{align}
   [\s{P}(\s{W}-\bE[\s{W}])]_{ij}\leq K_q\sqrt{\sum_{s=1}^n \s{P}_{is}^2\log n},
\end{align}
for a large enough constant $K_q$. Because $\s{P}$ is a projector, $\sum_{s=1}^n \s{P}_{is}^2 = [\s{P}^2]_{ii} = [\s{P}]_{ii} \leq \norm{\s{U}}_{2,\infty}^2 = \s{coh}(\s{U}) r/n$. Hence, by a union bound over all pairs $(i,j) \in [n]\times [n]$, with probability at least $1-n^{-1}$,
\begin{align}
\norm{\s{P}(\s{W}-\bE[\s{W}])}_{\ell_\infty} \leq K_q\sqrt{\frac{\s{coh}(\s{U}) r}{n} \log n},
\end{align}
Likewise, with $\s{Z}=\s{Y}=\frac{\s{W}-\bE[\s{W}]}{\|\s{W}-\bE[\s{W}]\|_{\s{op}}}$ we get
\begin{align}
\|\calP_{\calT}(\s{Y})\|_{\ell_\infty}
&\leq\p{2 + \sqrt{ \frac{|v(\Gamma)|  \s{coh}(\s{U})r}{n}}}\cdot
\frac{K_q}{\|\s{W}-\bE[\s{W}]\|_{\s{op}}}\sqrt{\frac{\s{coh}(\s{U}) r}{n} \log n}.
\label{eqn:PTnormY}
\end{align}
Looking at \eqref{eqn:ZerothBoundCon}, it is only left to deal $\|\s{Diag}(\mathbf{s}-\mathbf{s}^\star)\|_{\ell_1}$. Let $\delta(\Gamma^\star)$ denote the minimum degree of $\Gamma^\star$. The edge--vertex coupling constraints $\s{X}_{ij}\leq \min(\mathbf{s}_i,\mathbf{s}_j)$ together with $\innerP{\bar{\s{J}},\s{X}}=2|e(\Gamma)|$ imply
\begin{align}
\|\s{Diag}(\mathbf{s}-\mathbf{s}^\star)\|_{\ell_1}\leq\frac{1}{\delta(\Gamma^\star)}\|\s{X}-\s{X}^\star\|_{\ell_1}.
\label{eqn:couplingBound}
\end{align}
Indeed, let $\calN_i$ denote the set of neighbors of $i$ in $\Gamma^\star$. For any $i\in v(\Gamma^\star)$,
\begin{align}
\sum_{j\in\mathcal N_i} \s{X}_{ij}
\leq\sum_{j\in\mathcal N_i}\min\{\mathbf{s}_i,\mathbf{s}_j\}
\leq\s{deg}_{\Gamma^\star}(i)\mathbf{s}_i.
\end{align}
Since $ \sum_{j\in\mathcal N_i}\s{X}^\star_{ij}=\s{deg}_{\Gamma^\star}(i) $, we get
\begin{align}
\sum_{j\in\mathcal N_i}\p{\s{X}^\star_{ij}-\s{X}_{ij}}
\geq\s{deg}_{\Gamma^\star}(i)(1-\mathbf{s}_i).
\end{align}
Taking positive parts only makes the left-hand side larger, so
\begin{align}
\sum_{j\in\mathcal N_i}\p{\s{X}^\star_{ij}-\s{X}_{ij}}_+
\geq\s{deg}_{\Gamma^\star}(i)(1-\mathbf{s}_i)
\geq\delta(\Gamma^\star)(1-\mathbf{s}_i),\label{eqn:LLNeig}
\end{align}
where $ \delta(\Gamma^\star)\triangleq\min_{k\in v(\Gamma^\star)}\s{deg}_{\Gamma^\star}(k) $. Now, for any scalars $ a,b $, it holds that $ |a-b|=(a-b)_++(b-a)_+\ge(b-a)_+ $. Summing \eqref{eqn:LLNeig} over ordered pairs $ (i,j) $ with $ j\in\mathcal{N}_i $ and using symmetry,
\begin{align}
\|\s{X}-\s{X}^\star\|_{\ell_1}
=\sum_{i,j}|\s{X}_{ij}-\s{X}^\star_{ij}|
\geq2\sum_{i\in v(\Gamma^\star)}\sum_{j\in\mathcal N_i}\p{\s{X}^\star_{ij}-\s{X}_{ij}}_+
\geq2\delta(\Gamma^\star)\sum_{i\in v(\Gamma^\star)}(1-\mathbf{s}_i).
\end{align}
By feasibility $ \sum_i \mathbf{s}_i=\sum_i \mathbf{s}_i^\star=|v(\Gamma^\star)| $, and thus
\begin{align}
\sum_{i\in v(\Gamma^\star)}(1-\mathbf{s}_i)
=\sum_{i\notin v(\Gamma^\star)} \mathbf{s}_i
=\frac{1}{2}\sum_{i\in[n]} |\mathbf{s}_i-\mathbf{s}_i^\star|
=\frac{1}{2}\|\s{Diag}(s-s^\star)\|_{\ell_1}.
\end{align}
Combining together reveals that,
\begin{align}
\|\s{X}-\s{X}^\star\|_{\ell_1}
\geq2\delta(\Gamma^\star)\cdot \frac{1}{2}\|\s{Diag}(\mathbf{s}-\mathbf{s}^\star)\|_{\ell_1}
=\delta(\Gamma^\star)\|\s{Diag}(\mathbf{s}-\mathbf{s}^\star)\|_{\ell_1},
\end{align}
which is equivalent to
\begin{align}
\|\s{Diag}(\mathbf{s}-\mathbf{s}^\star)\|_{\ell_1}\leq\frac{1}{\delta(\Gamma^\star)}\|\s{X}-\s{X}^\star\|_{\ell_1},
\end{align}
which yields \eqref{eqn:couplingBound}. Thus, plugging \eqref{eqn:FirstBoundCon}, \eqref{eqn:SecondBoundCon}, \eqref{eqn:ProjectionBound}, \eqref{eqn:PTnormY}, and \eqref{eqn:couplingBound} in \eqref{eqn:ZerothBoundCon}, we get with probability at least $1-n^{-1}$,
\begin{align}
    &\innerP{\s{X}^\star-\s{X},\s{W}}\geq\left[\frac{c}{2}
    -\frac{C'_qr}{\sqrt{n}}\s{coh}(\s{U})
    -\p{2 + \sqrt{ \frac{|v(\Gamma)| \s{coh}(\s{U})r}{n}}}K_q\sqrt{\frac{\s{coh}(\s{U}) r}{n} \log n}\right.\nonumber\\
    &\left.
    -\frac{\alpha}{\delta(\Gamma)}\left(
    \frac{C'_qr}{\sqrt n}\s{coh}(\s{U})
    + \p{2 + \sqrt{ \frac{|v(\Gamma)| \s{coh}(\s{U})r}{n}}}K_q\sqrt{\frac{\s{coh}(\s{U}) r}{n} \log n}
    \right)\right]\norm{\s{X}-\s{X}^\star}_{\ell_1}.
\end{align}
Accordingly, since $\alpha$ is a fixed constant, and $\delta(\Gamma)\geq1$, there exist constants $c_1,c_2>0$ such that if
\begin{align}
    &\s{coh}(\s{U})r\leq c_1\sqrt{n}\label{eqn:TwoCond0}\\
    &\s{coh}(\s{U})r\sqrt{|v(\Gamma)|}\leq c_2\frac{n}{\sqrt{\log n}},\label{eqn:TwoCond1}
    %&\alpha\frac{1}{\delta(\Gamma)}\left(
    %\frac{\s{coh}(\s{U}) r}{\sqrt n}
    %+ \Big(2 + \sqrt{ \tfrac{|v(\Gamma)| \s{coh}(\s{U})r}{n}}\Big)\sqrt{\tfrac{\s{coh}(\s{U}) r}{n} \log n}
    %\right)\le\ c_3,\label{eqn:TwoCond2}
\end{align}
then, for any $\s{X}\neq\s{X}^\star$, we have $\innerP{\s{X}^\star-\s{X},\s{W}}>0$. Finally, let us check when \eqref{eqn:TwoCond0} dominates \eqref{eqn:TwoCond1}. Assuming that \eqref{eqn:TwoCond0} holds, it can be seen that \eqref{eqn:TwoCond1} is satisfied if 
\begin{align}
    |v(\Gamma)|\leq \frac{c_2^2}{c_1^2}\frac{n}{\log n},
\end{align}
otherwise, \eqref{eqn:TwoCond1} dominates, which concludes the proof.
\end{proof}

We finally prove Lemmata~\ref{lem:maxNormCohe} and \ref{lem:infBoundSymY}.
\begin{proof}[Proof of Lemma~\ref{lem:maxNormCohe}]
Since $\s{U}$ has orthonormal columns, $\s{P}=\s{U}\s{U}^{\top}$ is the orthogonal projector onto $\s{range}(\s{U})$. Its entries satisfy
\begin{align}
\s{P}_{ij} = \s{U}_{i,:}\s{U}_{j,:}^{\top} = u_i^{\top}u_j,
\end{align}
where $u_i\triangleq \s{U}_{i,:}\in\mathbb{R}^r$. Also $\s{P}_{ii}=\|u_i\|_2^2=\ell_i$ and $\sum_{i=1}^n\ell_i=\s{trace}(\s{P})=r$. By Cauchy--Schwarz inequality,
\begin{align}
|\s{P}_{ij}| = |u_i^{\top}u_j|
\leq \|u_i\|_2\|u_j\|_2
\leq \max_k \|\s{U}_{k,:}\|_2^2
= \max_k \ell_k
= \frac{r}{n}\s{coh}(\s{U}).
\end{align}
Taking the maximum over $i,j$ gives $\|\s{P}\|_{\ell_\infty}\leq\frac{r}{n}\s{coh}(\s{U})$. For the lower bound, note that the Frobenius norm satisfies
\begin{align}
\|\s{P}\|_F^2
= \s{trace}(\s{P}^2)
= \s{trace}(\s{P})
= r.
\end{align}
Since $\|\s{P}\|_F^2 = \sum_{i,j} \s{P}_{ij}^2 \leq n^2 \|\s{P}\|_{\ell_\infty}^2$, we obtain
\begin{align}
n^2 \|\s{P}\|_{\ell_\infty}^2 \geq r,
\end{align}
and thus $\|\s{P}\|_{\ell_\infty}\geq\frac{\sqrt{r}}{n}$. This proves the claim.
\end{proof}

\begin{proof}[Proof of Lemma~\ref{lem:infBoundSymY}]
By triangle inequality, note that,
\begin{align}
\norm{\calP_{\calT}(\s{Y})}_{\ell_\infty} \leq \norm{\s{P}\s{Y}}_{\ell_\infty}+\norm{\s{Y}\s{P}}_{\ell_\infty}+\norm{\s{P}\s{Y}\s{P}}_{\ell_\infty}
\end{align}
Since $\s{Y}$ is symmetric, $\norm{\s{Y}\s{P}}_{\ell_\infty} = \norm{\s{P}\s{Y}}_{\ell_\infty}$. Moreover,
\begin{align}
\norm{\s{P}\s{Y}\s{P}}_{\ell_\infty} \leq\norm{\s{P}}_{\ell_\infty \to \ell_\infty} \norm{\s{P}\s{Y}}_{\ell_\infty}.
\end{align}
Therefore,
\begin{align}
\norm{\calP_{\calT}(\s{Y})}_{\ell_\infty} \leq (2 + \norm{\s{P}}_{\ell_\infty \to \ell_\infty}) \cdot \norm{\s{P}\s{Y}}_{\ell_\infty}.\label{eqn:boundLLP0}
%\leq (2 + \sqrt{\s{coh}(\s{U}) r}) \cdot \norm{\s{P}\s{Y}}_{\ell_\infty}.
\end{align}
Let us bound $\norm{\s{P}}_{\ell_\infty \to \ell_\infty}$. Note that $\mathsf{U}$ is supported on a set $S \subset [n]$ of size $|v(\Gamma)|$ (i.e., $\mathsf{U}_{i,:} = 0$ for $i \notin S$). Thus, for any $i \notin S$, we have $\mathsf{P}_{i,:} = 0$. For $i \in S$,
\begin{align}
\| \mathsf{P}_{i,:} \|_1 = \sum_{j \in S} | \langle \mathsf{U}_{i,:}, \mathsf{U}_{j,:} \rangle |
\leq \sum_{j \in S} \| \mathsf{U}_{i,:} \|_2 \cdot \| \mathsf{U}_{j,:} \|_2
= \| \mathsf{U}_{i,:} \|_2 \cdot \sum_{j \in S} \| \mathsf{U}_{j,:} \|_2,
\end{align}
by Cauchy--Schwarz on each inner product. Apply Cauchy--Schwarz to the sum over $ j \in S $:
\begin{align}
\sum_{j \in S} \| \mathsf{U}_{j,:} \|_2
\leq \sqrt{|v(\Gamma)|} \left( \sum_{j \in S} \| \mathsf{U}_{j,:} \|_2^2 \right)^{1/2}
= \sqrt{|v(\Gamma)|} \cdot \| \mathsf{U} \|_F
= \sqrt{|v(\Gamma)| r},
\end{align}
since $ \| \mathsf{U} \|_F^2 = \mathrm{tr}( \mathsf{U}^\top \mathsf{U} ) = r $, and rows outside $ S $ are zero. From the definition of coherence,
\begin{align}
\| \mathsf{U}_{i,:} \|_2 \leq \max_t \| \mathsf{U}_{t,:} \|_2 = \sqrt{ \frac{\s{coh}(\s{U})r}{n} }.
\end{align}
Combining the above,
\begin{align}
\| \mathsf{P}_{i,:} \|_1 \leq \sqrt{ \frac{\s{coh}(\s{U})r}{n} } \cdot \sqrt{k r}
= \sqrt{ \frac{|v(\Gamma)| \s{coh}(\s{U})r}{n} }.
\end{align}
Taking the maximum over $ i $ gives
\begin{align}
\norm{\s{P}}_{\ell_\infty \to \ell_\infty}
= \max_i \| \mathsf{P}_{i,:} \|_1
\leq \sqrt{ \frac{|v(\Gamma)|\s{coh}(\s{U})r}{n}}.\label{eqn:boundLLP}
\end{align}
Thus, substituting \eqref{eqn:boundLLP} in \eqref{eqn:boundLLP0} we obtain
\begin{align}
\norm{\calP_{\calT}(\s{Y})}_{\ell_\infty} \leq \p{2 +\sqrt{ \frac{|v(\Gamma)|\s{coh}(\s{U})r}{n}}}\cdot \norm{\s{P}\s{Y}}_{\ell_\infty}.
\end{align}
Finally, for a diagonal matrix $\s{D}$, we have $[\s{P}\s{D}]_{ij} = \s{P}_{ij} \s{D}_{jj}$, for any $i,j\in[n]$, and thus
\begin{align}
\|\s{P}\s{D}\|_{\ell_\infty} \leq \norm{\s{P}}_{\ell_\infty} \cdot \norm{\s{D}}_{\ell_\infty}
\leq \|\s{U}\|_{2,\infty}^2 \cdot \norm{\s{D}}_{\ell_\infty},
\end{align}
so that
\begin{align}
\norm{\calP_{\calT}(\s{D})}_{\ell_\infty} \leq 3 \norm{\s{P}\s{D}}_{\ell_\infty}
\leq 3 \|\s{U}\|_{2,\infty}^2 \cdot \norm{\s{D}}_{\ell_\infty}
= 3\frac{\s{coh}(\s{U}) r}{n} \cdot\norm{\s{D}}_{\ell_\infty}.
\end{align}
\end{proof}

\section{Computational Bounds}\label{sec:complowerbounds}

\subsection{Lower bound}

We follow the same notations and definitions established in Section~\ref{sec:mainresults}. Recall that our goal is to upper bound
\begin{align}
\s{Corr}_{\le D}^2\leq\sum_{\alpha\in\{0,1\}^N:|\alpha|\le D}\ \frac{\kappa_\alpha^2}{(q_n(1-p_n))^{|\alpha|}},
\end{align}
where $\kappa_\alpha=\mathbb E[x\cdot \s{X}^\alpha]-\sum_{0\le \beta\lneq \alpha}\kappa_\beta\mathbb E[\s{X}^{\alpha-\beta}]$. Equivalently, $\kappa_\alpha$ is the joint cumulant of one copy of $x$ and the set of coordinates $\{\s{X}_e:\alpha_e=1\}$. Therefore, we see that the task of upper bounding the correlation reduces to the problem of upper bounding the joint cumulants. To that end, we introduce a few important notations.
\begin{definition}[Rooted pattern]
A rooted pattern is a finite simple graph $\s{H}=(v(\s{H}),e(\s{H}))$ together with a distinguished
vertex $r^\star\in v(\s{H})$, called the root. We say that $\s{H}$ is \emph{connected} if the underlying graph is connected. %The size of $\s{H}$ is given by its number of vertices $|v(\s{H})|=t$ and number of edges $|e(\s{H})|=d$.
\end{definition}
We recall the definitions of injective homomorphism and embedding\begin{definition}[Injective homomorphism and embedding]
Let $\s{G}=(v(\s{G}),e(\s{G}))$ be a finite graph. An \emph{injective homomorphism} (or \emph{embedding}) $\psi:\s{H}\hookrightarrow \s{G}$ is an injective map $\psi:v(\s{H})\to v(\s{G})$ such that
\begin{align}
\{u,v\}\in e(\s{H}) \;\implies\; \{\psi(u),\psi(v)\}\in e(\s{G}).
\end{align}
Thus, an embedding preserves adjacency and does not identify distinct vertices of $\s{H}$.
\end{definition}

\begin{definition}[Root-preserving embedding]
Let $\s{H}=(v(\s{H}),e(\s{H}),r^\star)$ be a rooted pattern and let $v\in v(\s{G})$ be a vertex of a graph $v(\s{G})$. A \emph{root-preserving embedding} of $\s{H}$ into $\s{G}$ with root at $v$ is an injective homomorphism
$\psi:\s{H}\hookrightarrow \s{G}$ such that $\psi(r^\star)=v$.
\end{definition}

\begin{definition}[Rooted embedding counts]
Fix integers $t\ge 2$ and $d\ge 1$. For a template graph $\Gamma_n$ and a vertex $v\in v(\Gamma_n)$,
we define
\begin{align}
\s{Emb}_{t,d}(\s{H}\to\Gamma_n,v)&\triangleq\left|\left\{\psi:\s{H}\hookrightarrow \Gamma_n:\; \s{H}\; \s{is}\; \s{a}\; \s{connected}\; r^\star\text{-}\s{rooted}\; \s{pattern}\; \s{with}\; \right.\right.\nonumber\\
&\left.\left.\hspace{3.3cm}|v(\s{H})|=t,\ |e(\s{H})|=d,\ \psi(r^\star)=v\right\}\right|.
\end{align}
That is, $\s{Emb}_{t,d}(\s{H}\to\Gamma_n;v)$ is the number of root-preserving embeddings of all
connected rooted patterns with $t$ vertices and $d$ edges, where the root of the pattern is mapped to
the specified vertex $v$ of $\Gamma_n$. We also define the total count
\begin{align}
\s{Emb}_{t,d}(\s{H}\to\Gamma_n)
\triangleq \sum_{v\in v(\Gamma_n)}\s{Emb}_{t,d}(\s{H}\to\Gamma_n;v),
\end{align}
which sums over all possible root locations in $\Gamma_n$. Finally, we define
\begin{align}
\s{Emb}_{t,d}(\Gamma_n;v)&\triangleq\sum_{\substack{\s{H}:\,|v(\s{H})|=t,\,|e(\s{H})|=d\\\s{H}\;\s{is}\;\s{connected}\;\s{root}\;r^\star}}
\s{Emb}_{t,d}(\s{H}\to\Gamma_n;v),\\
\s{Emb}_{t,d}(\Gamma_n)&\triangleq\sum_{v\in v(\Gamma_n)} \s{Emb}_{t,d}(\Gamma_n;v).\label{eqn:defEmbDTotalG}
\end{align}
\end{definition}
The role of the root in the above definition is to enforce alignment with a fixed \emph{anchor} vertex in the ambient graph (say vertex $1\in[n]$). The distinction between $\s{Emb}_{t,d}(\s{H}\to\Gamma_n,v)$ and $\s{Emb}_{t,d}(\s{H}\to\Gamma_n)$ is that the former counts embeddings anchored at a specific template vertex, while the latter sums over all template roots. Next, let $\s{H}$ be a rooted pattern with $t$ vertices and $d$ edges. Let $\phi:v(\Gamma_n)\hookrightarrow [n]$ be a uniformly random injection, and set $\Gamma_n^\star=\phi(\Gamma_n)$. Consider the event
\begin{align}
\calE(\s{H})&\triangleq\left\{\exists\, v\in v(\Gamma_n)\; \s{and}\; \psi:\s{H}\hookrightarrow \Gamma_n\; \s{injective}\;\s{with}\;\psi(r^\star)=v,\right.\nonumber\\
&\left.\hspace{4.2cm} \s{and}\; (\phi\circ\psi)(r^\star)=v^\star\right\}.\label{eqn:calEHalpha}
\end{align}
%Equivalently, writing $f=\phi\circ\psi$ with $\psi:v(H)\hookrightarrow v(\Gamma_n)$ an injective homomorphism into the \emph{template} $\Gamma_n$, the event $\mathsf{E}(H)$ says: \emph{“There exists a root-preserving template embedding $\psi:H\hookrightarrow\Gamma_n$ and the random labeling $\phi$ maps the root $\psi(r)$ to the ambient anchor $1$.”}
Intuitively, $\calE(\s{H})$ is the event that the rooted pattern $\s{H}$ appears inside the planted copy $\Gamma_n^\star$ with the root landing at the ambient anchor $v^\star$. The following lemma is a key ingredient in the proof of our upper bound on $\s{Corr}_{\leq D}$.
\begin{lemma}\label{lem:probEmbedding}
For any connected rooted pattern $\s{H}$ as above,
\begin{align}
\pr\pp{\calE(\s{H})}
=
\frac{1}{|v(\Gamma_n)|}\cdot \frac{\s{Emb}_{t,d}(\s{H}\to\Gamma_n)}{(n-1)_{t-1}},
\end{align}
where $(n-1)_{t-1}\triangleq(n-1)(n-2)\cdots(n-t+1)$.
\end{lemma}
\begin{proof}[Proof of Lemma~\ref{lem:probEmbedding}]
Let $\calR\triangleq\phi^{-1}(v^\star)$ be the (random) template vertex mapped by $\phi$ to the ambient anchor $v^\star$. Because $\phi$ is a uniform injection, $\calR$ is uniform on $v(\Gamma_n)$, hence
\begin{align}
\pr\pp{\calE(\s{H})}=\frac{1}{|v(\Gamma_n)|}\sum_{v\in v(\Gamma_n)}\pr\pp{\calE(\s{H})\vert\calR=v}.
\end{align}
Fix $v\in v(\Gamma_n)$. Conditional on $\calR=v$, we have $\phi(v)=v^\star$, and the remaining $k-1$ template vertices are mapped to $[n]\setminus\{v^\star\}$ uniformly without replacement. For a \emph{fixed} root-preserving template embedding $\psi:\s{H}\hookrightarrow\Gamma_n$ with $\psi(r^\star)=v$, define the event
\begin{align}
\calE_\psi\triangleq\ppp{\phi(\psi(w))=a_w \ \text{for all } w\in v(\s{H})\setminus\{r^\star\}},
\end{align}
where $\{a_w:w\neq r\}\subset[n]\setminus\{v^\star\}$ is any \emph{fixed} set of $t-1$ distinct ambient vertices that we want the non-root pattern vertices to land on. Conditional on $\calR=v$, the probability that $\phi$ realizes $\calE_\psi$ is exactly
\begin{align}
\pr\pp{\calE(\s{H})\vert\calR=v}=\frac{1}{(n-1)_{t-1}},
\end{align}
because once $\phi(v)=v^\star$ is fixed, the remaining $t-1$ pattern vertices must occupy the $t-1$ distinct labels $\{a_w\}$ in a one-to-one fashion, and the remaining $|v(\Gamma_n)|-1$ labels are uniform without replacement.

Crucially, for a fixed $v$, the events $\{\calE_\psi\}_\psi$ over \emph{distinct} root-preserving embeddings $\psi$ are disjoint: two different $\psi$'s require $\phi$ to send (at least) one different template vertex to a \emph{specific} ambient label; since $\phi$ is injective and the target $\{a_w\}$ has size $t-1$, no single $\phi$ can satisfy two distinct $\psi$'s simultaneously. Therefore
\begin{align}
\pr\pp{\calE(\s{H})\vert\calR=v}
&=\sum_{\psi:\ \psi(r^\star)=v}\pr(\calE_\psi\vert\calR=v)\\
&=\s{Emb}^{\mathrm{vert}}_{t,d}(\s{H}\to\Gamma_n;v)\cdot \frac{1}{(n-1)_{t-1}}.
\end{align}
Averaging over $v$ gives
\begin{align}
\pr\pp{\calE(\s{H})}
&=\frac{1}{|v(\Gamma_n)|}\sum_{v\in v(\Gamma_n)}\s{Emb}^{\mathrm{vert}}_{t,d}(\s{H}\to\Gamma_n;v)\cdot \frac{1}{(n-1)_{t-1}}\\
&=\frac{1}{|v(\Gamma_n)|}\cdot \frac{\s{Emb}_{t,d}(\s{H}\to\Gamma_n)}{(n-1)_{t-1}},
\end{align}
which comcludes the proof.
\end{proof}
We are now ready to state and prove our upper bound on $\s{Corr}_{\leq D}$. 
\begin{lemma}\label{lem:LDPCorrUp}
Let $D\ge1$ and $0<q_n<p_n<1$. With the notation above, let $\lambda_n=(p_n-q_n)/(\sqrt{q_n(1-p_n)})$. Then
\begin{align}
\s{Corr}_{\le D}^2 \leq\bE^2[x]+\sum_{d=1}^D\sum_{t=2}^{(d+1)\wedge |v(\Gamma_n)|}[(d+1)t\lambda_n]^{2d}\frac{\s{Emb}_{t,d}(\Gamma_n)}{|v(\Gamma_n)|^2(n-1)_{t-1}},\label{eqn:CorrUpperBoundFinal}
\end{align}
where $(n-2)^{t-2}\triangleq(n-2)(n-3)\cdots(n-t)$.\label{eqn:MMSEUpperBoundFinal}
Consequently,
\begin{align}
\s{MMSE}_{\le D}\ge \s{Var}(x)-\sum_{d=1}^D\sum_{t=2}^{(d+1)\wedge |v(\Gamma_n)|}[(d+1)t\lambda_n]^{2d}\frac{\s{Emb}_{t,d}(\Gamma_n)}{|v(\Gamma_n)|^2(n-1)_{t-1}}.
\end{align}
\end{lemma}

\begin{proof}[Proof of Lemma~\ref{lem:LDPCorrUp}]
Recall that $\{\s{X}_e\}_e$ are the mean parameters of the Bernoulli observations, i.e., $\s{X}_e=\mathbb{E}[\s{Y}_e\vert\Gamma_n]\in\{q_n,p_n\}$, and $x=\Ind\{v^\star\in\Gamma_n\}$. Using \cite[Claim 2.14]{SchrammWein2022} and \cite[Prop. 2.13]{SchrammWein2022}, with $\s{I}_e\triangleq\Ind\{e\in \Gamma_n\}$ we have
\begin{align}
\kappa_\alpha=\kappa(x,\{\s{X}_e\}_{e\in\alpha})=(p_n-q_n)^{|\alpha|}\kappa(x,\{\s{I}_e\}_{e\in\alpha}),
\end{align}
and thus
\begin{align}
\s{Corr}_{\leq D}^2 \leq \sum_{\alpha\in\{0,1\}^N:|\alpha|\le D} \lambda_n^{2|\alpha|}(\kappa(x,\{\s{I}_e\}_{e\in\alpha}))^2,\label{eqn:GroupSum}
\end{align}
where $\lambda_n\triangleq\frac{p_n-q_n}{\sqrt{q_n(1-p_n)}}$. Next, fix any edge-set $\alpha\subset \binom{[n]}{2}$ of size $d = |\alpha|$, whose union with a fixed anchor $v^\star$ spans $t$ vertices; let $\s{H}_\alpha$ be the induced ambient rooted pattern (root at $v^\star$). Note that $t \leq d+1$. Define
\begin{align}
\s{Z}_0&\triangleq x=\Ind\{v^\star\in \Gamma_n\},\\
\s{Z}_j&\triangleq\Ind\{e_j\in \Gamma_n\},
\end{align}
for $1\le j\le d$, where $\{e_1,\dots,e_d\}=\alpha$. By the combinatorial formula for joint cumulants (see, \cite[Def. 2.10]{SchrammWein2022}),
\begin{align}
\kappa(\s{Z}_0,\dots,\s{Z}_d)=\sum_{\pi\in\calP_{d+1}} (|\pi|-1)!  (-1)^{|\pi|-1}\prod_{B\in b(\pi)}\mathbb{E}\pp{\prod_{j\in B}\s{Z}_j},
\end{align}
where $\mathcal{P}_{d+1}$ denotes the set of all partitions of $[d+1]$ (that is, partitions of $d+1$ labeled elements into nonempty, unlabeled blocks). For a partition $\pi\in\mathcal{P}_{d+1}$, we write $b(\pi)$ for the collection of its blocks and $|\pi|$ for the number of blocks. Taking absolute values and bounding $\mathbb{E}\pp{\prod_{j\in B}\s{Z}_j}\le1$, and using
\begin{align}
\sum_{\pi\in\calP_{d+1}} (|\pi|-1)! = \sum_{k=1}^{d+1} S(d+1,k)(k-1)!\le \sum_{k=1}^{d+1} k^d \le (d+1)^{d+1}\triangleq C_d,
\end{align}
we obtain
\begin{align}
|\kappa(\s{Z}_0,\dots,\s{Z}_d)| \le C_d\cdot\max_{\pi\in\calP_{d+1}} \prod_{B\in b(\pi)}\mathbb{E}\pp{\prod_{j\in B}\s{Z}_j}.
\end{align}
Among all blocks $B$, the block containing $\s{Z}_0=x$ yields the smallest event, hence for every partition $\pi$,
\begin{align}
\prod_{B\in b(\pi)}\mathbb{E}\pp{\prod_{j\in B}\s{Z}_j}&\leq\mathbb{E}\pp{x\prod_{e\in\alpha}\s{I}_e}\\
&=\mathbb{P}\ppp{\{v^\star\}\cup\alpha\subseteq \Gamma_n}.
\end{align}
Therefore
\begin{align}
|\kappa(x,\{\s{I}_e\}_{e\in\alpha})|\leq C_d\cdot\mathbb{P}\{\{v^\star\}\cup\alpha\subseteq \Gamma_n\}.
\end{align}
But the event $\{v^\star\}\cup\alpha\subseteq \Gamma_n$ is exactly $\calE(\s{H}_\alpha)$ in \eqref{eqn:calEHalpha}. Thus, by the Lemma~\ref{lem:probEmbedding}
\begin{align}
\mathbb{P}\ppp{\{v^\star\}\cup\alpha\subseteq \Gamma_n}&=\frac{1}{|v(\Gamma_n)|}\cdot \frac{\s{Emb}_{t,d}(\s{H}_\alpha\to\Gamma_n)}{(n-1)_{t-1}}\\
&\leq \frac{1}{|v(\Gamma_n)|}\cdot \frac{\s{Emb}_{t,d}(\Gamma_n)}{(n-1)_{t-1}},
\end{align}
where the inequality follows from the definition in \eqref{eqn:defEmbDTotalG}. Combining that above gives
\begin{align}
|\kappa(x,\{\s{I}_e\}_{e\in\alpha})|\leq \frac{C_d}{|v(\Gamma_n)|}\cdot \frac{\s{Emb}_{t,d}(\Gamma_n)}{(n-1)_{t-1}}.\label{eqn:KappaEmbdUpp}
\end{align}
Now, for fixed $d,t$, the number of ambient edge-sets $\alpha$ with $|\alpha|=d$ whose union with $v^\star$ spans exactly $t$ vertices is at most
\begin{align}
\binom{n-1}{t-1}\binom{\binom{t}{2}}{d}\leq (n-1)_{t-1}t^{2d}
\leq n^{t-1}t^{2d}.\label{eqn:EstEdgeVerE}
\end{align}
Indeed, choose the additional $t-1$ vertices among $n-1$ options, then choose $d$ edges among the $\binom{t}{2}$ possible on those $t$ vertices. The crude bounds $\binom{n-1}{t-1}\leq (n-1)_{t-1}\leq n^{t-1}$ and $\binom{\binom{t}{2}}{d}\leq \binom{t^2/2}{d}\leq t^{2d}$ yield the claim. Grouping the sum \eqref{eqn:GroupSum} by $(t,d)$, and applying \eqref{eqn:KappaEmbdUpp}--\eqref{eqn:EstEdgeVerE}, we obtain
\begin{align}
\s{Corr}^2_{\le D}
&\leq\bE^2[x]+\sum_{d=1}^D\sum_{t=2}^{(d+1)\wedge |v(\Gamma_n)|}  (n-1)_{t-1}t^{2d}\left((d+1)^d\lambda_n^d\frac{\s{Emb}_{t,d}(\Gamma_n)}{|v(\Gamma_n)|(n-1)_{t-1}}\right)^2\\
& =\bE^2[x]+\sum_{d=1}^D\sum_{t=2}^{(d+1)\wedge |v(\Gamma_n)|}[(d+1)t\lambda_n]^{2d}\frac{\s{Emb}^2_{t,d}(\Gamma_n)}{|v(\Gamma_n)|^2(n-1)_{t-1}}.
\end{align}
Finally, the identity $\s{MMSE}_{\le D}=\mathbb{E}[x]-\s{Corr}_{\le D}^2$ yields the MMSE bound.
\end{proof}
Next, to prove Theorem~\ref{thm:compLower}, we show that, roughly speaking, the slice corresponding to $(t,d) = (2,1)$ in the sum in \eqref{eqn:CorrUpperBoundFinal} dominates the entire expression. The slice $(t,d) = (2,1)$ corresponds to a single planted edge incident to the anchor. In this case, the root-averaged count equals the \emph{average degree} of the planted template:
\begin{align}
\s{Emb}_{2,1}(\Gamma_n)
=\sum_{v\in v(\Gamma_n)} \s{deg}_{\Gamma_n}(v)=2\cdot|e(\Gamma_n)|.
\end{align}
This in turn implies that $\s{Corr}^2_{\le D} = o(1)$ (and hence recovery is computationally hard) whenever
\begin{align}
    \frac{\s{Emb}^2_{2,1}(\Gamma_n)}{|v(\Gamma_n)|^2n}\ll1\implies\eta(\Gamma_n)=\frac{|e(\Gamma_n)|}{|v(\Gamma_n)|}\ll\sqrt{n}.
\end{align}
Let us now establish this rigorously.

\begin{lemma}\label{lem:anchor-avg-emb}
Let $\Gamma_n$ be a simple graph on $k$ vertices with average degree
\begin{align}
\bar{d}(\Gamma_n)=\frac{1}{k}\sum_{u\in v(\Gamma_n)}\s{deg}_{\Gamma_n}(u)=\frac{2|e(\Gamma_n)|}{k}.
\end{align}
Fix integers $t\ge 2$ and $d\ge t-1$. There exists an absolute constant $C\ge 1$ such that
\begin{equation}\label{eq:anchor-avg-emb}
\frac{1}{k}  \s{Emb}_{t,d}(\Gamma_n)
\leq 
\big(Ct\big)^{  C(d+t)}  \bar{d}(\Gamma_n)^{  t-1}.
\end{equation}
\end{lemma}

We begin with an intuitive sketch of the lemma before giving the formal proof. 
Our goal is to bound the anchor-averaged number of rooted embeddings of any $t$-vertex, $d$-edge connected pattern into $\Gamma_n$. The proof reinterprets every embedding as a \emph{growth process}: starting from the root, we reveal the pattern one vertex at a time according to a chosen 
\emph{exploration scheme} (a spanning tree together with an exposure order). At each step the number of valid extensions is controlled by the \emph{edge boundary} of the already-embedded set in~$\Gamma$. Averaging over relabelings makes every vertex of $\Gamma$ ``typical'', so that 
on average an $s$-set has boundary size about $s\cdot\bar{d}(\Gamma)$. This exchangeability turns the complicated boundary terms into a clean  multiplicative factor, leading to the bound
$\prod_{s=1}^{t-1} [s\cdot\bar{d}(\Gamma)]$. The only remaining combinatorial work is to account for the number of possible exploration schemes, which contributes an additional factor $(Ct)^{C(d+t)}$. 

\begin{proof}[Proof of Lemma~\ref{lem:anchor-avg-emb}] Let $\mathfrak H_{t,d}$ denote the family of connected rooted patterns on $t$ vertices and $d$ edges (root distinguished but otherwise labeled). The proof proceeds in five steps.

\paragraph{1) Exploration schemes.}
For a connected rooted pattern $H\in\mathfrak H_{t,d}$, fix a rooted spanning tree $T\subseteq H$ of size $t-1$, and fix an \emph{exploration order} $r=v_1,v_2,\dots,v_t$ in which each $v_{s+1}$ is adjacent in $T$ to some earlier vertex among $v_1,\dots,v_s$.
An \emph{exploration scheme} $\mathsf S$ comprises:
(i) the rooted spanning tree $T$;
(ii) the exploration order;
(iii) the choice of the extra (non-tree) edges of $H$ (there are $d-(t-1)$ of them).
Let $\mathcal S_{t,d}$ be the set of all exploration schemes over all $H\in\mathfrak H_{t,d}$.

A crude combinatorial bound suffices:
choose $T$ (at most $t^{  t-2}$ rooted labeled trees by Cayley),
choose an order (at most $t!$),
and choose the extra edges (at most $\binom{\binom{t}{2}}{d-(t-1)}\le (Ct)^{  2(d-(t-1))}$).
Thus
\begin{equation}\label{eq:scheme-count}
|\mathcal S_{t,d}|
\leq 
t^{  t-2}\cdot t!\cdot (Ct)^{  2(d-(t-1))}
\leq  (Ct)^{  C(d+t)}.
\end{equation}

\paragraph{2) Partial histories and the boundary recursion.}
Fix a relabeling $\sigma$ of $v(\Gamma_n)$, and write $\Gamma_n^\sigma$ for the relabeled graph.
For an exploration scheme $\mathsf S\in\mathcal S_{t,d}$ and an anchor $v\in v(\Gamma_n^\sigma)$, a \emph{partial history of length $s$} ($1\le s\le t$)
is a choice of an injective map sending $v_1,\dots,v_s$ to distinct vertices of $v(\Gamma_n^\sigma)$ that respects the (tree) adjacencies required by $\mathsf S$ among $v_1,\dots,v_s$.
Let $\mathcal F_s(\sigma)$ be the multiset of all partial histories of length $s$, formed by ranging over all anchors $v\in v(\Gamma_n^\sigma)$ and all schemes $\mathsf S\in\mathcal S_{t,d}$.

For $f\in\mathcal F_s(\sigma)$ write $U(f)\subseteq v(\Gamma_n^\sigma)$ for the current image (so $|U(f)|=s$), and
\begin{align}
\partial_{\Gamma_n^\sigma}(U)\triangleq\big|\{  \{x,y\}\in e(\Gamma_n^\sigma): x\in U,\ y\notin U  \}\big|
\end{align}
for the (undirected) edge boundary size of $U$.
Each extension from $s$ to $s+1$ chooses the next image vertex $v_{s+1}$ outside $U(f)$ that is adjacent, in $\Gamma_n^\sigma$, to the specified predecessor of $v_{s+1}$ in the tree $T$; consequently, for every $f\in\mathcal F_s(\sigma)$,
the number of admissible choices is at most $\partial_{\Gamma_n^\sigma}(U(f))$.
Summing over all $f$ yields the recursion
\begin{equation}\label{eq:histories-recursion-pointwise}
|\mathcal F_{s+1}(\sigma)|
\leq \sum_{f\in\mathcal F_s(\sigma)} \partial_{\Gamma_n^\sigma}\big(U(f)\big).
\end{equation}
Indeed, non-tree edges in $\mathsf S$ impose additional adjacency constraints and thus only \emph{decrease} the number of admissible choices, so \eqref{eq:histories-recursion-pointwise} remains valid.

\paragraph{3) Averaging over relabelings.}
By construction, $\s{Emb}_{t,d}(\Gamma_n)$ is invariant under relabelings of $v(\Gamma_n)$, so
\begin{align}
\s{Emb}_{t,d}(\Gamma_n)&=\mathbb E_\sigma\big[\s{Emb}_{t,d}(\Gamma_n^\sigma)\big]\\
&\leq \mathbb E_\sigma\big[  |\mathcal F_t(\sigma)|  \big],
\end{align}
because every full root-preserving embedding (for some $H$ and some $\mathsf S$) gives rise to at least one full history of length $t$.
Taking expectations of \eqref{eq:histories-recursion-pointwise}, we are reduced to bounding
\begin{align}
\mathbb E_\sigma\left[\sum_{f\in\mathcal F_s(\sigma)} \partial_{\Gamma_n^\sigma}\big(U(f)\big)\right]
\end{align}
in terms of $\mathbb E_\sigma\big[  |\mathcal F_s(\sigma)|  \big]$. We have the following lemma.

\begin{lemma}[Average boundary at size $s$]\label{cl:avg-boundary}
For each $1\le s\le t-1$,
\begin{equation}\label{eq:avg-boundary}
\mathbb E_\sigma\left[\sum_{f\in\mathcal F_s(\sigma)} \partial_{\Gamma_n^\sigma}\big(U(f)\big)\right]
\leq 
s  \bar{d}(\Gamma_n)\cdot \mathbb E_\sigma\big[  |\mathcal F_s(\sigma)|  \big].
\end{equation}
\end{lemma}

\begin{proof}[Proof of Lemma~\ref{cl:avg-boundary}]
For any $U\subseteq v(\Gamma_n^\sigma)$ we have $\partial_{\Gamma_n^\sigma}(U)\le\sum_{u\in U}\s{deg}_{\Gamma_n^\sigma}(u)$.
Thus
\begin{align}
\sum_{f\in\mathcal F_s(\sigma)} \partial_{\Gamma_n^\sigma}(U(f))
&\leq \sum_{f\in\mathcal F_s(\sigma)} \sum_{u\in U(f)} \s{deg}_{\Gamma_n^\sigma}(u)\\
&=\sum_{u\in v(\Gamma_n^\sigma)} \s{deg}_{\Gamma_n^\sigma}(u)\cdot N_s(u;\sigma),
\end{align}
where $N_s(u;\sigma)\triangleq|\{  f\in\mathcal F_s(\sigma): u\in U(f)  \}|$ is the multiplicity with which $u$ appears among the size-$s$ images.
By vertex-exchangeability under the uniform relabeling $\sigma$, $\mathbb E_\sigma[N_s(u;\sigma)]$ is the \emph{same} for all $u$.
Moreover,
\begin{align}
\sum_{u\in v(\Gamma_n^\sigma)} N_s(u;\sigma)&=\sum_{f\in\mathcal F_s(\sigma)} |U(f)|\\
&=s  |\mathcal F_s(\sigma)|,
\end{align}
hence $\mathbb E_\sigma[N_s(u;\sigma)]=(s/k)  \mathbb E_\sigma[|\mathcal F_s(\sigma)|]$ for every $u$. Taking expectations and using the fact that $\sum_{u}\s{deg}_{\Gamma_n^\sigma}(u)=2|e(\Gamma_n)|=k  \bar{d}(\Gamma_n)$ gives
\begin{align}
\mathbb E_\sigma\left[\sum_{f\in\mathcal F_s(\sigma)} \partial_{\Gamma_n^\sigma}(U(f))\right]
&\leq\sum_{u} \s{deg}_{\Gamma_n}(u)\cdot \frac{s}{k} \cdot \mathbb E_\sigma[  |\mathcal F_s(\sigma)|  ]\\
&=s  \bar{d}(\Gamma_n)\cdot\mathbb E_\sigma[  |\mathcal F_s(\sigma)|  ],
\end{align}
proving \eqref{eq:avg-boundary}.
\end{proof}

\paragraph{4) Induction and base size.}
Combining \eqref{eq:histories-recursion-pointwise} and Claim~\ref{cl:avg-boundary}, and iterating for $s=1,2,\dots,t-1$, we obtain
\begin{align}\label{eq:Ft-bound}
\mathbb E_\sigma\big[  |\mathcal F_t(\sigma)|  \big]
&\leq
\left(\prod_{s=1}^{t-1} s  \bar{d}(\Gamma_n)\right)\cdot \mathbb E_\sigma\big[  |\mathcal F_1(\sigma)|  \big]\\
&\leq t!  \bar{d}(\Gamma_n)^{  t-1}\cdot \mathbb E_\sigma\big[  |\mathcal F_1(\sigma)|  \big].
\end{align}
A partial history of length $1$ consists of choosing an anchor $v\in v(\Gamma_n^\sigma)$ and a scheme $\mathsf S\in\mathcal S_{t,d}$; hence $|\mathcal F_1(\sigma)|=k  |\mathcal S_{t,d}|$, deterministically.
Using \eqref{eq:scheme-count} and $t!\le (Ct)^t$, \eqref{eq:Ft-bound} yields
\begin{align}
\mathbb E_\sigma\big[  |\mathcal F_t(\sigma)|  \big]
\leq
k\cdot (Ct)^{  C(d+t)}\cdot \bar{d}(\Gamma_n)^{  t-1}.
\end{align}

\paragraph{5) From histories to embeddings.}
Every root-preserving embedding counted by $\s{Emb}_{t,d}(\Gamma_n^\sigma)$ corresponds to at least one full history (for some $\mathsf S$), so
$\s{Emb}_{t,d}(\Gamma_n^\sigma)\le |\mathcal F_t(\sigma)|$ and hence
\begin{align}
\s{Emb}_{t,d}(\Gamma_n)
&=\mathbb E_\sigma\big[\s{Emb}_{t,d}(\Gamma_n^\sigma)\big]\\
&\leq \mathbb E_\sigma\big[  |\mathcal F_t(\sigma)|  \big]\\
&\leq k\cdot (Ct)^{  C(d+t)}\cdot \bar{d}(\Gamma_n)^{  t-1}.
\end{align}
Dividing both sides by $k$ gives \eqref{eq:anchor-avg-emb}.
\end{proof}

We are now in a position to finish the proof of Theorem~\ref{thm:compLower}. We first recall our upper bound in Lemma~\ref{lem:LDPCorrUp}:
\begin{equation}\label{eq:corr-master}
\s{Corr}_{\leq D}^2
\leq\mathbb{E}^2[x]+
\sum_{d=1}^{D}\ \sum_{t=2}^{(d+1)\wedge k_n}
[(d+1)t\lambda_n]^{2d}\;
\frac{\mathsf{Emb}_{t,d}(\Gamma_n)^{2}}{k_n^{2}\,(n-1)_{t-1}}.
\end{equation}
Lemma~\ref{lem:anchor-avg-emb} state that
\begin{equation}\label{eq:emb-lemma}
\frac{1}{k_n}\mathsf{Emb}_{t,d}(\Gamma_n)\leq(C t)^{C(d+t)}\bar{d}(\Gamma_n)^{t-1},
\end{equation}
for a universal $C\geq1$. Since $t\leq d+1\leq D+1=o(n)$, we have the falling-factorial lower bound
\begin{equation}\label{eq:falling}
(n-1)_{t-1}\geq\Big(\frac{n}{2}\Big)^{t-1},
\end{equation}
for all large $n$. Define
\begin{equation}\label{eq:def-rn}
r_n\triangleq\frac{\bar{d}(\Gamma_n)^2}{n}\leq n^{-2\varepsilon},
\end{equation}
where the inequality follows from the assumptions of Theorem~\ref{thm:compLower}, for any $\varepsilon>0$. Plugging \eqref{eq:emb-lemma} and \eqref{eq:falling} into \eqref{eq:corr-master}, and noting the connectivity constraint $d\ge t-1$, we obtain
\begin{align}
\s{Corr}_{\le D}^2
&\le\ \mathbb{E}^2[x]+\sum_{t=2}^{D+1}\sum_{d=t-1}^{\min\{D,\binom{t}{2}\}}
[(d+1)t\lambda_n]^{2d}(C t)^{2C(d+t)}\cdot(2r_n)^{t-1}\label{eq:after-plug}\\
&=\mathbb{E}^2[x]+\sum_{t=2}^{D+1}\Phi_t\cdot(2r_n)^{t-1},\label{eq:def-Phi}
\end{align}
where for any $t\geq2$,
\begin{equation}\label{eq:Phi}
\Phi_t\triangleq\sum_{d=t-1}^{\min\{D,\binom{t}{2}\}}
[(d+1)t\lambda_n]^{2d}(C t)^{2C(d+t)}.
\end{equation}
Because we only sum over simple connected patterns, we always have $d\le \binom{t}{2}\le t^2/2$.
For such $d$,
\begin{align}
(d+1)t\lambda_n\leq ct^3,
\end{align}
for some $c\geq1$, and hence
\begin{align}
[(d+1)t\lambda_n]^{2d}\leq(ct^3)^{2d}\leq (ct^3)^{t^2}=\exp\p{t^2\log(c t^3)}\leq\exp\p{C_1\,t^2\log t},
\end{align}
for a universal $C_1$. Also,
\begin{align}
(C t)^{2C(d+t)}\leq(C t)^{2C(t^2/2+t)}\leq\exp\p{C_2t^2\log t}.
\end{align}
Multiplying the two results above and summing over at most $\binom{t}{2}-(t-1)+1=O(t^2)$ values of $d$, we can absorb the polynomial factor into the exponential envelope and get
\begin{align}\label{eq:Phi-envelope}
\Phi_t\leq\exp\p{C''t^2\log t},
\end{align}
for for some universal $C''\geq1$. Now, for $t\geq2$ define
\begin{align}\label{eq:at}
\alpha_t\triangleq\exp\p{C''t^2\log t}\cdot (2r_n)^{t-1}.
\end{align}
Then \eqref{eq:def-Phi} and \eqref{eq:Phi-envelope} give
\begin{align}\label{eq:corr-a_t}
\s{Corr}_{\le D}^2\leq\mathbb{E}^2[x]+\sum_{t=2}^{D+1}\alpha_t.
\end{align}
We next show that $\sum_{t=2}^{D+1} \alpha_t=o(1)$.
\paragraph{Case I.} Consider the regime where $D=o\p{\frac{\log n}{\log\log n}}$, and let $t_\star\triangleq D+1$. Since $r_n=n^{-2\varepsilon}$,
\begin{align}
\log \alpha_{t_\star}\leq C''D^2\log D-2\varepsilon D\log n+D\log 2.
\end{align}
Since $D\log D=o(\log n)$ in this regime, the negative term dominates, so $\log \alpha_{t_\star}=-\omega(1)$ and thus $\alpha_{t_\star}=n^{-\omega(1)}$. Because the sum in \eqref{eq:corr-a_t} has at most $D=o(\frac{\log n}{\log\log n})$ terms and the sequence $\{a_t\}$ is positive,
\begin{align}
\sum_{t=2}^{D+1} \alpha_t\leq D\cdot \alpha_{t_\star}=o(1).
\end{align}
\paragraph{Case II.} Consider the case where $D\le(\log n)^{\alpha}$, with a fixed $\alpha<1$. Again take $t_\star\triangleq D+1\leq 2(\log n)^{\alpha}$ and note that
\begin{align}
\log \alpha_{t_\star}\leq C''t_\star^2\log t_\star-2\varepsilon(t_\star-1)\log n+(t_\star-1)\log 2.
\end{align}
Using $t_\star\leq 2(\log n)^{\alpha}$ and $\log t_\star\leq \log[2(\log n)^{\alpha}]=O(\log\log n)$, we get
\begin{align}
\log \alpha_{t_\star}\leq C_3(\log n)^{2\alpha}\log\log n-2\varepsilon(\log n)^{\alpha+1}+O\p{(\log n)^{\alpha}}.
\end{align}
Because $\alpha+1>2\alpha$ for every $\alpha<1$, the negative term dominates, hence $\log \alpha_{t_\star}=-\omega(1)$ and $\alpha_{t_\star}=n^{-\omega(1)}$. Since the sum in \eqref{eq:corr-a_t} has at most $D\le(\log n)^{\alpha}$ terms,
\begin{align}
\sum_{t=2}^{D+1} \alpha_t\leq D\cdot \alpha_{t_\star}=o(1).
\end{align}
Combining Case I or II with \eqref{eq:corr-a_t} yields $\s{Corr}_{\leq D}^2\leq \mathbb{E}^2[x]+o(1)$. This concludes the proof of Theorem~\ref{thm:compLower}.

\subsection{Upper bounds}

\begin{comment}
  \begin{definition}[General planted-graph model]
Let $\Gamma=(V(\Gamma),E(\Gamma))$ be a simple graph on $k\triangleq |V(\Gamma)|$ vertices with degrees $d_\Gamma(u)$, average degree
\begin{align}
\eta(\Gamma)=\frac{1}{k}\sum_{u\in V(\Gamma)} d_\Gamma(u)=\frac{2|E(\Gamma)|}{k}.
\end{align}
Choose a uniform injection $\phi:V(\Gamma)\hookrightarrow[n]$ and set $E^\star=\{\{\phi(u),\phi(v)\}:\{u,v\}\in E(\Gamma)\}$. 
Conditional on $E^\star$, the edges are independent and
\begin{align}
Y_{ij}\sim\begin{cases}
\mathrm{Bern}(p), & \{i,j\}\in E^\star,\\
\mathrm{Bern}(q), & \{i,j\}\notin E^\star,
\end{cases}\qquad 1\le i<j\le n.
\end{align}
We define the target indicator $x\triangleq \Ind\{1\in \phi(V(\Gamma))\}\in\{0,1\}$. 
\end{definition}  
\end{comment}
\paragraph{Single-iteration power method.} Following the approach of \cite[Sec. 4.2]{SchrammWein2022}, to derive an upper bound on the truncated MMSE we analyze the performance of a simple algorithm: a single round of the power iteration method initialized with the all-ones vector, followed by thresholding. We begin by recalling two key results from \cite{SchrammWein2022}. Throught this section, we let $\nu_{p,q}\triangleq \min\{p,1-p,q,1-q\}$ and we recall that $\lambda\triangleq (p-q)/\sqrt{q(1-p)}$.
\begin{lemma}\cite[Prop. 4.1]{SchrammWein2022}\label{lem:poly-approx}
    There is a universal constant $c_0\in(0,\infty)$ such that for each $k\in\mathbb{N}$ there exists a degree-$(2k+1)$ polynomial $\tau_k:\mathbb{R}\to\mathbb{R}$ with:
\begin{equation}\label{eq:tau-approx}
\s{for}\; \ell\in\{0,1\},\quad |y-\ell|\le\Delta\le\tfrac12\ \Longrightarrow\ 
|\tau_k(y)-\ell|\le(k+\tfrac12)(c_0\Delta)^{k}.
\end{equation}
\end{lemma}
%\begin{proof}[Proof sketch]
%Construct $\tau_k$ from Chebyshev polynomials to approximate the indicator of $[1/2,\infty)$ uniformly on the two slabs $[0,\Delta]$ and $[1-\Delta,1]$; see Schramm–Wein, Prop.~4.1. The constant $c_0$ can be taken as $6$ in that construction, but any absolute $c_0$ suffices for our purposes.
%\end{proof}
\begin{lemma}\cite[Thm. 4.2]{SchrammWein2022}\label{lem:HC}
Let $f$ be a degree-$D$ polynomial in the independent edge-variables $\{Y_{ij}\}_{1\le i<j\le n}$, where each $Y_{ij}\in\{0,1\}$ with parameter $p$ or $q$. Then
\begin{equation}\label{eq:HC-ineq}
\bE\pp{f^4(Y)}\leq\p{\frac{9}{\nu_{p,q}}}^{D}\bE^2\pp{f^2(Y)}.
\end{equation}
\end{lemma}
The first lemma provides a subroutine for constructing a polynomial approximation to the threshold function. The second lemma is a hypercontractivity result for mixed Bernoulli random variables, which, roughly speaking, asserts that the moments of low-degree polynomials are well behaved.

Next, let us define the proposed estimator. Fix a normalization $d_\star>0$ to be specified. Define the centered/rescaled row-sum at vertex $1$:
\begin{equation}\label{eq:row-sum-g}
g(\s{Y})\triangleq\frac{1}{(p-q)d_\star}\sum_{i=2}^n\p{\s{Y}_{1i}-q}.
\end{equation}
We feed $g$ into the \emph{polynomial threshold} $\tau_k$ of degree $2k+1$ in Lemma~\ref{lem:poly-approx} and set
\begin{align}
f(\s{Y})\triangleq\tau_k\p{g(\s{Y})},\qquad \deg f=2k+1\triangleq D.
\end{align}
Recall that $\phi:v(\Gamma)\hookrightarrow[n]$ is the uniform injection, we set $e^\star=\{\{\phi(u),\phi(v)\}:\{u,v\}\in e(\Gamma)\}$, and $x\triangleq \Ind\{1\in \phi(v(\Gamma))\}$. We have the following lemma.
\begin{lemma}[Moments and tails of the row-sum]\label{lem:row-moments}
Condition on $\phi$. If $x=0$, then the $\{Y_{1i}-q\}_{i=2}^n$ are i.i.d. centered Bernoullis with variance $q(1-q)$, and
\begin{align}
\bE[g(\s{Y})\vert \phi,x=0]&=0,\\
\s{Var}(g(\s{Y})\vert \phi,x=0)&=\frac{q(1-q)}{(p-q)^2}\cdot\frac{n-1}{d_\star^2}.
\end{align}
If $x=1$ and $u\in v(\Gamma)$ is the (random) root with $\phi(u)=1$, then exactly $d_\Gamma(u)$ summands have mean $(p-q)$ and the rest have mean $0$, hence
\begin{align}
\bE[g(\s{Y})\vert \phi,x=1]&=\frac{d_\Gamma(u)}{d_\star},\\ 
\s{Var}(g(\s{Y})\vert \phi,x=1)&\leq \frac{p(1-p)+(n-2)q(1-q)}{(p-q)^2}\cdot\frac{1}{d_\star^2}\leq\frac{(p+q)(n-1)}{(p-q)^2}\cdot\frac{1}{d_\star^2}.
\end{align}
Moreover, for any $t>0$,
\begin{align}
\mathbb{P}\left(|g(\s{Y})|\ge t\vert\phi,x=0\right)
&\le 2\exp\pp{-\frac{t^2(p-q)^2 d_\star^2}{2q(n-1)+\frac{2}{3}t(p-q)d_\star}},\label{eq:bern-null}\\
\mathbb{P}\pp{\left.\abs{g(\s{Y})-\tfrac{d_\Gamma(u)}{d_\star}}\geq t\right|\phi,x=1}
&\le 2\exp\left(-\frac{t^2(p-q)^2 d_\star^2}{2(p+q)(n-1)+\frac{2}{3}t(p-q)d_\star}\right).\label{eq:bern-planted}
\end{align}
\end{lemma}

\begin{proof}[Proof of Lemma~\ref{lem:row-moments}]
All results follow from independence, the variance computations above, and Bernstein's inequality for sums of bounded mean-zero variables (each summand has range $\leq 1$ after centering by $q$ and rescaling).
\end{proof}
\begin{lemma}[From high probability to mean-square via hypercontractivity]\label{lem:HP-to-MSE}
Let $Z=f(\s{Y})-x$ for a degree-$D$ polynomial $f$. Suppose for some $\varepsilon>0$ and $\delta\in(0,1)$ we have
\begin{align}
\mathbb{P}\{Z^2\le\varepsilon\}\ge1-\delta.
\end{align}
Then, using \eqref{eq:HC-ineq},
\begin{align}
\bE[Z^2]\le\varepsilon+\p{\frac{9}{\nu_{p,q}}}^{D/2}\bE[Z^2]\sqrt{\delta}.
\end{align}
Consequently, if $\delta\le\tfrac{1}{4}(\nu_{p,q}/9)^{D}$ then $\bE[Z^2]\le 2\varepsilon$; if $\delta\le\tfrac{1}{16}(\nu_{p,q}/9)^{D}$ then $\bE[Z^2]\le\tfrac{4}{3}\varepsilon$.
\end{lemma}

\begin{proof}[Proof of Lemma~\ref{lem:HP-to-MSE}]
Decompose $\bE[Z^2]=\bE[Z^2\Ind_{\mathcal{G}}]+\bE[Z^2\Ind_{\mathcal{B}}]\le \varepsilon+\sqrt{\bE[Z^4]\mathbb{P}(\mathcal{B})}$ with $\mathcal{G}=\{Z^2\le\varepsilon\}$ and $\mathcal{B}=\mathcal{G}^c$. Apply \eqref{eq:HC-ineq} to bound $\bE[Z^4]\le (9/\nu_{p,q})^D \bE[Z^2]^2$ and rearrange.
\end{proof}
We are now ready to state and prove our main upper bound on the $\s{MMSE}_{\leq D}$.
\begin{theorem}[General upper bound]\label{thm:upperLDP}
Fix $0<r<1$, $k\in\mathbb{N}$, and set $D=2k+1$. Assume there exists $d_\star>0$ such that
\begin{equation}\label{eq:near-regular}
\max_{u\in V(\Gamma)} \frac{|d_\Gamma(u)-d_\star|}{d_\star}\le\frac{r}{12}.
\end{equation}
Assume further that
\begin{equation}\label{eq:SNR}
(p-q)^2 d_\star^2\ge\frac{C_1}{r^2}(p+q)(n-1)\pp{\log 4+D\log\frac{9}{\nu_{p,q}}},
\end{equation}
for an absolute constant $C_1\ge 432$. Let $g$ be as in \eqref{eq:row-sum-g}, let $\tau_k$ be as in Lemma~\ref{lem:poly-approx}, and set $f=\tau_k\circ g$. Then
\begin{align}
\s{MMSE}_{\leq D}\leq\bE\pp{(f(\s{Y})-x)^2}\leq C_2D^2r^{\,D-1},
\end{align}
for an absolute constant $C_2\geq 1$.
\end{theorem}

\begin{proof}[Proof of Theorem~\ref{thm:upperLDP}]
Define
\begin{align}
\Delta_{\mathrm{noise}}\triangleq\sqrt{\frac{3(p+q)\pp{\log 4 + D\,\log(9/\nu_{p,q})}}{(p-q)^2d_\star^2(n-1)}}.
\end{align}
By \eqref{eq:SNR} (with $C_1\ge 432$), we have $\Delta_{\mathrm{noise}}\leq r/12$. By Lemma~\ref{lem:row-moments}, using $\Delta_{\mathrm{noise}}\le r/12$ and a union bound over the two regimes $x\in\{0,1\}$,
\begin{align}
\mathbb{P}\pp{\left.\abs{g(\s{Y})-x}\le r/6\right|\phi}\geq1-\delta_D,
\end{align}
where $\delta_D\triangleq\tfrac{1}{4}\p{\frac{\nu_{p,q}}{9}}^{D}$; the choice of $\delta_D$ follows by setting the exponents in \eqref{eq:bern-null}–-\eqref{eq:bern-planted} equal to $\log 4 + D\log(9/\nu_{p,q})$ so that each tail at most $\tfrac12(\nu_{p,q}/9)^D$, and then union bound. Now, on the event $\{|g(\s{Y})-x|\leq r/6\}$, Lemma~\ref{lem:poly-approx} with $\Delta=r/6$ gives
\begin{align}
|f(\s{Y})-x|\leq(k+\tfrac12)\p{c_0\tfrac{r}{6}}^{k}\le Ckr^{k},
\end{align}
for $C$ absorbing $c_0^k$ and $6^{-k}$; we may fix $c_0=6$ to get $|f-x|\leq (k+\tfrac12)r^k$. Hence $(f-x)^2\leq (k+\tfrac12)^2 r^{2k}$ on the good event. Next, we apply Lemma~\ref{lem:HP-to-MSE} with $\varepsilon=(k+\tfrac12)^2 r^{2k}$ and $\delta=\delta_D\le \tfrac14(\nu_{p,q}/9)^D$. We obtain
\begin{align}
\s{MMSE}_{\leq D}\leq \bE[(f(\s{Y})-x)^2]\leq2(k+\tfrac12)^2r^{2k}\leq C_2\,D^2\,r^{D-1},
\end{align}
since $D=2k+1$ and $r\in(0,1)$ imply $r^{2k}\leq r^{D-1}$ and $(k+\tfrac12)^2\leq D^2$ up to a fixed constant factor. This completes the proof.
\end{proof}

Theorem~\ref{thm:upperLDP} holds true for any choice of $d_{\star}$. In many cases, the choice $d_\star = \eta(\Gamma_n) = \frac{|e(\Gamma_n)|}{|v(\Gamma_n)|}$ is optimal. 
\begin{corollary}\label{cor:half-plus-eps}
Fix $\epsilon>0$ and $0<q<p<1$. Let $\Gamma_n$ be any template sequence with
\begin{align}
\s{Dis}(\Gamma)\triangleq\max_{v\in v(\Gamma)} \frac{|d_\Gamma(v)-\eta(\Gamma)|}{\eta(\Gamma)}&\leq\frac{r}{12},\label{eqn:condDstar1}\\
\eta(\Gamma_n)&\geq n^{\frac12+\epsilon},\label{eqn:condDstar2}
\end{align}
for some fixed $0<r<1$ and all large $n$. If $D=D(n)\leq (\log n)^{\alpha}$ for any fixed $\alpha>0$, then %for all large $n$,
%\begin{align}
%(p-q)^2\eta(\Gamma_n)^2\ \ge\ \frac{C_1}{r^2}\,(p+q)\,(n-1)\,\Big[\log 4 + D\,\log\Big(\frac{9}{\nu_{p,q}}\Big)\Big],
%\end{align}
%and therefore the degree-$D$ estimator $f_n$ satisfies
\begin{align}
\bE\pp{(f(\s{Y})-x)^2}\leq C_2\,D^2r^{D-1}\to0,
\end{align}
as $n\to\infty$.
\end{corollary}
\begin{proof}[Proof of Corollary~\ref{cor:half-plus-eps}]
The proof follows from Theorem~\ref{thm:upperLDP}. Indeed, \eqref{eq:near-regular} is satisfied under assumption \eqref{eqn:condDstar1}, while \eqref{eqn:condDstar2} implies that the term on the left-hand side of \eqref{eq:SNR} is $\asymp n^{1+2\epsilon}$, while left-hand side of \eqref{eq:SNR} is $O(n(\log n)^{\alpha})$ for any fixed $\alpha$. Thus \eqref{eq:SNR} holds for large enough $n$.
\end{proof}
Thus, for any sequence of subgraphs $\Gamma = \Gamma_n$ such that $\s{Dis}(\Gamma)$ is bounded away from one, the MSE of the algorithm proposed above matches and complements the computational lower bound in Theorem~\ref{thm:compLower}. While the result above is certainly nontrivial, it does not cover the many cases in which $\s{Dis}(\Gamma)$ fails to satisfy the condition in \eqref{eqn:condDstar1}. Consider the following example.
\begin{example}
Split $v(\Gamma)$ into two sets $A$ and $B$ with $|A| = |B| = k/2$: $A$ forms a clique, $B$ is an independent set, and each $b \in B$ is adjacent to exactly $\sqrt{k}$ vertices in $A$. In this construction, $\eta(\Gamma) = \Theta(k)$ and $\s{Dis}(\Gamma) \asymp 1 \not\leq \tfrac{1}{12}$. Hence, Corollary~\ref{cor:half-plus-eps} does not apply, since condition \eqref{eqn:condDstar1} fails, whereas our lower bound yields $k \ll \sqrt{n}$. This simple example therefore reveals a gap.
\end{example}

\paragraph{Multi-iteration power method.} We now propose a stronger bound obtained by applying $L$ iterations of the power method (rather than the single iteration considered above). Let us define the proposed estimator. For simplicity of notations, we let $\s{Z}_{ij}\triangleq\s{Y}_{ij}-q$, for any $i,j\in[n]$. Fix $L\in\mathbb{N}$, and let $\calP_L$ denote the set of all simple undirected paths of length $L$ in the ambient complete graph starting at vertex 1 and pairwise distinct vertices, namely, $P=(u_0,u_1,\dots,u_\ell)$ with $u_0=1$. For each $P\in\calP_L$, define
\begin{align}
    \s{Z}(P)\triangleq\prod_{\ell=0}^{L-1} \s{Z}_{u_t u_{t+1}},
\end{align}
and the degree-$L$ \emph{walk polynomial}
\begin{equation}\label{eq:def-Hl}
\s{W}_L\triangleq\sum_{P\in\mathcal{P}_L}\s{Z}(P).
\end{equation}
Let $B_L\triangleq c_L n^{L/2}$ where $c_L>0$ is a constant depending only on $(L,p,q)$ to be chosen in the sequel, and set
\begin{equation}\label{eq:def-Xl}
\mathscr{X}_L\triangleq\frac{\s{W}_L}{B_L} = \frac{1}{c_L n^{L/2}}\s{W}_L.
\end{equation}
By Lemma~\ref{lem:poly-approx}, for an integer $m\ge 1$, let $\tau_m:\mathbb{R}\to[0,1]$ be a degree-$(2m{+}1)$ polynomial with the following uniform approximation property: for any $r\in(0,1/4]$,
\begin{equation}\label{eq:tau-approx2}
|y|\leq \tfrac{r}{2}\ \Longrightarrow\ |\tau_m(y)-0|\le (6r)^m,
\quad
y\geq r\ \Longrightarrow\ |\tau_m(y)-1|\le (6r)^m.
\end{equation}
%Define the estimator
%\begin{equation}\label{eq:def-estimator}
%f_{\ell,m}(Y)\ \triangleq \ \tau_m\big(X_\ell\big),\qquad \deg f_{\ell,m}=\ell+2m+1.
%\end{equation}
For $u\in v(\Gamma)$ let
\begin{align}
    W_L(\Gamma;u)&\triangleq\abs{\ppp{\s{simple}\;\s{paths}\;\s{of}\;\s{length}\;L\;\s{in}\;\Gamma\;\s{ starting}\;\s{at}\;u}},\\
    W_L^{\min}(\Gamma)&\triangleq\min_{u\in v(\Gamma)} W_L(\Gamma;u).
\end{align}
Recall that $\phi:v(\Gamma)\hookrightarrow[n]$ is the uniform injection, we set $e^\star=\{\{\phi(u),\phi(v)\}:\{u,v\}\in e(\Gamma)\}$, and $x\triangleq \Ind\{1\in \phi(v(\Gamma))\}$. We have the following lemma.
\begin{lemma}\label{lem:null-variance}
Let $C_L\triangleq 2[q(1-q)]^{L}$. Then,
\begin{align}
\bE[\s{W}_L\vert\vert \phi,x=0]&=0,\\
\s{Var}(\s{W}_L\vert \phi,x=0)&\leq C_Ln^{L}.
\end{align}
Consequently, with $B_L=c_L n^{L/2}$ and $c_L = \sqrt{C_L}$ we have $\s{Var}(\mathscr{X}_L\vert\phi,x=0)\le 1$.
\end{lemma}

\begin{proof}[Proof of Lemma~\ref{lem:null-variance}]
Conditioned on $\{\phi,x=0\}$ all edges are i.i.d. $\s{Bern}(q)$, hence $\bE[\s{Z}_{ij}]=0$ and
$\bE[\s{Z}_{ij}^2]=\s{Var}(\s{Y}_{ij})=q(1-q)\le 1/4$, for any $(i,j)\in\binom{[n]}{2}$. Thus, $\bE[\s{W}_L\vert \phi,x=0]=0$ follows since every $\s{Z}(P)$ is a product of mean-zero independent factors. For the variance, write
\begin{align}
\s{Var}(\s{W}_L\vert \phi,x=0)\ =\ \sum_{P\in\mathcal{P}_L}\bE[\s{Z}^2(P)\vert \phi,x=0]
\;+\;\sum_{\substack{P,P'\in\mathcal{P}_L\\ P\neq P'}} \bE[\s{Z}(P)\s{Z}(P')\vert \phi,x=0],
\end{align}
where we have used the fact that all $\s{Z}(P)$ are mean zero. If $P$ and $P'$ do \emph{not} use exactly the same set of edges, then there exists an edge $(i,j)$ that appears in exactly one of $P,P'$, hence independence and $\bE[\s{Z}_{ij}]=0$ imply $\bE[\s{Z}(P)\s{Z}(P')]=0$. When $P\neq P'$ but have the same undirected edge-set (i.e., the reversed path), we have
\begin{align}
    \bE[\s{Z}(P)\s{Z}(P')\vert \phi,x=0]&=\prod_{e\in P}\bE[\s{Z}_e^2]\\
    &\leq [q(1-q)]^{L}\leq 4^{-L}.
\end{align}
Note that the number of such pairs is at most $|\mathcal{P}_L|$. Similarly, for each $P$, 
\begin{align}
\bE[\s{Z}^2(P)\vert \phi,x=0]=\prod_{e\in P}\bE[\s{Z}_e^2\vert \phi,x=0]\leq 4^{-L}.
\end{align}
Since $|\mathcal{P}_L|=(n-1)_{L}=\Theta(n^L)$, we obtain $\s{Var}(\s{W}_L\vert \phi,x=0)\leq C_L n^L$ where $C_L = 2[q(1-q)]^{L}$. Scaling by $B_L=\sqrt{C_L} n^{L/2}$ yields $\s{Var}(\mathscr{X}_L\vert \phi,x=0)\leq 1$. 
\end{proof}

\begin{lemma}\label{lem:planted-mean}
Let $u\in v(\Gamma)$ be such that $\phi(u)=1$. Then
\begin{align}
\bE\pp{\s{W}_L\vert \phi,x=1}=(p-q)^LW_L(\Gamma;u).
\end{align}
\end{lemma}

\begin{proof}[Proof of Lemma~\ref{lem:planted-mean}]
If a path $P\in\mathcal{P}_L$ is \emph{not} fully contained in the planted edge-set, then some factor
$\s{Z}_{ij}$ on that path has mean zero (conditional on $\phi$), hence $\bE[\s{Z}(P)\vert \phi]=0$.
If $P$ is fully planted, then the $\s{Z}_{ij}$ along $P$ are i.i.d. centered with mean $(p-q)$,
so $\bE[\s{Z}(P)\vert \phi,x=1]=(p-q)^L$. The number of fully planted simple ambient paths equals the number of simple paths of length $L$ in $\Gamma$ starting at $u$, namely $W_L(\Gamma;u)$.  Summing over $P\in\mathcal{P}_L$ yields the claim.
\end{proof}

\begin{lemma}\label{lem:planted-second-moment}
Let $u\in v(\Gamma)$ be such that $\phi(u)=1$. Then there exists a constant $C_L=C_L(L,p,q)$ such that
\begin{equation}\label{eq:var-planted-correct}
\s{Var}\p{\s{W}_L\vert \phi,x=1}\leq C_L\p{n^{L}+k^{2L-1}},
\end{equation}
where $k\triangleq|v(\Gamma)|$.
\end{lemma}

\begin{proof}
We define $\mathcal P_L^{\s{pl}}$ as the set of simple ambient $L$-paths fully contained in the planted vertex set $\phi(v(\Gamma))$, and $\mathcal P_L^{\s{mix}}\triangleq\mathcal \calP_L\setminus \mathcal P_L^{\s{pl}}$. Decompose $\s{W}_L=\s{W}_L^{\s{pl}}+\s{W}_L^{\s{mix}}$ with the corresponding meanings. We have
\begin{align}
    \s{Var}\p{\s{W}_L\vert \phi,x=1}&\leq \s{Var}\p{\s{W}_L^{\s{pl}}\vert \phi,x=1}+\s{Var}\p{\s{W}_L^{\s{mix}}\vert \phi,x=1}\nonumber\\
    &\;+2\sqrt{\s{Var}\p{\s{W}_L^{\s{pl}}\vert \phi,x=1}\s{Var}\p{\s{W}_L^{\s{mix}}\vert \phi,x=1}},
\end{align}
and we next bound the variances $\s{W}_L^{\s{pl}}$ and $\s{W}_L^{\s{mix}}$.

\paragraph{Mixed part.} If $P\in\mathcal P_L^{\s{mix}}$, then $\s{Z}(P)$ contains at least one non-planted edge; since $\bE[\s{Z}_{ij}\vert\phi,x=1]=0$ on non-planted edges and the edges remain independent conditional on $\phi$, we have $\bE[\s{Z}(P)\vert\phi,x=1]=0$.
Hence expanding the variance,
\begin{align}
\s{Var}(\s{W}_L^{\s{mix}}\vert \phi,x=1)
=\sum_{P\in\mathcal P_L^{\s{mix}}}\bE[\s{Z}(P)^2\vert \phi,x=1]
+\sum_{\substack{P\neq P'\\ P,P'\in\mathcal P_L^{\s{mix}}}}\bE[\s{Z}(P)\s{Z}(P')\vert \phi,x=1].
\end{align}
Each diagonal term is $\bE[\s{Z}(P)^2\vert \phi,x=1]=\prod_{e\in P}\bE[\s{Z}_e^2\vert \phi,x=1]\le \max\{p(1-p),q(1-q)\}^{L}\le (1/4)^L$. For off-diagonals, independence implies $\bE[\s{Z}(P)\s{Z}(P')\vert \phi,x=1]=0$ unless \emph{every} edge that appears with multiplicity one across $P\cup P'$ is planted. In particular, any non-planted edge that appears in exactly one of $P,P'$ nullifies the expectation. Therefore, for a given $P\in\mathcal P_L^{\s{mix}}$ the number of $P'\in\mathcal P_L^{\s{mix}}$ with $\bE[\s{Z}(P)\s{Z}(P')\vert \phi,x=1]\neq 0$ is bounded by a constant depending only on $L$ (one must pick exactly the same set of non-planted edges, and there are $O_L(1)$ ways to complete to a simple path once those are fixed). Consequently,
\begin{align}
\s{Var}(\s{W}_L^{\s{mix}}\vert\phi,x=1)\leq C_L^{(1)}|\mathcal P_L|\leq C_L^{(1)}n^{L}.
\end{align}

\paragraph{Planted-only part.} Write $M\triangleq |\mathcal P_L^{\s{pl}}|\leq (k-1)_{L}=O(k^L)$.
For $P\in\mathcal P_L^{\s{pl}}$, all edges of $P$ are planted, hence $\bE[\s{Z}(P)\vert\phi,x=1]=(p-q)^L$, and $\s{Var}(\s{Z}(P)\vert \phi,x=1)=\prod_{e\in P}\bE[\s{Z}_e^2\vert \phi,x=1]-(p-q)^{2L}\le (1/4)^L$. Moreover, for $P\neq P'$, the covariance vanishes unless $P$ and $P'$ share at least one edge: if $e(P)\cap e(P')=\emptyset$, then $\s{Z}(P)$ and $\s{Z}(P')$ are independent (disjoint edge sets), so $\s{Cov}(\s{Z}(P),\s{Z}(P')\vert \phi,x=1)=0$. For a fixed $P$, the number of $P'$ sharing at least one edge with $P$ is $O_L(k^{L-1})$ (a shared edge reduces one free choice). Therefore,
\begin{align}
\s{Var}(\s{W}_L^{\s{pl}}\vert \phi,x=1)
&=\sum_{P}\s{Var}(\s{Z}(P)\vert \phi,x=1)+\sum_{\substack{P\neq P'\\ e(P)\cap e(P')\neq\emptyset}}\s{Cov}(\s{Z}(P),\s{Z}(P')\vert\phi,x=1)\\
&\leq C_L^{(2)}\p{k^L + k^{2L-1}},
\end{align}
since each covariance is $O(1)$ (uniform in $k$) and there are $O(k^{2L-1})$ overlapping pairs.

Finally, by Cauchy--Schwarz inequality,
\begin{align}
|\s{Cov}(\s{W}_L^{\s{mix}},\s{W}_L^{\s{pl}}\vert\phi,x=1)|
&\leq\sqrt{\s{Var}(\s{W}_L^{\s{mix}}\vert\phi,x=1)}\sqrt{\s{Var}(\s{W}_L^{\s{pl}}\vert\phi,x=1)}\\
&\leq\sqrt{C_L^{(1)}C_L^{(2)}}n^{L/2}\,k^{L-1/2}.
\end{align}
This is absorbed by $C_L(n^L+k^{2L-1})$. Combining the above three bounds yields \eqref{eq:var-planted-correct}.
\end{proof}

We are now in a position to define the proposed estimator. Let
\begin{align}
\mu_L^{\min}(\Gamma)&\triangleq\frac{(p-q)^L}{B_L}W_L^{\min}(\Gamma),\\
t_n&\triangleq\frac{\mu_L^{\min}(\Gamma)}{2}.
\end{align}
Define the rescaled statistic $\mathscr{Z}_L\triangleq \mathscr{X}_L/t_n$, and let $\tau_m$ be as in \eqref{eq:tau-approx2}. Define the estimator
\begin{align}
f_{L,m}(\s{Y})=\tau_m(\mathscr{Z}_L)\label{eqn:EstLStep}
\end{align}
of total degree $D=L+2m+1$. We have the following result.
\begin{lemma}\label{lem:HP-sep}
Fix $r\in(0,1/4]$. Assume
\begin{equation}\label{eq:W-ell-condition}
W_L^{\min}(\Gamma)\geq C^\star(L,p,q)\pp{n^{L/2}+k^{L-1/2}}\sqrt{\log n},
\end{equation}
where 
\begin{align}
C^\star(L,p,q)\geq \max\ppp{\frac{4c_L}{(p-q)^L},\frac{2\sqrt{C_L}}{c_L(p-q)^{L}}}.
\end{align}
Then for all large $n$,
\begin{align}
&\pr\pp{|\mathscr{Z}_L|\leq \tfrac{r}{2}\vert x=0}\ge1-\frac{1}{\log n},\\
&\pr\pp{\mathscr{Z}_L\geq r\vert x=1}\geq1-\frac{1}{\log n}.
\end{align}
\end{lemma}

\begin{proof}[Proof of Lemma~\ref{lem:HP-sep}]
By Lemma~\ref{lem:null-variance}, we have $\bE[\mathscr{X}_L\vert x=0]=0$ and $\s{Var}(\mathscr{X}_L\vert x=0)\le 1$. Hence
$\s{Var}(\mathscr{Z}_L\vert x=0)=\s{Var}(\mathscr{X}_L\vert x=0)/t_n^2\le 1/t_n^2$.
Chebyshev's inequality then gives
\begin{align}
\pr\pp{|\mathscr{Z}_L|\geq r/2\vert x=0}&\leq\frac{4\s{Var}(\mathscr{Z}_L)}{r^2}\\
&\leq\frac{4}{r^2t_n^2}\\
&=  \frac{16B_L^2}{r^2(p-q)^{2L}(W_L^{\min}(\Gamma))^2}\\
&\leq \frac{16B_L^2}{r^2(p-q)^{2L}[C^\star(L,p,q)]^2n^{L}\log n}\\
& = \frac{16c_L^2}{r^2(p-q)^{2L}[C^\star(L,p,q)]^2\log n}\\
&\leq\frac{1}{\log n},
\end{align}
where in the second inequality we have used \eqref{eq:W-ell-condition}, and in the last inequality we used the fact that $C^\star(L,p,q)\geq \frac{4c_L}{(p-q)^L}$. Under $x=1$ and conditioning on $\phi$ with $\phi(u)=1$, Lemma~\ref{lem:planted-mean} gives
$\bE[\mathscr{X}_L\vert \phi]=\mu_L(u)\triangleq \frac{(p-q)^L}{B_L} W_L(\Gamma;u)\ge \mu_L^{\min}(\Gamma)=2t_n$, and thus $\bE[\mathscr{Z}_L\vert \phi]\ge 2$. Lemma~\ref{lem:planted-second-moment} gives $\s{Var}\p{\s{W}_L\vert \phi,x=1}\leq C_L\p{n^{L}+k^{2L-1}}$ and so $\s{Var}(\mathscr{Z}_L\vert \phi,x=1)\leq (C_L\p{n^{L}+k^{2L-1}})/(c_L^2B_L^2t_n^2)$. Thus, Chebyshev's inequality yields
\begin{align}
\pr\pp{\mathscr{Z}_L<1\vert \phi,x=1}&\leq\frac{\s{Var}(\mathscr{Z}_L\vert \phi,x=1)}{(2-1)^2}\\
&\leq \frac{C_L\p{n^{L}+k^{2L-1}}}{c_L^2B_L^2t_n^2}\\
&= \frac{4C_L\p{n^{L}+k^{2L-1}}}{c_L^2B_L^2(\mu_L^{\min}(\Gamma))^2}\\
& = \frac{4B_L^2C_L\p{n^{L}+k^{2L-1}}}{c_L^2B_L^2(p-q)^{2L}(W_L^{\min}(\Gamma))^2}\\
&\leq \frac{4C_L\p{n^{L}+k^{2L-1}}}{c_L^2(p-q)^{2L}\pp{C^\star(L,p,q)\pp{n^{L/2}+k^{L-1/2}}\sqrt{\log n}}^2}\\
&\leq \frac{4C_L}{c_L^2(p-q)^{2L}[C^\star(L,p,q)]^2\log n}\\
&\leq\frac{1}{\log n},
\end{align}
where in the last inequality we have used the fact that $C^\star(L,p,q)\geq \sqrt{\frac{4C_L}{c_L^2(p-q)^{2L}}}$. Since $r\le 1/4$, the event $\{\mathscr{Z}_L\ge 1\}$ implies $\{\mathscr{Z}_L\ge r\}$. Averaging over $\phi$ proves the claim.
\end{proof}

\begin{theorem}[$L$-step low-degree upper bound for general $\Gamma$]\label{thm:LD-upper-ell}
Fix $L\in\mathbb{N}$ and $0<q<p<1$. Let $r\in(0,1/4]$.
Assume the rooted simple-path mass satisfies
\begin{equation}
W_L^{\min}(\Gamma)\geq C^\star(L,p,q)\pp{n^{L/2}+k^{L-1/2}}\sqrt{\log n},
\end{equation}
Consider the estimator $f_{L,m}(\s{Y})$ in \eqref{eqn:EstLStep}. Assume that $m=\omega(1)$ and $D\le C\log\log n$, for some constants $C>0$. Then
\begin{align}
\bE\pp{(f_{L,m}(\s{Y})-x)^2}\leq (\log n)^{-\Omega(1)} \xrightarrow[n\to\infty]{}0.
\end{align}
\end{theorem}

\begin{proof}[Proof of Theorem~\ref{thm:LD-upper-ell}]
By Lemma~\ref{lem:HP-sep}, with $\delta_n\triangleq 1/\log n$ we have
\begin{align}
&\pr\pp{|\mathscr{Z}_L|\leq \tfrac{r}{2}\vert x=0}\ge1-\delta_n,\\
&\pr\pp{\mathscr{Z}_L\geq r\vert x=1}\geq1-\delta_n.
\end{align}
By \eqref{eq:tau-approx2}, on $\{|\mathscr{Z}_L|\le r/2\}$ we have $|\tau_m(\mathscr{Z}_L)-0|\le (6r)^m$, and on $\{\mathscr{Z}_L\ge r\}$ we have $|\tau_m(\mathscr{Z}_L)-1|\le (6r)^m$.
Thus the pointwise squared error $(\tau_m(\mathscr{Z}_L)-x)^2$ is at most $\varepsilon_n^2\triangleq (6r)^{2m}$ on the respective ``good" events, whose complements have probability at most $\delta_n$.
Now, apply Lemma~\ref{lem:HP-to-MSE} with $\bar{f}=\tau_m(\mathscr{Z}_L)-x$, noting that $\bar{f}$ is a polynomial of total degree $D=L+2m+1$. We obtain
\begin{align}
\bE[(\tau_m(\mathscr{Z}_L)-x)^2]\leq\varepsilon_n^2+\p{\frac{9}{\nu}}^{D/2}\bE[(\tau_m(\mathscr{Z}_L)-x)^2]\sqrt{\delta_n}.
\end{align}
Choose $D\le C\log\log n$ with $C$ small enough so that $\sqrt{\delta_n}=(\log n)^{-1/2}\le \tfrac14(\nu/9)^{D/2}$ for large $n$, e.g., any $C<\frac{1}{\log(9/\nu)}$ works. Then the second term can be absorbed to the left, yielding $\bE[(\tau_m(\mathscr{Z}_L)-x)^2]\leq 2\varepsilon_n^2$. Finally, $\varepsilon_n^2=(6r)^{2m}=o(1)$ since $m=\omega(1)$, which proves the claim.
\end{proof}

\section{Conclusion and Outlook}\label{sec:Out}

This work studies the problem of exact recovery of an arbitrary planted subgraph $\Gamma_n$ in the Erd\H{o}s--R\'enyi random graph, in the dense regime where both the planting and noise edge probabilities, $p_n$ and $q_n$, are fixed independently of $n$. For this problem, we characterize the statistical limits of recovery: roughly speaking, we show that if a certain graph-theoretic quantity—termed the minimal maximum subgraph density $\mu_{\s{min}}(\Gamma_n)$—is below $\log n$, then recovery is statistically impossible, while it becomes possible (via exhaustive search) when it is above $\log n$. We then turn to the problem of recovery in polynomial-time. We first propose a general algorithm that applies to arbitrary $\Gamma_n$ and provide its statistical guarantees. Next, we derive computational lower bounds based on the low-degree polynomial framework recently developed in \cite{SchrammWein2022}. Finally, we discuss several extensions, including semi-random models and weaker notions of recovery.

We hope that our work inspires more questions than it answers. We conclude by outlining several avenues for future research:
\begin{enumerate}[leftmargin=*]
    \item Extending our results to settings in which the edge probabilities $p$ and $q$ depend on $n$, especially when they vanish or approach one (for instance, polynomially in $n$), entails technical and conceptual difficulties. Some of our techniques can be generalized to handle moderately sparse regimes (e.g., when $p$ and $q$ are as small as $n^{-\alpha}$ for certain ranges of $\alpha$), but pushing these arguments further to the general case requires new ideas.
    \item While we were able to tightly characterize the statistical limits, a gap remains between our lower and upper computational bounds. At present, the precise computational barrier for recovering an arbitrary planted subgraph is remains a mystery.
    \item Our semi-random analysis focused on an adversary limited to edge deletions outside the planted subgraph and edge additions inside it. Considering models in which the adversary is allowed to add a bounded number of edges, or to perform more general perturbations, leads to natural questions concerning robustness and reconstructability.
    \item We assume that the statistician observes the entire graph. Extending the analysis to settings where only part of the graph is observable—through adaptive or non-adaptive queries—poses an interesting and challenging problem for arbitrary planted subgraphs.
    \item It is of interest, both for the detection variant and for the recovery problem studied in this paper, to provide additional forms of evidence for the statistical–computational gaps that arise. One powerful technique is through average-case reductions. For example, can we show that the hardness of detecting or recovering an arbitrary planted subgraph can be reduced from the detection or recovery of a planted clique—given that the clique is the ``easiest" subgraph to infer?
\end{enumerate}

\section*{Acknowledgments}
Fruitful discussions with Dor Elimelech on the subgraph recovery problems and related topics are acknowledged with thanks.
\bibliographystyle{alpha}
\bibliography{bibfile}

\appendix

\section{Auxiliary Lemmata}\label{app:lemmata}

\subsection{Equivalence of worst-case and Bayes risks}\label{app:EquivBayesWorst}
This appendix establishes the equivalence between the worst-case and Bayesian error probabilities when the prior over $\Gamma^\star$ is uniform on $\calS_\Gamma$.
\begin{lemma}
\label{lem:WC=Bayes}
Fix $n\in\mathbb N$ and let $\Gamma_n$ be an arbitrary graph with no isolated vertices and at most $n$ vertices.
Denote by $\calS_{\Gamma_n}$ the set of (labelled) copies of\ $\Gamma_n$ inside the complete graph $\calK_n$. For $\Gamma^\star\in\calS_{\Gamma_n}$ write $\pr_{\Gamma^\star}$ for the distribution $\calG_{\Gamma^\star}(n,p_n,q_n)$ of the observed graph $\s{G}$. Given an estimator $\hat{\Gamma}_n:\{0,1\}^{\binom n2}\to\calS_{\Gamma_n}$, define the \emph{risk function}
\begin{align}
\s{R}_n(\hat{\Gamma}_n;\Gamma^\star)
      \triangleq
      \pr_{\Gamma^\star}
      \pp{\hat{\Gamma}_n(\s{G})\neq\Gamma^\star},
\quad
\Gamma^\star\in\calS_{\Gamma_n}.
\end{align}
Let
\begin{align}
\s E_n(\hat{\Gamma}_n)&\triangleq\sup_{\Gamma^\star\in\calS_{\Gamma_n}}
                           \s{R}_n(\hat{\Gamma}_n;\Gamma^\star),\\
\overline{\s E}_n(\hat{\Gamma}_n)&\triangleq
          \E_{\Gamma^\star\sim\mathrm{Unif}(\calS_{\Gamma_n})}
             \pp{\s{R}_n(\hat{\Gamma}_n;\Gamma^\star)}.
\end{align}
Then for every estimator $\hat{\Gamma}_n$ there exists an equivariant estimator
$\hat{\Gamma}_n^{\mathrm{eq}}$ such that 
\begin{align}
\label{eq:WC=Bayes-risk}
\s{R}_n(\hat{\Gamma}_n^{\mathrm{eq}};\Gamma^\star)
       \equiv \overline{\s E}_n(\hat{\Gamma}_n),
       ,
\end{align}
for all $\Gamma^\star\in\calS_{\Gamma_n}$, and
\begin{align}
    \s E_n(\hat{\Gamma}_n^{\mathrm{eq}})
       =\overline{\s E}_n(\hat{\Gamma}_n^{\mathrm{eq}})
        =\overline{\s E}_n(\hat{\Gamma}_n)
        \leq \s E_n(\hat{\Gamma}_n).\label{eqn:chinInEq}
\end{align}
Consequently
\begin{align}
      \inf_{\hat{\Gamma}_n:\{0,1\}^{n^{(2)}} \to \calS_{\Gamma_n}}\s E_n(\hat{\Gamma}_n)
      =
      \inf_{\hat{\Gamma}_n:\{0,1\}^{n^{(2)}} \to \calS_{\Gamma_n}}\overline{\s E}_n(\hat{\Gamma}_n),
\end{align}
for every $n\in\mathbb{N}$. That is, the uniform prior on $\calS_{\Gamma_n}$ is least–favorable, and the minimax error probability coincides with the Bayes error probability under that prior.
\end{lemma}

%--------------------------------------------------------------------
\begin{proof}
Let $\mathbb{S}_n$ be the permutation group on the vertex set $[n]=\{1,\dots,n\}$.
For $\pi\in \mathbb{S}_n$ and a graph $H$ on $[n]$, define
$\pi\circ H$ to be the graph whose adjacency indicator satisfies
\begin{align}
(\pi \circ H)_{ij}=H_{\pi^{-1}(i),\pi^{-1}(j)},\qquad 1\leq i<j\leq n.
\end{align}
The action extends to subgraphs: $\pi \circ\Gamma^\star\triangleq\pi(\Gamma^\star)$
for $\Gamma^\star\in\calS_{\Gamma_n}$. Because the generative rule depends \emph{only} on (i)~which edges belong to $\Gamma^\star$ and (ii)~the probabilities $(p_n,q_n)$, for every $\pi\in \mathbb{S}_n$ and every measurable $A\subseteq\{0,1\}^{\binom n2}$ 
\begin{align}
\pr_{\Gamma^\star}\pp{\s G\in A}
=\pr_{\pi \circ\Gamma^\star}
        \pp{\s G\in\pi  \circ A},
\label{eq:invariance}
\end{align}
where $\pi  \circ A=\{\pi  \circ H:H\in A\}$. Also, note that the indicator loss
$\Ind\{\hat{\Gamma}_n(\s G)\neq\Gamma^\star\}$ is invariant in the sense that
\begin{align}
\Ind\{\hat{\Gamma}_n(\s G)\neq\Gamma^\star\}
=\Ind\bigl\{\pi  \circ \hat{\Gamma}_n(\pi  \circ \s G)
          \neq\pi  \circ \Gamma^\star\bigr\}.
\label{eq:loss-invariance}
\end{align}
Now, given an arbitrary estimator $\hat{\Gamma}_n$, define a \emph{randomized} rule
\begin{align}
\hat{\Gamma}_n^{\mathrm{eq}}(\s G)
   \triangleq\Pi^{-1}  \circ 
     \hat{\Gamma}_n(\Pi  \circ \s G),
\end{align}
where
$\Pi\sim\mathrm{Unif}(\mathbb{S}_n)$ is an independent, auxiliary random permutation.
For any fixed $\pi\in \mathbb{S}_n$,
\begin{align}
\hat{\Gamma}_n^{\mathrm{eq}}(\pi  \circ \s G)
   =\Pi^{-1}  \circ \hat{\Gamma}_n(\Pi\pi  \circ \s G)
   \,{\buildrel d\over=}\,
   \pi  \circ \bigl[\Pi^{-1}  \circ \hat{\Gamma}_n(\Pi  \circ \s G)\bigr]
   =\pi  \circ \hat{\Gamma}_n^{\mathrm{eq}}(\s G),
\end{align}
hence the rule is \emph{equivariant}. Furthermore, for any $\Gamma^\star\in\calS_{\Gamma_n}$,
\begin{align}
\s{R}_n(\hat{\Gamma}_n^{\mathrm{eq}};\Gamma^\star)
&=\E_{\Gamma^\star}
   \pp{\Ind\bigl\{\Pi^{-1}  \circ \hat{\Gamma}_n(\Pi  \circ \s G)
               \neq\Gamma^\star\bigr\}}\\
&=\frac1{|\mathbb{S}_n|}
   \sum_{\pi\in \mathbb{S}_n}
   \E_{\Gamma^\star}
      \pp{\Ind \bigl\{\hat{\Gamma}_n(\pi  \circ \s G)
                      \neq\pi  \circ \Gamma^\star\bigr\}}\\
&=
   \frac1{|\mathbb{S}_n|}
   \sum_{\pi\in \mathbb{S}_n}
   \E_{\pi \circ\Gamma^\star}
      \pp{\Ind \bigl\{\hat{\Gamma}_n(\s G)
                      \neq\pi  \circ \Gamma^\star\bigr\}}\\
&=\E_{\widetilde{\Gamma}\sim\mathrm{Unif}(\calS_{\Gamma_n})}
      \s{R}_n(\hat{\Gamma}_n;\widetilde{\Gamma})
   =\overline{\s E}_n(\hat{\Gamma}_n),
\end{align}
which is independent of $\Gamma^\star$, and the third equality follows from \eqref{eq:invariance}--\eqref{eq:loss-invariance}. This proves \eqref{eq:WC=Bayes-risk}. Next, we compare the worst-case and Bayes risks. Because the equivariant rule has equal risk \emph{everywhere}, its worst–case and average risks coincide:
$\s E_n(\hat{\Gamma}_n^{\mathrm{eq}})  =\overline{\s E}_n(\hat{\Gamma}_n^{\mathrm{eq}})$. Moreover, Jensen's inequality applied to the averaging in Step~3 gives $\overline{\s E}_n(\hat{\Gamma}_n^{\mathrm{eq}})\leq\overline{\s E}_n(\hat{\Gamma}_n)$, while by definition $\s E_n(\hat{\Gamma}_n^{\mathrm{eq}})  \leq\s E_n(\hat{\Gamma}_n)$. This yields the chain of equalities and inequalities in \eqref{eqn:chinInEq}. Finally, taking the infimum over \emph{all} estimators on both sides of
\eqref{eq:WC=Bayes-risk} gives
\begin{align}
\inf_{\hat{\Gamma}_n:\{0,1\}^{n^{(2)}} \to \calS_{\Gamma_n}}\s E_n(\hat{\Gamma}_n)
  =\inf_{\hat{\Gamma}_n:\{0,1\}^{n^{(2)}} \to \calS_{\Gamma_n}}\overline{\s E}_n(\hat{\Gamma}_n),
\end{align}
so the uniform prior on $\calS_{\Gamma_n}$ achieves the minimax value.
\end{proof}

\subsection{Proof of Lemma~\ref{lem:onionUnique}}\label{app:onionUn}

Fix a step $\ell$ and let $\Gamma^{(\ell-1)}$ denote the graph remaining at the beginning of that step. Recall that for any subgraph $\Gamma^{(\ell-1)}\subsetneq\mathcal{K}$, we define
\begin{align}
  \eta(\mathcal{K}\vert\Gamma^{(\ell-1)})
  \triangleq\frac{|\mathcal{K}\setminus\Gamma^{(\ell-1)}|}
        {|v(\mathcal{K})\setminus v(\Gamma^{(\ell-1)})|}.
\end{align}
By construction, $\Gamma^{(\ell)}$ is a \emph{maximal} subgraph that attains the
maximum value
\begin{align}
  \eta^\star=
  \max_{\Gamma^{(\ell-1)}\subsetneq\mathcal{K}}
       \eta(\mathcal{K}\vert\Gamma^{(\ell-1)}).
\end{align}
Assume, toward a contradiction, that two distinct maximisers $\calA,\calB\subsetneq\Gamma^{(\ell-1)}$ exist, i.e.,
\begin{align}
  \eta(\calA\vert\Gamma^{(\ell-1)})
  =\eta(\calB\vert\Gamma^{(\ell-1)})
  =\eta^\star,
\end{align}
with $\calA\neq\calB$. Let $\mathcal{C}\triangleq\calA\cup\calB$ and define
\begin{align}
  e_{\calX}&\triangleq|\mathcal{X}\setminus\Gamma^{(\ell-1)}|,\\
  v_{\calX}&\triangleq|v(\mathcal{X})\setminus v(\Gamma^{(\ell-1)})|,
\end{align}
for $\calX\in\{\calA,\calB,\calC\}.$ Then $e_{\calC}=e_{\calA}+e_{\calB}-\delta$ and $v_{\calC}=v_{\calA}+v_{\calB}-\gamma$, where $\delta,\gamma\geq  0$ count the edges and vertices already shared by $\calA$ and $\calB$ outside~$\Gamma^{(\ell-1)}$. Because $\eta(\calA\vert\Gamma^{(\ell-1)})=\eta(\calB\vert\Gamma^{(\ell-1)})=\eta^\star$, we have
$e_{\calA}/v_{\calA} = e_{\calB}/v_{\calB} = \eta^\star$. Notice that by the maximiality of $\eta^\star$, we have $\frac{\delta}{\gamma}\leq\eta^\star$, and consequently
\begin{align}
  \eta(\mathcal{C}\vert\Gamma^{(\ell-1)})
  =\frac{e_{\calA}+e_{\calB}-\delta}{v_{\calA}+v_{\calB}-\gamma}
  \ge
  \frac{e_{\calA}+e_{\calB}}{v_{\calA}+v_{\calB}}
  = \eta^\star,
\end{align}
with equality only if $\frac{\delta}{\gamma}=\eta^\star$. Now, if the inequality is strict, namely, $\eta(\mathcal{C}\vert\Gamma^{(\ell-1)})>\eta^\star$, then this contradicts the optimality of $\eta^\star$. If equality holds, then $\mathcal{C}\supsetneq\calA,\calB$ attains the same maximal density, contradicting the \emph{maximality} of $\Gamma{^(\ell)}$ as $\calA\neq\calB$. Thus, no two distinct maximizers coexist.

\subsection{Proof of Lemma~\ref{lem:equivminimalMSD}}\label{app:equiv}

Recall that for any $\s{J}\subseteq\Gamma$
\begin{align}
%\mu(S)\;\equiv\;
\mu(\Gamma\vert\s{J}) \triangleq \max_{\s{J}\subsetneq \s{F}\subseteq \Gamma}\ \eta(\s{F}\vert\s{J}),
\qquad S\subseteq\Gamma .
\end{align}
Denote the onion decomposition of $\Gamma$ by
\begin{align}
\emptyset=\Gamma^{(0)}\subsetneq\Gamma^{(1)}\subsetneq\cdots\subsetneq \Gamma^{(M)}=\Gamma,
\end{align}
and set
\begin{align}
d_\ell \triangleq \eta \left(\Gamma^{(\ell)}\vert\Gamma^{(\ell-1)}\right)=\mu(\Gamma\vert\Gamma^{(\ell-1)}),
\end{align}
for $\ell=1,\dots,M$. We proceed through four simple lemmas. All subgraphs are understood to be subgraphs of $\Gamma$.
\begin{lemma}\label{lem:APP1}
    If $\s{J}\subseteq \s{J}'$ then $\mu(\Gamma\vert\s{J})\ge \mu(\mu(\Gamma\vert\s{J}')$.
\end{lemma}
\begin{proof}[Proof of Lemma~\ref{lem:APP1}]
The feasible set $\{\s{F}:\s{J}\subsetneq \s{F}\subseteq \Gamma\}$ contains $\{\s{F}:\s{J}'\subsetneq \s{F}\subseteq \Gamma\}$. Maximizing over a larger set cannot give a smaller value.
\end{proof}
\begin{lemma}\label{lem:APP2}
Fix $\s{J}$. Let $\mathcal{F}_{\s{J}}\triangleq\{\s{F}:\s{J}\subsetneq \s{F}\subseteq\Gamma,\ \eta(\s{F}\vert \s{J})=\mu(\Gamma\vert\s{J})\}$. If $\s{F}_1,\s{F}_2\in\mathcal{F}_{\s{J}}$ then $\s{F}_1\cup \s{F}_2\in\mathcal{F}_{\s{J}}$. Consequently, $\mathcal{F}_{\s{J}}$ has a unique inclusion-wise maximal element (its union).
\end{lemma}
\begin{proof}[Proof of Lemma~\ref{lem:APP2}]
With $d\triangleq\mu(\Gamma\vert\s{J})$, set $\s{U}\triangleq\s{F}_1\cup \s{F}_2$, $\s{I}\triangleq \s{F}_1\cap \s{F}_2$. Then
\begin{align}
|e(\s{U}\vert \s{J})|-d\cdot|v(\s{U}\vert \s{J})|&=\bigl[|e(\s{F}_1\vert\s{J})|-d\cdot|v(\s{F}_1\vert \s{J})|\bigr]+\bigl[|e(\s{F}_2\vert\s{J})|-d\cdot|v(\s{F}_2\vert \s{J})|\bigr]\nonumber\\
&\qquad-\bigl[|e(\s{I}\vert \s{J})|-d\cdot|v(\s{I}\vert\s{J})|\bigr]\ge 0,
\end{align}
because $\s{I}$ is feasible and hence $\eta(\s{I}\vert\s{J})\leq d$. Thus $\eta(\s{U}\vert\s{J})\ge d$, and since $d$ is maximal, equality holds.
\end{proof}
\begin{lemma}\label{lem:APP3}
Recall that $d_\ell \triangleq \eta \left(\Gamma^{(\ell)}\vert\Gamma^{(\ell-1)}\right)$, for $\ell=1,\ldots,M$. Then, $d_1>d_2>\cdots>d_M$. 
\end{lemma}
\begin{proof}[Proof of Lemma~\ref{lem:APP3}]
    Suppose $d_{\ell+1}\ge d_\ell$. Let $\s{F}$ maximize $\mu(\Gamma\vert\Gamma^{(\ell)})$, so $\eta(\s{F}\vert\Gamma^{(\ell)})=d_{\ell+1}\ge d_\ell$. Then
\begin{align}
\eta(\s{F}\vert\Gamma^{(\ell-1)})=\frac{|e(\Gamma^{(\ell)}\vert\Gamma^{(\ell-1)})|+|e(\s{F}\vert\Gamma^{(\ell)})|}{|v(\Gamma^{(\ell)}\vert\Gamma^{(\ell-1)})|+|v(\s{F}\vert\Gamma^{(\ell)})|}\ \ge\ d_\ell.
\end{align}
If the inequality is strict, this contradicts maximality of $d_\ell=\mu(\Gamma\vert\Gamma^{(\ell-1)})$. If equality holds, then $\s{F}$ is also a maximizer for seed $\Gamma^{(\ell-1)}$; by Lemma~\ref{lem:APP2} and the onion step's maximality, we cannot have a strict superset $\s{F}\supsetneq\Gamma^{(\ell)}$. This leads to a contradiction. 
\end{proof}
\begin{lemma}\label{lem:APP4}
    Fix $\ell$ and set $\s{A}\triangleq\Gamma^{(\ell-1)}$, $\s{B}\triangleq\Gamma^{(\ell)}$, $d\triangleq d_\ell=\eta(\s{B}\vert \s{A})=\mu(\Gamma\vert\s{A})$.  
Then for every $\s{J}$ with $\s{A}\subseteq \s{J}\subseteq \s{B}$, we have $\mu(\Gamma\vert\s{J})=d$.
\end{lemma}
\begin{proof}[Proof of Lemma~\ref{lem:APP4}]
On the one hand, by the inequality in the proof of Lemma~\ref{lem:APP3} with $\s{F}=\s{B}$,
\begin{align}
\eta(\s{B}\vert \s{J})\ge\eta(\s{B}\vert\s{A})=d\quad\Rightarrow\quad \mu(\Gamma\vert\s{J})\ge d.
\end{align}
On the other hand, by Lemma~\ref{lem:APP1}, since $\s{A}\subseteq \s{J}$,
\begin{align}
\mu(\Gamma\vert\s{J})\le\mu(\Gamma\vert\s{A})=d.
\end{align}
Together, $\mu(\Gamma\vert\s{J})=d$.
\end{proof}
We are now ready to prove Lemma~\ref{lem:equivminimalMSD}, and we start by proving that $T=M$. Indeed, for each $\ell$, $d_\ell=\mu(\Gamma\vert\Gamma^{(\ell-1)})\in \Lambda(\Gamma)$. Hence $\{d_1,\dots,d_M\}\subseteq \Lambda(\Gamma)$. In particular, the number $T$ of distinct values satisfies $T\ge M$. Now, let $\s{J}\subseteq\Gamma$ be arbitrary, and let $\ell$ be the minimal index with $\s{J}\subseteq \Gamma^{(\ell)}$. Then $\Gamma^{(\ell-1)}\subseteq \s{J}\subseteq \Gamma^{(\ell)}$. Lemma~\ref{lem:APP4} gives $\mu(\Gamma\vert\s{J})=d_\ell$. Thus $\Lambda(\Gamma)\subseteq\{d_1,\dots,d_M\}$, and so $T\le M$. Combining the above, we conclude $T=M$ and the sets of distinct values coincide:
\begin{align}
\{\lambda_1>\cdots>\lambda_M\}=\{d_1>\cdots>d_M\}.
\end{align}
By Lemma~\ref{lem:APP3}, both sides are strictly decreasing, and thus their $\ell$-th entries match:
\begin{align}
\lambda_\ell=d_\ell=\eta(\Gamma^{(\ell)}\vert\Gamma^{(\ell-1)}),\qquad \ell=1,\dots,M,
\end{align}
which concludes the proof.

\paragraph{Minimal maximum subgraph density.} We provide here an alternative proof for the fact that $\eta(\Gamma^{(M)}\vert \Gamma^{(M-1)})=\mu_{\s{min}}(\Gamma)$, where $\mu_{\s{min}}(\Gamma)$ is defined in \eqref{eqn:minimaxSubDen}. The proof follows from several facts established in \cite{LeePerniceRajaramanZadik2025}. Specifically, for any $\s{S}\subseteq\Gamma$, define $G(\s{S})\triangleq\max_{\s{S}\subsetneq\s{F}\subseteq\Gamma}\eta(\s{F}\vert\s{S})$. Then, we claim that
\begin{align}
\min_{\s{S}\subseteq \Gamma}G(\s{S})=
\min_{\alpha\in[0,1]}\min_{\substack{\s{S}\subseteq\Gamma\\ |\s{S}|\leq \alpha|\Gamma|}} G(\s{S}).
\end{align}
Indeed, for every $\alpha\in[0,1]$, the inner minimum is over a subset of all $\s{S}$, so
\begin{align}
\min_{\substack{\s{S}\subseteq\Gamma\\ |\s{S}|\leq \alpha|\Gamma|}} G(\s{S})\geq \min_{\s{S}\subseteq \Gamma}G(\s{S}).
\end{align}
Taking $\min_{\alpha\in[0,1]}$ over both sides we get
\begin{align}
\min_{\alpha\in[0,1]}\min_{\substack{\s{S}\subseteq\Gamma\\ |\s{S}|\leq \alpha|\Gamma|}} G(\s{S})\geq \min_{\s{S}\subseteq \Gamma}G(\s{S}).
\end{align}
Conversely, let $\s{S}^\star$ attain $\min_{\s{S}\subseteq \Gamma}G(\s{S})$ and take
$\alpha^\star\triangleq|\s{S}^\star|/|\Gamma|$; then $\s{S}^\star$ is feasible for $\alpha^\star$, so
\begin{align}
\min_{\substack{\s{S}\subseteq\Gamma\\ |\s{S}|\leq \alpha|\Gamma|}} G(\s{S})\leq G(\s{S}^\star)=\min_{\s{S}\subseteq \Gamma}G(\s{S}).
\end{align}
Thus equality holds. If for $\alpha\in[0,1]$ we define 
\begin{align}
    \phi(\alpha)\triangleq\min_{\substack{\s{S}\subseteq\Gamma\\ |\s{S}|\leq \alpha|\Gamma|}}\max_{\s{S}\subsetneq\s{F}\subset\Gamma}\eta(\s{F}\vert\s{S}),
\end{align}
then the above implies that
\begin{align}
\min_{\s{S}\subseteq \Gamma}\max_{\s{S}\subsetneq\s{F}\subset\Gamma}\eta(\s{F}\vert\s{S})
=
\min_{\alpha\in[0,1]}\ \phi(\alpha).\label{eqn:FirstChar}
\end{align}
Now, by \cite[Thm.~3.6(b)]{LeePerniceRajaramanZadik2025}, $\phi(\alpha)$ is piecewise constant in $\alpha$ with breakpoints $\alpha_i\triangleq|e(\Gamma^{(i)})|/|e(\Gamma)|$, and on the plateau $\alpha\in[\alpha_i,\alpha_{i+1})$,
\begin{align}
\phi(\alpha)=\eta(\Gamma^{(i+1)}\vert\Gamma^{(i)}).
\end{align}
Furthermore, by \cite[Lemma~5.4]{LeePerniceRajaramanZadik2025} these plateau values are \emph{non-increasing} in $i$, so the minimum over all $\alpha\in[0,1]$ is attained on the last plateau, i.e., for $\alpha\in(\alpha_{M-1},1]$, and equals
\begin{align}
\min_{\alpha\in[0,1]}\phi(\alpha)=\rho(\Gamma^{(M)}\vert\Gamma^{(M-1)}).\label{eqn:SecondChar}
\end{align}
Combining \eqref{eqn:FirstChar} and \eqref{eqn:SecondChar} completes the proof.

\subsection{Bounds on coherence}\label{app:cohereBounds}

\begin{lemma}\label{lem:CohUpLO}
Let $\s{U} \in \mathbb{R}^{n \times r} $ be a matrix with orthonormal columns, and supported on the set $\calS$ of size $k=|\calS|$; that is, the $i$th row of $\mathsf{U}$ satisfies $\mathsf{U}_{i,:} = 0$ for all $i \notin \calS$. Then
\begin{align}
\frac{n}{k}\leq \s{coh}(\s{U})\leq \frac{n}{r}.
\end{align}
\end{lemma}

\begin{proof}[Proof of Lemma~\ref{lem:CohUpLO}]
Since the columns of $ \s{U} $ are orthonormal, we have
\begin{align}
\sum_{i=1}^n \|\s{U}_{i,:}\|_2^2
&= \s{trace}(\s{U}\s{U}^{\top})\\
&= \s{trace}(\s{U}^{\top}\s{U}) = \s{rank}(\s{X}^\star).\label{eqn:avgEng}
\end{align}
To establish the upper bound, note that for any row $ i $, the squared norm satisfies:
\begin{align}
\|  \s{U}_{i,:} \|_2^2 \leq 1,
\end{align}
because $  \s{U}_{i,:} \in \mathbb{R}^r $ and the total squared norm across all rows sums to $ r $. Hence
\begin{align}
\s{coh}(\s{U})& = \frac{n}{r} \cdot \max_{i} \| \s{U}_{i,:} \|_2^2\\
&\leq \frac{n}{r} \cdot 1 = \frac{n}{r}.
\end{align}
This proves the upper bound. For the lower bound, due to \eqref{eqn:avgEng}, the average row norm squared over the $k$ non-zero rows is $\frac{r}{k}$. Thus
\begin{align}
    \s{coh}(\s{U})\geq \frac{n}{r}\cdot\frac{r}{k} = \frac{n}{k}.
\end{align}
\end{proof}

\subsection{Spectral-degree bound on coherence}
Consider the following result.
\begin{theorem}[Spectral/degree bound on coherence]
Let $\s{X}\in\{0,1\}^{n\times n}$ be symmetric (e.g., the adjacency matrix of an undirected graph),
let $r\triangleq\s{rank}(\s{X})\geq  1$, and let $\s{U}\in\mathbb{R}^{n\times r}$ have orthonormal columns spanning $\mathrm{range}(\s{X})$. Define the coherence
\begin{align}
\s{coh}(\s{U})\triangleq\frac{n}{r}\max_{1\leq  i\leq  n}\|\s{U}_{i,:}\|_2^2.
\end{align}
Let $\sigma_{\min}>0$ be the smallest nonzero singular value of $\s{X}$, and let
$d_i\triangleq\|x_i\|_2^2$ where $x_i^\top$ is the $i$-th row of $\s{X}$, with $d_{\max}\triangleq\max_i d_i$.
Then
\begin{align}
\s{coh}(\s{U})\leq\frac{n}{r}\frac{d_{\max}}{\sigma_{\min}^2}.
\end{align}
In particular, if $\s{X}$ is a $\{0,1\}$ adjacency matrix (with or without self-loops), then $d_i$ equals the (loop–inclusive) degree of vertex $i$.
\end{theorem}

\begin{proof}
Let $\s{P}\triangleq\s{U}\s{U}^\top$ be the orthogonal projector onto $\s{range}(\s{X})$.
The \emph{leverage scores} are $\ell_i\triangleq\|\s{U}_{i,:}\|_2^2=\s{P}_{ii}$, and
\begin{align}
\s{coh}(\s{U})=\frac{n}{r}\max_i \ell_i.
\end{align}
We express $\s{P}$ using only $\s{X}$. Since $\s{X}$ is symmetric, its singular values are the absolute values
of its (nonzero) eigenvalues. Writing the spectral decomposition as $\s{X}=Q\Lambda Q^\top$ with
$\Lambda=\mathrm{diag}(\lambda_1,\dots,\lambda_r,0,\dots,0)$, where $\lambda_k\neq 0$ for $k\leq r$),
one checks that
\begin{align}
\s{P}&=\s{U}\s{U}^\top \\
&= Q_r Q_r^\top \\
&= \s{X} \s{X}^\dagger\\
&= \s{X}(\s{X}^2)^\dagger \s{X},
\end{align}
where $Q_r\triangleq[q_1\ \cdots\ q_r]$ collects the eigenvectors associated with the nonzero eigenvalues,
and $\s{X}^\dagger$ denotes the Moore--Penrose pseudoinverse. Therefore, for each $i$,
\begin{align}\label{eq:leverage-variational}
\ell_i &= e_i^\top P e_i \\
&= e_i^\top \s{X} (\s{X}^2)^\dagger \s{X} e_i\\
&= x_i^\top (\s{X}^2)^\dagger x_i.
\end{align}
Let $\sigma_{\min}>0$ be the smallest nonzero singular value of $\s{X}$. Then the nonzero spectrum of $\s{X}^2$ is $\{\sigma_k^2\}_{k=1}^r$, so the operator norm of $(\s{X}^2)^\dagger$ equals $1/\sigma_{\min}^2$. Using \eqref{eq:leverage-variational} and the Cauchy--Schwarz/operator-norm bound,
\begin{align}
\ell_i &=x_i^\top (\s{X}^2)^\dagger x_i\\
&\leq \|(\s{X}^2)^\dagger\|_2\|x_i\|_2^2\\
&= \frac{\|x_i\|_2^2}{\sigma_{\min}^2}\\
&= \frac{d_i}{\sigma_{\min}^2}.
\end{align}
Taking the maximum over $i$ yields
\begin{align}
\max_i \ell_i\leq\frac{d_{\max}}{\sigma_{\min}^2},
\end{align}
and hence
\begin{align}
\s{coh}(\s{U})&=\frac{n}{r}\max_i \ell_i \\
&\leq \frac{n}{r}\frac{d_{\max}}{\sigma_{\min}^2}.
\end{align}
This completes the proof.
\end{proof}

\section{Derivation of the Maximum Likelihood Estimator}\label{app:0}

Consider the following recovery task. Pick a copy $\Gamma^\star\in\calS_{\Gamma}$. A random graph $\s{G}$ with $n$ vertices is formed as follows: keep the edges of $\Gamma$ with probability $p$, and the edges outside $\Gamma$ with probability $q$. We denote the ensemble of such planted graphs by $\calG_{\Gamma^\star}(n,p,q)$. Given $\s{G}$, the goal is to recover the graph $\Gamma^\star$. Then, we have
\begin{align}
    \pr_{\calG_{\Gamma^\star}(n,p,q)}(\s{G};\Gamma) &= \prod_{(i,j)\in e(\Gamma)}p^{\s{A}_{ij}}(1-p)^{1-\s{A}_{ij}}\prod_{(i,j)\in \binom{[n]}{2}\setminus e(\Gamma)}q^{\s{A}_{ij}}(1-q)^{1-\s{A}_{ij}}\\
    & \propto \prod_{(i,j)\in e(\Gamma)}\pp{\frac{p(1-q)}{q(1-p)}}^{\s{A}_{ij}}\\
    & = \pp{\frac{p(1-q)}{q(1-p)}}^{\sum_{(i,j)\in e(\Gamma)}\s{A}_{ij}}
\end{align}
Thus, we see that in the regime $p>q$, the MLE is given by
\begin{align}
\hat{\Gamma}_{\s{MLE}} = \arg\max_{\Gamma\in\calS_\Gamma}\sum_{(i,j)\in e(\Gamma)}\s{A}_{ij}.\label{eqn:combiDev}
\end{align}

\section{Lower Bound Via Bayes Risk Analysis}\label{app:secProofofITLOWEr}

In this appendix, we provide an alternative proof of the information-theoretic lower bounds on recovery, derived from hypothesis testing risk analysis. The resulting bound is tight for balanced graphs, i.e., graphs $\Gamma$ with $\mu(\Gamma) = \eta(\Gamma)$, and for which $\mu_{\s{min}}(\Gamma) = \mu(\Gamma)$, and with super-logarithmic maximum density, namely, $\mu(\Gamma)\geq\alpha_n\log|v(\Gamma)|$, for some $\alpha_n=\Omega(1)$ (which is, in fact, the interesting region).  Therefore, we focus on such graphs, although the method applies to arbitrary graphs as well. To that end, we begin by showing that it suffices to consider the canonical case where $p_n = 1$. We then provide brief preliminaries on the detection problem of an arbitrary planted subgraph in random graphs. Next, we propose and prove a generalized notion of the subgraph expectation threshold, which plays an important role in some of our later proofs.

We first observe that we can safely focus on the case where $p=1$. This follows from the fact that for any given instance of $\calG_{\Gamma_n}(n,p_n,q_n)$, there exist a fixed $\tilde{q}$ such that recovery is statistically easier over $\calG_{\Gamma_n}(n,1,\tilde{q}_n)$. Indeed, let $\hat{\Gamma}$ be any successful estimation algorithm for the recovery of $\Gamma$ over $\calG_{\Gamma_n}(n,p_n,q_n)$, i.e., $\pr_{\calG_{\Gamma_n}(n,p_n,q_n)}(\hat{\Gamma}\neq\Gamma)\leq\varepsilon$, for any $\varepsilon>0$. Now, let $\s{G}\sim\calG_{\Gamma_n}(n,1,q_n/p_n)$, and consider the random map $\phi:\{0,1\}^{\binom{n}{2}}\to\{0,1\}^{\binom{n}{2}}$, which receives the graph $\s{G}$ as an input, and if $(i,j)\in\s{G}$, then we keep it with probability $p_n$, otherwise, if $(i,j)\not\in\s{G}$, then it remains the same. It should be clear that the random graph $\phi(\s{G})$ is distributed as follows: if $(i,j)\in\Gamma$, then $\pr\pp{[\phi(\s{G})]_{ij}=1}=p_n\cdot 1$, otherwise, if $(i,j)\not\in\Gamma$, then $\pr\pp{[\phi(\s{G})]_{ij}=1}=p_n\cdot\p{q_n/p_n} = q_n$, and thus, $\phi(\s{G})\sim\calG_{\Gamma_n}(n,p_n,q_n)$. Furthermore, the estimator $\hat{\Gamma}(\phi(\s{G}))$ is successful by construction. Therefore, proving the impossibility of recovery over $\calG_{\Gamma_n}(n,1,q_n/p_n)$ implies immediately the impossibility of recovery over $\calG_{\Gamma_n}(n,p_n,q_n)$. Below, with abuse of notation, we focus on $\calG_{\Gamma_n}(n,1,q_n)$.

\subsection{Preliminaries on detection}

In statistical analysis of detection problems, one of the goals is to establish a lower bound on the optimal risk from below, thereby ruling out the possibility of successful detection. A general recipe for this is as follows. Recall the likelihood ratio functional,
\begin{equation}
    \s{L}(\s{G})\triangleq\frac{\mathrm{d}\P_{\calH_1}}{\mathrm{d}\P_{\calH_0}}(\s{G}),\label{eqn:LIKlihood}
\end{equation} 
which is the Radon-Nikodym derivative of $\P_{\calH_1}$ w.r.t. the measure $\P_{\calH_0}$. It is well known (see, e.g., \cite[Theorem 2.2]{tsybakov2004introduction}) that the optimal test $\phi^{\star}$ that minimizes the risk $\s{R}_n$ is the likelihood ratio test defined as,
\begin{align}
\phi^{\star}\left(\s{G}\right) \triangleq \begin{cases}
1,\ &\text{if }\mathsf{L}\left(\s{G}\right) \geq 1\\
0,\ &\text{otherwise},   
\end{cases}
\end{align}
and the associated optimal risk is $ \mathsf{R}^{\star} \triangleq \mathsf{R}\left(\phi^{\star}\left(\s{G}\right)\right)=1-d_{\s{TV}}(\P_{\calH_0},\P_{\calH_1})$. Recalling that $\chi^2(\P_{\calH_0},\P_{\calH_1})=\E_{\calH_0}[\s{L}(\s{G})^2]-1$, it can be shown that (see, e.g., \cite[Sec. 2]{tsybakov2004introduction} and \cite[Prop. 3]{sason2014bounds}),
\begin{align}
    \chi^2(\P_{\calH_0},\P_{\calH_1})\geq \max\p{\frac{1}{2\p{1-d_{\s{TV}}(\P_{\calH_0},\P_{\calH_1})}}-1,\p{2d_{\s{TV}}(\P_{\calH_0},\P_{\calH_1})}^2},
\end{align}
and thus,
\begin{align}
    \mathsf{R}^{\star}&=1-d_{\s{TV}}(\P_{\calH_0},\P_{\calH_1})\geq \max\p{1-\frac{1}{2}\sqrt{\chi^2(\P_{\calH_0},\P_{\calH_1})},\frac{1}{2(1+\chi^2(\P_{\calH_0},\P_{\calH_1}))}}.\label{eqn:lowerBoundSecond}
\end{align}
In particular, we see that $\s{R}^\star$ is bounded away from zero, namely, strong detection is impossible, if $\E[\s{L}(\s{G})^2]$ is bounded. Similarly, $\s{R}^\star$ converge to unity, i.e., weak detection is impossible if $\E_{\calH_0}[\s{L}(\s{G})^2]=1+o(1)$. Accordingly, to rule out the possibility of detection (either strong or weak) it suffices to upper bound the second moment of the likelihood function. 

Let us derive the likelihood function in our case, where the null distribution reflects the distribution of $\calG(n,q_n)$, and the alternative distribution is exactly the one we consider in the recovery problem. Then,
\begin{align}
    \s{L}(\s{G}) = \bE_{\Gamma}\pp{\frac{\pr_{\calH_1\vert\Gamma}(\s{G}\vert\Gamma)}{\pr_{\calH_0}(\s{G})}},
\end{align}
and
\begin{align}
    \frac{\pr_{\calH_1\vert\Gamma}(\s{G}\vert\Gamma)}{\pr_{\calH_0}(\s{G})} = q^{-|\Gamma|}\cdot\Ind\ppp{\Gamma\subseteq \s{G}}.
\end{align}
Thus,
\begin{align}
    \s{L}(\s{G}) = \frac{q^{-e(\Gamma)}}{|\calS_\Gamma|}\sum_{\ell=1}^{|\calS_\Gamma|}\Ind\ppp{\Gamma_\ell\subseteq \s{G}},
\end{align}
where $\{\Gamma_\ell\}_{\ell=1}^{|\calS_\Gamma|}$ are the $|\calS_\Gamma|$ possible copies of $\Gamma$ in $\calK_n$. Let $\s{N}_{\Gamma}(\s{G})\triangleq\sum_{\ell=1}^{|\calS_\Gamma|}\Ind\ppp{\Gamma_\ell\subseteq \s{G}}$ denote the graph copies enumerator. Then, we note that
\begin{align}
    \bE_{\calH_0}[\s{N}_{\Gamma}(\s{G})] = |\calS_\Gamma|\cdot q^{|\Gamma|},
\end{align}
and thus,
\begin{align}
    \s{L}(\s{G}) = \frac{\s{N}_{\Gamma}(\s{G})}{\bE_{\calH_0}[\s{N}_{\Gamma}(\s{G})]}.
\end{align}
Therefore, we obtain that
\begin{align}
    \bE_{\calH_0}[\s{L}^2(\s{G})] = \frac{\bE_{\calH_0}[\s{N}_{\Gamma}^2(\s{G})]}{\bE_{\calH_0}^2[\s{N}_{\Gamma}(\s{G})]}. \label{eqn:secondmomentEnumerator}
\end{align}
In particular, if, under some conditions we have $\bE_{\calH_0}[\s{L}^2(\s{G})]=1+o(1)$, implying that detection is statistically impossible, then, under the same conditions, we have,
\begin{align}
    \bE_{\calH_0}[\s{N}_{\Gamma}^2(\s{G})] &= (1+o(1))\cdot\bE_{\calH_0}^2[\s{N}_{\Gamma}(\s{G})] \\
    & = (1+o(1))\cdot|\calS_\Gamma|^2\cdot q^{2|\Gamma|}.\label{eqn:detectionGur}
\end{align}
It was recently shown in \cite{elimelech2025detecting} that these conditions are: If $\mu(\Gamma_n)\geq \alpha_n\cdot \log|v(\Gamma_n)|$, for some $\alpha_n=\Omega(1)$, then there exists a constant $\underline{C}>0$ such that weak detection is impossible if, 
    \begin{align}
        \mu(\Gamma_n)\leq \underline{C}\cdot \log n.\label{eq:condDenseStatproof}
    \end{align}
If, on the other hand, $\mu(\Gamma_n)=o(\log|v(\Gamma_n)|)$, then for every $\varepsilon>0$, weak detection is impossible if,
    \begin{align}
        |e(\Gamma_n)|\vee d^2_{\max}(\Gamma_n)\leq n^{1-\varepsilon}.\label{eq:condDenseStat2proof}
    \end{align}

\subsection{Generalized subgraph expectation threshold}\label{app:GeneralizedExpecTThreshold}

A central problem in probabilistic combinatorics is to understand, for a fixed graph $\Gamma$, the minimum edge probability $p$ such that the random graph $ G(n, p) $ contains $\Gamma$ as a subgraph with probability at least $ 1/2 $. This threshold is commonly known as the \emph{critical threshold} for the appearance of $\Gamma$, denoted by $ p_c(\Gamma) $. Formally, for graphs $\Gamma$ and $\s{G}$ let us denote the number of copies of $\Gamma$ in $\s{G}$ by $\calN(\Gamma,\s{G})$. The critical probability of $\Gamma$ is defined as,
\begin{align}
    q_c(\Gamma)&\triangleq\min\ppp{q\in[0,1] ~\bigg|~ \P_{\s{G}\sim \calG(n,q)}\pp{\calN(\Gamma,\s{G})\geq 1}\geq \frac{1}{2}}.
    \end{align}
A long-standing conjecture by Kahn and Kalai~\cite{kahn2007thresholds} suggests that this critical threshold is closely approximated—up to a logarithmic factor—by a more tractable quantity called the \emph{subgraph expectation threshold}. In \cite{mossel2022second}, the critical probability was bounded by a modified subgraph expectation threshold, defined as follows,
\begin{align}
    \Tilde{q}_E(\Gamma)\triangleq\min\ppp{q\in[0,1] ~\bigg|~ \E_{\s{G}\sim \calG(n,q)}\pp{\calN(\s{H},\s{G})}\geq \frac{\calN(\s{H},\Gamma)}{2}\text{ for all }\s{H}\subseteq \Gamma},
\end{align} 
where only subgraphs $\s{H}\subseteq \Gamma$ with no isolated vertices are considered.
\begin{theorem}\cite[Theorem 1]{mossel2022second} \label{th:Zadik} There exists a universal constant $C$ such that for any graph $\Gamma$, \begin{align}
    \Tilde{q}_E(\Gamma)\leq q_c(\Gamma)\leq C\cdot \Tilde{q}_E\cdot  \log|e(\Gamma)|.
\end{align}
\end{theorem}
We would like to generalize the above notions to the scenario of the appearance of $L\in\mathbb{N}$ distinct copies. Specifically, define the critical probability as,
\begin{align}
    q_{c}^{(L)}(\Gamma)&\triangleq\min\ppp{q\in[0,1] ~\bigg|~ \P_{\s{G}\sim \calG(n,q)}\pp{\calN(\Gamma,\s{G})\geq L}\geq \frac{1}{2}},\label{eqn:qCl}
    \end{align}
and accordingly,
\begin{align}
    \Tilde{q}_E^{(L)}(\Gamma)\triangleq\min\ppp{q\in[0,1] ~\bigg|~ \E_{\s{G}\sim \calG(n,q)}\pp{\calN(\s{H},\s{G})}\geq \frac{L\cdot \calN(\s{H},\Gamma)}{2}\text{ for all }\s{H}\subseteq \Gamma}.
\end{align} 
Note that for $q\geq q_{c}^{(L)}(\Gamma)$, by Markov's inequality, for any $\s{H}\subseteq\Gamma$,
\begin{align}
    \frac{1}{2}&\leq \P_{\s{G}\sim \calG(n,q)}\pp{\calN(\Gamma,\s{G})\geq L}\\
    &\leq \P_{\s{G}\sim \calG(n,q)}\pp{\calN(\s{H},\s{G})\geq L\cdot \calN(\s{H},\Gamma)}\\
    &\leq \frac{\E_{\s{G}\sim \calG(n,q)}\pp{\calN(\s{H},\s{G})}}{L\cdot \calN(\s{H},\Gamma)},
\end{align}
and thus, $\Tilde{q}_E^{(L)}(\Gamma)\leq q_{c}^{(L)}(\Gamma)$. Finally, note that we can rewrite,
\begin{align}
    \Tilde{q}_E^{(L)}(\Gamma)\triangleq\max\ppp{\p{\frac{L\cdot\calN(\s{H},\Gamma)}{2|\calS_{\s{H}}|}}^{1/|e(\s{H})|}:\;\s{H}\subseteq\Gamma}.\label{eqn:EquivSpreaPro}
\end{align}
We have the following result.
\begin{theorem}\label{th:SpreadGen} There exists a universal constant $C$ such that for any graph $\Gamma$, \begin{align}
    \Tilde{q}_E^{(L)}(\Gamma)\leq q_c^{(L)}(\Gamma)\leq C\cdot L\cdot\Tilde{q}_E^{(L)}\cdot  \log(L|e(\Gamma)|).
\end{align}
\end{theorem}
The proof of Theorem~\ref{th:SpreadGen} relies on a powerful probabilistic tool known as the \emph{spread lemma}, which has been instrumental in several recent breakthroughs. To describe the lemma in our context, consider a probability distribution $\pi$ supported on subgraphs of the complete graph $\calK_n$. Let $\alpha > 1$. We say that the distribution $\pi$ is \emph{$\alpha$-spread} if, for every (non-empty) subgraph $\s{H} \subseteq \calK_n$,
\begin{equation}
\pi(\s{H} \subseteq \Gamma) \leq  \alpha^{-|e(\s{H})|},
\label{eq:spread}
\end{equation}
where $\Gamma\sim \pi$. The following result—adapted from Theorem 1.6 in~\cite{frankl2021fractional}—gives a threshold condition under which a random graph drawn from $G(n,p)$ is likely to intersect a family of such subgraphs.
\begin{lemma}[Spread lemma {\cite[Lemma 2]{mossel2022second}}]
\label{lem:spread}
Fix integers $k\geq  1$ and $M\geq  1$.  
Let $\calG_M\triangleq\ppp{\s{G}_{1},\dots ,\s{G}_{M}}\subseteq K_{n}$ satisfy $|e(\s{G}_{i})|\leq  k$ for every $i$. Let $\pi$ be the uniform measure on $\calG_M$, and assume the
\emph{$\alpha$‑spread} condition in \eqref{eq:spread}. Then, there is an absolute constant $C>0$ such that if $q>C\frac{\log k}{\alpha}$, then a sample from $\calG(n,q)$ contains one of the $\s{G}_i$'s with probability at least $1/2$.
\end{lemma}
We now propose the following generalization of the above spread lemma.
\begin{lemma}[Generalized spread lemma]
\label{lem:Gspread}
Fix integers $k\geq  1$ and $M\geq  L\geq  1$.  
Let $\calG_M\triangleq\ppp{\s{G}_{1},\dots ,\s{G}_{M}}\subseteq K_{n}$ satisfy $|e(\s{G}_{i})|\leq  k$ for every $i$. Let $\pi$ be the uniform measure on $\calG_M$, and assume the
\emph{$\alpha$‑spread} condition in \eqref{eq:spread}. Then, there is an absolute constant $C>0$ such that if,
\begin{align}
  q>C\frac{L\log(kL)}{\alpha}
  \label{eq:threshold}
\end{align}
then,
\begin{align}
  \P_{\s{G}\sim \calG(n,q)}\pp{\s{G}\;
    \s{contains}\; L\; \s{distinct}\; \s{graphs}\;\s{among}\;
    \calG_M}\geq\frac{1}{2}.
\end{align}
\end{lemma}

\begin{proof}[Proof of Lemma~\ref{lem:Gspread}]
Define,
\begin{align}
  \Sigma &\triangleq
  \bigl\{ (i_{1},\dots ,i_{L})\in[M]^{L}:
          i_{1},\dots ,i_{L}\;\s{distinct}\bigr\},
\end{align}
and thus, $|\Sigma|=\frac{M!}{(M-L)!}$. 
%\begin{align}
%  |\Sigma|& = M(M-1)\cdots(M-L+1).
%\end{align}
For $\sigma=(i_{1},\dots ,i_{L})\in\Sigma$, set,
\begin{align}
  \s{G}_{\sigma}\triangleq\bigcup_{\ell=1}^L\s{G}_{i_{\ell}}.
\end{align}
Let
\begin{align}
  \pi^{\langle L\rangle} = \s{Unif}(\Sigma),
\end{align}
be the \emph{uniform} law on the index set $\Sigma$. Sampling $\sigma\sim\pi^{\langle L\rangle}$ and putting
$\s{G}^{\ast}=\s{G}_{\sigma}$ chooses an ordered $L$‑tuple without replacement and bundles its $L$ distinct graphs into one target. Now, fix a non‑empty subgraph $\s{H}\subseteq \calK_{n}$.  
For $j\in[L]$ let,
\begin{align}
  q_{j}\triangleq
  \pr_{\s{G}\sim\pi}\pp{\s{H}\subseteq \s{G}\vert\s{G}\notin\{\s{G}_{i_{1}},\dots ,\s{G}_{i_{j-1}}\}}.
\end{align}
Since $\pi$ is assumed $\alpha$‑spread, and because conditioning can only decrease the chance that $\s{H}\subseteq \s{G}$, we have,
\begin{align}
    q_{j}\leq \pr_{\s{G}\sim\pi}\pp{\s{H}\subseteq \s{G}}\leq\alpha^{-|e(\s{H})|}.
\end{align}
Therefore,
\begin{align}
  \pi^{\langle L\rangle}\p{\s{H}\subseteq \s{G}^{\ast}}&=  1-\prod_{j=1}^{L}(1-q_{j})\\
  &\leq L\alpha^{-|e(\s{H})|}\\
  & \leq (L/\alpha)^{|e(\s{H})|}.
\end{align}
Thus $\pi^{\langle L\rangle}$ is $(\alpha/L)$‑spread. Now, we invoke the single spread lemma in Lemma~\ref{lem:spread}. Note that each element $\s{G}^{\ast}$ has at most $kL$ edges, and there are at most
$|\Sigma|\leq  M^{L}$ such unions. Applying Lemma~\ref{lem:spread} with parameters $k' = kL$ and $\alpha' = \alpha/L$, we see that \eqref{eq:threshold} guarantees
$\P_{\s{G}\sim \calG(n,q)}\pp{\s{G}\supseteq \s{G}^{\ast}}\ge\tfrac12$, which concludes the proof.
\end{proof}
We are now in a position to prove Theorem~\ref{th:SpreadGen}.
\begin{proof}[Proof of Theorem~\ref{th:SpreadGen}]
Let $\pi = \pi_\Gamma$ denote the uniform distribution over all copies of a fixed graph $\Gamma$ within the complete graph $\calK_n$. Let $\Gamma'\sim\pi_\Gamma$ be a random sample from this distribution. Now, consider a subgraph $\s{H}\subset\Gamma$. Let $ \pi_{\s{H}} $ denote the uniform distribution over all copies of $\s{H}$ in $\calK_n$, and let $\s{H}'\sim\pi_{\s{H}}$. For any fixed instances $ \s{H}_0 \subseteq \calK_n $ and $ \Gamma_0 \subseteq \calK_n $, copies of $ \s{H} $ and $ \Gamma $, respectively, we can compute the inclusion probability in two equivalent ways,
\begin{align}
\pi_{\Gamma}(\s{H}_0 \subseteq \Gamma') = \pi_{\s{H}}(\s{H}' \subseteq \Gamma_0) = \frac{\calN(\s{H},\Gamma)}{|\calS_{\s{H}}|}.
\label{eq:symmetry}
\end{align}
Now, combining \eqref{eq:symmetry} with \eqref{eqn:EquivSpreaPro}, we obtain,
\begin{align}
\pi_{\Gamma}(\s{H}_0 \subseteq \Gamma') &=\frac{\calN(\s{H},\Gamma)}{|\calS_{\s{H}}|} \\
&\leq  \frac{2}{L} \cdot \pp{\widetilde{q}^{(L)}_E(\Gamma)}^{|e(\s{H})|}\\
&\leq  \left( \frac{1}{2\widetilde{q}^{(L)}_E(\Gamma)} \right)^{-|e(\s{H})|}.
\end{align}
Since the bound holds uniformly over all subgraphs $\s{H}_0$, we conclude that $\pi_\Gamma$ is $\alpha$-spread with $\alpha = \frac{1}{2\widetilde{q}^{(L)}_E(\Gamma)}$. An application of Lemma~\ref{lem:Gspread} with $k = |e(\Gamma)|$ then completes the proof of Theorem~\ref{th:SpreadGen}.

\end{proof}

Next, we prove the following result.
\begin{theorem}
\label{theorem:sublogProb}
    Fix $L\geq1$. Let $\Gamma=(\Gamma_n)$ be a sequence of graphs such that $\omega(1) \leq |v(\Gamma)|\leq n $  and
    \begin{align}
        (1+\varepsilon)\mu(\Gamma_n)\cdot [\log\log(L|e(\Gamma_n)|)+\log(CL)]+\log|v(\Gamma_n)|\leq \log n,\label{eq:condCopy}
    \end{align}
    for some $\varepsilon>0$ and $C>0$. Then, for any fixed $q$, a sample from $\calG(n,q)$ contains at least $L$ isomorphic copies of $\Gamma$, with probability at least $1/2$.
\end{theorem}
In order to prove Theorem~\ref{theorem:sublogProb} we need to bound the probability that a uniform random copy of $\Gamma$ in $\calK_n$ contains an arbitrary isomorphic copy of a subgraph $\s{H}\subseteq \Gamma$ in $\calK_n$, namely, $\P_{\Gamma}[\s{H}\subseteq \Gamma]$.  
\begin{lemma}\label{lem:probcalc2} For any $\s{H}\subseteq \Gamma$,
\begin{align}
    \P_{\Gamma}[\s{H}\subseteq \Gamma]\leq \p{\frac{|v(\Gamma)|}{n}}^{|v(\s{H})|}.
\end{align}
\end{lemma}
\begin{proof}[Proof of Lemma~\ref{lem:probcalc2}]
Let $m = |v(\Gamma)|$ and $k = |v(\s{H})|$. Fix a particular copy $\s{H}_0$ of $\s{H}$ in $\calK_n$ with vertex set $U = \{u_1,\dots,u_k\}$. Generate a uniformly random copy of $\Gamma$ by first picking its vertex set $S \subset [n]$ uniformly among all $m$-subsets (the internal labeling of $\Gamma$ can only lower the probability of containing $\s{H}_0$, so it suffices to control this step). If the sampled copy contains $\s{H}_0$, then necessarily $U \subseteq S$. Thus
\begin{align}
\pr_{\Gamma}[\s{H}_0 \subseteq \Gamma]&\leq \pr[U \subseteq S]\\
&= \prod_{i=1}^{k} \pr[u_i \in S \vert u_1,\dots,u_{i-1} \in S].
\end{align}
Conditioned on $u_1,\dots,u_{i-1} \in S$, there are $n-(i-1)$ remaining vertices and $m-(i-1)$ remaining slots in $S$, so
\begin{align}
\pr[u_i \in S \vert u_1,\dots,u_{i-1} \in S] 
&= \frac{m-(i-1)}{n-(i-1)}\\
&\leq \frac{m}{n}. 
\end{align}
Multiplying these $k$ bounds gives
\begin{align}
\pr_{\Gamma}[\s{H} \subseteq \Gamma] &\leq\pr[U \subseteq S]\\
&\leq \left(\frac{m}{n}\right)^k = \left(\frac{|v(\Gamma)|}{n}\right)^{|v(\s{H})|}.
\end{align}
    %The proof relies on two simple argument. First, we note that by symmetry considerations, the probability that a uniform random of copy $\Gamma$ contains a fixed copy of $\s{H}$ equals to the probability that a uniform random copy of $\s{H}$ in $\calK_n$ is contained in some fixed copy of $\Gamma$ (see, \cite[Lemma~2]{elimelech2025detecting}, for a formal proof). Second, we note that the event that the random copy $\s{H}$ is contained in some fixed copy $\Gamma'$ contains the event that $v(\s{H})\subseteq v(\Gamma')$. Thus, since $v(\s{H})\sim \s{Unif}\binom{[n]}{|v(\s{H})|}$ we have
    %\begin{align}
    %     \P_{\Gamma}[\s{H}\subseteq \Gamma]&\leq \P_{\s{H}}[v(\s{H})\subseteq v(\Gamma)]=\frac{\binom{|v(\Gamma)|}{|v(\s{H})|}}{\binom{n}{|v(\s{H})|}}\\
    %     &=\frac{|v(\Gamma)|\cdot (|v(\Gamma)|-1)\cdots (|v(\Gamma)|-|v(\s{H})|+1) }{n(n-1)\cdots (n-|v(\s{H})|+1)}\leq \p{\frac{|v(\Gamma)|}{n}}^{|v(\s{H})|},
    %\end{align}
    %where the last equality follows since the function $f(i)=\frac{k-i}{n-i}$ is monotone decreasing for $k<n$.
\end{proof}
We are now ready to prove Theorem~\ref{theorem:sublogProb}.
\begin{proof}[Proof of Theorem~\ref{theorem:sublogProb}]
    From the definition of $q_{c}^{(L)}(\Gamma)$ in \eqref{eqn:qCl}, our goal is to understand for which $\Gamma$'s we have $q_{c}^{(L)}(\Gamma)\leq q$, for any fixed $q\in (0,1]$. By Theorem~\ref{th:SpreadGen}, it is sufficient to show that $ \Tilde{q}_E^{(L)}\leq \frac{q}{CL\log(L|e(\Gamma)|)}$ for any fixed $q$, which holds if and only if,
    \begin{align}
        \inf_{\s{H}\subseteq \Gamma} \frac{\E_{\s{G}\sim\calG(n,\tilde{q})}\pp{\calN(\s{H},\s{G})}}{\calN(\s{H},\Gamma)}\geq \frac{L}{2},\label{eq:condinf}
    \end{align}
    where $\tilde{q}\triangleq\frac{q}{CL\log(L|e(\Gamma)|)}$. We note that,
    \begin{align}
    {\E_{\s{G}\sim\calG(n,\tilde{q})}\pp{\calN(\s{H},\s{G})}}=|\calS_{\s{H}}|\cdot \tilde{q}^{|\s{H}|},
    \end{align}
    where we recall that $|\calS_{\s{H}}|=\calN(\s{H},\calK_n)$ denotes the number of copies of $\s{H}$ in the complete graph, and $|\s{H}|$ denotes the number of edges in $\s{H}$. Thus, \eqref{eq:condinf} holds if and only if for any $\s{H}\subseteq \Gamma$ we have,
    \begin{align}
        \frac{\calN(\s{H},\Gamma)}{|\calS_{\s{H}}|}\leq \frac{2}{L} \tilde{q}^{|\s{H}|}=\frac{2}{L}\p{\frac{q}{CL\log(L|e(\Gamma)|)}}^{|\s{H}|}.\label{eq:probRC}
    \end{align}
    An easy combinatorial argument shows that the expression on the left-hand side of \eqref{eq:probRC} equals $\P_{\Gamma}[\s{H}\subseteq \Gamma]$ (see, \cite[Lemma~2]{elimelech2025detecting}). Thus, \eqref{eq:probRC} holds if,
    \begin{align}
        \log \P_{\Gamma}[\s{H}\subseteq \Gamma]-\log\frac{2}{L}-|\s{H}|\p{\log q -\log(CL)-\log\log(L|e(\Gamma)|)}\leq 0.
    \end{align}
    By the assumptions that $|v(\Gamma)|=\omega(1)$, and that $\Gamma$ has no isolated vertices, for any $\varepsilon$ and for sufficiently large $n$, the above holds if, 
    \begin{align}
        \log \P_{\Gamma}[\s{H}\subseteq \Gamma]+|\s{H}|\log(CL)+(1+\varepsilon)|\s{H}|\log\log(L|e(\Gamma)|)\leq 0.\label{eq:almostF}
    \end{align}
    Finally, using Lemma~\ref{lem:probcalc2}, the left-hand side of \eqref{eq:almostF} can be upper bounded by,
    \begin{align}
        &\log \P_{\Gamma}[\s{H}\subseteq \Gamma]+|\s{H}|\log(CL)+(1+\varepsilon)|\s{H}|\log\log(L|e(\Gamma)|\\
        &\leq |v(\s{H})|\p{\log|v(\Gamma)|-\log n}+|\s{H}|\log(CL)+(1+\varepsilon)|\s{H}|\log\log(L|e(\Gamma)|)\\
        &=|v(\s{H})|\cdot \p{\log|v(\Gamma)|+\frac{|\s{H}|}{|v(\s{H})|}\log(CL) +(1+\varepsilon)\frac{|\s{H}|}{|v(\s{H})|} \log \log (L|e(\Gamma)|) -\log n}\\
        &\overset{(a)}{\leq}|v(\s{H})|\cdot \p{\log|v(\Gamma)|+(1+\varepsilon)\mu(\Gamma) [\log \log |e(\Gamma)|+\log(CL)] -\log n }\leq 0,
    \end{align}
    where $(a)$ follows from the definition of the maximal subgraph density, and the last inequality follows from \eqref{eq:condCopy}. This concludes the proof.
\end{proof}

\subsection{Proof of Theorem~\ref{thm:lowerp1}}\label{app:prrpf}

\noindent\textbf{Recovery through detection.} As mentioned before we focus on the case where $\Gamma$ is on balanced. First note that the worst-case error probability can be lower bounded by the average risk as follows,
Accordingly, 
\begin{align}
    \inf_{\hat{\Gamma}^M}\max_{\Gamma^\star\in\calS_{\Gamma}}\pr[\hat{\Gamma}^M(\s{G};\Gamma^{(M-1)})\neq\Gamma^{(M)}] = \inf_{\hat{\Gamma}^M}\max_{\Gamma^{(M)}\in\calM(\Gamma^{(M-1)},\Gamma)}\pr[\hat{\Gamma}^M(\s{G};\Gamma^{(M-1)})\neq\Gamma^{(M)}]
\end{align}
\begin{align}
    \max_{\Gamma^\star\in\calS_{\Gamma}}\pr[\hat{\Gamma}(\s{G})\neq\Gamma^\star]\geq \bE_{\Gamma\sim\s{Unif}(\calS_\Gamma)}\pr[\hat{\Gamma}(\s{G})\neq\Gamma],\label{eqn:worst-caseRec}
\end{align}
and thus we next focus on the case where $\Gamma$ is drawn uniformly at random. It proves more convenient to analyze the probability of correct recovery, i.e.,
\begin{align}
     \pr_{\calG_{\Gamma_n}}\pp{\hat{\Gamma}(\s{G})=\Gamma} &= \bE\pp{\pr\pp{\left.\hat{\Gamma}(\s{G})=\Gamma\right|\s{G}}},
\end{align}
where we have used the law of total expectation. Let us analyze the distribution of $\Gamma$ \emph{given} $\s{G}$. First, we note that for any $\Gamma'\in\calS_{\Gamma}$,
\begin{align}
    \pr_{\calG_{\Gamma_n}}(\s{G}\vert\Gamma=\Gamma') = q^{\binom{n}{2}-|\Gamma|}\Ind\ppp{\Gamma'\subseteq \s{G}},
\end{align}
and thus, 
\begin{align}
    \pr_{\calG_{\Gamma_n}}(\Gamma=\Gamma'\vert \s{G}) &= \frac{\pr_{\calG_{\Gamma_n}}(\s{G},\Gamma=\Gamma')}{\pr_{\calG_{\Gamma_n}}(\s{G})}\\
     &= \frac{\pr_{\calG_{\Gamma_n}}(\s{G}\vert\Gamma=\Gamma')\pr(\Gamma=\Gamma')}{\sum_{\Gamma''\in\calS_{\Gamma}}\pr_{\calG_{\Gamma_n}}(\s{G}\vert\Gamma=\Gamma'')\pr(\Gamma=\Gamma'')}\\
     &=\frac{q^{\binom{n}{2}-|\Gamma|}\Ind\ppp{\Gamma'\subseteq \s{G}}}{\sum_{\Gamma''\in\calS_{\Gamma}}q^{\binom{n}{2}-|\Gamma''|}\Ind\ppp{\Gamma''\subseteq \s{G}}}\\
    & = \frac{\Ind\ppp{\Gamma'\subseteq \s{G}}}{\s{N}_{\Gamma}(\s{G})},
\end{align}
where the second equality follows from the fact that $\Gamma$ is drawn uniformly over $\calS_{\Gamma}$, the last inequality is because $|\Gamma'|=|\Gamma''|$, for any $\Gamma',\Gamma''\in\calS_{\Gamma}$, and we recall that $\s{N}_{\Gamma}(\s{G})$ counts the number of copies of $\Gamma$ in $\s{G}$. Thus, using the above we see that,
\begin{align}
    \pr_{\calG_{\Gamma_n}}\pp{\hat{\Gamma}(\s{G})=\Gamma} &= \bE\pp{\pr\pp{\left.\hat{\Gamma}(\s{G})=\Gamma\right|\s{G}}}\\
    & = \bE\pp{\frac{\Ind\ppp{\hat{\Gamma}(\s{G})\subseteq \s{G}}}{\s{N}_{\Gamma}(\s{G})}}\\
    &\leq \bE\pp{\s{N}^{-1}_{\Gamma}(\s{G})}.\label{eqn:recCount}
\end{align}
Luckily, we understand very well the distribution of $\s{N}_{\Gamma}(\s{G})$ when $\s{G}\sim\calG(n,q_n)$; here, on the other hand, the expectation is taken w.r.t. $\calG_{\Gamma_n}$. Nonetheless, we note to the following simple observation. Let $\s{G}\setminus\Gamma^\star$ be defined as the graph obtained by removing from $\s{G}$, the planted subgraph $\Gamma^\star$, all the edges between the vertices of $\Gamma^\star$ in $\s{G}$, and the edges from $\Gamma^\star$ to the vertices in $\s{G}$. Then, it is clear that $\s{G}\setminus\Gamma^\star\sim\calG(n-|v(\Gamma)|,q_n)$. Furthermore, if we let $\bar{\s{N}}_{\Gamma}(\bar{\s{G}})\triangleq\s{N}_{\Gamma}(\s{G}\setminus\Gamma^\star)$, then clearly, $\s{N}_{\Gamma}(\s{G}\setminus\Gamma^\star)\leq\s{N}_{\Gamma}(\s{G})$. Thus, if we fix $\ell\in\mathbb{N}$, then,
\begin{align}
    \bE\pp{\frac{1}{\s{N}_{\Gamma}(\s{G})}}& = \bE\pp{\frac{1}{\s{N}_{\Gamma}(\s{G})}\Ind\ppp{\s{N}_{\Gamma}(\s{G})>\ell}}+\bE\pp{\frac{1}{\s{N}_{\Gamma}(\s{G})}\Ind\ppp{\s{N}_{\Gamma}(\s{G})\leq\ell}}\\
    &\leq \ell^{-1}+\pr\pp{\s{N}_{\Gamma}(\s{G})\leq\ell}\\
    &\leq \ell^{-1}+\pr\pp{\s{N}_{\Gamma}(\s{G}\setminus\Gamma^\star)\leq\ell},\label{eqn:ExactRecov}
\end{align}
where in the second inequality we used the fact that $\s{N}_{\Gamma}(\s{G}) \geq1$ with probability one.

Finally, we prove Theorem~\ref{thm:lowerp1}, again, assuming that $\Gamma$ is balanced. Recalling \eqref{eqn:secondmomentEnumerator}, let $\bar{\s{L}}(\s{G})\triangleq\frac{\bar{\s{N}}_{\Gamma}(\bar{\s{G}})}{\bE\bar{\s{N}}_{\Gamma}(\bar{\s{G}}}$ denote the respective likelihood function. By Chebyshev's inequality, for $0<\alpha<1$,
\begin{align}
    \pr\pp{\bar{\s{N}}_{\Gamma}(\bar{\s{G}})\leq\alpha\cdot\bE\bar{\s{N}}_{\Gamma}(\bar{\s{G}})}&\leq \pr\pp{\abs{\bar{\s{N}}_{\Gamma}(\bar{\s{G}})-\bE\bar{\s{N}}_{\Gamma}(\bar{\s{G}})}\geq(1-\alpha)\cdot\bE\bar{\s{N}}_{\Gamma}(\bar{\s{G}})}\\
    &\leq\frac{\s{Var}\p{\bar{\s{N}}_{\Gamma}(\bar{\s{G}})}}{(1-\alpha)^2\bE^2\bar{\s{N}}_{\Gamma}(\bar{\s{G}})}.
\end{align}
Also,
\begin{align}
    \frac{\s{Var}\p{\bar{\s{N}}_{\Gamma}(\bar{\s{G}})}}{\bE^2\bar{\s{N}}_{\Gamma}(\bar{\s{G}})} &= \frac{\bE\bar{\s{N}}^2_{\Gamma}(\bar{\s{G}})}{\bE^2\bar{\s{N}}_{\Gamma}(\bar{\s{G}})}-1\\
    & = \bE[\bar{\s{L}}^2(\bar{\s{G}})]-1,
\end{align}
where the expectation is w.r.t. $\bar{\s{G}}\sim\calG(n-|v(\Gamma)|,q_n)$. The above second moment is, almost, the quantity we bound when lower bounding the risk of the corresponding \emph{detection problem}; it is almost, because here the underlying graphs are over $n-|v(\Gamma)|$ rather than $n$ vertices. However, this can be accounted for by replacing $n$ with $n-|v(\Gamma)|$ in \eqref{eq:condDenseStatproof}. Specifically, for sufficiently large $n$, it should be clear that there exist constants $\underline{C}$ and $\varepsilon>0$, such that \eqref{eq:condDenseStatproof} hold with $n$ replaced by $n-|v(\Gamma)|$, and accordingly, under these conditions, using \eqref{eqn:detectionGur},
\begin{align}
    \frac{\s{Var}\p{\bar{\s{N}}_{\Gamma}(\bar{\s{G}})}}{\bE^2\bar{\s{N}}_{\Gamma}(\bar{\s{G}})} &=\bE[\bar{\s{L}}^2(\bar{\s{G}})]-1=o(1).
\end{align}
Thus, with $\ell =  \alpha\cdot\bE\bar{\s{N}}_{\Gamma}(\bar{\s{G}})$,
\begin{align}
    \pr\pp{\s{N}_{\Gamma}(\s{G})\leq\ell}=o(1).
\end{align}
Finally, notice that, if we let $k=|v(\Gamma)|$, then,
\begin{align}
    \ell &= \alpha\cdot\bE\bar{\s{N}}_{\Gamma}(\bar{\s{G}}) \\
    &= \binom{n-k}{k}\frac{k!}{|\s{Aut}(\Gamma)|}\cdot q^{|\Gamma|}\\
    &\geq \p{\frac{n-k}{k}}^kq^{|\Gamma|}\\
    & = 2^{k\log\frac{n-k}{k}-|\Gamma|\log\frac{1}{q}}\\
    & = 2^{k\p{\log\frac{n-k}{k}-\frac{|\Gamma|}{k}\log\frac{1}{q}}}\\
    & \geq 2^{k\p{\log\frac{n-k}{k}-\mu(\Gamma_n)\log\frac{1}{q}}}\to\infty,
\end{align}
for some $\underline{C}$ under \eqref{eq:condDenseStat} for graphs with super-logarithmic density. Thus, using the above fact we get from \eqref{eqn:worst-caseRec}, \eqref{eqn:recCount}, and \eqref{eqn:ExactRecov}, that,
\begin{align}
    \max_{\Gamma^\star\in\calS_{\Gamma}}\pr[\hat{\Gamma}(\s{G})\neq\Gamma^\star]\geq & 1-o(1),\label{eqn:ExactRecov2}
\end{align}
under the same conditions as in \eqref{eq:condDenseStatproof} albeit with a different constant $\underline{C}$.

%\paragraph{Proof of \eqref{eq:condDenseStat2}.} To prove \eqref{eq:condDenseStat2}, we use Theorem~\ref{theorem:sublogProb}. Specifically, let $n'\triangleq n-|v(\Gamma)|$. Then, using Theorem~\ref{theorem:sublogProb}, we have that for any $q$ and sequence of graphs $\Gamma=(\Gamma_n)$ such that
%    \begin{align}
%        (1+\varepsilon)\mu(\Gamma_n)\cdot [\log\log(\ell|e(\Gamma_n)|)+\log(C\ell)]+\log|v(\Gamma_n)|\leq \log n',\label{eq:condCopyP}
%    \end{align}
%a sample from $\calG(n-|v(\Gamma)|,q)$ contains at least $\ell$ isomorphic copies of $\Gamma$, with probability at least $1/2$, namely, $\pr\pp{\s{N}_{\Gamma}(\s{G}\setminus\Gamma^\star)\leq\ell}\leq 1/2$. Thus, for any $\ell\in\mathbb{N}$, we get from \eqref{eqn:ExactRecov} that,
%\begin{align}
%    \bE\pp{\frac{1}{\s{N}_{\Gamma}(\s{G})}}& \leq \frac{1}{\ell}+\frac{1}{2},\label{eqn:ExactRecovIns}
%\end{align}
%and in light of \eqref{eqn:worst-caseRec} and \eqref{eqn:recCount} implies that,
%\begin{align}
%    \max_{\Gamma^\star\in\calS_{\Gamma}}\pr[\hat{\Gamma}(\s{G})\neq\Gamma^\star]\geq\frac{1}{2}-\frac{1}{\ell},
%\end{align}
%which is positive for any $\ell\geq3$. Finally, it is evident that \eqref{eqn:ExactRecovIns} holds under \eqref{eq:condDenseStat2}, for sufficiently large $n$.

\section{Almost-exact recovery}\label{app:weakRec}
In this appendix, we prove Theorems~\ref{thm:subLogAlmostExactMain}. To that end, we start with some preliminaries. 

\subsection{Preliminaries and auxiliary results}

\paragraph{Equivalent notions of recovery.} Recall that an estimator $\hat{\Gamma}$ almost-exactly recover $\Gamma^\star$ if, as $n \to \infty$, $\s{d}_{\s{H}}(\hat{\Gamma},\Gamma^\star)/|e(\Gamma)| \to 0$ in probability. We have the following result.
\begin{lemma}\label{lem:equivWEProbMean}
    An estimator $\hat{\Gamma}$ almost-exactly recover $\Gamma^\star$ if and only if $\bE\s{d}_{\s{H}}(\hat{\Gamma},\Gamma^\star)/|e(\Gamma)| \to 0$.
\end{lemma}
\begin{proof}[Proof of Lemma~\ref{lem:equivWEProbMean}]
The forward implication follows immediately, since $L_1$-convergence entails convergence in probability \cite[Proposition 6.14]{Folland1999}. For the other direction, assume there is $\hat{\Gamma}$ with
\begin{align}
\frac{d_\s{H}(\Gamma^\star,\hat{\Gamma})}{m} \xrightarrow{\mathsf{P}} 0,
\end{align}
where $m\triangleq|e(\Gamma)|$. Let $\Pi_m(\hat{\Gamma})$ be any $m$-sparse vector minimizing Hamming distance to $\hat{\Gamma}$, i.e., delete or add ones to reach size $m$). Then, by triangle inequality,
\begin{align}
d_{H}(\Gamma^\star,\Pi_m(\hat{\Gamma})) &\leq d_{\s{H}}(\Gamma^\star,\hat{\Gamma}) + d_{\s{H}}(\hat{\Gamma},\Pi_m(\hat{\Gamma}))\\
&= d_{\s{H}}(\Gamma^\star,\hat{\Gamma}) + \big||\hat{\Gamma}| - m\big| \\
& = d_{\s{H}}(\Gamma^\star,\hat{\Gamma}) + \big||\s{supp}(\hat{\Gamma})| - |\s{supp}(\Gamma^\star)|\big|\\
&\leq 2\cdot d_{\s{H}}(\Gamma^\star,\hat{\Gamma}).
\end{align}
Furthermore, $|\Gamma^\star| = |\Pi_m(\hat{\Gamma})| = m$ gives $d_{H}(\Gamma^\star,\Pi_m(\hat{\Gamma})) \leq 2m$. Hence
\begin{align}
\frac{d_{H}(\Gamma^\star,\Pi_m(\hat{\Gamma}))}{m} \leq 2\left(\frac{d_{\s{H}}(\Gamma^\star,\hat{\Gamma})}{m} \wedge 1\right),
\end{align}
with probability one. Now, for any nonnegative $X$, we have $X \wedge 1 = \int_0^1 \Ind\ppp{X > t}\mathrm{d}t$. Therefore
\begin{align}
\mathbb{E}\left[\frac{d_{H}(\Gamma^\star,\Pi_m(\hat{\Gamma}))}{m}\right] &\leq 2\mathbb{E}\left[\left(\frac{d_{\s{H}}(\Gamma^\star,\hat{\Gamma})}{m} \wedge 1\right)\right]\\
&= 2 \int_0^1 \mathbb{P}\left(\frac{d_{\s{H}}(\Gamma^\star,\hat{\Gamma})}{m} > t\right)\mathrm{d}t.
\end{align}
Because $\frac{d_\s{H}(\Gamma^\star,\hat{\Gamma})}{m} \xrightarrow{\mathsf{P}} 0$, for every fixed $t>0$ we have $\mathbb{P}(d_\s{H}(\Gamma^\star,\hat{\Gamma}) > mt) \to 0$; the integrand is bounded by $1$. By dominated convergence,
\begin{align}
 \int_0^1 \mathbb{P}\left(\frac{d_{\s{H}}(\Gamma^\star,\hat{\Gamma})}{m} > t\right)\mathrm{d}t\to 0,
\end{align}
which implies that $\mathbb{E}[d_{H}(\Gamma,\Pi_m(\hat{\Gamma}))]/m \to 0$. Thus, the estimator $\tilde{\Gamma} \triangleq  \Pi_m(\hat{\Gamma})$ satisfies $\mathbb{E}[d_{H}(\Gamma^\star,\Pi_m(\hat{\Gamma}))/m] \to 0$, as required. 
\end{proof}
%Finally, note that $\bE d_{\s{H}}(\Gamma,\hat{\Gamma}) = |v(\Gamma)|-\s{ov}(\hat{\Gamma})$, so $\s{ov}(\hat{\Gamma}) = |v(\Gamma)|(1-o(1))$ implies weak recovery. But we need to change that for edges not vertices. We have the following result.

\paragraph{Genie argument for recovery.}
To prove our lower bounds, we rely on the following genie argument. Consider the onion decomposition in Definition~\ref{def:onionDec}, and for each $\Gamma^{(\ell)}$ define its onion tail as $T_\ell(\Gamma) \triangleq \Gamma \setminus \Gamma^{(\ell)}$. In light of Lemma~\ref{lem:equivWEProbMean}, the (global) minimax risk is
\begin{align}
\s{E}^\star_{\s{almost}}\triangleq \inf_{\hat\Gamma}\sup_{\Gamma^\star\in\mathcal \calS_\Gamma}
\mathbb E\left[\frac{d_{\s{H}}(\hat\Gamma,\Gamma^\star)}{|e(\Gamma^\star)|}\right].
\end{align}
Define the \emph{genie minimax risk} at level $\ell$ as follows. The oracle reveals the \emph{true} $\Gamma^{(\ell)}$ and we only need to recover the tail $T_\ell(\Gamma^\star)$;  the loss is the \emph{normalized} Hamming error on the tail
\begin{align}
\s{e}_{\s{almost}}^\star(\ell)\triangleq \inf_{\hat T}\sup_{\Gamma^\star\in\calS_\Gamma}
\mathbb E\left[\frac{d_{\s{H}}\big(\hat T(\s{G},\Gamma^{(\ell)}),\,T_\ell(\Gamma^\star)\big)}{|e(T_\ell(\Gamma^\star))|}\right].\label{eqn:lossGenPar}
\end{align}
We have the following result.
\begin{lemma}\label{lem:genWeak}
    For any fixed index $\ell\in[M(\Gamma)]$,
\begin{align}
\s{E}^\star_{\s{almost}}\geq\frac{|e(T_\ell(\Gamma^\star))|}{|e(\Gamma^\star)|}\cdot\s{e}_{\s{almost}}^\star(\ell).
\end{align}
%Equivalently, \emph{pointwise} for each $\Gamma^\star$,
%\begin{align}
%\inf_{\hat\Gamma}\;\sup_{\Gamma^\star}
%\mathbb E \left[\frac{d_{\s{H}}(\hat\Gamma,\Gamma^\star)}{|E(\Gamma^\star)|}\right]
%\;\;\ge\;\;
%\frac{|E(T_i(\Gamma^\star))|}{|E(\Gamma^\star)|}\; r_n(i).
%\end{align}
%(Use the $\inf_{\Gamma^\star}$ of the factor if you want a bound uniform over the class.)
\end{lemma}

\begin{proof}[Proof of Lemma~\ref{lem:genWeak}]
Fix any estimator $\hat\Gamma$. For each $\Gamma^\star$, let $\Pi_{T_\ell(\Gamma^\star)}(\hat{\Gamma})$ denote the projection of $\hat{\Gamma}$ into $T_\ell(\Gamma^\star)$. Then, disagreements on the tail are a subset of global disagreements, and so
\begin{align}
d_{\s{H}}(\hat\Gamma,\Gamma^\star)\geq d_{\s{H}}\p{\hat\Gamma|_{T_\ell(\Gamma^\star)},T_\ell(\Gamma^\star)}.
\end{align}
Divide by $|e(\Gamma^\star)|$ and take expectation:
\begin{equation}
\mathbb E\left[\frac{d_{\s{H}}(\hat\Gamma,\Gamma^\star)}{|e(\Gamma^\star)|}\right]
\geq
\frac{|e(T_\ell(\Gamma^\star))|}{|e(\Gamma^\star)|}
\mathbb E\left[\frac{d_{\s{H}}\p{\hat\Gamma|_{T_\ell(\Gamma^\star)},T_\ell(\Gamma^\star)}}{|e(T_i(\Gamma^\star))|}\right].\label{eqn:ssTar}
\end{equation}
%From $\hat\Gamma$, build a genie estimator for the tail:
%\begin{align}
%\widehat T^{(\hat\Gamma)}(\text{data}, S) \triangleq  \hat\Gamma(\text{data})\big|_{T_i(S)}.
%\end{align}
%When the oracle provides $S=\Gamma^{(i)}$, for each $\Gamma^\star$ the genie loss of $\widehat T^{(\hat\Gamma)}$ 
%equals the \emph{tail} loss on the RHS of~($\star$). 
Consequently, by the definition of \eqref{eqn:lossGenPar}, we clearly have
\begin{align}
\sup_{\Gamma^\star\in\mathcal \calS_\Gamma}
\mathbb E\left[\frac{d_{\s{H}}\p{\hat\Gamma|_{T_\ell(\Gamma^\star)},T_\ell(\Gamma^\star)}}{|e(T_i(\Gamma^\star))|}\right]
\geq \s{e}_{\s{almost}}^\star(\ell).
\end{align}
Taking $\sup_{\Gamma^\star}$ in \eqref{eqn:ssTar} and then $\inf_{\hat\Gamma}$,
\begin{align}
\begin{aligned}
\s{E}^\star_{\s{almost}}
&=\inf_{\hat\Gamma}\sup_{\Gamma^\star\in\mathcal \calS_\Gamma}
\mathbb E\left[\frac{d_{\s{H}}(\hat\Gamma,\Gamma^\star)}{|e(\Gamma^\star)|}\right] \\
&\geq \inf_{\hat\Gamma}\sup_{\Gamma^\star\in\mathcal \calS_\Gamma}
\frac{|e(T_\ell(\Gamma^\star))|}{|e(\Gamma^\star)|}
\mathbb E\left[\frac{d_{\s{H}}\p{\hat\Gamma|_{T_\ell(\Gamma^\star)},T_\ell(\Gamma^\star)}}{|e(T_\ell(\Gamma^\star))|}\right] \\
& = \frac{|e(T_\ell(\Gamma^\star))|}{|e(\Gamma^\star)|} \inf_{\hat\Gamma}\sup_{\Gamma^\star\in\mathcal \calS_\Gamma}
\mathbb E\left[\frac{d_{\s{H}}\p{\hat\Gamma|_{T_\ell(\Gamma^\star)},T_\ell(\Gamma^\star)}}{|e(T_\ell(\Gamma^\star))|}\right]\\
&\ge \frac{|e(T_\ell(\Gamma^\star))|}{|e(\Gamma^\star)|}\cdot \s{e}_{\s{almost}}^\star(\ell),
\end{aligned}
\end{align}
where the second equality follows from the fact that $\frac{|e(T_\ell(\Gamma^\star))|}{|e(\Gamma^\star)|}$ is the same for all $\Gamma^\star\in\calS_\Gamma$.
\end{proof}
Define
\begin{align}
r_\ell^{(n)}\triangleq\frac{|e(\Gamma_n^\star\setminus \Gamma_n^{(\ell)})|}{|e(\Gamma^\star_n)|}.
\end{align}
Fix any sequence $\varepsilon_n\downarrow 0$ and define
\begin{align}
\ell_{\s{LB}}(n)\triangleq\max\{\ell: r_\ell^{(n)}>\varepsilon_n\}.
\end{align}
Heuristically, $\ell_{\s{LB}}$ is the last index where the ratio $r_\ell^{(n)}$ is $\Omega(1)$. Then, using Lemma~\ref{lem:genWeak} we have
\begin{align}
    \s{E}^\star_{\s{almost}}\geq \frac{|e(T_{\ell_{\s{LB}}}(\Gamma^\star))|}{|e(\Gamma^\star)|}\cdot\s{e}_{\s{almost}}^\star(\ell_{\s{LB}}),
\end{align}
and by the definition of $\ell_{\s{LB}}$, we know that the multiplicative factor $\frac{|e(T_{\ell_{\s{LB}}}(\Gamma^\star))|}{|e(\Gamma^\star)|}$ is strictly positive for all $n$. Hence, to rule out the possibility of almost-exact recovery, it suffices to prove that $\s{e}_{\s{almost}}^\star(\ell_{\s{LB}}) = \Omega(1)$. In fact, we first lower bound the worst-case error probability $\s{e}_{\s{almost}}$ by its average-case counterpart as follows:
\begin{align}
\s{e}_{\s{almost}}^\star(\ell)&= \inf_{\hat T}\sup_{\Gamma^{\ell}}\sup_{T_\ell\in\calM(\Gamma^{(\ell)},\Gamma^\star)}
\mathbb E\left[\frac{d_{\s{H}}\big(\hat T(\s{G},\Gamma^{(\ell)}),\,T_\ell\big)}{|e(T_\ell)|}\right]\\
&\geq \inf_{\hat T}\sup_{\Gamma^{\ell}}\bE_{T_\ell\sim\pi}
\mathbb E\left[\frac{d_{\s{H}}\big(\hat T(\s{G},\Gamma^{(\ell)}),\,T_\ell\big)}{|e(T_\ell)|}\right]\\
&\triangleq\bar{\s{e}}_{\s{almost}}^\star(\ell),\label{eqn:lossGenParAvg}
\end{align}
where $\pi \triangleq \s{Unif}(\calM(\Gamma^{(\ell)},\Gamma^\star))$ is the uniform measure over $\calM(\Gamma^{(\ell)},\Gamma^\star)$, and the inequality follows from the fact that the worst-case risk is lower bounded by the average-case risk.

\subsection{Lower bound through detection}

By the results of the previous subsection, it suffices to establish lower bounds for recovering $\bar{\Gamma}_n \triangleq \Gamma^\star_n \setminus \Gamma_n^{\ell_{\s{LB}}}$. Let $\s{G}'\triangleq\s{G}\setminus\Gamma_n^{\ell_{\s{LB}}}$ and $n'\triangleq n-v(|\Gamma_n^{\ell_{\s{LB}}}|)$. We have the following result.
\begin{theorem}\label{thm:lowerp2weak}
Assume that $p_n,q_n=\Theta(1)$. %Then, the following hold.
\begin{enumerate}
\item If $\mu(\bar\Gamma_n)\geq \alpha_n\cdot \log|v(\bar\Gamma_n)|$, for some $\alpha_n=\Omega(1)$, then there exists a constant $\underline{C}>0$ such that almost-exact recovery of $\bar\Gamma$ is impossible if
    \begin{align}
        \mu(\bar\Gamma_n)\leq \underline{C}\cdot \log n'.\label{eq:condDenseStatweakApp}
    \end{align}
\item If $\mu(\bar\Gamma_n)=o(\log|v(\bar\Gamma_n)|)$, then for every $\varepsilon>0$, almost-exact recovery of $\bar\Gamma$ is impossible if
    \begin{align}
        |e(\bar\Gamma_n)|\vee d^2_{\max}(\bar\Gamma_n)\leq n^{1-\varepsilon}.\label{condDenseStat2weakApp}
    \end{align}
\end{enumerate}
\end{theorem}

\begin{proof}[Proof of Theorem~\ref{thm:lowerp2weak}]
Let $k\triangleq|v(\bar\Gamma)|$, and $L_n\triangleq|\calM(\Gamma_n^{\ell_{\s{LB}}},\Gamma^\star)|$ where $\calM(\Gamma_n^{\ell_{\s{LB}}},\Gamma^\star)$ is the number of ways $\Gamma_n^{\ell_{\s{LB}}}$ can be extended to a copy of $\Gamma^\star$ in $\s{K}_n$, or the number of copies of $\Gamma^\star$ in $\s{K}_n$ that contain $\Gamma_n^{\ell_{\s{LB}}}$. Define the following quantity:
\begin{align}
\s{Int}_{\s{G}'}(\bar\Gamma)\triangleq \frac{1}{\calN^2_{\bar\Gamma}}\sum_{\ell_1=1}^{L_n}\sum_{\ell_1=2}^{L_n}\Ind\pp{\bar{\Gamma}_{\ell_1}\cap\bar{\Gamma}_{\ell_2}\neq\emptyset,\;\bar{\Gamma}_{\ell_1},\bar{\Gamma}_{\ell_2}\in\s{G}'},
\end{align}
where $\bar{\Gamma}_1,\bar{\Gamma}_2,\ldots,\bar{\Gamma}_{L_n}$ are all possible subgraph copies, and $\calN_{\bar\Gamma}\triangleq \sum_{\ell=1}^{L_n}\Ind\pp{\bar{\Gamma}_\ell\in\s{G}}$. Namely, $\s{Int}_{\s{G}}(\bar\Gamma)$ is the proportion of pairs copies of $\bar\Gamma$ in $\s{G}$ whose intersection is nonempty. Let $\s{N}_\ell\triangleq\Ind\pp{\bar{\Gamma}_\ell\in\s{G}}$. Then
\begin{align}
\bE_{\calH_0}[\calN^2_{\bar\Gamma}] &= \sum_{\ell_1,\ell_2}\pr_{\calH_0}[\s{N}_{\ell_1}\s{N}_{\ell_2}] \\
&= \sum_{\ell_1,\ell_2:\bar{\Gamma}_{\ell_1}\cap\bar{\Gamma}_{\ell_2}=\emptyset}\pr_{\calH_0}[\s{N}_{\ell_1}\s{N}_{\ell_2}]+\sum_{\ell_1,\ell_2:\bar{\Gamma}_{\ell_1}\cap\bar{\Gamma}_{\ell_2}\neq\emptyset}\pr_{\calH_0}[\s{N}_{\ell_1}\s{N}_{\ell_2}]\\
& \triangleq \s{A}+\s{B}.
\end{align}
We can easily compute $\s{A}$. Indeed
\begin{align}
\s{A} = L_n\cdot L_{n-k}\cdot q^{2e(\bar\Gamma)}\sim [\bE_{\calH_0}(\calN_{\bar\Gamma})]^2.
\end{align}
Now, recall that the likelihood function is defined as $\s{L}(\s{G}') = \calN_{\bar\Gamma}/\bE_{\calH_0}(\calN_{\bar\Gamma})$, and that under conditions \eqref{eq:condDenseStatproof}--\eqref{eq:condDenseStat2proof}, we have $\bE_{\calH_0}[\s{L}(\s{G}')]^2=1+o(1)$. Therefore, it follows that
\begin{align}
\frac{\s{B}}{\pp{\bE_{\calH_0}(\calN_{\bar\Gamma})}^2}=o(1).\label{eqn:Basymp}
\end{align}
Next, we note that
\begin{align}
\s{Int}_{\s{G}'}(\bar\Gamma)= \frac{1}{\calN^2_{\bar\Gamma}}\sum_{\ell_1,\ell_2:\bar{\Gamma}_{\ell_1}\cap\bar{\Gamma}_{\ell_2}\neq\emptyset}\s{N}_{\ell_1}\s{N}_{\ell_2}.
\end{align}
Therefore, we can decompose $\s{Int}_{\s{G}'}(\bar\Gamma)$ as follows
\begin{align}
\s{Int}_{\s{G}'}(\bar\Gamma) &= \s{Int}_{\s{G}'}(\bar\Gamma)\Ind\pp{\s{L}^2(\s{G}')\geq1/2}+ \s{Int}_{\s{G}'}(\bar\Gamma)\Ind\pp{\s{L}^2(\s{G}')<1/2}\\
& = \frac{\sum_{\ell_1,\ell_2:\bar{\Gamma}_{\ell_1}\cap\bar{\Gamma}_{\ell_2}\neq\emptyset}\s{N}_{\ell_1}\s{N}_{\ell_2}}{\bE[\calN^2_{\bar\Gamma}]}\frac{1}{\s{L^2(\s{G}')}}\Ind\pp{\s{L}^2(\s{G}')\geq1/2}+ \s{Int}_{\s{G}'}(\bar\Gamma)\Ind\pp{\s{L}^2(\s{G}')<1/2},
\end{align}
and thus
\begin{align}
\bE_{\calH_0}[\s{Int}_{\s{G}'}(\bar\Gamma)]&\leq \bE_{\calH_0}\pp{\frac{\sum_{\ell_1,\ell_2:\bar{\Gamma}_{\ell_1}\cap\bar{\Gamma}_{\ell_2}\neq\emptyset}\s{N}_{\ell_1}\s{N}_{\ell_2}}{\bE[\calN^2_{\bar\Gamma}]}\frac{1}{\s{L^2(\s{G}')}}\Ind\pp{\s{L}^2(\s{G}')\geq1/2}}\nonumber\\
&\quad+\bE_{\calH_0}\pp{\s{Int}_{\s{G}'}(\bar\Gamma)\Ind\pp{\s{L}^2(\s{G}')<1/2}}\\
&\leq 2\cdot\bE_{\calH_0}\pp{\frac{\sum_{\ell_1,\ell_2:\bar{\Gamma}_{\ell_1}\cap\bar{\Gamma}_{\ell_2}\neq\emptyset}\s{N}_{\ell_1}\s{N}_{\ell_2}}{\bE[\calN^2_{\bar\Gamma}]}}+\bE_{\calH_0}\pp{\Ind\pp{\s{L}^2(\s{G}')<1/2}}\\
&= 2\cdot \frac{\s{B}}{\pp{\bE_{\calH_0}(\calN_{\bar\Gamma})}^2}+\pr_{\calH_0}\pp{\s{L}^2(\s{G}')<1/2}\\
& \leq o(1),\label{eqn:ExpectedOverlap}
\end{align}
where in the inequality we have used the fact that $\s{Int}_{\s{G}'}(\bar\Gamma)\leq1$ with probability one, and the last equality is due to \eqref{eqn:Basymp} and Chebyshev's inequality. In order to prove that almost-exact recovery is impossible we will look the following overlap measure
\begin{align}
\s{over}(\hat{\Gamma})\triangleq \sum_{(i,j)\in\binom{[n]}{2}}\pr_{\calH_1}[(i,j)\in\bar\Gamma\cap\hat{\bar\Gamma}],
\end{align}
where $\hat{\bar\Gamma}$ is any possible estimator of $\bar\Gamma$. Note that $\bE d_{\s{H}}(\bar\Gamma,\hat{\bar\Gamma}) = 2|e(\bar\Gamma)|-2\s{over}(\hat{\bar\Gamma})$. Thus to rule out almost-exact recovery, it suffices to prove that $\s{over}(\hat{\bar\Gamma})=o(|e(\bar\Gamma)|)$. To that end, we note that $\s{over}(\hat{\bar\Gamma})$ can be rewritten as follows
\begin{align}
\s{over}(\hat{\bar\Gamma}) &= \sum_{\s{G}'}\pr_{\calH_1}(\s{G}')\sum_{\ell=1}^{L_n}\pr_{\calH_1}(\bar{\Gamma}_\ell\vert\s{G}')|\bar{\Gamma}_\ell\cap\hat{\bar\Gamma}|\\
& = \sum_{\s{G}'}\pr_{\calH_1}(\s{G}')\sum_{\ell=1}^{L_n}\frac{|\bar{\Gamma}_\ell\cap\hat{\bar\Gamma}|}{\calN_{\bar\Gamma}}\\
& = \sum_{\s{G}'}\pr_{\calH_0}(\s{G}')\sum_{\ell=1}^{L_n}\frac{|\bar{\Gamma}_\ell\cap\hat{\bar\Gamma}|}{\calN_{\bar\Gamma}}+\sum_{\s{G}'}[\pr_{\calH_1}(\s{G}')-\pr_{\calH_0}(\s{G}')]\sum_{\ell=1}^{L_n}\frac{|\bar{\Gamma}_\ell\cap\hat{\bar\Gamma}|}{\calN_{\bar\Gamma}}\\
&\leq \sum_{\s{G}'}\pr_{\calH_0}(\s{G}')\sum_{\ell=1}^{L_n}\frac{|\bar{\Gamma}_\ell\cap\hat{\bar\Gamma}|}{\calN_{\bar\Gamma}}+|e(\bar\Gamma)|\cdot\s{TV}(\pr_{\calH_0},\pr_{\calH_1}),
\end{align}
where in the last inequality we have used the definition of the total-variation distance, and the fact that $|\bar{\Gamma}_\ell\cap\hat{\bar\Gamma}|\leq |e(\bar\Gamma)|$, for any $\bar{\Gamma}_\ell$ and $\hat{\bar\Gamma}$. Since $\s{TV}(\pr_{\calH_0},\pr_{\calH_1})\leq\sqrt{\bE_{\calH_0}(\s{L}^2(\s{G}'))-1}$, conditions \eqref{eq:condDenseStatproof}--\eqref{eq:condDenseStat2proof} (with $n$ replaced by $n'$) imply that $\s{TV}(\pr_{\calH_0},\pr_{\calH_1}) = o(1)$, and therefore
\begin{align}
\s{over}(\hat{\bar\Gamma})\leq \sum_{\s{G}'}\pr_{\calH_0}(\s{G}')\sum_{\ell=1}^{L_n}\frac{|\bar{\Gamma}_\ell\cap\hat{\bar\Gamma}|}{\calN_{\bar\Gamma}}+o(|e(\bar\Gamma)|).
\end{align}
Next, we can write
\begin{align}
\s{over}(\hat{\bar\Gamma})&\leq \bE_{\calH_0}\pp{\sum_{\ell=1}^{L_n}\frac{|\bar{\Gamma}_\ell\cap\hat{\bar\Gamma}|}{\calN_{\bar\Gamma}}}+o(|e(\bar\Gamma)|)\\
& = \sum_{(i,j)\in\binom{[n]}{2}}^n\bE_{\calH_0}\pp{\Ind\pp{(i,j)\in\hat{\bar\Gamma}}\sum_{\ell=1}^{L_n}\frac{\Ind\pp{(i,j)\in\bar{\Gamma}_\ell}}{\calN_{\bar\Gamma}}}+o(|e(\bar\Gamma)|),
\end{align}
and we note that
\begin{align}
\pp{\sum_{\ell=1}^{L_n}\frac{\Ind\pp{i\in\bar{\Gamma}_\ell}}{\calN_{\bar\Gamma}}}^2 &= \sum_{\ell_1=1}^{L_n}\sum_{\ell_2=1}^{L_n}\frac{\Ind\pp{(i,j)\in\bar{\Gamma}_{\ell_1}}\Ind\pp{(i,j)\in\bar{\Gamma}_{\ell_2}}}{\calN^2_{\bar\Gamma}}\\
&\leq \s{Int}_{\s{G}'}(\bar\Gamma).
\end{align}
Thus
\begin{align}
\s{over}(\hat{\bar\Gamma})&\leq    \sum_{(i,j)\in\binom{[n]}{2}}\bE_{\calH_0}\pp{\Ind\pp{(i,j)\in\hat{\bar\Gamma}}\sqrt{\s{Int}_{\s{G}'}(\bar\Gamma)}}+o(|e(\bar\Gamma)|)\\
&\leq |e(\bar\Gamma)|\cdot \bE_{\calH_0}\pp{\sqrt{\s{Int}_{\s{G}'}(\bar\Gamma)}}+o(|e(\bar\Gamma)|)\\
&\leq |e(\bar\Gamma)|\cdot\sqrt{\bE_{\calH_0}\pp{\s{Int}_{\s{G}'}(\bar\Gamma)}}+o(|e(\bar\Gamma)|)\\
&\leq o(|e(\bar\Gamma)|),
\end{align}
where the third inequality follows from Jensen's inequality, and the last inequality is due to \eqref{eqn:ExpectedOverlap}. This concludes the proof.

\end{proof}

\subsection{Lower bound in sub-logarithmic density regime}

As it turns out, the proof technique in the previous subsection is not strong enough to capture the correct behavior in the sub-logarithmic density regime. Specifically, Theorem~\ref{thm:lowerp2weak} shows that, in this regime, recovery is impossible if \eqref{condDenseStat2weakApp} holds. However, as we show next, recovery is in fact always impossible in this regime.
\begin{theorem}\label{thm:subLogAlmostExact}
Assume that $q\in(0,1)$ is fixed and that $|v(\Gamma^\star\setminus\Gamma^{(\ell_{\s{LB}}),\star})|=o(n)$. If 
\begin{align}
\mu(\Gamma^\star\vert\Gamma^{(\ell_{\s{LB}}),\star})=o\p{\frac{\log|v(\Gamma^\star\setminus\Gamma^{(\ell_{\s{LB}}),\star})|}{\log \log|v(\Gamma^\star\setminus\Gamma^{(\ell_{\s{LB}}),\star})|}},\label{eq:condCopyLB}
\end{align}
then almost-exact recovery is impossible. 
\end{theorem}
To prove Theorem~\ref{thm:subLogAlmostExact}, we need a generalization of the subgraph expectation threshold that was already discussed in Appendix~\ref{app:GeneralizedExpecTThreshold0}. Specifically, the original subgraph expectation threshold in Theorem~\ref{th:Zadik} analyzes the threshold for the appearance of \emph{any possible} copy of the planted subgraph $\Gamma$ (in the complete graph $\calK_n$) within $\s{G}\sim\calG(n,q)$. For our purposes, however, some copies are precluded. Let us explain this in detail. 

In a nutshell, recall that we are in the scenario where the recovery problem is supplied with $\Gamma^{(\ell_{\s{LB}}),\star}$ and tasked with finding $\Gamma^\star$. Therefore, the admissible $\Gamma^\star$'s are only those that extend $\Gamma^{(\ell_{\s{LB}}),\star}$, namely, the copies contained in $\calM(\Gamma^{(\ell_{\s{LB}}),\star},\Gamma^\star)$. We resolve this by deriving a generalization of the subgraph expectation threshold that accounts for the appearance of a constrained set of subgraph copies. The details are provided in Appendix~\ref{app:GeneralizedExpecTThreshold0}, where we also establish the following key result.
\begin{lemma}
\label{lem:sublogProbLB}
Let $\Gamma=(\Gamma_n)$ be a sequence of graphs such that $\omega(1) \leq |v(\Gamma)|\leq n $, and
\eqref{eq:condCopyLB} holds. Then, for any fixed $q$, a sample from $\calG(n,q)$ contains an isomorphic copy of $\Gamma\in\calM(\Gamma^{(\ell_{\s{LB}}),\star},\Gamma^\star)$, with probability at least $1/2$.
\end{lemma}
Let us now prove Theorem~\ref{thm:subLogAlmostExact}, and then move forward to the proof of Lemma~\ref{lem:sublogProbLB}.

\begin{proof}[Proof of Theorem~\ref{thm:subLogAlmostExact}]
We prove that
\begin{align}
\inf_{\hat T}\sup_{\Gamma^\star\in\calS_\Gamma}
\pr_{\Gamma^\star}\left[d_{\s{H}}\big(\hat T(\s{G},\Gamma^{(\ell)}),\,T_\ell(\Gamma^\star)\big)\geq|e(T_\ell(\Gamma^\star))|\right]\geq\frac{1}{4},\label{eqn:ProbHig}
\end{align}
where we use $\pr_{\Gamma^\star}$ to emphasize that $\s{G}\sim\calG_{\Gamma^\star}(n,1,q)$, with $\Gamma^\star$ being the underlying planted subgraph. Proving \eqref{eqn:ProbHig} then readily implies that almost-exact recovery is impossible. 
%\begin{align}
%\pr_s\big[d(\widehat{E}(G),E(s))\ \ge\ m\big]\ \ge\ \frac{c_0}{2}.
%\end{align}
Recall that by Lemma~\ref{lem:sublogProbLB}, for all sufficiently large $n$, and since $|v(\Gamma^{(\ell_{\s{LB}}),\star})|\leq|v(\Gamma)|=o(n)$, under the condition in \eqref{eq:condCopyLB}, a typical sample from $\calG(n-|v(\Gamma^{(\ell_{\s{LB}}),\star})|,q)$ contains an isomorphic copy of $\Gamma^\star\in\calM(\Gamma^{(\ell_{\s{LB}}),\star},\Gamma^\star)$, with probability at least $1/2$. Denote this copy by $\Gamma'$, and note that by construction, both $\Gamma^\star$ and $\Gamma'$ are extensions of the same subgraph $\Gamma^{(\ell_{\s{LB}}),\star}$, and furthermore, $v(\Gamma^\star\setminus\Gamma^{(\ell_{\s{LB}}),\star})\cap v(\Gamma'\setminus\Gamma^{(\ell_{\s{LB}}),\star})=\emptyset$, namely, their tails are vertex-disjoint. Accordingly, define the event
\begin{align}
&\mathcal{F}_{\Gamma^\star}\triangleq\left\{\exists\;\Gamma'\in\calM(\Gamma^{(\ell_{\s{LB}}),\star},\Gamma^\star)\;\s{with}\;v(\Gamma^\star\setminus\Gamma^{(\ell_{\s{LB}}),\star})\cap v(\Gamma'\setminus\Gamma^{(\ell_{\s{LB}}),\star})=\emptyset ,e(\Gamma')\subseteq e(\s{G})\right\}.
\end{align}
Then $\pr_{\Gamma^\star}[\mathcal{F}_{\Gamma^\star}]\geq1/2$ provided that \eqref{eq:condCopyLB} holds. We will need the following simple observation, which we prove at the end.
\begin{lemma}\label{lem:symmetryTwoCop}
Let $\Gamma',\Gamma^\star$ be two vertex-disjoint copies of $\Gamma$. For any graph $\s{g}\in\{0,1\}^{\binom{n}{2}}$, such that all edges in both $\Gamma'$ and $\Gamma^\star$ appear in $\s{g}$, we have
\begin{align}
\pr_{\Gamma^\star}(\s{G}=\s{g})=\pr_{\Gamma'}(\s{G}=\s{g}).
\end{align}
\end{lemma}
We are now in a position to prove Theorem~\ref{thm:subLogAlmostExact}. Fix any estimator $\hat{T}$, and a subgraph $\Gamma^\star$. On the event $\calF_{\Gamma^\star}$, pick one disjoint extra copy $\Gamma'$ whose edges are all 1 in $\s{G}$, (existence guaranteed by $\calF_{\Gamma^\star}$). Consider the two distributions $\pr_{\Gamma^\star}$ and $\pr_{\Gamma'}$ and their equal mixture
\begin{align}
\mathbb{M}\triangleq\frac{1}{2} \pr_{\Gamma^\star} + \frac{1}{2} \pr_{\Gamma'} .
\end{align}
Define the error indicator under a randomly and uniformly chosen planted location $\widetilde{\Gamma}\in\{\Gamma^\star,\Gamma'\}$:
\begin{align}
\calE(\s{G},\widetilde{\Gamma})\triangleq\Ind\left\{d_{\s{H}}\big(\hat T(\s{G}),T_\ell(\widetilde{\Gamma})\big)\geq|e(T_\ell(\Gamma^\star))|\right\}.
\end{align}
By Lemma~\ref{lem:symmetryTwoCop}, for any realized graph $\s{g}\in\{0,1\}^{\binom{n}{2}}$ such that both $e(\Gamma^\star)$ and $e(\Gamma')$ are 1 in $\s{g}$,
\begin{align}
\pr_{\Gamma^\star}(\s{G}=\s{g})=\pr_{\Gamma'}(\s{G}=\s{g})\quad\Longrightarrow\quad
\mathbb{M}(\widetilde{\Gamma}=\Gamma^\star\vert \s{G}=\s{g})=\mathbb{M}(\widetilde{\Gamma}=\Gamma'\vert \s{G}=\s{g})=\frac{1}{2}.
\end{align}
Triangle inequality implies that
\begin{align}
d_{\s{H}}(\hat T(\s{G}),T_\ell(\Gamma^\star))+d_{\s{H}}(\hat T(\s{G}),T_\ell(\Gamma'))\geq d_{\s{H}}(T_\ell(\Gamma'),T_\ell(\Gamma^\star)) = 2|e(T_\ell(\Gamma^\star))|,
\end{align}
where we have used the fact that $v(T_\ell(\Gamma^\star))\cap v(T_\ell(\Gamma'))=\emptyset$. Therefore
\begin{align}
\max\{d_{\s{H}}(\hat T(\s{G}),T_\ell(\Gamma^\star)),d_{\s{H}}(\hat T(\s{G}),T_\ell(\Gamma'))\}\geq |e(T_\ell(\Gamma^\star))|.
\end{align}
Thus, whenever both copies $\Gamma^\star$ and $\Gamma'$ are present in $\s{g}$, we obtain
\begin{align}
    \calE(\s{g},\Gamma^\star)+\calE(\s{g},\Gamma')\geq1.\label{eqn:LLB1SumE}
\end{align}
Next, let $\calC$ be the event ``both $\Gamma^\star$ and $\Gamma'$ appear in $\s{G}$". Notice that $\calC$ coincides with $\calF_{\Gamma^\star}$ when $\s{G}\sim\pr_{\Gamma^\star}$, and with $\calF_{\Gamma'}$ when $\s{G}\sim\pr_{\Gamma'}$. Hence
\begin{align}
\mathbb{M}(\calC)=\frac{1}{2}\pr_{\Gamma^\star}(\calF_{\Gamma^\star})+\frac{1}{2}\pr_{\Gamma'}(\calF_{\Gamma'})\geq\frac{1}{2}.
\end{align}
Conditioning on $\s{G}$ and then averaging over $\widetilde{\Gamma}$ under $\mathbb{M}$,
\begin{align}
\mathbb{E}_{\mathbb{M}}\big[\calE(\s{G},\widetilde{\Gamma})\Ind\ppp{\calC}\big]
&= \mathbb{E}_{\mathbb{M}}\left[\mathbb{E}_{\mathbb{M}}\big[\calE(\s{G},\widetilde{\Gamma})\vert\s{G}\big]\ \Ind\ppp{\calC}\right]\\
&\geq
\mathbb{E}_{\mathbb{M}}\left[\frac{1}{2}\Ind\ppp{\calC}\right]\\
&= \frac{1}{2}\mathbb{M}(\mathcal{C})\geq\frac{1}{4},
\end{align}
where the inequality follows from the facts that on $\calC$ the posterior on $\{\Gamma^\star,\Gamma'\}$ is uniform, and \eqref{eqn:LLB1SumE} holds. Since $\calC\geq 0$, we can drop $\Ind\ppp{\calC}$ on the left to get
\begin{align}
\mathbb{E}_{\mathbb{M}}\big[\calE(\s{G},\widetilde{\Gamma})\big]\geq\frac{1}{4}.
\end{align}
But, we note that
\begin{align}
  \mathbb{E}_{\mathbb{M}}[\calE(\s{G},\widetilde{\Gamma})] &= \frac{1}{2}\pr_{\Gamma^\star}\pp{d_{\s{H}}\big(\hat T(\s{G}),T_\ell(\Gamma^\star)\big)\geq|e(T_\ell(\Gamma^\star))|}\nonumber\\
  &\qquad+\frac{1}{2}\pr_{\Gamma'}\pp{d_{\s{H}}\big(\hat T(\s{G}),T_\ell(\Gamma')\big)\geq|e(T_\ell(\Gamma^\star))|}.\label{eqn:sumProbab1/4}
\end{align}
Therefore, at least one of those two probabilities at the right-hand side of \eqref{eqn:sumProbab1/4} is at least $1/4$. This, in turn, implies that
\begin{align}
    &\sup_{\bar\Gamma\in\calS_\Gamma}
\pr_{\bar\Gamma}\left[d_{\s{H}}\big(\hat T(\s{G},\Gamma^{(\ell)}),\,T_\ell(\bar\Gamma)\big)\geq|e(T_\ell(\bar\Gamma))|\right]\nonumber\\
&\hspace{4cm}\geq\sup_{\bar\Gamma\in\{\Gamma^\star,\Gamma'\}}
\pr_{\bar\Gamma}\left[d_{\s{H}}\big(\hat T(\s{G},\Gamma^{(\ell)}),\,T_\ell(\bar\Gamma)\big)\geq|e(T_\ell(\bar\Gamma))|\right]\geq\frac{1}{4}.\label{eqn:almostExactEndProof}
\end{align}
Since $\hat{T}$ was arbitrary, the infimum over estimators of the left-hand side in \eqref{eqn:almostExactEndProof} is also at least $1/4$, which concludes that proof. 
\end{proof}
Finally, we prove Lemma~\ref{lem:symmetryTwoCop}.
\begin{proof}[Proof of Lemma~\ref{lem:symmetryTwoCop}]
Under $\pr_{\Gamma^\star}$, edges in $e(\Gamma^\star)$ are fixed to $1$, and edges in $\binom{[n]}{2}\setminus e(\Gamma^\star)$ are independent $\s{Bern}(q)$. Hence, for any $\s{g}\in\{0,1\}^{\binom{n}{2}}$,
\begin{align}
\pr_{\Gamma^\star}(\s{G}=\s{g})=
\begin{cases}
\prod_{e\in \binom{[n]}{2}\setminus e(\Gamma^\star)} q^{\s{g}_e}(1-q)^{1-\s{g}_e}, & \text{if } \s{g}_e=1\ \forall e\in e(\Gamma^\star),\\
0, & \s{otherwise.}
\end{cases}
\end{align}
An analogous formula holds for $\pr_{\Gamma'}$. Assume $\s{g}$ has $\s{g}_e=1$ for every $e\in e(\Gamma^\star)\cup e(\Gamma')$. Let
\begin{align}
\calR \triangleq \binom{[n]}{2}\setminus\p{e(\Gamma^\star)\cup e(\Gamma')}.
\end{align}
Then
\begin{align}
\pr_{\Gamma^\star}(\s{G}=\s{g})
&= \prod_{e\in e(\Gamma')} q^{\s{g}_e}(1-q)^{1-\s{g}_e}\prod_{e\in \calR} q^{\s{g}_e}(1-q)^{1-\s{g}_e}\\
&= q^{|e(\Gamma')|}\prod_{e\in \calR}q^{\s{g}_e}(1-q)^{1-\s{g}_e},
\end{align}
because $\s{g}_e=1$ for all $e\in e(\Gamma')$. Similarly,
\begin{align}
\pr_{\Gamma'}(\s{G}=\s{g})
= q^{|e(\Gamma^\star)|}\prod_{e\in \calR}q^{\s{g}_e}(1-q)^{1-\s{g}_e}.
\end{align}
Since $|e(\Gamma^\star)|=|e(\Gamma')|$, these expressions are equal:
\begin{align}
\pr_{\Gamma^\star}(\s{G}=\s{g})=q^{|e(\Gamma^\star)|}\prod_{e\in\calR} q^{\s{g}_e}(1-q)^{1-\s{g}_e}
=\pr_{\Gamma'}(\s{G}=\s{g}).
\end{align}
\end{proof}

\subsubsection{Proof of Lemma~\ref{lem:sublogProbLB}}\label{app:GeneralizedExpecTThreshold0}

In this subsection we prove Lemma~\ref{lem:sublogProbLB}. To that end, let us introduce a few definitions and notations. Fix a graph $\Gamma$ and a subgraph $\s{J} \subsetneq \Gamma$. Recall that $\calM(\s{J},\Gamma)$ is the set of copies of $\Gamma$ in $\calK_n$ that contain $\s{J}$. Let
\begin{align}
    \calN(\s{J},\Gamma,\s{G})\triangleq\sum_{\Gamma'\in\calM(\s{J},\Gamma)}\Ind\ppp{\Gamma'\in\s{G}},
\end{align}
which counts the number of copies of $\Gamma'\in\calM(\s{J},\Gamma)$ which appear in $\s{G}$. The critical probability of $\Gamma$ w.r.t. to $\calM$ is defined as
\begin{align}
    q_{c}(\s{J};\Gamma)&\triangleq\min\ppp{q\in[0,1] ~\bigg|~ \P_{\s{G}\sim \calG(n,q)}\pp{\calN(\s{J},\Gamma,\s{G})\geq 1\vert\s{J}\in\s{G}}\geq \frac{1}{2}}.\label{eqn:qcJGamma}
    \end{align}
Define the modified subgraph expectation threshold w.r.t. to $\calM$ as
\begin{align}
    \Tilde{q}_E(\s{J};\Gamma)&\triangleq\min\left\{q\in[0,1] ~\bigg|~ \E_{\s{G}\sim \calG(n,q)}\pp{\calN(\s{J},\s{H},\s{G})\vert\s{J}\in\s{G}}\geq \frac{\calN(\s{H},\Gamma)}{2},\;\forall\s{J}\subseteq\s{H}\subseteq \Gamma\right\},
\end{align} 
where only subgraphs $\s{J}\subseteq\s{H}\subseteq \Gamma$ with no isolated vertices are considered. We prove the following result.
\begin{theorem}\label{th:Zadik2} 
There exists a universal constant $C$ such that for any graph $\Gamma$, \begin{align}
    \Tilde{q}_E(\s{J};\Gamma)\leq q_{c}(\s{J};\Gamma)\leq C\cdot \Tilde{q}_E(\s{J};\Gamma)\cdot  \log|e(\Gamma\setminus\s{J})|.
\end{align}
\end{theorem}
Note that for $q\geq \Tilde{q}_E(\s{J};\Gamma)$, by Markov's inequality, for any $\s{J}\subseteq\s{H}\subseteq \Gamma$,
\begin{align}
    \frac{1}{2}&\leq \P_{\s{G}\sim \calG(n,q)}\pp{\calN(\s{J},\Gamma,\s{G})\geq 1\vert\s{J}\in\s{G}}\\
    &\leq \P_{\s{G}\sim \calG(n,q)}\pp{\calN(\s{J},\s{H},\s{G})\geq \calN(\s{H},\Gamma)\vert\s{J}\in\s{G}}\\
    &\leq \frac{\E_{\s{G}\sim \calG(n,q)}\pp{\calN(\s{J},\s{H},\s{G})\vert\s{J}\in\s{G}}}{\calN(\s{H},\Gamma)},
\end{align}
and thus, $\Tilde{q}_E(\s{J};\Gamma)\leq q_{c}(\s{J};\Gamma)$. Finally, note that we can rewrite
\begin{align}
    \Tilde{q}_E(\s{J};\Gamma)\triangleq\max\ppp{\p{\frac{\calN(\s{H},\Gamma)}{2|\calM(\s{J},\s{H})|}}^{1/|e(\s{H}\setminus\s{J})|}:\;\s{J}\subseteq\s{H}\subseteq \Gamma}.\label{eqn:EquivSpreaProHen}
\end{align}
We are now in a position to prove Theorem~\ref{th:Zadik2}.
\begin{proof}[Proof of Theorem~\ref{th:Zadik2}]
Let $\pi_\Gamma=\s{Unif}(\calM(\s{J},\Gamma))$ denotes the uniform measure over all copies of a fixed graph $\Gamma$ within $\calK_n$ that contain $\s{J}$, and let $\Gamma'\sim\pi_\Gamma$ be a random sample from this distribution. Now, consider a subgraph $\s{H}$ such that $\s{J}\subseteq\s{H}\subset\Gamma$. Let $\pi_{\s{H}}=\s{Unif}(\calM(\s{J},\s{H}))$ denotes the uniform measure over all copies of $\s{H}$ within $\calK_n$ that contain $\s{J}$, and let $\s{H}'\sim\pi_{\s{H}}$. For any fixed instances $\s{H}_0,\Gamma_0\subseteq\calK_n$, copies of $\s{H}$ and $\Gamma$, respectively, we can compute the inclusion probability in two equivalent ways:
\begin{align}
\pi_{\Gamma}(\s{H}_0 \subseteq \Gamma') = \pi_{\s{H}}(\s{H}' \subseteq \Gamma_0) = \frac{|\calM(\s{H},\Gamma)|}{|\calM(\s{J},\Gamma)|}=\frac{\calN(\s{H},\Gamma)}{|\calM(\s{J},\s{H})|}.
\label{eq:symmetry2}
\end{align}
Now, combining \eqref{eq:symmetry2} with \eqref{eqn:EquivSpreaProHen}, we obtain
\begin{align}
\pi_{\Gamma}(\s{H}_0 \subseteq \Gamma') =\frac{\calN(\s{H},\Gamma)}{|\calM(\s{J},\s{H})|} \leq 2\cdot \pp{\Tilde{q}_E(\s{J};\Gamma)}^{|e(\s{H}\setminus\s{J})|} \leq  \left( \frac{1}{2\Tilde{q}_E(\s{J};\Gamma)} \right)^{-|e(\s{H}\setminus\s{J})|}.
\end{align}
Since the bound holds uniformly over all subgraphs $\s{H}_0$, we conclude that $\pi_\Gamma$ is $\alpha$-spread with $\alpha = 1/[2\Tilde{q}_E(\s{J};\Gamma)]$. An application of Lemma~\ref{lem:Gspread} with $k = |e(\Gamma\setminus\s{J})|$ then completes the proof of Theorem~\ref{th:Zadik2}.

\end{proof}

Next, we prove the following result.
\begin{theorem}
\label{theorem:sublogProbALE}
Let $\Gamma=(\Gamma_n)$ and $\s{J} = (\s{J}_n)$ be sequences of graphs such that $\omega(1) \leq |v(\Gamma)|\leq n$, $\s{J}\subseteq\Gamma$, and
\begin{align}
  (1+\varepsilon)\mu(\Gamma_n\vert\s{J}_n)\cdot [\log\log|e(\Gamma_n\setminus\s{J})|+\log C]+\log|v(\Gamma_n)\setminus v(\s{J}_n)|\leq \log n,\label{eq:condCopyALE}
\end{align}
for some $\varepsilon>0$ and $C>0$. Then, for any fixed $q$, a sample from $\calG(n,q)$ contains an isomorphic copy of $\Gamma'\in\calM(\s{J},\Gamma)$, with probability at least $1/2$.
\end{theorem}
\begin{proof}[Proof of Theorem~\ref{theorem:sublogProbALE}]
    From the definition of $q_{c}(\s{J};\Gamma)$ in \eqref{eqn:qcJGamma}, our goal is to understand for which $\Gamma$'s we have $q_{c}(\s{J};\Gamma)\leq q$, for any fixed $q\in (0,1]$. By Theorem~\ref{th:Zadik2}, it is sufficient to show that $ \Tilde{q}_E(\s{J};\Gamma)\leq \frac{q}{C\log(|e(\Gamma\setminus\s{J})|)}$ for any fixed $q$, which holds if and only if
    \begin{align}
        \inf_{\s{J}\subseteq\s{H}\subseteq \Gamma} \frac{\E_{\s{G}\sim\calG(n,\tilde{q})}\pp{\calN(\s{J},\s{H},\s{G})\vert\s{J}\in\s{G}}}{\calN(\s{H},\Gamma)}\geq \frac{1}{2},\label{eq:condinfALE}
    \end{align}
    where $\tilde{q}\triangleq\frac{q}{C\log(|e(\Gamma\setminus\s{J})|)}$. We note that
    \begin{align}
    \E_{\s{G}\sim\calG(n,\tilde{q})}\pp{\calN(\s{J},\s{H},\s{G})\vert\s{J}\in\s{G}}=|\calM(\s{J},\s{H})|\cdot \tilde{q}^{|e(\s{H}\setminus\s{J})|},
    \end{align}
    Thus, \eqref{eq:condinfALE} holds if and only if for any $\s{J}\subseteq\s{H}\subseteq \Gamma$ we have
    \begin{align}
        \frac{\calN(\s{H},\Gamma)}{|\calM(\s{J},\s{H})|}\leq 2 \tilde{q}^{|e(\s{H}\setminus\s{J})|}=2\p{\frac{q}{C\log(|e(\Gamma\setminus\s{J})|)}}^{|e(\s{H}\setminus\s{J})|}.\label{eq:probRCALE}
    \end{align}
    As we have seen in the proof of Theorem~\ref{th:Zadik2}, the expression on the left-hand side of \eqref{eq:probRCALE} equals $\pi_{\Gamma}[\s{H}\subseteq \Gamma]$, where $\pi_\Gamma=\s{Unif}(\calM(\s{J},\Gamma))$ denotes the uniform measure over all copies of a fixed graph $\Gamma$ within $\calK_n$ that contain $\s{J}$. Thus, \eqref{eq:probRCALE} holds if,
    \begin{align}
        \log \pi_\Gamma[\s{H}\subseteq \Gamma]-\log2-|e(\s{H}\setminus\s{J})|\p{\log q -\log C-\log\log|e(\Gamma\setminus\s{J})|}\leq 0.
    \end{align}
    By the assumptions that $|v(\Gamma)|=\omega(1)$, and that $\Gamma$ has no isolated vertices, for any $\varepsilon$ and for sufficiently large $n$, the above holds if, 
    \begin{align}
        \log \pi_\Gamma[\s{H}\subseteq \Gamma]+|e(\s{H}\setminus\s{J})|\log C+(1+\varepsilon)|e(\s{H}\setminus\s{J})|\log\log|e(\Gamma\setminus\s{J})|\leq 0.\label{eq:almostFALE}
    \end{align}
    Finally, using Lemma~\ref{lem:probcalc2Alsmot}, the left-hand side of \eqref{eq:almostFALE} can be upper bounded by,
    \begin{align}
        &\log \pi_\Gamma[\s{H}\subseteq \Gamma]+|e(\s{H}\setminus\s{J})|\log C+(1+\varepsilon)|e(\s{H}\setminus\s{J})|\log\log|e(\Gamma\setminus\s{J})|\\
        &\leq |v(\s{H})\setminus v(\s{J})|\p{\log|v(\Gamma)\setminus v(\s{J})|-\log n}+|e(\s{H}\setminus\s{J})|\log C\nonumber\\
        &\qquad+(1+\varepsilon)|e(\s{H}\setminus\s{J})|\log\log|e(\Gamma\setminus\s{J})|\\
        &=|v(\s{H})\setminus v(\s{J})|\cdot \left(\log|v(\Gamma\setminus v(\s{J}))|+\frac{|e(\s{H}\setminus\s{J})|}{|v(\s{H})\setminus v(\s{J})|}\log C\right.\nonumber\\
        &\hspace{5cm}\left.+(1+\varepsilon)\frac{|e(\s{H}\setminus\s{J})|}{|v(\s{H})\setminus v(\s{J})|} \log \log|e(\Gamma\setminus\s{J})| -\log n\right)\\
        &\overset{(a)}{\leq}|v(\s{H})\setminus v(\s{J})|\cdot\pp{\log|v(\Gamma)\setminus v(\s{J})|+(1+\varepsilon)\mu(\Gamma\vert\s{J})\pp{\log\log|e(\Gamma\setminus\s{J})|+\log C} -\log n}\\
        &\leq 0,
    \end{align}
    where $(a)$ follows from the definition of the maximum subgraph relative density in \eqref{eqn:maxDensityRelative}, and the last inequality follows from \eqref{eq:condCopyALE}. This concludes the proof of the statement of Theorem~\ref{theorem:sublogProbALE}, and it is left to prove the following lemma.
\begin{lemma}\label{lem:probcalc2Alsmot} For any $\s{J}\subseteq\s{H}\subseteq \Gamma$,
\begin{align}
    \pi_{\Gamma}[\s{H}\subseteq \Gamma]\leq \p{\frac{|v(\Gamma)\setminus v(\s{J})|}{n}}^{|v(\s{H}\setminus\s{J})|},
\end{align}
where $\pi_\Gamma=\s{Unif}(\calM(\s{J},\Gamma))$.
\end{lemma}
\begin{proof}[Proof of Lemma~\ref{lem:probcalc2Alsmot}]
Let $m = |v(\Gamma)\setminus v(\s{J})|$ and $k = |v(\s{H}\setminus v(\s{J})|$. Fix a particular copy $\s{H}_0$ of $\s{H}$ in $\calK_n$ with vertex set $U = \{u_1,\dots,u_k\}$. Generate a uniformly random copy of $\Gamma$ by first picking its vertex set $S \subset [n]$ uniformly among all $m$-subsets (the internal labeling of $\Gamma$ can only lower the probability of containing $\s{H}_0$, so it suffices to control this step). If the sampled copy contains $\s{H}_0$, then necessarily $U \subseteq S$. Thus
\begin{align}
\pi_{\Gamma}[\s{H}_0 \subseteq \Gamma] \leq \pr[U \subseteq S]
= \prod_{i=1}^{k} \pr[u_i \in S \vert u_1,\dots,u_{i-1} \in S].
\end{align}
Conditioned on $u_1,\dots,u_{i-1} \in S$, there are $n-(i-1)$ remaining vertices and $m-(i-1)$ remaining slots in $S$, so
\begin{align}
\pr[u_i \in S \vert u_1,\dots,u_{i-1} \in S] 
= \frac{m-(i-1)}{n-(i-1)} \leq \frac{m}{n}. 
\end{align}
Multiplying these $k$ bounds gives
\begin{align}
\pi_{\Gamma}[\s{H} \subseteq \Gamma] \leq\pr[U \subseteq S] \leq \left(\frac{m}{n}\right)^k = \p{\frac{|v(\Gamma)\setminus v(\s{J})|}{n}}^{|v(\s{H}\setminus\s{J})|}.
\end{align}
\end{proof}
\end{proof}
We are now in a position to prove Lemma~\ref{lem:sublogProbLB}. To that end, we set $\s{J} = \Gamma^{(\ell_{\s{LB}})}$ and $\Gamma = \Gamma^\star$ in Theorem~\ref{theorem:sublogProbALE}. Then it is clear that condition~\eqref{eq:condCopyALE} is satisfied under~\eqref{eq:condCopyLB}, which completes the proof.

\end{document}